\documentclass[12pt,a4paper]{article}
\usepackage{amsmath,amsthm}
\RequirePackage[psamsfonts]{amssymb}

\usepackage{mathrsfs}

\usepackage{tikz} 
\usepackage{subcaption} 
\usepackage{listings}
\usepackage{aliascnt} 
\usepackage{mathtools}
\usepackage{dirtytalk} 
\usepackage{commath}
\usepackage{epigraph}
\usepackage{blindtext} 
\usepackage{multirow}
\usepackage{times}
\usepackage{graphicx} 
\usepackage{bbm} 
\usepackage{float}
\usepackage{lipsum}
\usepackage{latexsym}
\usepackage{listings}
\usepackage{verbatim}
\usepackage[titletoc,toc,title]{appendix} 
\usepackage{color,xcolor}

\usepackage{wrapfig}
\usepackage{lscape}
\usepackage{rotating}
\usepackage{epstopdf}
\usepackage[colorlinks=true,
linkcolor=blue,
urlcolor=blue,
citecolor=blue]{hyperref}
\usepackage[capitalize]{cleveref}
\usepackage{adjustbox}

\usepackage{fullpage}
\usepackage[numbers]{natbib}

\numberwithin{equation}{section}
\newtheorem{conj}{Conjecture}
\newtheorem{theorem}[conj]{Theorem}

\newtheorem{proposition}[conj]{Proposition}
\newtheorem{lemma}{Lemma}[section]
\newtheorem{ass}{Assumption}[section]
\theoremstyle{remark}
\newtheorem{remark}{Remark}

\theoremstyle{definition}
\newtheorem{definition}{Definition}[section]

\theoremstyle{definition}
\newtheorem{model}{Model}

   {\par\noindent\adjustbox{margin=1ex,margin=0ex \medskipamount}\bgroup\minipage\linewidth\verbatim}%
   {\endverbatim\endminipage\egroup}
   
\definecolor{mygreen}{rgb}{0,0.6,0}
\definecolor{mygray}{rgb}{0.5,0.5,0.5}
\definecolor{mymauve}{rgb}{0.58,0,0.82}

\lstdefinestyle{customc}{
  belowcaptionskip=1\baselineskip,
  breaklines=true,
  frame=L,
  xleftmargin=\parindent,
  language=R,
  showstringspaces=false,
  basicstyle=\footnotesize\ttclass,
  keywordstyle=\bfseries\color{green!40!black},
  commentstyle=\itshape\color{purple!40!black},
  identifierstyle=\color{blue},
  stringstyle=\color{orange},
}

\lstdefinestyle{customasm}{
  belowcaptionskip=1\baselineskip,
  frame=L,
  xleftmargin=\parindent,
  language=[x86masm]Assembler,
  basicstyle=\footnotesize\ttclass,
  commentstyle=\itshape\color{purple!40!black},
}

\makeatletter
\newcommand{\distas}[1]{\mathbin{\overset{#1}{\kern\z@\sim}}}%
\newsavebox{\mybox}\newsavebox{\mysim}
\newcommand{\distras}[1]{%
  \savebox{\mybox}{\hbox{\kern3pt$\scriptstyle#1$\kern3pt}}%
  \savebox{\mysim}{\hbox{$\sim$}}%
  \mathbin{\overset{#1}{\kern\z@\resizebox{\wd\mybox}{\ht\mysim}{$\sim$}}}%
}
\makeatother
\newcommand\independent{\protect\mathpalette{\protect\independenT}{\perp}}
\def\independenT#1#2{\mathrel{\rlap{$#1#2$}\mkern2mu{#1#2}}}

\newcommand*\oline[1]{%
  \vbox{%
    \hrule height 0.5pt
    \kern0.25ex
    \hbox{%
      \kern-0.1em
      \ifmmode#1\else\ensuremath{#1}\fi
      \kern-0.1em
    }
  }
}

\newcommand{\T}{\top}
\newcommand{\tr}{{\rm tr}}
\def\wh{\widehat}
\def\wt{\widetilde}
\def\Cov{{\rm Cov}}
\def\Var{{\rm Var}} 
\def\i{\infty}
\def\diag{{\rm diag}}
\def\rank{{\rm rank}}
\def\op{{\rm op}}

\def\oZ{\oline{Z}}
\def\oX{\oline{X}}
\def\oW{\oline{W}}
\def\distrto{\overset{{\rm d}}{\longrightarrow}}
 
\def\rI{{\rm I }}
\def\rII{{\rm II}}
\def\rIII{{\rm III}}

\def\cO{\mathcal{O}}
\def\cN{\mathcal{N}}
\def\cE{\mathcal{E}}

\def\b1{\mathbbm{1}}
\def\bI{\mathbf{I}}
\def\bbO{\mathbb{O}}

\newcommand{\EE}{\mathbb{E}}
\newcommand{\PP}{\mathbb{P}}
\newcommand{\RR}{\mathbb{R}}
\def\cE{\mathcal{E}}
\def\oSigma{\bar{\Sigma}}
\def\uSigma{\underline{\Sigma}}
\def\eps{\varepsilon}
\def\cH{\mathscr{H}}

\title{High-Dimensional Invariant Tests of Multivariate Normality Based on Radial Concentration}

\author{Xin Bing \thanks{Department of 
Statistical Sciences, University of Toronto, CA. E-mail: \texttt{xin.bing@utoronto.ca}}
\hspace{2cm} 
Derek Latremouille\thanks{Division of Biostatistics, University of Toronto, CA. E-mail: \texttt{dlatremouille7@gmail.com}}}

\begin{document}

\maketitle

\begin{abstract}
    While the problem of testing multivariate normality has received considerable attention in the classical low-dimensional setting where the sample size $n$ is much larger than the feature dimension $d$ of the data, there is presently a dearth of existing tests which are valid in the high-dimensional setting where $d$ is of comparable or larger order than $n$. This paper studies the hypothesis testing problem of determining whether $n$ i.i.d. samples are generated from a $d$-dimensional multivariate normal distribution, in settings where $d$ grows with $n$ at some rate under a broad regime. To this end, we propose a new class of computationally efficient tests which can be regarded as a high-dimensional adaptation of the classical radial approach to testing normality. A key member of this class is a range-type test which, under a very general rate of growth of $d$ with respect to $n$, is proven to achieve both type I error-control and consistency for three important classes of alternatives; namely, finite mixture model, non-Gaussian elliptical, and leptokurtic alternatives. Extensive simulation studies demonstrate the superiority of our test compared to existing methods, and two gene expression applications demonstrate the effectiveness of our procedure for detecting violations of multivariate normality which are of potentially practical significance. 
\end{abstract} 

\section{Introduction} \label{sec_intro}

 The multivariate normal model arguably constitutes the most important distributional family in statistics. Assuming normality of the observed data is ubiquitous, with use of this condition originating in classical statistical problems and continuing to have prominence in modern data analysis. 
 Consequently, the availability of tests and graphical diagnostics for assessing this assumption is crucial \cite{Thode_Book, ChenXia, Gnanadesikan, Henze2002}. However, while this problem has been extensively studied historically, resulting in the development of numerous procedures for testing this condition in the classical setting where the data dimension $d$ is small relative to the sample size $n$, there is a notable lack of analogous valid procedures in the high-dimensional setting, frequently characteristic of modern data analysis, where $d$ grows at some rate with $n$ \cite{ChenXia, Elliptical_GoF}. 
 In particular, as recently demonstrated in \cite{ChenXia}, classical normality tests typically exhibit type I error inflation as $d/n$ increases. The absence of valid multivariate normality tests in high-dimensional regimes has potentially serious practical consequences, as the performance of many procedures used to analyze high-dimensional data critically depends on the appropriateness of this assumption. For example, numerous methodologies developed in the high-dimensional setting for problems including one- and two-sample testing, gene-set and pathway analysis, and Gaussian graphical models for network inference are rendered invalid or exhibit marked degradation in performance when multivariate normality is violated \cite{Finegold, Hirose, Witten_GGM, Yang_Elliptical_Sphere, Chang, Ho, Maleki}.

 To address this issue, we seek to develop normality testing procedures which possess rigorous theoretical guarantees in the high-dimensional regime. Specifically, let $X_1,\ldots, X_n$ be \text{i.i.d.} copies of a random vector $X\in\RR^d$ with some unknown mean vector $\mu := \EE[X]$ and unknown covariance matrix $\Sigma := \Cov(X)$. We consider the problem of testing the null hypothesis 
\begin{equation}\label{def_null}
        \mathscr{H}_0:\quad X \sim \cN_d(\mu, \Sigma),
\end{equation} 
against general alternatives, in settings where the dimension $d = d(n)$ increases at some rate with the sample size $n \to \i$. This testing problem encompasses a broad range of high-dimensional methodologies, including Gaussian graphical models and network inference \cite{XiaCai, van_de_Geer, SILGGM}, one- and two-sample testing \cite{LedoitWolf2002, Nishiyama2013, Fisher2010}, covariance matrix estimation \cite{DunsonPati, Fisher2011_CovEstim}, MANOVA \cite{SriMANOVA, Schott2007}, gene-set and pathway analysis \cite{Ho, Maleki}, sparse linear regression \cite{JM2018}, 
discriminant analysis \cite{FanBook}, variable selection \cite{knockoff, forward_regress}, causal inference \cite{Klaassen2023, HD_Causal}, semi-supervised learning \cite{Couillet2018}, and canonical correlation analysis \cite{Canonical_Correlation_HD}. These approaches are widely applied in contexts where the number of variables may far exceed the number of samples, such as microarray and RNA-Seq gene expression studies, proteomics, finance, and brain imaging \cite{FanBook, XiaCai, Fritsch, SILGGM}. 

\subsection{Pre-Existing Literature} \label{preExisting_work}

The problem of testing multivariate normality has received extensive attention historically, particularly in the classical low-dimensional regime, making it difficult to provide a comprehensive review of the literature. Instead, we refer the interested reader to classical references such as \cite{Henze2002, Thode_Book, Shapiro_Review}, as well as the recent reviews provided in \cite{ChenXia, EbnerHenze2020}.

For our purposes, it suffices to note that the principal approach to developing tests for $\cH_0$ involves the use of test statistics and associated graphical diagnostics which encapsulate certain geometric properties of the data. As discussed in \cref{rem_invariance}, this is in part related to the fact that testing $\cH_0$ is classically treated as an invariant testing problem with respect to arbitrary non-singular affine transformation of $X$ \cite{Henze2002, Cox}. Let $\oX \in \mathbb{R}^d$ and $\wh \Sigma \in \mathbb{R}^{d \times d}$ respectively denote the sample mean and sample covariance matrix, and $\wh\Sigma^{1/2}$ denote the symmetric square root of $\wh\Sigma$.
The \textit{scaled radii},
\begin{align}\label{scaled_radii}
{R^*_i} :=   \| \wh \Sigma^{-1/2} (X_i - \oX)\|_2 \qquad \text{for } i \in [n]:= \{1,\ldots, n\},    
\end{align} 
arguably most commonly constitute the basis for classical tests and graphical diagnostics for multivariate normality. Well-known examples include Mardia's kurtosis test \cite{Mardia}, the uniformly most powerful test against outlier-type alternatives \cite{Ferguson, Wilks}, and multivariate adaptations of the Cram\'{e}r--von Mises \cite{Koizol82, Kolmogorov_Smirnov}, Shapiro-Wilk \cite{Thode_Book}, Kolmogorov-Smirnov \cite{Kolmogorov_Smirnov}, and Anderson-Darling tests \cite{AndersonDarling}, among others \cite{Henze2002, Thode_Book, Brenner}. Moreover, to complement these formal tests, a well-known diagnostic technique for assessing multivariate normality is based on quantile plots of the scaled radii \cite{Small, Brenner}. Beyond the aforementioned invariance criterion, the theoretical basis for testing $\cH_0$ using the scaled radii derives from the fact that when $n \gg d$, the joint behavior of $R^*_1, \ldots, R^*_n$ under $\cH_0$ is approximately equivalent to that of the Euclidean norms of i.i.d. realizations $Z_1,\ldots, Z_n$ of $\cN_d(0_d, \bI_d)$ \cite{Gnanadesikan}, and yields tests and diagnostics with desirable power properties against a broad array of pertinent alternatives \cite{Henze2002, Ferguson, Brenner, Barnett_Book}. 

 However, despite the abundance of existing tests of $\cH_0$, few of them, if any, are suitable for modern high-dimensional data \cite{ChenXia, Elliptical_GoF}. In particular, \cite{ChenXia} demonstrates that conventional tests of multivariate normality possess critical limitations beyond the low-dimensional setting, with existing methods exhibiting marked inflation of type I error or power-loss as $d / n$ increases.
    This can typically be attributed to difficulties in estimating the parameters $\mu$ and $\Sigma$, which, in the classical setting, are effectively estimated by $\oX$ and $\wh \Sigma$, respectively. For example, as discussed in \cref{rem_invariance}, any affine-invariant test of $\cH_0$, such as those based on the scaled radii \eqref{scaled_radii}, is not well-defined when $d \geq n$ due to the singularity of $\wh \Sigma$, and this issue cannot be remedied via a generalized inverse \cite{PiresBranco} or some regularized estimator \cite{ChenXia}. 

  Only recently, \cite{ChenXia} developed the first test of $\cH_0$ with a type I error-control guarantee in a regime where $d$ may increase at some rate with $n \to \i$, and demonstrated its superiority over classical tests. However, their test possesses several critical limitations. First, the type I error theory in \cite{ChenXia} is only established in the regime $d = o(\sqrt{n})$, and under  restrictive conditions on $\Sigma$. 
  Both requirements appear to be essential to the validity of their test, as the simulation studies in \cref{sec_sims_moderate}, \cref{sec_sims_high}, and \cite{Elliptical_GoF} show that their test exhibits significant type I error inflation when either $d/n$ increases or the required conditions on $\Sigma$ are violated. 
  Second, no theoretical guarantees regarding power are established for their proposed test. 
  Lastly, the test of \cite{ChenXia} is computationally intensive when either $d$ or $n$ is large (see \cref{rem_comp}) and does not satisfy fundamental invariance properties for the problem of testing $\cH_0$ (see \cref{rem_invariance}).


Finally, it is worth mentioning another recent work \cite{Elliptical_GoF}, which proposes a goodness-of-fit test for centered elliptical distributions and derives a type I error-control guarantee in a high-dimensional regime with $n \asymp d$. However, for  testing $\cH_0$ specifically, this implies that their test has low or trivial power against a general class of non-Gaussian elliptical alternatives. 


\subsection{Our Contributions} \label{sec_our_contributions}

    We summarize our main contributions in this section. 

    \subsubsection{A High-Dimensional Radial-Based Approach for Testing Multivariate Normality.} As discussed in \cref{preExisting_work}, existing tests of $\cH_0$, such as those based on the scaled radii $R^*_1,\ldots,R^*_n$ and the recent test of \cite{ChenXia}, are typically plagued by issues involving estimation of $\Sigma$ or its inverse as the dimension $d$ increases. Our first contribution is to introduce a new class of tests for $\cH_0$ which effectively adapts the classical radial-based approach so as to benefit from increasing dimensionality. Specifically, we show that, as long as $d$ exceeds a logarithmic factor of $n$, the {\em radii}
     \begin{equation}\label{def_Ri}
        R_i := \| X_i - \oX \|_2,\quad \text{for }i \in [n],
      \end{equation}
     after suitable normalization, behave similarly to $\{\|Z_i\|_2\}_{i\in[n]}$ under $\cH_0$ and mild conditions on $\Sigma$.  
     Thus, instead of using the scaled radii $R^*_i$ as is done classically, our proposed test statistics are based on the normalized radii $R_i$, thereby circumventing the challenging task of estimating $\Sigma^{-1}$.  

    To obtain the scale-type parameter used to normalize the radii, we first note that, for reasons discussed in \cref{app_Radii_SqRadii}, the test statistics are based on $R_i$ instead of $R_i^2$.    However, while a closed-form expression for the variance of $R_i^2$ can readily be derived, the variance of $R_i$ is analytically intractable in general.
    Thus, we adopt the \textit{dispersion index} of $\|X - \mu\|_2^2$, 
    \begin{equation}\label{def_Delta}
         \Delta_2 :=  \frac{\Var(\| X - \mu \|_2^2)}{\EE \| X - \mu \|_2^2 } =  \frac{\Var(\| X - \mu \|_2^2)}{\tr(\Sigma)},
    \end{equation} 
     to quantify the variance of $R_i$. Indeed, as proposed and established in a companion working paper, the dispersion index parameter $\Delta_2$ serves as a sharp \textit{generic} proxy for $\Var(\|X - \mu\|_2)$, in the sense that $\Var(\|X - \mu\|_2) \leq \Delta_2$, with equality achieved for some random vector $X$, and only requires the existence of the fourth moments of the coordinates of $X$. Moreover, this companion work establishes that $\Delta_2$ determines the asymptotic variance of the limiting distribution of $\| X - \mu \|_2$ as $d \to \i$ for a relatively general class of random vectors. 
    Due to both the marginal kurtosis and the dependence structure of the multivariate normal model, under $\cH_0$ the dispersion index $\Delta_2$ takes the form
     \begin{equation}\label{def_Delta_null}
        \Delta \equiv \frac{2\tr(\Sigma^2)}{\tr(\Sigma)},
    \end{equation}
    and the variance proxy for the radii $R_i$ is simply $(n - 1)\Delta / n$. When the distribution of $X$ is non-Gaussian, the dispersion index $\Delta_2$ in \eqref{def_Delta} will not generally be of the form in \eqref{def_Delta_null} due to discrepancies arising from either non-Gaussian kurtosis or dependence properties. 
    In \cref{sec_method}, we propose an estimator $\wh \Delta$ of $\Delta$ which can be computed efficiently and is shown to be ratio-consistent with a fast rate of convergence under both $\cH_0$ and a broad class of alternatives. 
    
    Equipped with the estimator $\wh\Delta$, let $R_{(1)} \le \cdots \le R_{(n)}$ be the ordered radii. Given a pair of \textit{symmetric} empirical quantiles $1\le \underline{q}<\bar q \le n$, and some  deterministic normalizing sequences $a_n, b_n \ge 0$, our proposed class of test statistics is of the form 
    \begin{equation}\label{test_stat_class}
         2a_n 
        ~ \wh \Delta^{-1/2}\bigl(R_{(\bar q)} - R_{(\underline q)}\bigr) - 2a_n   b_n .
    \end{equation} 
In comparison to the estimator $\wh \Delta$ of the dispersion index parameter in \eqref{def_Delta_null}, the quantile contrast $R_{(\bar q)} - R_{(\underline q)}$ is a distinct measure of dispersion of the distribution of the radii under both $\cH_0$ and non-Gaussian alternatives. Thus, test statistics of the class \eqref{test_stat_class} are characterized by a ratio of two scale-type estimators of the radial distribution. 
Test statistics defined by a ratio of two scale estimators, with one such estimator constructed via some symmetric contrast of order statistics, have an extensive history in the classical problem of testing univariate normality \cite{Pearson_Normal, Pearson_Stephens, Thode_Book}.  The effectiveness of such test statistics derives from their tractability, invariance properties, and the fact that the relationship between the two scale estimators exhibits under- or over-dispersion under a broad class of alternatives. 
Our proposed tests naturally inherit these advantages and further, as discussed in \cref{rem_invariance}, satisfy an important form of invariance for the problem of testing $\cH_0$. 

Using symmetric quantile contrasts in \eqref{test_stat_class} also eliminates a nuisance centering parameter in the marginal asymptotic distribution of $R_i$, which itself can be difficult to estimate at an adequate rate in high dimensions. The choice of quantiles $\bar q$ and $\underline{q}$ determines the normalizing sequences $a_n$ and $b_n$ in \eqref{test_stat_class}. For reasons discussed in \cref{sec_method} and \cref{app_general_theory_AppendixA}, in this paper we primarily consider the range-type specification of \eqref{test_stat_class}, corresponding to $\bar q = n$ and $\underline{q} = 1$, as well as the interquartile range specification, corresponding to $\bar q = \lfloor  3n/4 \rfloor$ and $\underline{q} = \lfloor  n/4 \rfloor$, with their respective normalizing constants provided in \cref{sec_method}. Other choices of quantile contrasts are discussed in \cref{rem_general_class} and \cref{app_general_theory_AppendixA}. Based on the established limiting distributions of both test statistics, we specify their rejection regions in \cref{sec_method}, and our proposed testing procedure combines these two tests using a Bonferroni correction.

Finally, we remark that although a similar idea of replacing Mahalanobis distance with Euclidean distance has been employed in high-dimensional two-sample testing problems \cite{ChenQin, BaiSarandasa}, where Hotelling’s $\text{T}^2$ statistic is traditionally used when $n > d$, applying it to our setting requires analyzing the joint distribution of $R_1, \ldots, R_n$, determining the appropriate normalization, deriving the properties of its estimator under $\cH_0$ and non-Gaussian alternatives, and constructing the final form of the test statistics along with the corresponding rejection regions. These steps, which also provide the basis for the accompanying graphical diagnostics for assessing multivariate normality and detecting outliers in high dimensions (see \cref{sec_real_data} for more detail), constitute our main methodological contributions.

    \subsubsection{Asymptotic Type I Error Control of the Proposed Testing Procedure}
   Our second contribution is to prove that the proposed test achieves valid asymptotic type I error control. To this end, we establish the asymptotic distribution of the proposed range-type test statistic under $\cH_0$ in a general high-dimensional regime where $n, d \to \infty$ (see, also, \cref{app_general_theory_AppendixA} for the limiting distribution of the interquartile range test statistic as well as more general test statistics of the class \eqref{test_stat_class} and their combination). \Cref{thm_range_limit} in \cref{sec_theory_null} presents a Gaussian approximation result that bounds the Kolmogorov distance between the proposed range test statistic and the normalized range of $n$ \text{i.i.d.} standard Gaussian random variables when $\Delta$ is known. A key quantity in our analysis is the \textit{effective rank}, $\rho_1(\Sigma^2)$, of the covariance matrix $\Sigma$ (see \cref{def_rhos}). Our results in \cref{thm_range_limit} are non-asymptotic in nature and are valid provided that $\rho_1(\Sigma^2) \gg \log^5(nd)$, which is a mild condition also ensuring that the Kolmogorov distance sufficiently small (see \cref{rem_cond_thm1}). When, for example, $\Sigma$ has bounded eigenvalues, the condition reduces to $d \gg \log^5 n$, thereby allowing $d$ to increase with $n$ at a particularly general rate. In conjunction with the ratio-consistency of the proposed estimator $\wh\Delta$ of $\Delta$ established in \cref{prop_Delta_Null}, \cref{thm_range_stat_limit} derives an analogous Gaussian approximation result for the proposed test statistic. As discussed in \cref{sec_theory_null}, this directly yields the proposed rejection region outlined in \cref{sec_method}, based on which \cref{thm_range_typeI} establishes theoretical type I error control of our test. To the best of our knowledge, our procedure is the first test of $\cH_0$ with theoretical control of the type I error when $d$ may grow proportionately to, or substantially exceed, $n$.
    Moreover, as discussed in \cref{rem_cond_thm1}, the condition on $\rho_1(\Sigma^2)$ is mild and encompasses standard assumptions on $\Sigma$ commonly used in high-dimensional analyses. 

    

    \subsubsection{Consistency of the Proposed Testing Procedure for a Broad Class of Alternatives} \label{sec_contributions_consistency} In addition to type I error control, \cite{Henze2002, EbnerHenze2020} argues that any test of multivariate normality ought to be accompanied by theory identifying relevant alternatives for which it is consistent. While general alternatives are of interest, recent theoretical developments on power in high-dimensional testing \cite{Kock} suggest that even when universal testing consistency is achievable for a problem in the low-dimensional setting, it may not be attainable for its high-dimensional analog. This indicates the importance developing tests prioritizing specific types of alternatives which are of greatest practical interest. 
    Our third contribution is thus to establish consistency of our proposed test in \cref{sec_theory_alter} against a broad class of alternatives which are of both theoretical and methodological relevance, including finite mixture, non-Gaussian elliptical, and leptokurtic alternatives. 
    
    The power analysis is based on the fact that, as $n,d \to \i$, the radii \eqref{def_Ri} have a distinct relationship with the null dispersion index $\Delta$ \eqref{def_Delta_null} under general non-Gaussian alternatives compared to that under $\cH_0$. Thus, a key step in proving consistency involves establishing the ratio-consistency of the estimator $\wh \Delta$ of $\Delta$ under the aforementioned alternatives, which is the content of \cref{prop_Delta_Alternatives}. Similar to the type I error theory, our consistency results in \cref{thm_loc_mix_sG,thm_cov_mix_sG,thm_ellip,thm_kurtosis} of \cref{sec_theory_alter} are derived in a general high-dimensional regime, only requiring that the relevant effective rank quantity exceeds a logarithmic factor of $nd$ in conjunction with a signal-to-noise ratio (SNR) condition -- both of which are specific to the type of alternative. Our theory shows that the SNR condition becomes less stringent as the effective rank of the relevant covariance matrix increases, hence revealing a \textit{blessing of dimensionality} effect for the power of our test. To the best of our knowledge, establishing consistency theory for tests of $\cH_0$ in high dimensions remains an open problem, and our work provides the first such results for important classes of alternatives. 

    In \cref{sec_sims_moderate,sec_sims_high}, we corroborate our theoretical guarantees through extensive simulation studies, which demonstrate that the proposed test achieves superior type I error control and power compared to leading existing tests of $\cH_0$, including the recently proposed high-dimensional normality test of \cite{ChenXia}, across both moderate- and high-dimensional settings. To further illustrate the practical utility of our test and its associated graphical diagnostics for high-dimensional data analysis, we analyze two gene expression datasets in \cref{sec_real_data} and \cref{app_real_data} as case studies, demonstrating how our methodology can effectively detect critical departures from $\cH_0$ in practice. While both datasets have been previously analyzed, our results offer new insights. 

    This paper is organized as follows. The proposed testing procedure is described in \cref{sec_method}. 
    Valid type I error control is established in \cref{sec_theory_null}, while consistency against pertinent classes of alternatives is developed in \cref{sec_theory_alter}. \cref{sec_sims_moderate} conducts simulation analyses to corroborate the type I error and power theory of \cref{sec_theory}, and provide comparison of our test's performance to that of leading pre-existing tests of $\cH_0$ (see, also, \cref{app_add_simulation} for additional simulation studies). \cref{sec_real_data} and \cref{app_real_data} demonstrate the use of our procedure in applied problems via the analysis of two gene expression datasets. All proofs are deferred to the Appendix. 

\paragraph{Notation.} For any distribution function $F:\RR\to [0,1]$ and any $\alpha \in [0,1]$, its $\alpha$-quantile is $F^{-1}(\alpha) := \inf\{x\in \RR: F(x) \ge \alpha\}$.
For any positive integer $d$, we write $[d] := \{1,\ldots, d\}$. For any number $x\ge 0$, we write its integer part as $\lfloor x\rfloor$.
For any vector $v$, $\|v\|_p$ denotes its $\ell_p$ norm for $1 \le p \le \i$ and $v_{(q)}$ denotes its $q^\text{th}$ smallest value for each $q\in [d]$. The vector $0_d$ (and $1_d$) contains entries all equal to 0 (and $1$). We use $\bI_d$ to denote the $d\times d$ identity matrix  and $\bbO^d$ to denote the set of $d \times d$ orthogonal matrices. For any $A \in \mathbb{R}^{m \times k}$, $\| A \|_\op$ denotes its operator norm and $\| A \|_F$ denotes its Frobenius norm. For the spectral decomposition $\Sigma = U \Lambda U^{\T}$ of any symmetric, positive semi-definite $\Sigma \in \mathbb{R}^{d \times d}$, the diagonal entries of $\Lambda$, $\lambda_1 \geq \cdots \ge \lambda_d \geq 0$, represent the eigenvalues in non-increasing order, and $\Sigma^{1/2}$ denotes its symmetric square root. 
For any two sequences $a_n$ and $b_n$, we write $a_n\lesssim b_n$ if there exists some constant $C$ such that $a_n \le Cb_n$. 
The notation $a_n\asymp b_n$ corresponds to $a_n \lesssim b_n$ and $b_n \lesssim a_n$. Additionally, $a_n = \omega(b_n)$ denotes the property that $a_n / b_n\to \i$ as $n \to \i$. Analogously, for a sequence of random variables $Y_n$, $Y_n = \omega_{\mathbb{P}}(a_n)$ means that $Y_n / a_n \to \i$ in probability as $n \to \i$. For any $a, b \in \mathbb{R}$, we write $a\wedge b = \min\{a, b\}$ and $a\vee b =\max\{a,b\}$. 
Finally, we use $c,c',C,C'$ to denote positive finite absolute constants that, unless otherwise indicated, can change from line to line.

\section{Methodology}\label{sec_method}


Building upon the motivation for the proposed class of test statistics \eqref{test_stat_class} in \cref{sec_our_contributions}, in this section we provide additional detail and discussion pertinent to the implementation of our test of $\cH_0$. 
Recall from \cref{sec_our_contributions}  that the proposed class of test statistics is of the form
\begin{equation}\label{eq_general_T}
    2a_n~ 
    \wh \Delta^{-1/2}\bigl(R_{(\bar q)} - R_{(\underline q)}\bigr) - 2 a_n b_n,
\end{equation}
where $R_1,\ldots, R_n$ are the radii as defined in \eqref{def_Ri}, $\underline{q}<\bar q$ are a pair of symmetric empirical quantiles, and  $a_n, b_n \ge 0$ are some deterministic normalizing sequences. 
The final procedure we propose for testing $\cH_0$ is a composite test involving two particular test statistics of the form \eqref{eq_general_T}. Due to the distinct choice of empirical quantiles associated with the two test statistics, their normalizing constants $a_n, b_n \geq 0$ and rejection regions are determined separately. However, the estimator $\wh \Delta$ of $\Delta$, proposed below, is applicable to any test statistic of the form \eqref{eq_general_T}.

\paragraph*{Estimation of the Dispersion Index.} As discussed in \cref{sec_our_contributions}, the dispersion index parameter $\Delta$ serves as a proxy of the variance of $R_i$ under $\cH_0$, and is a critical component of our test statistics. Thus, we seek an estimator of $\Delta$ which is ratio-consistent with a fast rate of convergence under both the null and a broad class of alternatives. 
To this end, we propose to estimate $\Delta$ by
\begin{equation}\label{def_Delta_hat}
    \wh \Delta =   \frac{2\widehat{\tr(\Sigma^2)}}{\tr(\wh \Sigma_{\text{D}})}, 
\end{equation}
where $\wh \Sigma_{\text{D}}$ is the sample covariance matrix $\wh \Sigma =(n - 1)^{-1} \sum_{i = 1}^n (X_i - \overline{X}) (X_i - \overline{X})^\T$ when $n > d$, or the centered Gramian matrix $\wh \Sigma_{\text{G}} \in \mathbb{R}^{n \times n}$ with  its $(i,j)$th entry equal to $(n - 1)^{-1}(X_i - \overline{X})^\T (X_j - \overline{X})$ for $i,j\in [n]$ when $n \leq d$, and
\begin{align}\label{def_tr_Sigma_hat}
     \widehat{\tr(\Sigma^2)} :=  {n - 1 \over n (n - 2) (n - 3)} \Bigl( (n - 1) (n - 2) \tr(\wh \Sigma_{\text{D}}^2) + \tr^2(\wh \Sigma_{\text{D}}) - {n \over n - 1} \sum_{i = 1}^n R_i^4 \Bigr)
\end{align}
is equivalent to a standard estimator of $\tr(\Sigma^2)$ developed in \cite{Chen2010}. 
As discussed in \cref{rem_comp}, the form we use in \eqref{def_tr_Sigma_hat} is 
designed to further accelerate computation when $n \ll d$ or $n \gg d$, owing to the specification of $\wh \Sigma_{\text{D}}$. In \cref{sec_theory}, the estimator $\wh \Delta$ is shown to be ratio-consistent under both $\cH_0$ (\cref{prop_Delta_Null}) and a broad class of alternatives (\cref{prop_Delta_Alternatives}).


Below, we propose two specific test statistics of the form \eqref{eq_general_T}, corresponding to two choices of $\bar q$ and $\underline q$, together with their associated sequences of normalizing constants $a_n, b_n \ge 0$.

\paragraph*{The Range-Type Test.}
    Motivated by the range test for univariate normality \cite{Pearson_Normal, Pearson_Stephens} as well as the uniformly most powerful test of $\cH_0$ against outlier-type alternatives in the classical $n > d$ setting \cite{Ferguson, Wilks, Barnett_Book}, our first proposed test statistic of the class \eqref{eq_general_T} is constructed using the range of the radii:
    \begin{equation}\label{def_T_range}
        T := 2 a_n  ~ \wh\Delta^{-1/2} \bigl(R_{(n)} -  R_{(1)}\bigr) -  2 a_n b_n,
    \end{equation}
    where the normalizing constants are specified by
    \begin{equation}\label{def_an_bn}
        a_n := \sqrt{2 \log n},\qquad b_n := a_n - {(\log \log n + \log 4 \pi) \over 2a_n}.
    \end{equation}
    See the discussion after \cref{thm_range_limit} for a detailed explanation of the choice of normalizing constants. The distributional properties of $T$ under $\cH_0$ are established in \cref{sec_theory_null}. Based on the theory of \cref{sec_theory}, we propose to reject the null hypothesis at level $\alpha \in (0,1)$  if and only if 
    \begin{equation}\label{reject_range}
        T \notin \left( \wh F^{-1}_{M,n}(\alpha/2), ~  \wh F^{-1}_{M,n}(1-\alpha/2)\right),
    \end{equation}
    where, for a specified number of Monte Carlo replications $M \in \mathbb{Z}^+$ and percentile $\alpha_0 \in (0,1)$, $\wh F^{-1}_{M,n}(\alpha_0)$ denotes the $\alpha_0$-quantile of the empirical distribution $\wh F_{M,n}$ of $M$ i.i.d. realizations of 
    \begin{equation}\label{def_Un}
        U_n  ~ =  ~ a_n\bigl(S_{(n)} - S_{(1)}\bigr) -  2 a_nb_n, 
    \end{equation} 
    with $S_{(1)}\le \cdots \le S_{(n)}$ denoting the order statistics of $S \sim \cN_n(0_n, \bI_n)$. \cref{thm_range_typeI} informs the determination of a suitable number of Monte Carlo replications $M$, and simulation analysis further indicates that $M \sim 10,000$ replications is sufficient.

  \paragraph*{The Interquartile-Range-Type Test.}  
As discussed further in \cref{app_sec_IQR}, for some alternatives, incorporating information near the central quantiles of the distribution of the radii can further increase efficiency in comparison to solely using the range-type test. To this end, we propose a second type of test statistic based on the interquartile range (IQR) of the radii, corresponding to $\bar q = \lfloor 3n/4 \rfloor$ and $\underline{q} = \lfloor n/4 \rfloor$ in \eqref{eq_general_T}, given by
    \begin{equation}\label{def_T_IR}
        T_* :=  2 \sqrt{n} \left[  \wh\Delta^{-1/2} \left( R_{(\lfloor  3n/4 \rfloor)} - R_{(\lfloor n/4 \rfloor)} \right) - \ \Phi^{-1}(3/4) \right].
    \end{equation}
    We write  $\phi(x)=\exp(-x^2/2) / \sqrt{2\pi}$ for the standard normal density, and denote its cumulative distribution and quantile functions by $\Phi(x)$ and $\Phi^{-1}(x)$, respectively. Due to the choice of $\bar q$ and $\underline{q}$, we have different normalizing constants $a_n = \sqrt{n}$ and $b_n = \Phi^{-1}(3/4)$ in \eqref{def_T_IR}. 
    We show in \cref{thm_IQR} of \cref{app_sec_IQR} that the distribution of $T_*$ has an explicit and known normal limit under $\cH_0$, based on which we propose the following rejection region. 
    For a given level $\alpha \in (0,1)$, the test based on $T_*$ rejects the null hypothesis if and only if 
    \begin{equation}\label{reject_IQR}
        T_* \notin \left(
            \sigma_* \Phi^{-1}(\alpha/2), ~ \sigma_* \Phi^{-1}(1-\alpha/2)
        \right),\quad \text{with }~  \sigma_*^{-1} :=  2 \phi(\Phi^{-1}(3/4)).
    \end{equation}

    Our final testing procedure is a composite test that combines the range-based statistic $T$ and the IQR-based statistic $T_*$, together with a Bonferroni correction. This proposal is motivated by the differential sensitivity of the two types of tests for various classes of alternatives, as discussed in \cref{rem_composite_test} and corroborated by the simulation results in \cref{app_Sim_RangeIQR}. Finally, we remark that the derivation of the asymptotic properties of $T$ and $T_*$ can be used to establish the basis of the complementary graphical diagnostics presented in \cref{sec_real_data}.

    \begin{remark}[Computational Complexity]\label{rem_comp}
        The computation involved in the proposed testing procedure consists of calculating $T$ and $T_*$ and performing the Monte Carlo approximation for the rejection region \eqref{reject_range}, thus yielding an overall complexity of $\cO(nd(n\wedge d) + Mn)$. 
        Conversely, the testing procedure of \cite{ChenXia} has computational complexity at least of order $\cO(M' d (n^2 + d^2))$, making it computationally intensive when either $n$ or $d$ is large. Here, $M'$ is the number of replications for simulating $n$ independent $d$-dimensional Gaussian random vectors required by their algorithm.
    \end{remark}

     \begin{remark}[Invariance Properties]\label{rem_invariance} 
    As mentioned in \cref{sec_intro}, testing $\cH_0$ is classically stipulated to be an invariant problem with respect to the group of non-singular affine transformations of $X$ in the absence of any problem-specific justification, due to the closure of the multivariate normal model under this transformation group \cite{Henze2002, EbnerHenze2020}.
        However, \cite{Cox} argues that there is sometimes a practical basis for restricting the required invariance to a narrower subgroup of transformations. This consideration is particularly important, and even necessary, in the high-dimensional setting when $d \geq n$ because any affine-invariant test statistic for $\cH_0$ in the classical setting is a function of  $\{(X_i - \overline{X})^\T \widehat{\Sigma}^{-1} (X_j - \overline{X}) \}_{i,j \in [n]}$ \cite{Henze2002}, and the singularity of $\wh{\Sigma}$ when $d \geq n$ cannot be resolved via a generalized inverse \cite{PiresBranco} or regularized estimator \cite{ChenXia}. Furthermore, methodology for high dimensional data, including those based on multivariate normality, often critically depend on assumptions (see, for example, \cref{rem_cond_thm1}) which preclude $\Sigma$ from being of low effective rank \cite{Sri13, Ma15, Heinavaara, Hub_GGM, CondNum_GGM, Info_GGM, GGM_JMLR, Chang, LatentVar_GGM, DAG_GGM}.   
         Thus, in lieu of the general affine transformation group, a suitable form of invariance for testing $\cH_0$ in high-dimensional settings is defined with respect to its similarity transformation subgroup \cite{ChenQin, Chen2010, BaiSarandasa} 
         \begin{equation} \label{invariance_cond}
            X \mapsto \sigma V X + w, \qquad \text{for any } \sigma > 0, V \in \bbO^d, w \in \RR^d,
        \end{equation}
        which our proposed test satisfies, whereas the principal existing test of $\cH_0$ in the high-dimensional setting proposed by \cite{ChenXia} does not. 
    \end{remark}

    \begin{remark} (Other Choices of Quantile Contrasts) \label{rem_general_class} While this paper focuses primarily on the range-type and IQR-type tests, both our method and theory accommodate more general extreme quasi-range and central quantile range based test statistics of the class \eqref{test_stat_class}, and combinations thereof. Asymptotic distributional properties of $T_*$ as well as more general test statistics of the class \eqref{test_stat_class} and their weighted combination are established in \cref{app_general_theory_AppendixA}. The advantages associated with different quantile contrast specifications are briefly discussed in \cref{rem_composite_test,rem_general_quantile_stats}, but deserve extensive investigation, which is thus left for future research. 
    \end{remark}

\section{Theoretical Guarantees}\label{sec_theory}

We provide theoretical guarantees for the proposed test in this section. Results pertaining to control of the type I error are stated in \cref{sec_theory_null}, while those characterizing consistency against different classes of alternatives are presented in \cref{sec_theory_alter}. Our theory uses the following notions of the \textit{effective rank} of a matrix. 
\begin{definition}[Effective Ranks]\label{def_rhos}
    For any non-null positive semi-definite matrix $A \in \RR^{d\times d}$, 
    define two notions of its effective rank via
   $
    \rho_1(A) := {\tr(A) / \|A\|_\op}$ and 
    $\rho_2(A) := {\tr^2(A) / \tr(A^2)}$.  
\end{definition}
\noindent The theoretical guarantees for the proposed test are based on $\rho_r(\Sigma^s)$ for $r,s \in \{1,2\}$, where we note that each of these quantities constitutes a bona fide effective rank of $\Sigma$ in the sense that, for each choice of $r, s \in \{1, 2\}$, $\rho_r(\Sigma^s)$ is invariant under the transformation group \eqref{invariance_cond} and satisfies $1\le \rho_r(\Sigma^s) \le \rank(\Sigma)$.  Relations between these effective ranks are formally established in \cref{lem_ranks} from which, for future reference, we remark that 
\begin{equation}\label{cond_rhos}
    \rho_1(\Sigma^2) \le \rho_1(\Sigma) \le \rho_2(\Sigma) \le \rho_1^2(\Sigma), \qquad  \rho_1(\Sigma^2)  \le  \rho_2(\Sigma^2) \le \rho_2(\Sigma).
\end{equation} 
Our theory for both the type I error control and power of the proposed test is developed in an asymptotic regime where the effective rank of some relevant covariance matrix exceeds a logarithmic factor of $nd$. As detailed in \cref{rem_cond_thm1}, this asymptotic regime encompasses a wide range of high-dimensional settings. 


\subsection{Type I Error-Control for the Proposed Testing Procedure}\label{sec_theory_null}

In this section, we establish type I error control for the testing procedure proposed in \cref{sec_method}. 
To study the range-type test statistic $T$ in \eqref{def_T_range}, we first derive its limiting distribution under $\cH_0$ when the true parameter $\Delta$ is used in place of $\wh \Delta$; that is, we first consider
$
    \bar T :=  2 a_n  \Delta^{-1/2} ( R_{(n)} - R_{(1)}) -2 a_n b_n,
$
with $a_n$ and $b_n$ given by \eqref{def_an_bn}. Recall $U_n$ from \eqref{def_Un}.

\begin{theorem}\label{thm_range_limit}
    Grant the null $\cH_0$ and suppose that 
    \begin{equation}\label{cond_Sigma}
        \rho_1(\Sigma^2) = \omega\left( \log^5(n d) \right) ,\quad \text{as }n\to \i.
    \end{equation}
    Then, there exists some absolute constant $C>0$ such that 
    \begin{equation}\label{rate_CLT_range}
        \sup_{t\in \RR} \left|
        \PP\left(  \bar T \le t\right) - \PP \left(
            U_n \le t
        \right) 
        \right| ~ \le~  C \left(
        \log^5(nd)  \over \rho_1(\Sigma^2)
        \right)^{1/4}  + C \left( \log n\over n\right).
    \end{equation}
    Moreover, we have 
    $
        \bar T \distrto  E+ E',
    $
    where $E$ and $E'$ are independent random variables with the same ditribution
     $   \PP\left\{
        E \le x
        \right\} = \exp(-\exp(-x)), \
    $
    for any $-\i < x < \i$.
\end{theorem} 
 
The bound in \eqref{rate_CLT_range} controls the Kolmogorov distance between $\bar T$ and $U_n$. This result is non-asymptotic in nature, for which condition \eqref{cond_Sigma} can be stated as $\rho_1(\Sigma^2)  \ge C \log^5(n d)$ for some sufficiently large constant $C>0$. \cref{thm_range_limit} further states that $\bar T$ converges in distribution to the convolution of two independent standard Gumbel distributions, which follows from \eqref{rate_CLT_range} and classical results on the extreme order statistics of \text{i.i.d.} standard normal random variables \cite{David}. It is for this reason that the normalizing sequences $a_n$ and $b_n$ are specified via \eqref{def_an_bn}. As discussed following \cref{thm_range_stat_limit}, while the Gaussian approximation \eqref{rate_CLT_range} holds for more general $a_n, b_n \geq 0$ and our approach to constructing the rejection region \eqref{reject_range} could in principle be accomplished without requiring the Gumbel-based limiting distribution, the specification in \eqref{def_an_bn} ensures that the consistency results of \cref{sec_theory_alter} for non-Gaussian alternatives can be derived by establishing that $T \to \pm \i$ in probability as $n \to \i$.

\begin{remark}[The Effective Rank Condition of \cref{thm_range_limit}] \label{rem_cond_thm1} As discussed in \cref{rem_invariance}, effective rank conditions on $\Sigma$ often play a critical role in high-dimensional methodology based on multivariate normality. Condition \eqref{cond_Sigma} places a restriction on the effective rank $\rho_1(\Sigma^2)$, which is also needed to ensure the right hand side of \eqref{rate_CLT_range} vanishes as $n\to \i$. Since $\rho_1(\Sigma^2) \le d$, it implies  $d \gg \log^{5}(n)$, a mild condition which is frequently characteristic of high-dimensional data.
When $n \lesssim d^{\gamma}$ for some $\gamma \in (0, \i)$, \eqref{cond_Sigma} simplifies to $\rho_1(\Sigma^2) 
= \omega(\log^5 d)$. A special case of this is the bounded eigenvalue condition
\begin{equation}\label{cond_ident}
    0< c \le \lambda_d(\Sigma) \le \lambda_1(\Sigma) \le C < \i,
\end{equation}
which is widely assumed by high-dimensional methodology in conjunction with $\cH_0$ \cite{XiaCai, Klaassen2023, van_de_Geer, JM2018, FanBook, Couillet2018} 
and is also one of the conditions adopted by the recent high-dimensional normality test of \cite{ChenXia}. It is worth noting that under the stronger condition \eqref{cond_ident}, we have $\rho_1(\Sigma^2) \asymp d$ and the order $[\rho_1(\Sigma^2)]^{-1/4}$ in \eqref{rate_CLT_range} can be improved to a $d^{-1/2}$ rate of convergence, up to logarithmic factors, using the Gaussian approximation results of \cite{lopes2022central, kuchibhotla2020high}. 
More generally, since $\rho_1(\Sigma^2) = \omega(\log^5 d)$ is equivalent to $\tr(\Sigma^4) = o(\tr^2(\Sigma^2) )$ up to a logarithmic factor, \eqref{cond_Sigma} also encompasses other conditions commonly assumed alongside $\cH_0$ in high-dimensional inference problems, such as $\tr(\Sigma^k) \asymp d, \ $ for $k \in [4]$ (see, for instance, \cite{LedoitWolf2002, Fisher2010, Nishiyama2013, SriMANOVA, DunsonPati, Fisher2011_CovEstim, Schott2007}).
\end{remark}

In view of \cref{thm_range_limit}, deriving the asymptotic distribution of $T$ requires establishing a suitable rate of convergence of $\Delta / \wh  \Delta$ to unity. This is the content of the following proposition.

\begin{proposition}\label{prop_Delta_Null}
    Under $\cH_0$, one has  that for all $t > 0$, 
    \[\PP\left\{
    \Bigl|
    \sqrt{\Delta \over  \wh \Delta} - 1 
    \Bigr| ~ \ge~  {t\over  n} +   {t\over \sqrt{n \rho_2(\Sigma^2)}}
    \right\}  = \cO\left(
    {1\over  t^2}
    \right).
    \]
\end{proposition}

The ratio consistency of $\wh \Delta$ depends on the effective rank $\rho_2(\Sigma^2)$ which, according to the relation in \eqref{cond_rhos}, is bounded from below by $\rho_1(\Sigma^2)$. It is evident that the rate of convergence in \cref{prop_Delta_Null} improves as $\rho_2(\Sigma^2)$ increases, ranging from $\cO_\PP(n^{-1/2})$ to $\cO_\PP(n^{-1})$. By combining \cref{thm_range_limit} and \cref{prop_Delta_Null}, we establish a Gaussian approximation for our proposed range-type statistic $T$ in the following theorem.
\begin{theorem}\label{thm_range_stat_limit} 
   Grant condition \eqref{cond_Sigma} of \cref{thm_range_limit}. Under $\cH_0$,  one has
    \begin{align*}
        \sup_{t\in \RR} \left|
        \PP(T \le t) - \PP(U_n \le t)
        \right|  
        & = \cO\left(
            \left(\log^5(n d)  \over \rho_1(\Sigma^2)
            \right)^{1/4}  +  {\log n\over \sqrt n}  
        \right).
    \end{align*} 
    Furthermore, under $\cH_0$, we have 
    $
    T \distrto E+E', 
    $ 
    where $E$ and $E'$ are specified in \cref{thm_range_limit}. 
\end{theorem}


The second part of \cref{thm_range_stat_limit} provides the {\em explicit} limiting distribution of $T$ under $\cH_0$, based on which an asymptotically valid rejection region could be derived. However, as discussed after  \cref{thm_range_limit}, this limiting distribution originates from 
$
    U_n = a_n \left(S_{(n)} - S_{(1)}\right) - 2a_nb_n \distrto E+E'
$
based on extreme value theory. Since the rate of this convergence is prohibitively slow \cite{Hall79,David}, constructing rejection regions based on quantiles of the distribution of $E + E'$ yields inadequate finite-sample performance. We therefore propose to construct the rejection region based on the first Gaussian approximation of \cref{thm_range_stat_limit}. Due to the analytical intractability of the exact distribution of $U_n$ \cite{David}, we employ a Monte Carlo sampling algorithm to approximate it, resulting in a rejection region of the form specified by \eqref{reject_range}. The asymptotic validity of this rejection region for controlling the type I error is established in the following theorem.  




\begin{theorem}\label{thm_range_typeI} 
    Grant condition \eqref{cond_Sigma} of \cref{thm_range_limit}. Under $\cH_0$, for any given level $\alpha \in (0,1)$,
    \[
        \left|\PP \left(T \notin \left( \wh F^{-1}_{M,n}(\alpha/2), ~  \wh F^{-1}_{M,n}(1-\alpha/2)\right) \right)  -  \alpha\right| = \cO \left( 
            \left(\log^5(n d)  \over \rho_1(\Sigma^2)
            \right)^{1/4}\!\!\!  +  {\log n\over \sqrt n}  + {1 \over \sqrt  M}
        \right).
    \] 
\end{theorem}

\cref{thm_range_typeI} shows that the proposed range-type test of \eqref{reject_range} maintains control of the type I error as $\rho_1(\Sigma^2) = \omega(\log^5(nd))$ and $M\to \i$ (see \cref{rem_cond_thm1}). 
It also informs specification of the number of replications $M$ for the Monte Carlo approximation used to construct the rejection region \eqref{reject_range}. 



\subsection{Power \& Consistency of the Proposed Testing Procedure}\label{sec_theory_alter}

As discussed in \cref{sec_our_contributions}, identifying classes of alternatives for which a proposed test of $\cH_0$ is consistent is an important task, particularly in the high-dimensional setting. In this section, we establish the consistency of our test for classes of finite-mixture, non-Gaussian elliptical, and leptokurtic alternatives in \cref{sec_power_mix,sec_power_ellip,sec_power_lep}, respectively. These types of alternatives comprise a broad class of nonparametric alternatives and constitute particularly problematic departures from the assumed normal model in various methodological contexts. 

\subsubsection{Consistency for Finite Mixture Alternatives}\label{sec_power_mix}

We first examine the power of our test under finite mixture models, a widely used class of distributions which constitutes a critical type of departure from multivariate normality \cite{Hirose, Thode_Book, Barnett_Book, Fritsch}. As detailed below, we consider the mixture components to be sub-Gaussian distributions, with Gaussian mixture models serving as a specific instance. Our results can be extended to mixture components satisfying milder moment conditions; see \cref{rem_general_finite_alternatives} and \cref{app_sec_power_BS} for further detail. For simplicity, we assume that the \say{standardized marginals} have equal fourth moments within each mixture component. This assumption is not essential and can be relaxed. 
\begin{model}[Sub-Gaussian Mixture Alternatives]\label{model_subG}
    Suppose there exists some integer $K\ge 2$, mean vectors $\mu_1, \ldots, \mu_K \in \RR^d$, and covariance matrices $\Sigma_1, \ldots, \Sigma_K \in \RR^{d\times d}$ such that   
    $(X_i \mid C_i = k)   = \mu_k + \Sigma_k^{1/2} Z_i
    $ and $\PP(C_i = k) = \pi_k$ for all $k\in [K]$ and $i\in [n]$,
    where $Z_1,\ldots, Z_n$ are independent isotropic sub-Gaussian random vectors in $\RR^d$ with bounded sub-Gaussian constants and independent entries. For each $k \in [K]$, assume $\EE(Z^4_{ij} \mid C_i = k) = \kappa_k$, for all $i \in [n]$ and $j \in [d]$, and $\pi_k \ge c$ for some universal constant $c>0$.
\end{model}

Under \cref{model_subG}, the {\em unconditional} covariance matrix of $X$ satisfies $\Sigma = \sum_{k < m}^K \pi_k \pi_m (\mu_k - \mu_m) (\mu_k - \mu_m)^{\T} + \sum_{k = 1}^K \pi_k \Sigma_k$. In the following, we establish consistency of our proposed test  for two types of alternatives under \cref{model_subG}; namely, location-mixtures and covariance-type mixtures. The former is first examined, where it is only assumed that there is discernible location-based separation between at least two of the $K$ mixture components. While \cref{thm_loc_mix_sG} below assumes equal covariance matrices for simplicity, its proof in \cref{app_proof_thm_loc_mix_sG} is based on a more general setting permitting distinct component-specific covariance matrices $\Sigma_k \neq \Sigma_{\ell}$, with $k, \ell \in [K]$.

\begin{theorem}[Location-Type Mixtures]\label{thm_loc_mix_sG} 
    Under \cref{model_subG} with $\Sigma_* := \Sigma_k$ for all $k\in [K]$, suppose that
        $\rho_1(\Sigma_*^2) \ge \log n$ and
        \begin{equation}\label{cond_mean_sep_sG}
            \max_{k,\ell \in [K]} { \|\mu_k-\mu_\ell\|_2^2 \over \tr(\Sigma_*)}  = \omega\left({1  / \sqrt{\rho_2(\Sigma_*)}} \right).
        \end{equation}
    Then, for arbitrary choice of level $\alpha \in (0,1)$,
    $
        \lim_{n\to \i} \PP \left(\cH_0 \text{ is rejected}\right) =  1.
    $
\end{theorem}

 

Analogous to the type I error analysis of \cref{sec_theory_null}, \cref{thm_loc_mix_sG} imposes requirements on the effective ranks of the {\em conditional} covariance matrices. In addition to the condition $\rho_1(\Sigma_*^2)\ge \log n$, \eqref{cond_mean_sep_sG} introduces a location-based separation requirement for at least two mixture components. In particular, 
the left hand side of \eqref{cond_mean_sep_sG} can be regarded as a \textit{signal-to-noise ratio} (SNR) based on the maximum location separation, $\max_{k,\ell}\|\mu_k-\mu_\ell\|_2^2$, relative to the total within-class variance, $\tr(\Sigma_*)$. Notably, the SNR requirement in \eqref{cond_mean_sep_sG} becomes less stringent as the effective rank $\rho_2(\Sigma_*)$ increases, thereby exhibiting a {\em blessing of dimensionality} phenomenon. To see this, suppose $K=2$  and $\Sigma_*$ satisfies \eqref{cond_ident} in lieu of $\Sigma$. In this case, $\tr(\Sigma_*) \asymp \rho_1(\Sigma_*^2) \asymp \rho_2(\Sigma_*) \asymp d$, implying that $\rho_1(\Sigma_*^2)\ge \log(n)$ is satisfied provided that $d \geq C \log(n)$ for some constant $C > 0$. Further assuming that $\mu_{j2} = \mu_{j1} + \delta_n$ for each $j\in [d]$ and some deterministic sequence $\delta_n>0$, condition \eqref{cond_mean_sep_sG} reduces to the marginal separation constraint 
$
    \delta^2_n = \omega(
         d^{-1/2}
    ).
$


 



When none of the mixture components are distinguishable based solely on location, consistency of our test can still be ensured if at least two mixture components are sufficiently distinct with respect to their total variances. This is the content of the next theorem, stated for the special case where all mixture components share the same mean vector, but proven in \cref{app_proof_thm_cov_mix_sG} for the more general setting with potentially distinct mean vectors $\mu_k \neq \mu_{\ell}$, such that $k, \ell \in [K]$.

\begin{theorem}[Covariance-Type Mixtures]\label{thm_cov_mix_sG}
 Under \cref{model_subG} with $\mu_1 = \cdots = \mu_K$, suppose that 
 \begin{equation}\label{cond_cov_sep_sG}
    {\max_{k,\ell \in [K]} ~  \tr(\Sigma_k -\Sigma_\ell)  \over \max_{k\in [K]} ~  \tr(\Sigma_k)}   ~ = ~  \omega\left(  { \sqrt{\log(n) /  \min_{k \in [K]}\rho_2(\Sigma_k)}}  \right).
  \end{equation}
 Then, for arbitrary choice of level $\alpha \in (0,1)$, 
    $
        \lim_{n\to \i}\PP \left(\cH_0 \text{ is rejected}\right) = 1.
    $
\end{theorem}

Analogous to \eqref{cond_mean_sep_sG}, condition \eqref{cond_cov_sep_sG} is a signal-to-noise ratio condition based on the maximum relative separation of the mixture components with respect to total variance. Note that it implies  $\rho_2(\Sigma_k) = \omega(\log n)$, for all $k\in [K]$. Since \eqref{cond_cov_sep_sG}  becomes milder as $\min_k\rho_2(\Sigma_k)$ increases, we observe a similar blessing of dimensionality phenomenon for covariance-type mixtures.   


\begin{remark}[Consistency for General Finite Mixtures]\label{rem_general_finite_alternatives}
The consistency of our test for the mixture alternatives of \cref{model_subG} established in \cref{thm_loc_mix_sG,thm_cov_mix_sG} can be extended to more general finite mixtures of distributions satisfying milder moment and dependence conditions, at the expense of stronger regularity conditions compared to that of \eqref{cond_mean_sep_sG} and \eqref{cond_cov_sep_sG}. See \cref{app_sec_power_BS} for details. 
\end{remark}


    



    \subsubsection{Consistency for Non-Gaussian Elliptical Alternatives}\label{sec_power_ellip}
    
    To characterize the power of our test for additional types of critical departures from normality, such as those exhibiting diverse heavy-tailed and tail dependence structure, 
    we now establish the consistency of our test for an important class of non-Gaussian elliptical alternatives. 
    Methods based on multivariate normality often exhibit substantial performance degradation under such alternative departures \cite{Yang_Elliptical_Sphere, Finegold, Fritsch, Ho, Brenner}. These alternatives, formally defined below, are generated via scale mixtures of multivariate normal distributions. 

    \begin{model}[Heavy-Tailed Elliptical Alternatives]\label{elliptical_stochastic_rep}
    Suppose there exists some mean vector $\mu \in \mathbb{R}^d$ and some positive semi-definite $\Sigma_* \in \RR^{d \times d}$ such that  
    $
        X_i = \mu + \eps_i~ \Sigma_*^{1/2} Z_i \ 
    $
    for each $i \in [n]$,
    where $Z_1,\ldots, Z_n$ are i.i.d. from $\cN_d(0_d, \bI_d)$ and $\eps_1,\ldots, \eps_n$ are \text{i.i.d.} mixing scale random variables in $\RR_{\geq 0}$, drawn from some non-degenerate distribution $F_\eps$ with $\EE[\eps_i^4] \leq C < \infty$, for some universal constant $C > 0$. Further, suppose that $\{Z_i\}_{i \in [n]}$ and $\{\eps\}_{i \in [n]}$ are independent.
    
    \end{model}


\cref{elliptical_stochastic_rep} constitutes a general class of nonparametric alternatives, well-known instances of which include the multivariate \textit{t}-distribution, the heavy-tailed multivariate power-exponential distributions such as the multivariate Laplace distribution, multivariate symmetric stable distributions, the semi-symmetric multivariate inverse Gaussian, countably infinite Gaussian scale-mixtures, and scale mixtures of these distributions. 
Notice that under \cref{elliptical_stochastic_rep}, we have $\Sigma = \EE[\eps^2] \Sigma_*$, implying that the effective ranks of $\Sigma$ are equal to those of $\Sigma_*$. 

The following theorem states that the proposed test is consistent against the heavy-tailed elliptical alternatives of \cref{elliptical_stochastic_rep}, provided that the distribution of the random mixing scales does not degenerate to a Dirac measure too rapidly.  Let $\eps_{(1)} \le \cdots \le \eps_{(n)}$ be the ordered random variables.

\begin{theorem}[Elliptical Alternatives]\label{thm_ellip}
Under \cref{elliptical_stochastic_rep}, 
suppose that $\rho_1(\Sigma_*^2) \ge \log n$ and 
\begin{equation}\label{cond_W_gap}
     {\eps_{(n)}-  \eps_{(1)} \over \eps_{(n)}} = \omega_{\mathbb{P}}\left(   \sqrt{{\log(n) / \rho_2(\Sigma_*)}}  \right),\quad \text{as }n\to \i. 
\end{equation}
Then, for arbitrary choice of level $\alpha \in (0,1)$,
$
    \lim_{n\to \i}\PP \left(\cH_0 \text{ is rejected}\right) = 1.
$
\end{theorem}


Condition \eqref{cond_W_gap} directly parallels the signal-to-noise ratio constraints of the finite-mixture alternatives considered in \eqref{cond_mean_sep_sG} and \eqref{cond_cov_sep_sG}, revealing an analogous blessing-of-dimensionality effect. Under mild conditions on the order of growth of $\eps_{(n)}$ such as $\eps_{(n)} = o_\PP(\sqrt{\rho_2(\Sigma_*)/\log n})$, \eqref{cond_W_gap} allows the distribution of the mixing scale random variable $\eps_i$ to approach a Dirac measure, thereby permitting the distribution of $X$ to converge to that of the null model. For example, consider the variance-inflation continuous scale-mixture alternative where we take $F_\eps$ to be $\text{Unif}(\sigma_0, \sigma_0 + \delta_n)$ in \cref{elliptical_stochastic_rep} for some $\sigma_0 > 0$, a positive sequence $\delta_n = o(1)$, and $\Sigma_*$ satisfying $\tr(\Sigma_*^{k}) \asymp d$ for $k \in [3]$ so that $\rho_2(\Sigma_*) \asymp d$.
Noting that  $\eps_{(n)} = \cO(1)$ with probability one and the fact that $(\eps_{(n)} -  \eps_{(1)})$ has the same distribution as $\delta_n (\bar \eps_{(n)} -  \bar \eps_{(1)})$, where $\bar \eps_1,\ldots, \bar \eps_n$ are \text{i.i.d.} from $\text{Unif}(0, 1)$, condition \eqref{cond_W_gap} simplifies to
$
\delta_n =  \omega( \sqrt{\log(n) / d} )$,
which, as $d/\log(n)$ increases, allows the distribution of $X$ to converge to a multivariate Gaussian distribution more rapidly.

\subsubsection{Consistency for Leptokurtic Alternatives}\label{sec_power_lep}

Having developed consistency theory pertaining to alternative classes constituting departures from $\cH_0$ which are essentially multivariate in nature, we now consider the asymptotic power of our test for alternatives whose discrepancy with the null model arise at the univariate level. In particular, we consider the class of univariate-based departures associated with excess kurtosis marginals. 

\begin{model}[Leptokurtic Alternatives] \label{kurtosis_stochastic_rep}
    Suppose there exists a vector $\mu \in \RR^d$, an orthogonal matrix $U \in \bbO^{d}$, and a diagonal matrix $\Lambda^{1/2} \in \RR^{d \times d}$ with non-negative diagonal entries such that 
    $
    X_i = \mu + U \Lambda^{1/2} Z_i,
    $
    where $\{Z_i \}_{i \in [n]}$ are \text{i.i.d.} random vectors consisting of independent sub-Gaussian entries with bounded sub-Gaussian constants. Furthermore, suppose these entries satisfy   $\EE[Z_{ij}] = 0$, $\EE[Z^2_{ij}] = 1$, and $\EE[Z^4_{ij}] = 3 + \delta_n$, for some deterministic sequence $\delta_n > 0$.  
\end{model}

 The quantity $\delta_n$ in \cref{kurtosis_stochastic_rep} is known as the {\em excess kurtosis} of each $Z_{ij}$. The multivariate normal distribution is a limiting case of \cref{kurtosis_stochastic_rep} when $\delta_n \to 0$. The following result establishes the consistency of our testing procedure for \cref{kurtosis_stochastic_rep}. 

\begin{theorem}[Leptokurtic Alternatives]\label{thm_kurtosis}
Under \cref{kurtosis_stochastic_rep}, assume $\rho_1(\Sigma^2) = \omega( \log^5(n d))$ and 
   $\delta_n  = \omega(1/ \log n)$ as $n \to \i$.
Then, for any $\alpha \in (0,1)$,
$
    \lim_{n\to \i}\PP \left(\cH_0 \text{ is rejected}\right) = 1.
$
\end{theorem}


    In contrast to the consistency theory developed for alternatives of the preceding sections, the condition on the signal $\delta_n$ 
     does not depend on an effective rank of $\Sigma$. As elucidated by the proofs of both \cref{thm_range_limit,thm_kurtosis}, this is a consequence of the fact that the alternatives of \cref{kurtosis_stochastic_rep} possess critical dependence and moment properties which are nearly identical to that of $\cH_0$. This yields analogous concentration properties in the asymptotic distribution of the radii as well as the rate of convergence of $\wh \Delta / \Delta$ to unity (see \cref{prop_Delta_Alternatives}). However, due to the presence of non-zero excess kurtosis $\delta_n$, a normalization discrepancy arises from using $\wh \Delta$ in place of 
    $
    \Delta_2 = (2 + \delta_n) \tr(\Sigma^2) / \tr(\Sigma),
    $ 
    which is the correct dispersion index parameter \eqref{def_Delta} under \cref{kurtosis_stochastic_rep}, as opposed to $\Delta$ in \eqref{def_Delta_null}. Since the normalization discrepancy $\Delta_2 / \Delta = (2 + \delta_n) / 2$ is relative in nature and is only compared with the normalizing sequences $a_n$ and $b_n$ in \eqref{def_T_range}, it exhibits a dimension-free effect in perturbing the limiting distribution of $T$ under the null. 

Given the current absence of consistency theory for testing $\cH_0$ in the high-dimensional setting, the suitability of $\delta_n=\omega(1/\log n)$ can be appreciated by considering the power of conventional \textit{nonparametric} procedures for the two-sample testing problem in a high-dimensional setting, under the invariance structure \eqref{invariance_cond} discussed in \cref{rem_invariance}. Despite aiming to detect general distributional differences, such procedures often have trivial power for detecting discrepancies based on kurtosis, even when such marginal differences are non-vanishing \cite{Zhu, Sarkar}. In contrast, \cref{thm_kurtosis} establishes that our test does not suffer from an analogous issue for the class of univariate kurtosis-based departures from $\cH_0$. This is corroborated by simulation analysis in \cref{sec_sims_moderate}, where our test has higher power for alternatives generated via independent $\chi^2_{\nu}$ random variables than the recent high-dimensional normality test of \cite{ChenXia}. Finally, the condition $\delta_n  = \omega( 1 / \log n)$ in \cref{thm_kurtosis} can be relaxed to $\delta_n  = \omega 
( n^{-1/2} )$ under a stronger condition on the effective rank; see \cref{rem_general_class} and \cref{app_sec_IQR} for further detail.

\begin{remark}[Extensions of \cref{thm_kurtosis}] \label{rem_kurtosis} 
    Consistency for \textit{platykurtic} alternatives, with $\EE[Z_{ij}^4] = 3 - \delta_n \ $ for a sequence $\delta_n \in (0, 2)$, can be similarly established. 
    The sub-Gaussian assumption in \cref{kurtosis_stochastic_rep} can be relaxed to one of bounded eighth-moments, but may require a stronger regularity condition on the effective rank. \cref{thm_kurtosis} is stated under common excess kurtosis $\delta_n$ for simplicity but can be generalized so as to allow distinct excess kurtosis parameters $\EE Z_{ij}^4 = 3 + \delta_{n,j}$, $j \in [d]$.
\end{remark}

\subsubsection{Ratio Consistency of the Dispersion Index Estimator under Alternatives}

Proof of the consistency results in \cref{thm_loc_mix_sG,thm_cov_mix_sG,thm_ellip,thm_kurtosis} depends on the ratio-consistency of our estimator $\wh \Delta$ as specified in \eqref{def_Delta_hat} for the null dispersion index $\Delta$ given by \eqref{def_Delta_null}. The following proposition formally establishes the rate of convergence of $\wh\Delta/\Delta$ to unity under the alternatives specified by Models \ref{model_subG}, \ref{elliptical_stochastic_rep}, and \ref{kurtosis_stochastic_rep}. We note that it can be generalized to incorporate a broader class of alternative models; see \cref{Bai-Sarandasa}, \cref{finiteMixture_stochastic_rep}, and \cref{app_Delta_Alternative_proof}, for example. 

\begin{proposition}\label{prop_Delta_Alternatives}
    Under either \cref{model_subG}, \cref{elliptical_stochastic_rep}, or \cref{kurtosis_stochastic_rep}, one has 
    ${\wh \Delta / \Delta} = 1 + \cO_\PP(n^{-1/2}).
    $
\end{proposition}


Since \cref{prop_Delta_Alternatives} makes no assumptions regarding the effective ranks of the covariance-type matrices under the alternatives, the rate is of order $\cO_\PP(n^{-1/2})$, coinciding with the worst-case scenario in \cref{prop_Delta_Null} under the null. For establishing consistency, this rate is sufficient, but if stronger conditions on the effective ranks are imposed, the rate in \cref{prop_Delta_Alternatives} can be improved.

\section{Simulation Studies}\label{sec_sims_moderate}

In this section, we compare our proposed test with existing procedures in settings where $n$ is proportional to or larger than $d$. Since the procedure of \cite{ChenXia} is the primary available method for testing $\cH_0$ in high dimensions, we perform a direct comparison with their simulation results. In particular, we examine the type I error and power of our test for these examples relative to the \textit{Chen-Xia test} of \cite{ChenXia}, as well as the classical tests which \cite{ChenXia} identifies as possessing the best performance in high dimensions; namely, the extended Friedman-Rafsky test \cite{ExtendedFriedman}, the multivariate Shapiro-Wilk test \cite{MultivariateShapiro}, and Fisher's test \cite{ChenXia}. We note that the latter three tests are of modified form as per \cite{ChenXia}, with the sample covariance matrix replaced by a regularized estimator in their respective test statistics. Additional simulation studies for settings with $n \ll d$ are deferred to \cref{sec_sims_high}, where the test proposed by \cite{ChenXia} has limited applicability.

In evaluating the type I error, we set the mean vector to be $\mu = 0_d$ and the covariance matrices as follows:
 (a) $\Sigma_1 = \bI_d$; (b) $\Sigma_2 =  (\rho^{|i - j|})_{i,j \leq d}$ with $\rho = 0.5$; (c)  $\Sigma_3 = (\Sigma^* + \delta ~ \bI_d) / (1 + \delta)$, where $\Sigma^* = (\sigma^*_{ij})_{i, j \in [d]}$, with $\sigma^*_{jj} = 1$ for $j \in [d]$, $\sigma^*_{ij} = \sigma^*_{ji} \sim \text{Unif}[0,1] * \text{Bernoulli}(0.02)$ for $i < j$, and  $\delta = \max \{- \lambda_{\text{min}}(\Sigma^*), 0 \} + 0.05$; (d) $\Sigma_4 =  W W^\T / d  $, where $W \in \mathbb{R}^{d \times d}$ has \text{i.i.d.} standard normal entries. 
The choices of $\Sigma_1$, $\Sigma_2$, $\Sigma_3$ are those considered in \cite{ChenXia}, whereas $\Sigma_4$ is an additional conventional covariance structure.  As in \cite{ChenXia}, we consider $d\in \{20, 100, 300\}$ and $n\in \{100,150\}$ with significance level $\alpha = 0.05$. The results of our test are based on 10,000 replications, whereas those of the other tests are based on 1000 replications due to computational constraints (see \cref{rem_comp}).

\cref{tab_type_I_low_d} reports the averaged type I errors of each method. We find that 
the size of our test is maintained at the appropriate level in all settings. However, both the test of \cite{ChenXia} and the extended Friedman–Rafsky test show substantial type I error inflation under $\Sigma_4$ when $d$ is comparable to $n$, highlighting the importance of the covariance and sample-size conditions $(n \ll \sqrt{d})$ required for the theoretical guarantees in \cite{ChenXia}, as discussed in \cref{preExisting_work}. A similar size-distortion issue for the test of \cite{ChenXia} is also observed in the simulation studies of \cite{Elliptical_GoF} and ours in \cref{sec_sims_high}. 

\begin{table}[ht!]
\centering
\caption{Type I errors of each method for $\Sigma_j$ with $1\le j\le 4$. \textbf{Bold} figures indicate inflation of the type I error beyond the acceptable 0.1 threshold, as stipulated by \cite{ChenXia}.}
\label{tab_type_I_low_d}
\resizebox{\textwidth}{!}{
\renewcommand{\arraystretch}{1}{
    \begin{tabular}
    {|p{0.5 cm}|p{4cm}|p{1.5cm}|p{1.5cm}|p{1.5cm}|p{1.5cm}|p{1.5cm}|p{1.5cm}|}
\hline
\multicolumn{2}{|c|}{} & \multicolumn{3}{c}{$n = 100$} & \multicolumn{3}{|c|}{$n = 150$} \\
\cline{3-8}
\multicolumn{2}{|c|}{} &$d = 20$ &$d = 100$ &$d = 300$ &$d = 20$ &$d = 100$ &$d = 300$\\
\hline
\multirow{5}{*}{$\Sigma_1$} & Our Test  &0.048 & 0.045 & 0.046 &0.046 &0.052 &0.047\\
& Chen-Xia Test &0.043 &0.043 & 0.051 & 0.039 & 0.048 & 0.056 \\
& Friedman-Rafksy Test &0.043 &0.048 & 0.043 &0.048 & 0.054 & 0.037 \\
& Shapiro-Wilk Test & 0.064 & 0.046 & 0.051 & 0.04 & 0.048 & 0.053 \\
& Fisher's Test &0.06 &0.043 &0.044 &0.037 &0.047 &0.051 \\
\hline
\multirow{5}{*}{$\Sigma_2$} & Our Test &0.059 &0.049 & 0.049 &0.063 &0.051 & 0.049 \\
& Chen-Xia Test &0.059 & 0.043 & 0.064 &0.059 & 0.058 & 0.049\\
& Friedman-Rafksy Test &0.043 & 0.051 & 0.048  &0.039 & 0.07 & 0.069\\
& Shapiro-Wilk Test &0.056 &0.062 &\textbf{0.107} &0.054 &0.069 &0.062\\
& Fisher's Test &0.053 &0.066 &\textbf{0.104} &0.044 &0.07 & 0.056 \\
\hline
\multirow{5}{*}{$\Sigma_3$} & Our Test & 0.05 & 0.05 &0.046 &0.046 &0.052 &0.048 \\
& Chen-Xia Test &0.059 & 0.052 & 0.063 &0.048 & 0.053 & 0.049 \\
& Friedman-Rafksy Test &0.054 & 0.048 & \textbf{0.197} &0.048 & 0.04 & \textbf{0.152} \\ 
& Shapiro-Wilk Test & 0.047 & 0.058 &0.059 &0.034 &0.043 &0.073 \\
& Fisher's Test & 0.052 & 0.053 & 0.053 & 0.038 & 0.044 & 0.061 \\
\hline
\multirow{5}{*}{$\Sigma_4$} & Our Test & 0.056 & 0.048 & 0.049 & 0.064 & 0.047 & 0.05 \\  
& Chen-Xia Test & 0.072 & \textbf{0.417} & \textbf{0.586} & 0.05 & \textbf{0.115} & \textbf{0.788} \\  
& Friedman-Rafksy Test & 0.054 & \textbf{0.885} & \textbf{0.996} & 0.042 & 0.057 & \textbf{0.999} \\
& Shapiro-Wilk Test & 0.056 & 0.062 & 0.052 & 0.048 & 0.07 & 0.078 \\
& Fisher's Test & 0.05 & 0.062 & 0.044 & 0.044 & 0.066 & 0.068\\
\hline
    \end{tabular}
}}
\end{table}

To compare the power of different methods, we adopt the same alternatives considered in \cite{ChenXia} as presented in Tables \ref{tab_power_low_d_mix} -- \ref{tab_power_low_d_Gauss-t}, along with additional location-mixture alternatives in Table \ref{tab_pow_balanced_LocMix}.
Specifically, let $\Sigma_k$ with $k\in \{1,2,3\}$ be the covariance matrices introduced above and consider $d\in \{20, 100, 300\}$ and $n\in \{100, 150\}$. 
\cref{tab_power_low_d_mix} summarizes the empirical power of each method under the two-component Gaussian scale-mixture alternatives 
$
    0.5 \cN_d(0_d, (1 + a_d) \Sigma_k) + 0.5 \cN_d(0_d, (1 - a_d) \Sigma_k)
$
with $ a_d :=  1.8/\sqrt{d}$ and $k\in \{1,2,3\}$. In \cref{tab_power_low_d_t} we present the empirical power of each method under multivariate \textit{t}-distribution alternatives $t_{d}(0_d,  \Sigma_k, \nu_d)$ with $\nu_d := d / 2$ degrees of freedom (\text{d.o.f.}) and $k\in \{1,2,3\}$.  The third type of alternative  is given by $X = \Sigma_3^{1/2} (Y - \nu 1_d) / \sqrt{2 \nu}$ with $Y_j$ for $j\in [d]$ being \text{i.i.d.} from $\chi_{\nu}^2$, which becomes closer to $\cH_0$ as $\nu$ increases.   \cref{tab_power_low_d_chi} compares the power of our test with that of \cite{ChenXia} and the extended Friedman-Rafksy test  for $n = d = 100$ and $\nu \in \{3,5,10,20\}$. 
\cref{tab_power_low_d_Gauss-t} presents results for mixed marginal alternatives, where $(1 - \pi_t)d$ dimensions are generated from a standard multivariate normal distribution and $\pi_t d$ dimensions are generated from a multivariate-\textit{t} distribution $t_{\pi_t d}(0_{\pi_t d}, \Sigma_1, 25)$ with 25 d.o.f. for $\pi_t \in \{0.5, 0.4, 0.3, 0.2, 0.1\}$ and $n = d = 100$. 
Finally, Table \ref{tab_pow_balanced_LocMix} compares the power of our test to that of \cite{ChenXia} for two-component location-mixture alternatives with $\mu_1 = 0_d$, $\mu_2 = (2.15d^{-1/4}) 1_d$, mixing weights  $(\pi_1, \pi_2) = (0.5, 0.5)$, and covariance matrix $\Sigma_k$ for $k \in [3]$. 
The results in Tables \ref{tab_power_low_d_mix} -- \ref{tab_pow_balanced_LocMix} indicate superior power of our proposed test compared to the existing methods, while also maintaining better overall control of the type I error, as per \cref{tab_type_I_low_d}. 


\begin{table}[ht!]
\centering
\caption{Power comparison under the two-component Gaussian scale-mixture alternatives of \cite{ChenXia}.}
\label{tab_power_low_d_mix}
\resizebox{\textwidth}{!}{
\renewcommand{\arraystretch}{1}{
    \begin{tabular}{ |p{0.5 cm}|p{4cm}|p{1.5cm}|p{1.5cm}|p{1.5cm}|p{1.5cm}|p{1.5cm}|p{1.5cm}|    }
\hline
\multicolumn{2}{|c|}{} & \multicolumn{3}{c}{$n = 100$} & \multicolumn{3}{|c|}{$n = 150$} \\
\cline{3-8}
\multicolumn{2}{|c|}{} &$d = 20$ &$d = 100$ &$d = 300$ &$d = 20$ &$d = 100$ &$d = 300$\\
\hline
\multirow{5}{*}{$\Sigma_1$} & Our Test  & 0.9998 & 0.9999 & 1 &1  &1  &1  \\
& Chen-Xia Test &0.458 & 0.817 & 0.816 & 0.663 & 0.958 & 0.948 \\
& Friedman-Rafksy Test & 0.066 & 0.039 & 0.053 & 0.074 & 0.037 & 0.057 \\ 
& Shapiro-Wilk Test &0.559 &0.125 &0.08 &0.687 &0.171 &0.105 \\
& Fisher's Test &0.56 &0.127 &0.08 &0.698 &0.168 &0.103\\
\hline
\multirow{5}{*}{$\Sigma_2$} & Our Test &0.957 & 0.962 & 0.962 & 0.993  & 0.996  & 0.995  \\
& Chen-Xia Test & 0.153 & 0.619 & 0.719 & 0.267 & 0.672 & 0.819\\
& Friedman-Rafksy Test & 0.063 & 0.056 & 0.049 & 0.104 & 0.064 & 0.065\\
& Shapiro-Wilk Test &0.463 &0.177 &0.192 &0.633 &0.197 &0.127 \\
& Fisher's Test &0.473 &0.183 &0.192 &0.638 &0.198 &0.123 \\
\hline
\multirow{5}{*}{$\Sigma_3$} & Our Test & 0.999 &0.998 &0.999 &1 &1 &1  \\
& Chen-Xia Test & 0.45 & 0.754 & 0.869 & 0.64 & 0.908 & 0.947 \\
& Friedman-Rafksy Test &0.065 & 0.045 & 0.217 & 0.075 & 0.038 & 0.149 \\
& Shapiro-Wilk Test &0.532 &0.19 &0.105 &0.701 &0.189 &0.139 \\
& Fisher's Test &0.55 &0.193 &0.098 &0.705 &0.185 &0.132 \\
\hline
    \end{tabular}
}}
\end{table}

\begin{table}[ht!]
\centering
\caption{Power comparison under the multivariate-\textit{t} alternatives of \cite{ChenXia}. }
\label{tab_power_low_d_t}
\resizebox{\textwidth}{!}{
\renewcommand{\arraystretch}{1}{
    \begin{tabular}{ |p{0.5 cm}|p{4cm}|p{1.5cm}|p{1.5cm}|p{1.5cm}|p{1.5cm}|p{1.5cm}|p{1.5cm}|    }
\hline
\multicolumn{2}{|c|}{} & \multicolumn{3}{c}{$n = 100$} & \multicolumn{3}{|c|}{$n = 150$} \\
\cline{3-8}
\multicolumn{2}{|c|}{} &$d = 20$ &$d = 100$ &$d = 300$ &$d = 20$ &$d = 100$ &$d = 300$\\
\hline
\multirow{5}{*}{$\Sigma_1$} & Our Test  & 0.9998 & 0.9997 &0.9997 &1 &1 & 1 \\
& Chen-Xia Test &0.585 &0.913 &0.93 &0.799 &0.985 &0.992 \\
& Friedman-Rafksy Test &0.067 &0.037 &0.064 &0.102 &0.039 &0.047 \\
& Shapiro-Wilk Test &0.863 &0.212 &0.09 &0.965 &0.29 &0.123 \\
& Fisher's Test &0.862 &0.208 &0.087 &0.967 &0.301 &0.121 \\
\hline
\multirow{5}{*}{$\Sigma_2$} & Our Test & 0.972 & 0.972  &0.971 &0.995 & 0.996  &0.997 \\
& Chen-Xia Test &0.202 &0.713 &0.86 & 0.322 & 0.864 &0.942\\
& Friedman-Rafksy Test &0.118 &0.054 &0.052 &0.15 &0.046 &0.06\\
& Shapiro-Wilk Test &0.754 &0.266 &0.213 &0.923 &0.309 &0.171 \\
& Fisher's Test &0.758 &0.272 &0.209 &0.926 &0.312 &0.16\\
\hline
\multirow{5}{*}{$\Sigma_3$} & Our Test &0.999  &0.998 &0.999 &1 &1 &1 \\
& Chen-Xia Test &0.565 &0.879 &0.949 &0.741 &0.979 &0.982
 \\
& Friedman-Rafksy Test &0.067 &0.048 &0.184 &0.11 &0.035 &0.112 \\
& Shapiro-Wilk Test &0.849 &0.288 &0.106 &0.966 &0.303 &0.168 \\
& Fisher's Test &0.856 &0.285 &0.108 &0.965 &0.31 &0.173\\
\hline
    \end{tabular}
}}
\end{table}

\begin{table}[H]
\centering
 \caption{Power comparison under the standardized $\chi_\nu^2$ coordinates alternatives of \cite{ChenXia}.}
\label{tab_power_low_d_chi}
\renewcommand{\arraystretch}{1}{
    \begin{tabular}{|l|c|c|c|c|}
\hline
\multicolumn{1}{|c|}{} &$\nu = 3$ &$\nu = 5$ &$\nu = 10$ &$\nu = 20$ \\
\hline
Our Test &0.995 &0.899 &0.471 & 0.191 \\
Chen-Xia Test &0.451 &0.252 &0.106 &0.065\\
Friedman-Rafksy Test &0.094 &0.068 &0.065 &0.051\\
\hline
\end{tabular}
}\vspace{-5mm}
\end{table}


\begin{table}[H]
\centering
 \caption{Power comparison for alternatives with a $\pi_t$ proportion of non-Gaussian dimensions \cite{ChenXia}.}
\label{tab_power_low_d_Gauss-t}
\renewcommand{\arraystretch}{1}{
    \begin{tabular}{|l|c|c|c|c|c|}
\hline
\multicolumn{1}{|c|}{} &$\pi_t = 0.5$ &$\pi_t = 0.4$ &$\pi_t = 0.3$ &$\pi_t = 0.2$ & $\pi_t = 0.1$ \\
\hline
Our Test &0.964 &0.808 &0.491 & 0.191 & 0.065 \\
Chen-Xia Test &0.481 &0.259 &0.129 &0.07 & 0.042\\
Friedman-Rafksy Test &0.037 &0.052 &0.048 &0.05 & 0.041\\
\hline
\end{tabular}
}
\end{table}

\begin{table}[H]
\centering
\caption{Power comparison under the Gaussian location-mixture alternatives.}
\label{tab_pow_balanced_LocMix}
\resizebox{\textwidth}{!}{
\renewcommand{\arraystretch}{1}{
    \begin{tabular}
    {|p{0.6 cm}|p{3.5cm}|p{1.5cm}|p{1.5cm}|p{1.5cm}|p{1.5cm}|p{1.5cm}|p{1.5cm}|}
\hline
\multicolumn{2}{|c|}{} & \multicolumn{3}{c}{$n = 100$} & \multicolumn{3}{|c|}{$n = 150$} \\
\cline{3-8}
\multicolumn{2}{|c|}{} &$d = 20$ &$d = 100$ &$d = 300$ &$d = 20$ &$d = 100$ &$d = 300$\\
\hline
\multirow{2}{*}{$\Sigma_1$} & Our Test  & 0.828 & 0.916 & 0.945 &0.948  &0.983 & 0.992 \\
& Chen-Xia Test & 0.063 & 0.068 & 0.055 & 0.069  & 0.155 & 0.05 \\
\hline
\multirow{2}{*}{$\Sigma_2$} & Our Test  & 0.532 & 0.569 & 0.627 & 0.704  & 0.732 & 0.787 \\
& Chen-Xia Test & 0.048 & 0.103 & 0.068 & 0.049  & 0.124 & 0.084 \\
\hline
\multirow{2}{*}{$\Sigma_3$} & Our Test  & 0.784 & 0.809 & 0.842 & 0.927  & 0.94 & 0.951 \\
& Chen-Xia Test & 0.046 & 0.091 & 0.062 & 0.065  & 0.117 & 0.043 \\
\hline
    \end{tabular}
}}
\end{table}

  \section{Real Data Analysis}\label{sec_real_data}

As discussed in \cref{sec_intro}, \cref{rem_invariance}, and \cref{sec_theory_alter}, violations of the normality assumption $\cH_0$ can have severe consequences for conventional methodologies used in high-dimensional data analysis. In this section, we present a genomic application to demonstrate the proposed test's capacity to detect critical departures from the assumed multivariate normal model. The formal test of $\cH_0$ is complemented by associated graphical diagnostics, which we present to illustrate their use in aiding the identification of the potential source of the detected departure. In \cref{app_real_data_lungCancer} we present a second gene expression application, where we compare the analysis based on our methodology to the findings reported by \cite{ChenXia} based on their test of $\cH_0$.

Gene co-expression network analysis is an active area of research and application in modern biology and frequently involves data where $n \ll d$ \cite{SILGGM}. The network structure is characterized by conditional independence relationships among genes. A principal approach for estimating this structure is based on Gaussian Graphical Models (GGMs), which leverage the relationship between the precision matrix $\Sigma^{-1}$ and conditional independence under $\cH_0$.
However, as discussed before, the performance of GGMs can be highly sensitive to violations of their model assumptions, which typically include $\cH_0$ along with regularity conditions on $\Sigma$ such as \eqref{cond_ident} \cite{XiaCai}.  
 We thus demonstrate the utility of our proposed methodology in performing diagnostic analysis for gene co-expression network inference based on GGMs.
 
 As an illustrative example, we consider the analysis in \cite{SILGGM}, which implements several state-of-the-art GGM methodologies for high-dimensional data and applies them for estimating large-scale gene co-expression networks. In particular, \cite{SILGGM} analyzes a study on the genetic basis of childhood asthma, involving $n = 258$ patients and $d = 1953$ genes. Their goal is to identify which genes are connected to the \textit{CLK1} gene, a hub gene known to be associated with asthma. To this end, the authors use the estimated global network structure to extract a biologically meaningful local sub-network centered around \textit{CLK1}. The resulting local network, including the top 20 most significant connections to \textit{CLK1}, is shown in \cref{fig:sub_GeneNetwork}.

However, our test rejects $\cH_0$ at the $0.05$ significance level. To ascertain potential sources of departure from the assumed model, we examine the graphical plots pertaining to the empirical distribution of the radii $\{ R_i \}_{i \in [n]}$ displayed in \cref{fig:asthma_diagnostic_plots}, where the ordered \textit{standardized radii}
\[
    V_i := 2 \wh\Delta^{-1/2}\bigl(R_i - \tr^{1/2}(\wh \Sigma_{\text{D}})\bigr),\qquad \text{for each $ i \in [n]$},
\]
are also plotted against the corresponding standard normal quantiles. This plot is informally justified by the marginal convergence of the $2 \Delta^{-1/2} ( R_i - \tr^{1/2}(\Sigma) )$ variates, for $i\in [n]$, to the standard normal distribution, their approximate independence, the ratio-consistency of $\wh\Delta$, and the consistency of $\tr^{1/2}(\wh \Sigma_{\text{D}})$, under $\cH_0$ and standard conditions on $\Sigma$ (see \cref{rem_cond_thm1}, for example). The asymptotic normality and approximate independence can be deduced from our proof of \cref{thm_range_limit} (see, also, \cref{thm_IQR} and its proof), while the consistency properties are a consequence of the proof of \cref{prop_Delta_Null}. Thus, analogous to the classical use of empirical c.d.f. and quantile plots for the scaled radii \eqref{scaled_radii} \cite{Small, Brenner}, we use the graphical diagnostics of \cref{fig:asthma_diagnostic_plots} 
as a supplementary tool to assess potential sources of departures from $\cH_0$ detected by our test.

\begin{figure}[!htb]
\centering
\vspace{-2mm}
\begin{minipage}{0.325\textwidth}
    \centering
    \includegraphics[width=\linewidth]{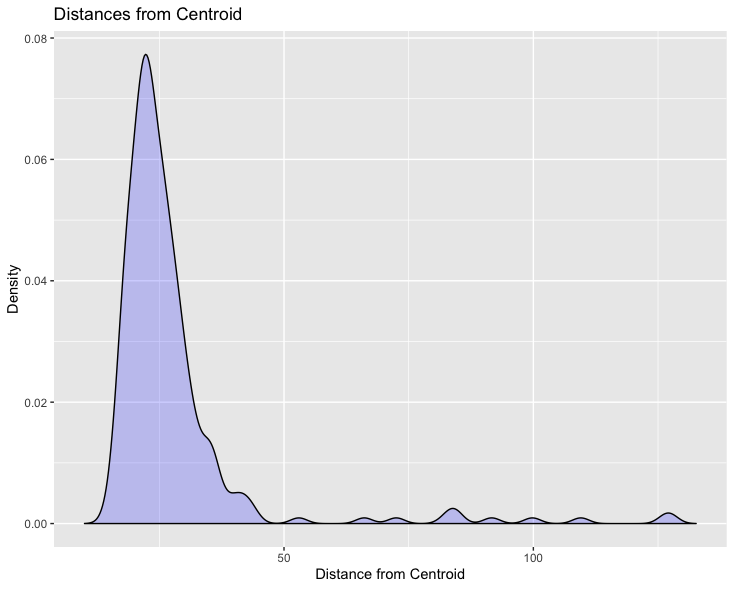}
\end{minipage}\hfill
\begin{minipage}{0.325\textwidth}
    \centering
    \includegraphics[width=\linewidth]{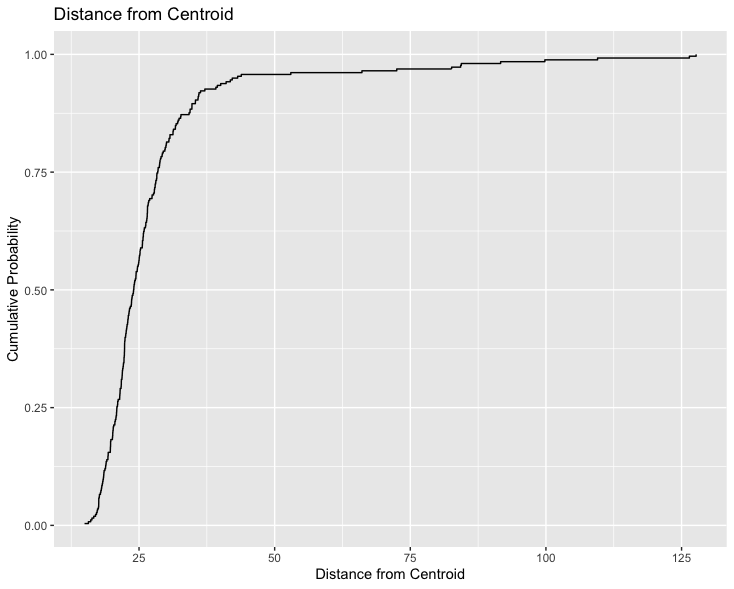}
\end{minipage}\hfill
\begin{minipage}{0.325\textwidth}
    \centering
    \includegraphics[height = 0.205\textheight,keepaspectratio=false]{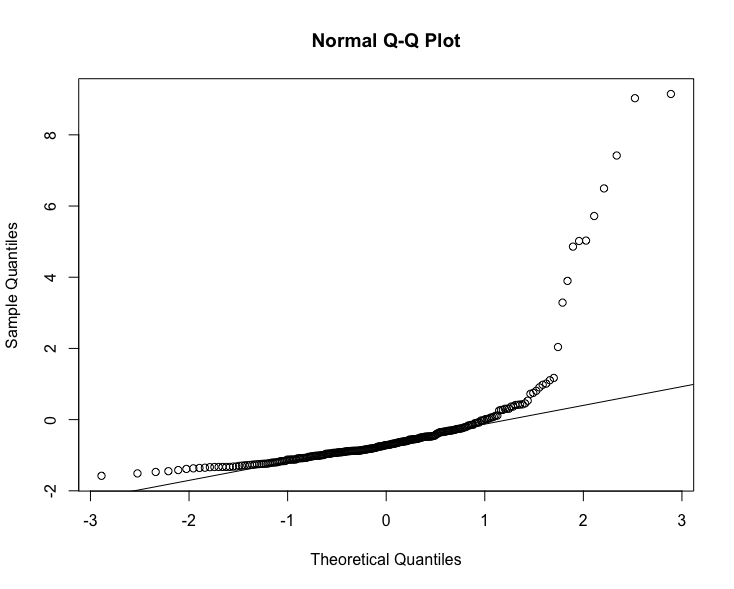}
\end{minipage}
\vspace{-3mm}
\caption{Left: Density estimate for the radii $\{ R_i \}_{i \in [n]}$. Middle: Empirical c.d.f. of the radii $\{ R_i \}_{i \in [n]}$. Right: Normal QQ plot for the standardized radii $\{ V_i \}_{i \in [n]}$.}
\label{fig:asthma_diagnostic_plots}
\end{figure}



\cref{fig:asthma_diagnostic_plots} suggests that the radii possess an empirical distribution with a notably heavy upper tail. Specifically, 11 samples -- approximately $4.3\%$ of the data -- deviate markedly from the expected behavior under $\cH_0$ and the bulk of the distribution, as determined by the gap criterion for outlier assessment \cite{Barnett_Book}. This is corroborated by inspection of marginal univariate and bivariate plots of the genes as well as a formal outlier-detection procedure we develop in a companion work. Thus, to examine the effect of these extreme samples on the results obtained by \cite{SILGGM}, we perform their analysis again after removing these observations. This exploratory-type comparative analysis is presented for illustration purposes. 
Note that without these 11 samples, our test fails to reject $\cH_0$.

\cref{fig:sub_GeneNetwork} contrasts the estimated local graph structure for the \textit{CLK1} gene obtained using the original data with that inferred when the extreme samples are absent. First, we note that when the complete dataset is used, 37 significant edges are inferred, whereas only 15 edges for the \textit{CLK1} gene are detected when the extreme samples are not present. Secondly, when we compare the network consisting of the top 20 most significant edges in \cref{fig:sub_GeneNetwork}, as presented in \cite{SILGGM}, we find that only $55\%$ of these genes appear in the set of significant genes identified when the extreme observations are omitted. As discussed in \cref{app_realData_network_supp} in further detail, this difference in the genes identified may correspond to potentially biologically meaningful differences in the relationships between the expression of certain genes with childhood asthma and its comorbidities. These considerations further illustrate the fact that the marked discrepancies in network structures inferred based on the presence of the extreme samples, as depicted in \cref{fig:sub_GeneNetwork}, may significantly affect the practical interpretation obtained in the gene co-expression network analysis of \cite{SILGGM}.




    


Overall, this exploratory-type analysis briefly demonstrates how our methodology can detect potentially critical departures from $\cH_0$, and can be used to guide follow-up analysis in conjunction with domain knowledge and recommended practices \cite{Barnett_Book, Gnanadesikan} for conducting analysis in the presence of, for example, potential outlier or contaminated mixture based violations of the assumed model.

\begin{figure}[!htb]
\centering
\begin{minipage}{0.325\textwidth}
    \centering
    \includegraphics[width=\linewidth]{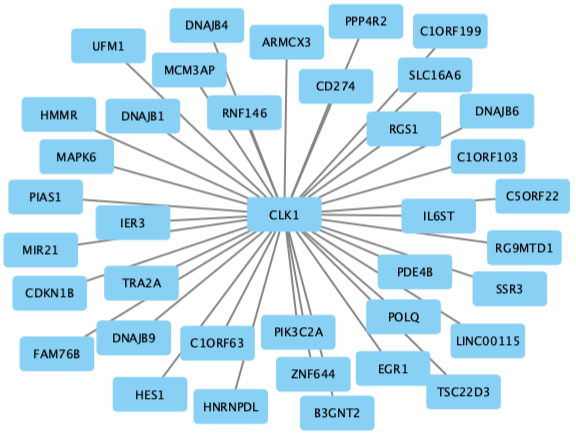}
\end{minipage}\hfill
\begin{minipage}{0.325\textwidth}
    \centering
    \includegraphics[width=\linewidth]{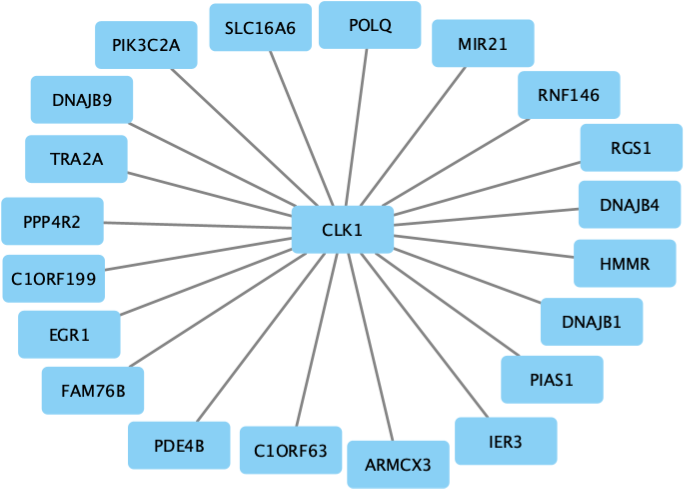}
\end{minipage}\hfill
\begin{minipage}{0.325\textwidth}
    \centering
    \includegraphics[width=\linewidth]{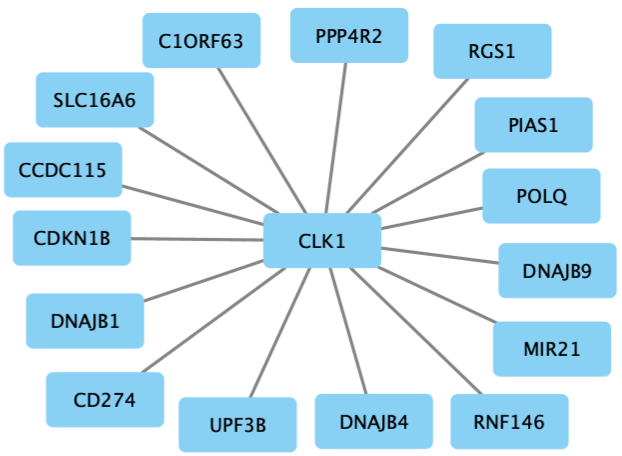}
\end{minipage}

\caption{Left: Significant edges inferred from the original data. Middle: Top 20 edges inferred from the original data. Right: Significant edges inferred when the extreme samples are absent.}
\label{fig:sub_GeneNetwork}
\end{figure}

\paragraph{Acknowledgments} The authors thank Kengo Kato for helpful discussion pertaining to the application of Gaussian approximation results to the range-type test statistic, as well as David Brenner for guidance and feedback regarding the direction of the project.




{\small 
\setlength{\bibsep}{0.85pt}{
    \bibliographystyle{abbrvnat}
    \bibliography{ref}
}}

\newpage

\appendix



\newcommand{\anon}{0}



\if1\anon
{ 
  \title{\bf Supplement to ``High-Dimensional Invariant Tests of Multivariate Normality Based on Radial Concentration''}
  \author{Xin Bing\\
    Department of Statistical Sciences, University of Toronto\\
    and \\
    Derek Latremouille \\ Division of Biostatistics, University of Toronto}
  \maketitle
} \fi

\if0\anon
{
  \bigskip
  \bigskip
  \bigskip
  \begin{center}
    {\LARGE\bf Supplement to ``High-Dimensional Invariant Tests of Multivariate Normality Based on Radial Concentration''}
\end{center}
  \medskip
} \fi

\bigskip

\appendix

The Appendix is structured as follows. 
\cref{app_sec_IQR} considers tests of the proposed class \eqref{test_stat_class} based on central quantile contrasts and combinations thereof, and develops associated asymptotic theory pertaining to type I error control in a high-dimensional $n,d \to \i$ regime. In particular, the type I error control theory for the interquartile range type test statistic $T_*$ specified by \eqref{def_T_IR} is developed. \cref{app_quasiRange_theory} presents methodology for implementing tests based on arbitrary extreme quasi-range specifications of the proposed class \eqref{test_stat_class}.     
\cref{app_add_simulation} presents additional simulation studies for our test and the test of \cite{ChenXia} when $n \ll d$ as well as some simulation results which are pertinent to  justification of using the radii instead of the squared radii in our test statistics (\cref{app_Radii_SqRadii}) and \cref{rem_composite_test} (\cref{app_Sim_RangeIQR}). \cref{app_real_data_lungCancer} presents an additional gene expression application for the purpose of comparing the results of our methodology to that of \cite{ChenXia}. \cref{app_realData_network_supp} provides further information pertaining to the gene co-expression network analysis of \cref{sec_real_data}. \cref{app_sec_power_BS} presents consistency results for a class of finite mixture models which generalizes the finite mixture alternatives considered in \cref{sec_power_mix}. Finally, \cref{app_proofs} contains the proofs of the results presented in \cref{sec_theory}, \cref{app_sec_IQR}, and \cref{app_sec_power_BS}.

\section{General Asymptotic Theory for the Proposed 
Class of Tests}\label{app_general_theory_AppendixA}

\subsection{Asymptotic Distribution of the IQR and Central Quantile Contrast Test Statistics}\label{app_sec_IQR}

    As discussed in \cref{rem_general_class}, while the main paper primarily addressed the theoretical properties of the range-type test based on $T$ from the general class of proposed test statistics \eqref{test_stat_class}, the analysis of \cref{sec_theory_null} can be extended so as to yield tests based on any finite number of central quantile contrast specifications in \eqref{test_stat_class} with associated asymptotic type I error-control guarantees under a high-dimensional $n, d \to \i$ regime. 
    

    The following theorem establishes that $T_*$ as defined in \eqref{def_T_IR} has an explicit and known normal limiting distribution under the null hypothesis. To state the result, we first define an additional notion of effective rank to complement those defined in \cref{def_rhos}: 
    \begin{align}
        \rho_3(\Sigma) := {\tr^3(\Sigma^2) \over \tr^2(\Sigma^3)}.
    \end{align}
    The relationship of $\rho_3(\Sigma)$ to those of \cref{def_rhos} is formally established in \cref{lem_ranks}, and the effective rank condition \eqref{effectiveRank_IQR} of \cref{thm_IQR} is further discussed in \cref{Rem_cond_thm_IQR}.


    \begin{theorem}\label{thm_IQR}
        Under $\cH_0$, suppose that, as $n \to \i$, either
        \begin{equation}\label{effectiveRank_IQR}
        \rho_1(\Sigma^2) = \omega(n) \ \ \ \ \text{or} \ \ \ \ \rho_3(\Sigma) = \omega(n^2 \log^2 n).
        \end{equation}
        Then, defining $\sigma_* = [2 \phi(\Phi^{-1}(0.75)) ]^{-1}$, one has 
        $
            T_* \distrto \cN
            (0, \sigma^2_*). 
        $
    \end{theorem}
    \begin{proof}
        The proof appears in \cref{app_proof_thm_IQR}.
    \end{proof}

     \cref{thm_IQR} justifies the usage of the rejection region specified in \eqref{reject_IQR}. 
    While the proof of \cref{thm_IQR} follows the same general structure as that of \cref{thm_range_limit}, the techniques used are different. Specifically, due to the use of the interquartile range as opposed to the range, the first step of the proof involves bounding the difference 
    \[
            \sup_{t\in \RR}\left|\PP\left(
				 Y_{(\lfloor 3n/4\rfloor)} - Y_{(\lfloor n/4\rfloor)}  \le t
				\right) - \PP(\wt U_{n} \le t)\right|,
    \]
    where $\wt U_n \sim \cN(2\Phi^{-1}(3/4), \sigma_*^2/n)$.
    In this case, the Gaussian approximation theory used for the range-type statistic $T$ is not applicable. Instead, in the case that $\rho_1(\Sigma^2) = \omega(n)$, our arguments rely on newly derived theory, stated in \cref{thm_order_clt}, pertaining to the asymptotic joint distribution of any finite number of central order statistics of $Y_1,\ldots, Y_n$, which is of interest in its own right. In particular, since \(Y_1, \ldots, Y_n\) are not independent and the distribution of each \(Y_i\) only converges to a particular absolutely continuous distribution $F$ (which, in this case, is Gaussian) in the limit, our new theory generalizes the classical result that only applies to the asymptotic joint distribution of a finite number of central order statistics from \(n\) \text{i.i.d.} realizations of $F$. On the other hand, when $\rho_3(\Sigma) = \omega(n^2 \log^2 n)$, we invoke Yurinskii coupling with respect to the sup-norm, as stated in \cref{lem_Yurinskii}, in conjunction with the 1-Lipschitz property of order statistics with respect to the sup-norm as established in \cref{lem_Lipschitz}. This coupling argument yields a Gaussian approximation from which the desired quantile convergence properties follow.

\begin{remark}[The Effective Rank Condition in \cref{thm_IQR}] \label{Rem_cond_thm_IQR}

In contrast to the theory developed in \cref{sec_theory} for the range-type test based on $T$ (requiring $d/\log^\gamma(n) \to \i$ for some constant $\gamma > 0$), the type I error-control for the IQR-type test based on $T_*$ involves the so-called \textit{high-dimensional medium sample size}, or \textit{ultra high-dimensional}, asymptotic regime, where $d / n \to \i$ at some rate as $n \to \i$. This asymptotic regime is commonly considered in establishing the theoretical properties of methodology developed for the numerous types of modern data-analysis applications involving a sample size which is of smaller order than the number of variables under consideration; see, for instance, \cite{TwoSampleCov_HDMSS, Zhu, AoshimaReview, Fan_2008, ShenReview}. 
In consideration of \cref{rem_cond_thm1} and \cref{lem_ranks}, when $d \gg n$, \eqref{effectiveRank_IQR} holds under a standard condition such as \eqref{cond_ident}, and when $d \gg n^2 \log^2 n$ it holds under a condition such as $\tr(\Sigma^k) \asymp d, \ $ for $k \in [4]$ (see, for instance, \cite{LedoitWolf2002, Fisher2010, Nishiyama2013, SriMANOVA, DunsonPati, Fisher2011_CovEstim, Schott2007}). However, as discussed in \cref{rem_composite_test}, the test based on $T_*$ exhibits sound empirical performance across both $n \ll d$ and $n \gtrsim d$ settings.  
\end{remark}

    \begin{remark} [A Composite Test in Practice] \label{rem_composite_test}
       Recall that our proposed test combines the range-type statistic $T$ and the IQR-type statistics $T_*$, with a Bonferroni correction. The proposal of such a combined test is motivated by the differential sensitivity of the two types of tests arising for various classes of alternatives, as exemplified empirically in the simulation analyses of \cref{tab_pow_unbalanced_mix,tab_pow_balanced_mix} in \cref{app_Sim_RangeIQR}.  For instance, the range-type test typically possesses higher sensitivity to location-type mixtures, unbalanced mixture alternatives, and departures from $\cH_0$ associated with the presence of a small proportion of samples of atypical distance to the centroid in the data. As illustrated in \cref{tab_pow_unbalanced_mix}, the latter phenomenon arises in mixture models with a relatively high degree of imbalance in the mixture proportions as well as outlier-contaminated data. As discussed in \cref{sec_our_contributions} and \cref{sec_method}, beyond its analogy to the range test of univariate normality \cite{Pearson_Normal, Pearson_Stephens, Serfling}, the use of the extremal radii in consideration of such departures is comparable to the classical uniformly most powerful test of multivariate normality against outlier-type alternatives, which is based the extreme scaled radius $R^*_{(n)}$ \cite{Ferguson, Wilks, Barnett_Book, Thode_Book, Gnanadesikan, Syed}. On the other hand, dispersion of the radii at the level of moderate deviations is often more effectively encapsulated by central quantile contrasts. This motivates using an additional statistic from the proposed class \eqref{test_stat_class} based on central quantiles, such as the IQR-type $T_*$, to complement the range-type $T$, so as to increase efficiency for detecting certain types of alternatives \cite{David, Lockhart}. In fact, consistency of the IQR test for sub-Gaussian covariance-type mixtures with suitably \textit{balanced} mixture proportions can be established under a slightly milder condition than that of \eqref{cond_cov_sep_sG} in \cref{thm_cov_mix_sG}. The superior power of $T_*$ compared to $T$ under \textit{balanced} finite scale-mixture alternatives, as shown in \cref{tab_pow_balanced_mix}, highlights its effectiveness for covariance-type mixture models with relatively homogeneous mixture proportions. 
     Finally, it is worth mentioning that, as suggested by our numerical experiments in \cref{app_Sim_RangeIQR}, the proposed composite test often outperforms the two individual tests, particularly in higher dimensions, and is only slightly inferior than the better of the two individual tests in other cases. Furthermore, while the use of this composite test is theoretically justified  when $n \ll d$ due to the inclusion of the IQR-type test based on $T_*$, the simulation analyses of \cref{sec_sims_moderate,sec_sims_high,app_Sim_RangeIQR} further support its use in cases where either $n \gtrsim d$ or $n \ll d$, as the empirical type I error is well-controlled and high power is attained against a broad class of alternatives.  
    \end{remark}

    
    \begin{remark} [Extension to More General Central Quantile Tests] \label{rem_general_quantile_stats}
    Let $Q  \le  \lfloor n/2 \rfloor$ denote a specified number of quantiles and  $1/2 < \pi_1^* < \cdots < \pi_{Q}^* < 1$ be any given upper-percentiles.  A general central quantile based test statistic can be defined via
            \begin{align}\label{Central Quantile class}
                T_{*, \pi^*} := 
                2 \sqrt{n} \sum_{q = 1}^{Q} \Bigl[ \wh\Delta^{-1/2}  \bigl( R_{(\lfloor \pi_q^* n \rfloor)} - R_{(\lfloor(1 - \pi_q^*)n \rfloor)} \bigr)  -  \Phi^{-1}(\pi_q^*) \Bigr],
            \end{align}
        and its asymptotic distribution under $\cH_0$ can be derived in analogy to that of $T_*$. The advantages associated with different test statistics of this general class deserve full investigation, which is thus left for future research. 
    \end{remark}

\subsection{Asymptotic Distribution of Extreme Quasi-Range Test Statistics}\label{app_quasiRange_theory}

The methodology and supporting theory for the range-type test of \cref{reject_range} can be generalized to accommodate arbitrary extreme \textit{quasi-range} test statistics of the proposed class \eqref{test_stat_class}. Specifically, for any integer $q \le n/2$ that is constant with respect to $n$, the $q^{\text{th}}$-order quasi-range test statistic corresponds to \eqref{eq_general_T} with $\bar q = n - q + 1$, $\underline{q} = q$, and $a_n,b_n$ specified by \eqref{def_an_bn}. The associated decision rule is given by the procedure used to construct the rejection region specified in \eqref{reject_range} and \eqref{def_Un}, except with the Monte Carlo distribution based on the $q^{\text{th}}$-order quasi-range instead of the range. That is, for this specification of $\bar q$, $\underline{q}$, $a_n$, and $b_n$, we define the corresponding quasi-range test statistic via
\begin{align}\label{Quasi Range Test Statistic}
                T_q := 2 a_n  ~ \wh\Delta^{-1/2} \bigl(R_{(n - q + 1)} -  R_{(q)}\bigr) -  2 a_n b_n.
    \end{align}
The null hypothesis $\cH_0$ is then rejected at level $\alpha \in (0,1)$ iff 
\begin{equation}\label{reject_quasiRange}
        T_q \notin \left( \wh F^{-1}_{M,n,q}(\alpha/2), ~  \wh F^{-1}_{M, n, q}(1-\alpha/2)\right),
    \end{equation}
    where, for a specified number of Monte Carlo replications $M \in \mathbb{Z}^+$ and percentile $\alpha_0 \in (0,1)$, $\wh F^{-1}_{M,n, q}(\alpha_0)$ denotes the $\alpha_0$-quantile of the empirical distribution $\wh F_{M,n,q}$ of $M$ i.i.d. realizations of 
    \begin{equation}
        U_{n,q}  ~ =  ~ a_n\bigl(S_{(n - q + 1)} - S_{(q)}\bigr) -  2 a_nb_n, 
    \end{equation} 
    with $S_{(1)}\le \cdots \le S_{(n)}$ denoting the order statistics of $S \sim \cN_n(0_n, \bI_n)$. The general proof techniques used to establish the results of \cref{sec_theory_null} pertaining to the range-type test based on $T$ in conjunction with the recently developed Gaussian approximation for general extreme order statistics \cite{Gauss_approx_extremeOrderStats} can be used to derive analogous asymptotic type I error theory for the test specified by \eqref{reject_quasiRange} in a high-dimensional regime.

\section{Additional Simulation Studies} \label{app_add_simulation}

\subsection{Empirical Performance of the Proposed Test when $n \ll d$}\label{sec_sims_high}



    


In this section, we examine the performance our proposed test in settings where $n \ll d$. Here, we primarily restrict the simulation analysis to our test alone. This is due to the fact that, while the recently proposed test of \cite{ChenXia} is the principal test of $\cH_0$ with valid type I error control in a high-dimensional $n, d \to \i$ setting, the theoretical guarantee for their test requires that $d \ll \sqrt{n}$, and the computationally-intensive nature of their procedure (see \cref{rem_comp}) renders extensive comparison to our test across a broad range of $d \gg n$ settings infeasible. Nonetheless, we begin by conducting a small-scale simulation analysis to examine the type I error of their test as $d/n$ increases, based on 500 Monte Carlo replications for their test and 10,000 replications for ours. \cref{tab_type_I_CX_hd} presents the empirical type I error of their test for $n \in \{40, 100 \}$ and $d \in \{600, 1000\}$ under $\Sigma_2$ as specified in \cref{sec_sims_moderate}. Note that, despite restricting $d \leq 1000$, their test exhibits noticeable size distortion. Moreover, if we were to consider $\Sigma_2$ with $\rho = 0.9$ instead of $\rho = 0.5$, the type I error of their test increases to nearly 1 across all $(n,d)$ configurations considered while it remains well-controlled for our test. Failure to maintain adequate control of the type I error was also identified as an issue for their test in the simulation studies of \cref{sec_sims_moderate} and \cite{Elliptical_GoF}. 

\begin{table}[H]
\centering
\caption{Comparison of the type I error of our test to that of the Chen-Xia test \cite{ChenXia} under covariance matrix $\Sigma_2$ specified in \cref{sec_sims_moderate} when $n \ll d$. \textbf{Bold} figures indicate inflation of the type I error beyond the acceptable 0.1 threshold, as stipulated by \cite{ChenXia}.}
\label{tab_type_I_CX_hd}
    \centering
    \renewcommand{\arraystretch}{1}{
    \begin{tabular}{ |c|l|c|c|c|c| }
        \hline
        \multicolumn{2}{|c|}{} & \multicolumn{2}{c|}{$n = 40$} & \multicolumn{2}{c|}{$n = 100$} \\
        \hline
        \multicolumn{2}{|c|}{}  & $d = 600$ &$d = 1000$ &$d = 600$ &$d = 1000$ \\
        \hline
        \multirow{2}{*}{$\Sigma_2$} & Our Test & 0.051 & 0.047 & 0.048 & 0.05  \\
        & Chen-Xia Test &\textbf{0.288} & \textbf{0.322} & 0.087  & 0.094 \\
        \hline
    \end{tabular}
    }
\end{table}

 
We now examine the performance of our test under a broader set of $n \ll d$ settings. The type I error of our testing procedure is examined first. We set $\mu = 0_d$ without loss of generality due to the invariance properties of our test (see \cref{rem_invariance}) and consider each of the following choices for the covariance matrix $\Sigma$: $\Sigma_1$, $\Sigma_2$ with $\rho = 0.9$, and $\Sigma_4$ as in \cref{sec_sims_moderate}, as well as $\Sigma_5  = \diag(\lambda_1,\ldots,\lambda_d)$ with $\lambda_j  = 0.93^j$ for $j\in [d]$.

\cref{tab_type_I_high_d} displays the empirical type I errors of the proposed test for $d\in \{2000, 5000, 10000\}$ and $n\in \{50, 100, 250\}$ across 10,000 replications. We see that the type I error of our test is well-controlled at the $\alpha = 0.05$ level in all settings. 
 

\begin{table}[H]
\centering
\caption{Empirical type I errors of our test based on $10,000$ replications.}
\label{tab_type_I_high_d}
\renewcommand{\arraystretch}{1.1}{
    \begin{tabular}{cc|ccc|ccc|ccc|}
\hline
\multicolumn{2}{|c|}{} & \multicolumn{3}{c|}{$n = 50$} & \multicolumn{3}{c|}{$n = 100$} & \multicolumn{3}{c|}{$n = 250$} \\
\hline
\multicolumn{2}{|c|}{$d$} &$2000$ &$5000$ &$10,000$ &$2000$ &$5000$ &$10,000$ &$2000$ &$5000$ &$10,000$\\
\hline
\multicolumn{2}{|c|}{$\Sigma_1$} &0.048 &0.051 &0.048 &0.048 & 0.05 & 0.048 & 0.048 & 0.048 & 0.052 \\
\hline
\multicolumn{2}{|c|}{$\Sigma_2$} &0.046 &0.052 &0.051 & 0.05 & 0.049 &  0.049 & 0.051 & 0.051 &  0.043 \\
\hline
\multicolumn{2}{|c|}{$\Sigma_4$} &0.05 &0.051 &0.051 &0.042 &0.051  & 0.039 & 0.05 & 0.058  & 0.041 \\
\hline
\multicolumn{2}{|c|}{$\Sigma_5$} &0.051 &0.049 &0.05 &0.049  &0.052 & 0.052 & 0.053 & 0.054 & 0.056 \\
\hline
\end{tabular}
}
\end{table}


To evaluate the power of our testing procedure, we consider the following four alternatives, which correspond to Theorems \ref{thm_loc_mix_sG}, \ref{thm_cov_mix_sG}, \ref{thm_ellip}, and \ref{thm_kurtosis}, respectively.  
\begin{itemize}
    \setlength{\itemsep}{0pt}
    \item[(1)] {\em Loc-mixture:} $X\sim 0.5 \cN_d(0_d, \bI_d) + 0.5 \cN_d(a_d 1_d, \bI_d)$ with $a_d   =  (2.15) d^{-1/4}$.

    \item[(2)] {\em Cov-mixture:} $X\sim 0.5 \cN_d(0_d, (1 + a_d) \bI_d) + 0.5 \cN_d(0_d, (1 - a_d) \bI_d)$ with $a_d = 1.4/\sqrt{d}$.


    \item[(3)] {\em Multivariate-t:} $X$ follows the multivariate $t$-distribution $t_{d}(0_d,  \bI_d, \nu_d)$ with $\nu_d = d$. 
    
    \item[(4)] {\em $\chi^2$ marginals:}  $X_j \sim \chi^2_6$ independently for $j\in [d]$. 

\end{itemize}

\cref{tab_power_high_d} reports the empirical power of our testing procedure under each of the above alternatives. We observe that the empirical power tends to increase with the sample size. The results for the location-mixture (1) and covariance-mixture (2) examples indicate that the test can reliably detect both location- and covariance-based signals which are of relatively low strength marginally. We note that the signal-to-noise ratio quantities introduced in \cref{sec_power_mix} are set to decay at a $\delta \asymp d^{-1/2}$ rate for both of these examples. The performance of the test for the multivariate \textit{t}-distribution alternative of (3) suggests sensitivity of our test for non-Gaussian elliptical alternatives, even when the univariate and low-dimensional marginal distributions are approximately normal and the covariance matrix is scale-identity, thereby demonstrating its high sensitivity for detecting non-linear dependence structure. Finally, the simulation results obtained for the chi-squared marginal model (4) indicates that our test, despite being multivariate in nature, has the capacity to detect kurtosis-based departures from $\cH_0$ which arise at the univariate level.

\begin{table}[H]
\centering 
\caption{Empirical power of our test based on $10,000$ replications.}
\label{tab_power_high_d}
\renewcommand{\arraystretch}{1.1}{
\resizebox{\textwidth}{!}{
     \begin{tabular}{cc|ccc|ccc|ccc|}\hline
\multicolumn{2}{|c|}{} & \multicolumn{3}{c|}{$n = 50$} & \multicolumn{3}{c}{$n = 100$} & \multicolumn{3}{|c|}{$n = 250$} \\
\hline
\multicolumn{2}{|c|}{$d$} &$2000$ &$5000$ &$10,000$ &$2000$ &$5000$ &$10,000$ &$2000$ &$5000$ &$10,000$\\
\hline
\multicolumn{2}{|c|}{Loc-mixture} & 0.764 & 0.766 & 0.774 &0.967 & 0.967 & 0.971 & 1 & 1 & 1 \\
\hline
\multicolumn{2}{|c|}{Cov-mixture} & 0.793 & 0.799 & 0.803 & 0.967 & 0.966 & 0.968 & 1 & 0.999 & 0.999 \\
\hline
\multicolumn{2}{|c|}{Multivariate-$t$} & 0.756 & 0.749 & 0.757 & 0.928 & 0.926  & 0.943  & 0.999  & 0.999  & 0.999  \\
\hline\multicolumn{2}{|c|}{$\chi^2$ marginals} & 0.743  & 0.747  & 0.748 & 0.928 & 0.931 & 0.929 & 0.999 & 0.998 & 0.999 \\
\hline
\end{tabular}
}
}
\end{table}




\subsection{Comparing the Radii and Squared Radii Based Tests} \label{app_Radii_SqRadii}

The proof of \cref{thm_range_limit} implicitly supports the option of testing $\cH_0$ using a statistic analogous to $T$ based on the range of the {\em squared} radii $R_1^2,\ldots, R_n^2$ instead of the radii $R_1,\ldots,R_n$. In this section  we conduct extensive simulation analyses to support the use of the  proposed test over its counterpart based on the squared radii. 

The range-type test statistic based on the squared radii is defined via
\[ T_2 := a_n \left[\left(2 \widehat{\tr(\Sigma^2)}\right)^{-1/2}\left(R^2_{(n)} - R^2_{(1)}\right)  - 2 b_n \right], \]
whereas the IQR-type statistic based on the squared radii is 
\[T_{*,2} := \sqrt{n} \left[\left(2 \widehat{\tr(\Sigma^2)}\right)^{-1/2}\left(R^2_{(\lfloor  3n/4 \rfloor)} - R^2_{(\lfloor n/4 \rfloor)} \right) - 2 \Phi^{-1}(0.75) \right]. 
\]
Their rejection rules are identical to that specified by \eqref{reject_range} and \eqref{reject_IQR}, respectively. In the following subsections \cref{sec_sims_sqRadii_moderate} and \cref{sec_sims_sqRadii_high}, \textit{squared radii} refers to the composite test involving $T_2$ and $T_{*,2}$, with a Bonferroni correction applied to control the overall type I error. 

\cref{sec_sims_sqRadii_moderate} and \cref{sec_sims_sqRadii_high} compares the empirical type I error rate and power of the squared radii test to that of our proposed test under the simulation settings considered in \cref{sec_sims_moderate} and \cref{sec_sims_high}, respectively. In contrast to our proposed test, we find that the test based on the squared radii exhibits a persistent size distortion issue under $\cH_0$, with an empirical type I error rate $\wh \alpha_* > 0.05$ exceeding the nominal $\alpha = 0.05$ level across the entire range of $(n,d)$ and covariance matrix configurations considered; see \cref{tab_sqRadii_TypeI} and \cref{tab_sqRadii_TypeI_high}. In several cases, the squared radii test exhibits a particularly high degree of type I error inflation, with $\wh \alpha_* > 0.1$, whereas our proposed test does not. On the other hand, the power of the squared radii test is comparable to that of our proposed test across the alternatives considered in \cref{sec_sims_moderate} and \cref{sec_sims_high}. 

A heuristic theoretical justification for this is as follows: Since each squared radius $R_i^2$ under $\cH_0$ has an exact distribution equal to that of a linear combination of $d$ independent $\chi_1^2$ random variables, use of the square root transformation improves the Gaussian approximation to its distribution, hence providing better finite-sample control of the type I error. This improvement is analogous to the fact that the $\chi_d$ distribution provides a better normal approximation than the $\chi^2_d$ distribution \cite{Johnson_chi} due to the reduction of right skewness and kurtosis.

\subsubsection{Comparison under Simulation Settings in \cref{sec_sims_moderate}} \label{sec_sims_sqRadii_moderate}

\begin{table}[H]
\centering
\caption{Type I errors under the examples of \cite{ChenXia} as well as that under the null model with covariance matrix $\Sigma_4$. \textbf{Bold} figures indicate inflation of the type I error beyond the acceptable 0.1 threshold.}
\label{tab_sqRadii_TypeI}
\resizebox{\textwidth}{!}{
    \begin{tabular}
    {|p{0.9 cm}|p{4cm}|p{1.5cm}|p{1.5cm}|p{1.5cm}|p{1.5cm}|p{1.5cm}|p{1.5cm}|}
\hline
\multicolumn{2}{|c|}{} & \multicolumn{3}{c}{$n = 100$} & \multicolumn{3}{|c|}{$n = 150$} \\
\cline{3-8}
\multicolumn{2}{|c|}{} &$d = 20$ &$d = 100$ &$d = 300$ &$d = 20$ &$d = 100$ &$d = 300$\\
\hline
\multirow{2}{*}{$\Sigma_1$} & Our Test  &0.048 & 0.045 & 0.046 &0.046 &0.052 &0.047\\
& Squared Radii &0.07 & 0.058 & 0.055 &0.077 &0.058 &0.052 \\
\hline
\multirow{2}{*}{$\Sigma_2$} & Our Test &0.059 &0.049 & 0.049 &0.063 &0.051 & 0.049 \\
& Squared Radii & \textbf{0.138} & 0.063 & 0.055 & \textbf{0.143} &0.072 &0.057 \\
\hline
\multirow{2}{*}{$\Sigma_3$} & Our Test & 0.05 & 0.05 &0.046 &0.046 &0.052 &0.048 \\
& Squared Radii &0.088 & 0.057 & 0.052 & 0.094 &0.058 &0.052 \\
\hline
\multirow{2}{*}{$\Sigma_4$} & Our Test & 0.056 & 0.048 & 0.049 & 0.064 & 0.047 & 0.05 \\ 
& Squared Radii &\textbf{0.168} & 0.065 & 0.059 & \textbf{0.161} &0.07 &0.054 \\
\hline
    \end{tabular}
}
\end{table}

\begin{table}[H]
\centering
\caption{Power under the two-component Gaussian scale-mixture alternatives of \cite{ChenXia}.}
\label{tab_sqRadii_mixPow}
\resizebox{\textwidth}{!}{
    \begin{tabular}
    {|p{0.9 cm}|p{4cm}|p{1.5cm}|p{1.5cm}|p{1.5cm}|p{1.5cm}|p{1.5cm}|p{1.5cm}|}
\hline
\multicolumn{2}{|c|}{} & \multicolumn{3}{c}{$n = 100$} & \multicolumn{3}{|c|}{$n = 150$} \\
\cline{3-8}
\multicolumn{2}{|c|}{} &$d = 20$ &$d = 100$ &$d = 300$ &$d = 20$ &$d = 100$ &$d = 300$\\
\hline
\multirow{2}{*}{$\Sigma_1$} & Our Test  & 0.9999 & 0.9998 & 1 &1  &1  &1  \\
& Squared Radii  & 0.9998 & 0.9999 & 0.9999 & 1 & 1 & 1 \\
\hline
\multirow{2}{*}{$\Sigma_2$} & Our Test &0.957 & 0.962 & 0.962 & 0.993  & 0.996  & 0.995  \\
& Squared Radii & 0.934 & 0.947 & 0.948 & 0.986 & 0.991 & 0.992 \\
\hline
\multirow{2}{*}{$\Sigma_3$}& Our Test & 0.999 &0.998 &0.999 &1 &1 &1  \\
& Squared Radii & 0.9992 & 0.998 & 0.997 & 1 & 1 & 1 \\
\hline
    \end{tabular}
}
\end{table}


\begin{table}[H]
\centering
\caption{Power under the multivariate \textit{t}-distribution alternatives of \cite{ChenXia}.} 
\label{tab_sqRadii_tPow}
\resizebox{\textwidth}{!}{
    \begin{tabular}
    {|p{0.9 cm}|p{4cm}|p{1.5cm}|p{1.5cm}|p{1.5cm}|p{1.5cm}|p{1.5cm}|p{1.5cm}|}
\hline
\multicolumn{2}{|c|}{} & \multicolumn{3}{c}{$n = 100$} & \multicolumn{3}{|c|}{$n = 150$} \\
\cline{3-8}
\multicolumn{2}{|c|}{} &$d = 20$ &$d = 100$ &$d = 300$ &$d = 20$ &$d = 100$ &$d = 300$\\
\hline
\multirow{2}{*}{$\Sigma_1$} & Our Test  & 0.9998 & 0.9997 &0.9997 &1 &1 & 1 \\
& Squared Radii & 0.9997 & 0.9995 & 0.9997 & 1 & 1 & 1 \\
\hline
\multirow{2}{*}{$\Sigma_2$} & Our Test & 0.972 & 0.972  &0.971 &0.995 & 0.996  &0.997 \\
& Squared Radii & 0.977 & 0.967 & 0.964 & 0.996 & 0.994 & 0.995 \\
\hline
\multirow{2}{*}{$\Sigma_3$} & Our Test &0.999  &0.998 &0.999 &1 &1 &1 \\
& Squared Radii & 0.998 & 0.997 & 0.997 & 0.9997 & 1 & 0.9997 \\
\hline
    \end{tabular}
}
\end{table}

\begin{table}[H]
\centering
 \caption{Power comparison under the standardized $\chi_\nu^2$ coordinates alternative of \cite{ChenXia}, with $n = d = 100$ and covariance matrix $\Sigma_3$.}
\renewcommand{\arraystretch}{1.1}{
    \begin{tabular}{|l|c|c|c|c|}
\hline
\multicolumn{1}{|c|}{} &$\nu = 3$ &$\nu = 5$ &$\nu = 10$ &$\nu = 20$ \\
\hline
Our Test &0.995 &0.89 &0.471 & 0.191 \\
Squared Radii & 0.992 &0.882 &0.454 & 0.177\\
\hline
\end{tabular}
}
\end{table}

\begin{table}[H]
\centering
 \caption{Power comparison for alternatives with a $\pi_t$ proportion of non-Gaussian dimensions \cite{ChenXia}.}
\label{tab_power_low_d_Gauss-t}
\renewcommand{\arraystretch}{1.1}{
    \begin{tabular}{|l|c|c|c|c|c|}
\hline
\multicolumn{1}{|c|}{} &$\pi_t = 0.5$ &$\pi_t = 0.4$ &$\pi_t = 0.3$ &$\pi_t = 0.2$ & $\pi_t = 0.1$ \\
\hline
Our Test &0.964 &0.808 &0.491 & 0.191 & 0.065 \\
Squared Radii & 0.955 & 0.791 & 0.468 & 0.177 & 0.067 \\
\hline
\end{tabular}
}
\end{table}

\begin{table}[H]
\centering
\caption{Power comparison under the Gaussian location-mixture alternatives.}
\label{tab_pow_balanced_LocMix}
\resizebox{\textwidth}{!}{
\renewcommand{\arraystretch}{1.1}{
    \begin{tabular}
    {|p{0.6 cm}|p{3.5cm}|p{1.5cm}|p{1.5cm}|p{1.5cm}|p{1.5cm}|p{1.5cm}|p{1.5cm}|}
\hline
\multicolumn{2}{|c|}{} & \multicolumn{3}{c}{$n = 100$} & \multicolumn{3}{|c|}{$n = 150$} \\
\cline{3-8}
\multicolumn{2}{|c|}{} &$d = 20$ &$d = 100$ &$d = 300$ &$d = 20$ &$d = 100$ &$d = 300$\\
\hline
\multirow{2}{*}{$\Sigma_1$} & Our Test  & 0.828 & 0.916 & 0.945 &0.948  &0.983 & 0.992 \\
& Squared Radii & 0.872 & 0.936 & 0.961 & 0.956  & 0.988 & 0.994 \\
\hline
\multirow{2}{*}{$\Sigma_2$} & Our Test  & 0.532 & 0.569 & 0.627 & 0.704  & 0.732 & 0.787 \\
& Squared Radii & 0.585 & 0.605 & 0.673 & 0.729  & 0.748 & 0.817 \\
\hline
\multirow{2}{*}{$\Sigma_3$} & Our Test  & 0.784 & 0.809 & 0.842 & 0.927  & 0.94 & 0.951 \\
& Squared Radii & 0.832 & 0.838 & 0.869 & 0.936  & 0.939 & 0.954 \\
\hline
    \end{tabular}
}}
\end{table}

\subsubsection{Comparison under Simulation Settings in \cref{sec_sims_high}}  \label{sec_sims_sqRadii_high}

\begin{table}[H]
\centering
\caption{Empirical type I errors based on 10,000 replications. \textbf{Bold} figures indicate inflation of the type I error beyond the acceptable 0.1 threshold.} 
\label{tab_sqRadii_TypeI_high}
\Large
\resizebox{\textwidth}{!}{
    \begin{tabular}
    {|c|l|ccc|ccc|ccc|}
\hline
\multicolumn{2}{|c|}{} & \multicolumn{3}{c|}{$n = 50$} & \multicolumn{3}{c|}{$n = 100$} & \multicolumn{3}{c|}{$n = 250$} \\
\cline{3-11}
\multicolumn{2}{|c|}{} &$d = 2000$ &$d = 5000$ &$d = 10,000$ &$d = 2000$ &$d = 5000$ &$d = 10,000$ &$d = 2000$ &$d = 5000$ &$d = 10,000$ \\
\hline
\multirow{2}{*}{$\Sigma_1$}  & Our Test &0.048 &0.051 &0.048 &0.048 & 0.05 & 0.048 & 0.048 & 0.048 & 0.052 \\
& Squared Radii &0.055 &0.053 & 0.053 & 0.054 & 0.053 & 0.054 & 0.053 & 0.053 & 0.052 \\
\hline
\multirow{2}{*}{$\Sigma_2$} & Our Test &0.046 &0.052 &0.051 & 0.05 & 0.049 &  0.049 & 0.051 & 0.051 &  0.043 \\
& Squared Radii &0.055 &0.056 & 0.054 & 0.057 & 0.056 & 0.051 & 0.056 & 0.056 & 0.053 \\
\hline
\multirow{2}{*}{$\Sigma_4$} & Our Test &0.05 &0.051 &0.051 &0.042 &0.051  & 0.039 & 0.05 & 0.058  & 0.041 \\
& Squared Radii &0.054 &0.052 & 0.053 & 0.052 & 0.053 & 0.053 & 0.053 & 0.051 & 0.057 \\
\hline
\multirow{2}{*}{$\Sigma_5$} & Our Test &0.051 &0.049 &0.05 &0.049  &0.052 & 0.052 & 0.053 & 0.054 & 0.056 \\
& Squared Radii &0.074 & 0.072 & 0.069 & 0.083 & 0.08 & 0.084 & \textbf{0.103} & \textbf{0.11} & \textbf{0.102} \\
\hline
    \end{tabular}
}
\end{table}

\begin{table}[H]
\centering
\caption{Empirical power based on $10,000$ replications.} 
\Large
\resizebox{\textwidth}{!}{
    \begin{tabular}
    {|c|l|ccc|ccc|ccc|}
\hline
\multicolumn{2}{|c|}{} & \multicolumn{3}{c|}{$n = 50$} & \multicolumn{3}{c|}{$n = 100$} & \multicolumn{3}{c|}{$n = 250$} \\
\cline{3-11}
\multicolumn{2}{|c|}{} &$d = 2000$ &$d = 5000$ &$d = 10,000$ &$d = 2000$ &$d = 5000$ &$d = 10,000$ &$d = 2000$ &$d = 5000$ &$d = 10,000$ \\
\hline
\multirow{2}{*}{Loc-Mixture}  & Our Test & 0.764 & 0.766 & 0.774 &0.967 & 0.967 & 0.971 & 1 & 1 & 1 \\
& Squared Radii & 0.823 & 0.822 & 0.834 &0.974 & 0.976 & 0.981 & 1 & 1 & 1 \\
\hline
\multirow{2}{*}{Cov-Mixture} & Our Test & 0.793 & 0.799 & 0.803 & 0.967 & 0.966 & 0.968 & 1 & 0.999 & 0.999 \\
& Squared Radii & 0.729 & 0.737 & 0.728 & 0.952 & 0.954 &  0.953 & 1 & 0.999 & 1 \\
\hline
\multirow{2}{*}{Multivariate-$t$} & Our Test & 0.756 & 0.749 & 0.757 & 0.928 & 0.926  & 0.943  & 0.999  & 0.999  & 0.999  \\
& Squared Radii &0.668 &0.679 & 0.671 & 0.904 & 0.902 & 0.912 & 0.999 & 0.998 & 0.998 \\
\hline
\multirow{2}{*}{$\chi^2$ Marginals} & Our Test & 0.743  & 0.747  & 0.748 & 0.928 & 0.931 & 0.929 & 0.999 & 0.998 & 0.999 \\
& Squared Radii & 0.674 & 0.668 & 0.681 & 0.903 & 0.899 & 0.906 & 0.998 & 0.997 & 0.999 \\
\hline
    \end{tabular}
}
\end{table}

\subsection{Power Comparison: Range Versus IQR Tests} \label{app_Sim_RangeIQR}

As discussed in \cref{rem_composite_test}, in this section we report an empirical power comparison for the proposed combined test, the range-type test (\cref{sec_method}), the IQR-type test (\cref{app_sec_IQR}), and the test of \cite{ChenXia} under \textit{unbalanced} two-component Gaussian \textit{location-mixture} alternatives in \cref{tab_pow_unbalanced_mix} and \textit{balanced} two-component Gaussian \textit{covariance-mixture} alternatives in \cref{tab_pow_balanced_mix}. We find that the range test outperforms the IQR test and the test of \cite{ChenXia} for the former type of alternative, while the IQR test has higher power for alternatives of the latter type. Overall, the combined test tends to have the highest power, indicating the benefit of using both the range- and IQR-based tests together. 


\begin{table}[H]
\centering
\caption{Power comparison under \textit{unbalanced} two-component Gaussian \textit{location-mixture} alternatives, with $\mu_1 = 0_d$, $\mu_2 = 1_d$, and mixture proportions $(\pi_1, \pi_2) = (0.95, 0.05)$.}
\label{tab_pow_unbalanced_mix}
\resizebox{\textwidth}{!}{
    \begin{tabular}
    {|p{0.6 cm}|p{3cm}|p{1.5cm}|p{1.5cm}|p{1.5cm}|p{1.5cm}|p{1.5cm}|p{1.5cm}|}
\hline
\multicolumn{2}{|c|}{} & \multicolumn{3}{c}{$n = 100$} & \multicolumn{3}{|c|}{$n = 150$} \\
\cline{3-8}
\multicolumn{2}{|c|}{} &$d = 20$ &$d = 100$ &$d = 300$ &$d = 20$ &$d = 100$ &$d = 300$\\
\hline
\multirow{4}{*}{$\Sigma_1$} & Combined Test  &0.214 & 0.963 & 0.993 &0.256  &0.987 & 0.9995 \\
& IQR Test &0.049 & 0.061 & 0.312 & 0.051 & 0.076 & 0.426 \\
& Range Test & 0.274 & 0.975 & 0.988 & 0.321 & 0.993 & 0.998 \\ 
& Chen-Xia Test &0.056 &0.152 &0.418 &0.048 &0.26 &0.606 \\
\hline
\multirow{4}{*}{$\Sigma_2$} & Combined Test & 0.127 & 0.76  & 0.989 &0.158 &0.853  & 0.999   \\
& IQR Test & 0.057 & 0.06 &  0.183 &  0.07  & 0.077 & 0.265 \\
& Range Test & 0.155 &  0.832 & 0.989 & 0.178 & 0.909 &  0.998 \\ 
& Chen-Xia Test &0.064 &0.128 &0.415 &0.076 &0.14 &0.506 \\
\hline
\multirow{4}{*}{$\Sigma_3$} & Combined Test & 0.209 & 0.933  & 0.994 & 0.249 &0.975 & 0.9994  \\
& IQR Test & 0.052 & 0.068 & 0.267 & 0.05 &  0.075   & 0.375 \\
& Range Test & 0.267 &  0.96 & 0.99 &  0.35 & 0.989 & 0.998 \\  
& Chen-Xia Test &0.052 &0.146 &0.562 &0.05 & 0.182 &0.614 \\
\hline
    \end{tabular}
}
\end{table}

\begin{table}[H]
\centering
\caption{Power comparison under the \textit{balanced} two-component Gaussian \textit{covariance-mixture} alternatives of \cite{ChenXia} (see \cref{sec_sims_moderate}). All tests are considered at the $\alpha = 0.05$ level.}
\label{tab_pow_balanced_mix}
\resizebox{\textwidth}{!}{
    \begin{tabular}
    {|p{0.6cm}|p{3cm}|p{1.5cm}|p{1.5cm}|p{1.5cm}|p{1.5cm}|p{1.5cm}|p{1.5cm}|}
\hline
\multicolumn{2}{|c|}{} & \multicolumn{3}{c}{$n = 100$} & \multicolumn{3}{|c|}{$n = 150$} \\
\cline{3-8}
\multicolumn{2}{|c|}{} &$d = 20$ &$d = 100$ &$d = 300$ &$d = 20$ &$d = 100$ &$d = 300$\\
\hline
\multirow{4}{*}{$\Sigma_1$} & Combined Test  &0.9999 & 0.9998 & 1 & 1  & 1 & 1\\
& IQR Test & 0.9997 & 0.9999 & 1 & 1 & 1 & 1 \\
& Range Test &0.823 & 0.841 & 0.852 & 0.865 & 0.892 & 0.898 \\
& Chen-Xia Test &0.458 & 0.817 & 0.816 & 0.663 & 0.958 & 0.948 \\  
\hline
\multirow{4}{*}{$\Sigma_2$} & Combined Test & 0.957 & 0.962 &0.962 &0.993  &0.996  &0.995  \\
& IQR Test &0.957 & 0.962 & 0.965 & 0.995 & 0.996 & 0.995 \\
& Range Test &0.499 & 0.531 & 0.542 & 0.529 & 0.584 & 0.597 \\
& Chen-Xia Test & 0.153 & 0.619 & 0.719 & 0.267 & 0.672 & 0.819\\   
\hline
\multirow{4}{*}{$\Sigma_3$} & Combined Test &0.999 & 0.998 & 0.999 & 1 & 1 & 1  \\
& IQR Test & 0.999 & 0.998 & 0.998 & 1 & 1 & 0.9999 \\
& Range Test &0.766 & 0.739 & 0.746 & 0.802 & 0.791 & 0.803 \\
& Chen-Xia Test & 0.45 & 0.754 & 0.869 & 0.64 & 0.908 & 0.947 \\ 
\hline
    \end{tabular}
}
\end{table}

\section{Real Data Analysis: An Additional Application and Supplementary Detail for \cref{sec_real_data}} \label{app_real_data}

\subsection{The Lung Cancer Gene Expression Application Considered in \cite{ChenXia}} \label{app_real_data_lungCancer}

\indent Microarray gene expression data frequently involves a sample size in the tens to low hundreds, with the expression levels of up to thousands or tens of thousands of genes included in the analysis, which is often based on the multivariate normal model. For the purpose of comparison, we consider the lung cancer gene expression data of \cite{Gordon2002}, whose multivariate normality was tested in \cite{ChenXia} using their high-dimensional test of $\cH_0$. The data consists of $n = 150$ patients and $d = 12,533$ genes, and is considered in \cite{ChenXia} because it has been analyzed using variable-selection and discrimination methods. 

Our test rejects $\cH_0$ at the $\alpha = 0.05$ level. While we note that, as discussed in \cref{preExisting_work} and \cref{sec_sims_moderate}, the test of \cite{ChenXia} encounters issues in controlling the type I error when $n \ll d$, our findings corroborate their conclusion regarding $\cH_0$. Moreover, our graphical diagnostics, as introduced in \cref{sec_real_data}, reveal additional pertinent structure in the data. In particular, the plots pertaining to the empirical distribution of the radii $\{ R_i \}_{i \in [n]}$, the \textit{standardized radii} $\{ V_i \}_{i \in [n]}$ as defined in \cref{sec_real_data}, and the interpoint distances $\{ \| X_i - X_j \|_2 \}_{i < j \in [n]}$ displayed in \cref{fig:Lung_Cancer} aid in the identification of two samples which are of anomalous distance from all other observations. 

Detection of outliers is crucial in the analysis of high-dimensional data, as many procedures used to address diverse scientific problems exhibit severe performance degradation in their presence, but identifying these anomalous observations in a rigorous manner is challenging \cite{Hirose, Fritsch, Barnett_Book}. By leveraging pertinent distance-based information contained in the sample, our proposed test and associated graphical diagnostics can assist in formally detecting such observations. These observations can then be further examined to determine whether steps such as sensitivity analysis or omission of the samples are warranted, based on domain knowledge and recommended best practices for conducting analysis in the presence of outliers \cite{Stephens, Barnett_Book, Gnanadesikan}.

\begin{figure}[!htb]
\centering
\vspace{-2mm}
\begin{minipage}{0.48\textwidth}
    \centering
    \includegraphics[width=\linewidth]{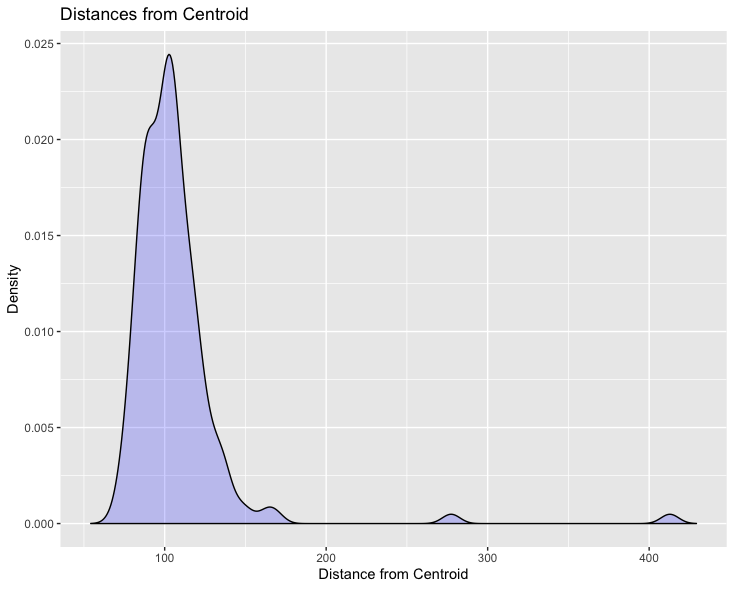}
\end{minipage}\hspace{5mm}
\begin{minipage}{0.48\textwidth}
    \centering
    \includegraphics[width=\linewidth]{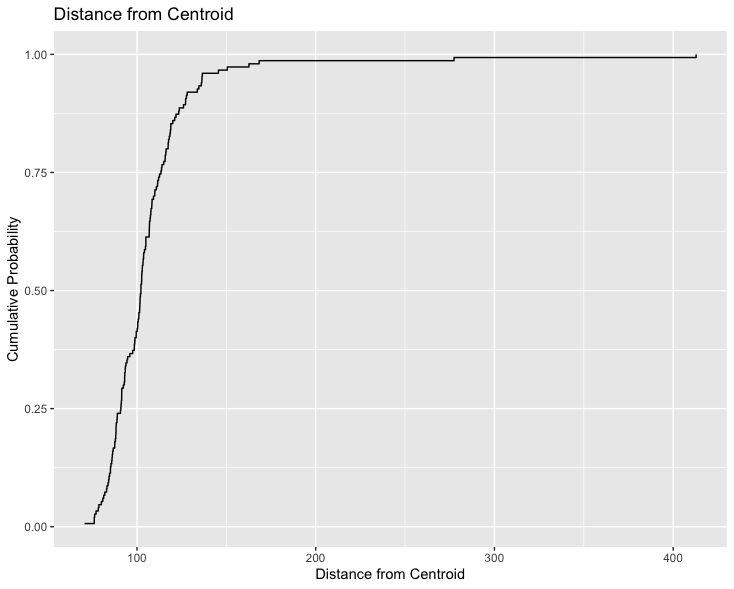}
\end{minipage}
\begin{minipage}{0.48\textwidth}
    \centering
    \includegraphics[width=\linewidth]{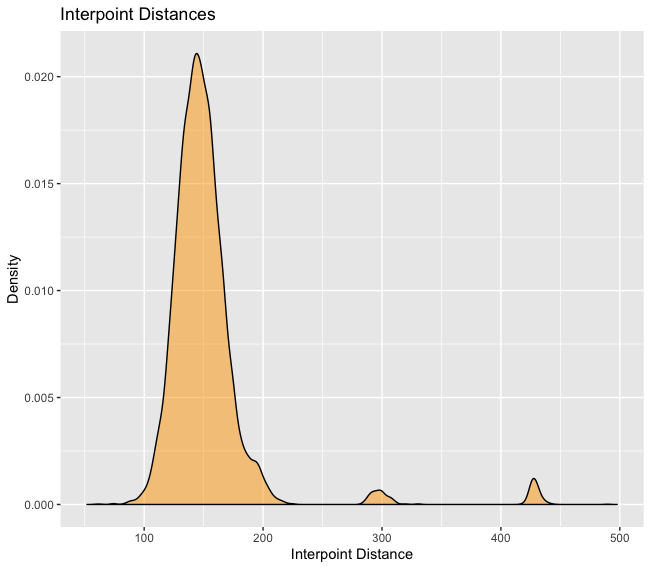}
\end{minipage} \hspace{0mm}
\begin{minipage}{0.48\textwidth}
    \centering
    \includegraphics[height = 0.325\textheight,keepaspectratio=false,width=1.1\linewidth]{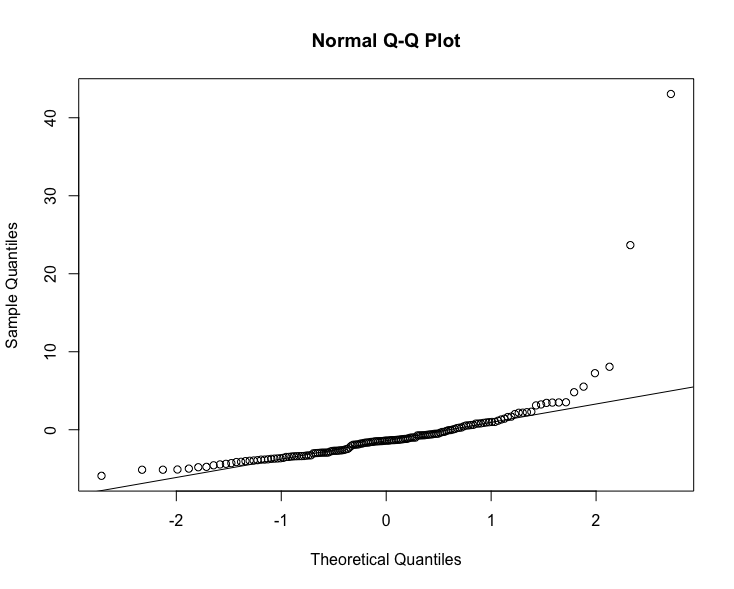}
\end{minipage}
\vspace{-3mm}
\caption{Panel 1: Density for the radii $\{ R_i \}_{i \in [n]}$. Panel 2: E.C.D.F. of the radii $\{ R_i \}_{i \in [n]}$. Panel 3: Density for the interpoint distances. Panel 4: Normal QQ plot for the standardized radii $\{ V_i \}_{i \in [n]}$.}
\label{fig:Lung_Cancer}
\end{figure}

\subsection{Supplementary Detail for the Application of \cref{sec_real_data}} \label{app_realData_network_supp}

As discussed in \cref{sec_real_data}, the gene co-expression network analysis of \cite{SILGGM} based on the childhood asthma data of \cite{Asthma_data1, Asthma_data2} is sensitive to the presence of 11 samples we determined to be anomalous and incongruous with the assumed model $\cH_0$. For reference, the inferred network structures both including and excluding these extreme samples is replicated in \cref{fig:sub_GeneNetwork_2}. We again note that when the complete dataset is used, 37 significant edges are inferred, whereas only 15 edges for the \textit{CLK1} gene are detected when the extreme samples are not present. Secondly, when we compare the network consisting of the top 20 most significant edges in \cref{fig:sub_GeneNetwork}, as presented in \cite{SILGGM}, we find that only $55\%$ of these genes appear in the set of significant genes identified when the extreme observations are omitted. This difference in the genes identified may correspond to potentially biologically meaningful differences in the relationships between the expression of certain genes with childhood asthma and its comorbidities. For example, the \textit{CDKN1B}, \textit{CCDC115}, \textit{CD274}, and \textit{UPF3B} genes are connected to the \textit{CLK1} gene after removal of these extreme samples, in contrast to the fact that these edges are not among the 20 most significant edges from the analysis of \cite{SILGGM} using the original data, as depicted in \cref{fig:sub_GeneNetwork}. These additional genes are associated with asthma, including in children in particular, \cite{CDKN1B_2, CDKN1B_1, CD274_1, CD274_2} as well as body height and other developmental issues \cite{Genehancer_UPF3B, UPF3B_1, CCDC115_1}\footnote{\url{https://www.genecards.org/cgi-bin/carddisp.pl?gene=UPF3B}} for which childhood asthma is a risk factor \cite{Asthma_Growth, Asthma_develop}. These considerations further illustrate the fact that the marked discrepancies in network structures inferred based on the presence of the extreme samples, as depicted in \cref{fig:sub_GeneNetwork}, may significantly affect the practical interpretation obtained in the gene co-expression network analysis of \cite{SILGGM}.

\begin{figure}[!htb]
\centering
\begin{minipage}{0.325\textwidth}
    \centering
    \includegraphics[width=\linewidth]{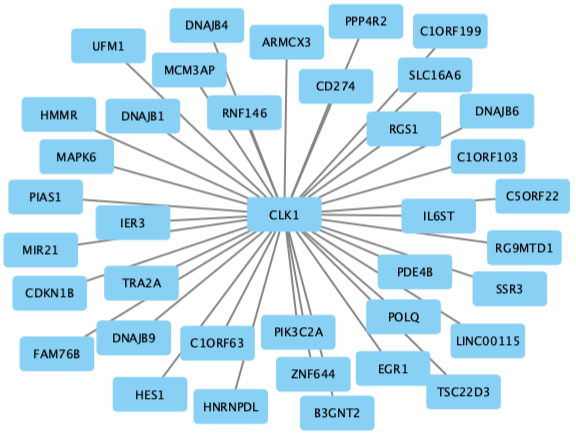}
\end{minipage}\hfill
\begin{minipage}{0.325\textwidth}
    \centering
    \includegraphics[width=\linewidth]{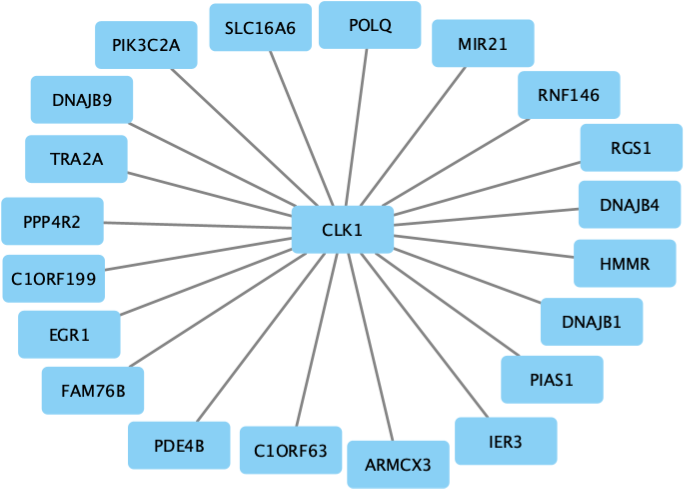}
\end{minipage}\hfill
\begin{minipage}{0.325\textwidth}
    \centering
    \includegraphics[width=\linewidth]{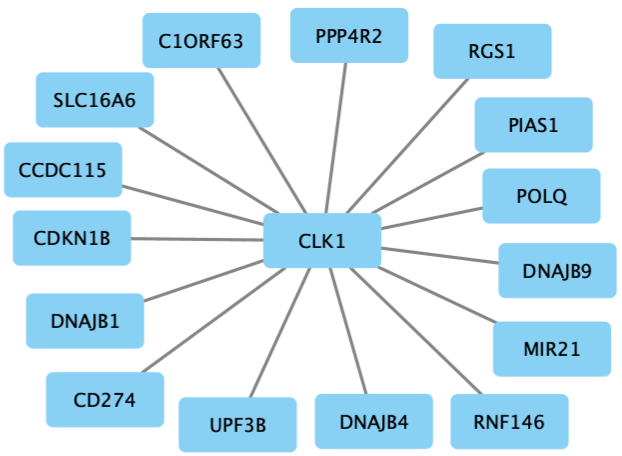}
\end{minipage}

\caption{Left: Significant edges inferred from the original data. Middle: Top 20 edges inferred from the original data. Right: Significant edges inferred when the extreme samples are absent.}
\label{fig:sub_GeneNetwork_2}
\end{figure}


\section{Consistency for Finite Mixture Alternatives under Mild Moment Conditions}\label{app_sec_power_BS}

In addition to the finite mixture of sub-Gaussian alternatives considered in \cref{model_subG} of \cref{sec_power_mix}, in the following we consider finite mixtures generated via mixture components from a more general class of distributions.

\begin{definition}[Bai-Sarandasa Type Distributions]\label{Bai-Sarandasa}
    We say that a random vector $Y\in \RR^d$ has a distribution of \emph{Bai-Sarandasa type}, denoted by $Y \sim \ \mathcal{B}_{d,m}(\mu, \Gamma)$, if there exists some integer $m \geq d$,  some mean vector $\mu \in \mathbb{R}^d$, and some matrix $\Gamma \in \mathbb{R}^{d \times m}$ such that:
    \begin{enumerate} 
        \item[(i)] $Y =  \mu + \Gamma  Z$ for some random vector $Z \in \mathbb{R}^m$ with $\EE[Z] = 0_d$, $\EE[ZZ^\T] = \bI_m$, $\EE[Z_\ell^4] = \kappa$, and $\EE [Z_\ell^8] \leq C < \infty$ for each $\ell \in [m]$ and for some constants $\kappa,C>0$.
        
        \item[(ii)] For any $\ell_1 \neq  \cdots \neq \ell_r \in [m]$ with $r \in [8]$ and exponents $\alpha_1,\ldots,\alpha_r \in \mathbb{Z}^+$ satisfying $\sum_{k = 1}^r \alpha_k \leq 8$, $\EE ( \prod_{k = 1}^r Z_{\ell_k}^{\alpha_k} ) = \prod_{k = 1}^r \EE[Z_{\ell_k}^{\alpha_k}]$ holds.
    \end{enumerate} 
\end{definition}
For any $Y \sim \mathcal{B}_{d,m}(\mu, \Gamma)$, we have $\EE[Y] = \mu$ and $\Cov(Y) = \Gamma \Gamma^{\T}$. This class of distributions routinely serves as a generic family of multivariate models in high-dimensional testing problems \cite{songchen_regress, songchen_glm, Chen2010, ChenQin, SongChen2014, empiricallikeli, ChenBanded}, and consists of nonparametric factor-analytic type models, where coordinates of the latent factor vectors $Z_i \in \RR^m$, $i \in [n]$, are not required to be independent. Adopting these distributions for mixture components provides a broad class of mixture alternatives, as the distribution, dependence, and number of latent factors used in each mixture component can be heterogeneous. In particular, this alternative class includes Gaussian mixture models and multivariate skew-normal mixtures \cite{Aosh2018}, the families of parametric mixture alternatives which conform most closely to the null multivariate normal model, but which have mixture components satisfying far stronger structural conditions than that specified by \cref{Bai-Sarandasa}.

The class of alternatives we consider below are finite mixture distributions whose mixture components are of Bai-Sarandasa type.

\begin{model}[Bai-Sarandasa Mixture Alternatives]\label{finiteMixture_stochastic_rep}
There exists some integer $K \geq 2$, some vectors $\mu_k\in \RR^d$ and some matrices $\Gamma_k\in \RR^{d\times m_k}$ with integers $m_k\ge d$ for $k\in [K]$ such that 
\begin{align*}
  X_1,\ldots, X_n  ~ \overset{\text{i.i.d.}}{\sim} ~ \sum_{k = 1}^K \pi_k \ \mathcal{B}_{d,m_k}(\mu_k, \Gamma_k),
\end{align*}
where the mixing probabilities satisfy $\min_{k\in [K]}\pi_k \geq c$ for some universal constant $c>0$. 
\end{model}

The following theorem establishes consistency under location-type Bai-Sarandasa mixtures. As discussed in \cref{sec_power_mix}, the condition of identical component-conditional covariance matrices $\Sigma_k = \Sigma_*$, for $k \in [K]$, can be relaxed so as to allow $\Sigma_k \neq \Sigma_{\ell}$; see the proof of \cref{thm_loc_mix} in \cref{app_proof_thm_loc_mix} for more detail.

\begin{theorem}[Location-Type Mixtures]\label{thm_loc_mix} 
    Under \cref{finiteMixture_stochastic_rep} with $\Sigma_*:= \Gamma_k \Gamma_k^\T$ for all $k\in [K]$, suppose that
    $\rho_1(\Sigma_*) = \omega(n)$ and
    \begin{equation}\label{cond_mean_sep}
        \max_{k,\ell \in [K]}{\|\mu_k-\mu_\ell\|_2^2\over \tr(\Sigma_*) }  = \omega\left({1\over \min\{\rho_1(\Sigma_*)/n, \sqrt{\rho_2(\Sigma_*)/n}~\}}\right).
    \end{equation}
    Then, for arbitrary choice of level $\alpha \in (0,1)$,
    \[ 
        \lim_{n\to \i} \PP \left(\cH_0 \text{ is rejected}\right) =  1.
    \]
\end{theorem}

\begin{proof}
    Its proof appears in \cref{app_proof_thm_loc_mix}.
\end{proof}

Compared to \cref{thm_cov_mix_sG}, the SNR condition in \eqref{cond_mean_sep} is stronger for general \cref{finiteMixture_stochastic_rep} alternatives as a consequence of the relaxed moment conditions on the mixture components. To observe the effect of dimensionality in \eqref{cond_mean_sep}, in the simple example considered after \cref{thm_cov_mix_sG} where 
$\mu_{j2} = \mu_{j1} + \delta_n$ for all $j\in [d]$,
\eqref{cond_mean_sep} reduces to 
$
    \delta^2_n = \omega(
         \sqrt{n / d}
    ),
$
implying that the marginal distinguishability condition on $\delta_n$ gets milder when $d$ is larger in order with respect to $n \to \i$. 

Similar to \cref{thm_cov_mix_sG}, we also have the following consistency results under covariance-type Bai-Sarandasa mixtures. As with \cref{thm_cov_mix_sG}, the result is stated for mixture components with an identical mean vector  $\mu_1 = \cdots = \mu_K$, but is proven in \cref{app_proof_thm_cov_mix} under a relaxed condition allowing distinct mean vectors $\mu_k \neq \mu_{\ell}$.

\begin{theorem}[Covariance-Type Mixtures]\label{thm_cov_mix}
 Under \cref{finiteMixture_stochastic_rep} with $\mu_1 = \cdots = \mu_K$, suppose that
 \begin{equation}\label{cond_cov_sep}
    {\max_{k,\ell\in [K]} \tr(\Sigma_k-\Sigma_\ell) \over \max_{k\in [K]}  \tr(\Sigma_k) }   ~ = ~  \omega\left(  { \sqrt{\log n} \over 
        \min_{k\in[K]} \min\{\rho_1(\Sigma_k)/n, \sqrt{\rho_2(\Sigma_k)/n}~\}}  \right).
  \end{equation}
 Then, for arbitrary choice of level $\alpha \in (0,1)$, 
    \[ 
        \lim_{n\to \i}\PP \left(\cH_0 \text{ is rejected}\right) = 1.
    \]
\end{theorem}

\begin{proof}
    Its proof appears in \cref{app_proof_thm_cov_mix}.
\end{proof}

Due to the relaxed moment conditions on Bai-Sarandasa mixture components, condition \eqref{cond_cov_sep} puts stronger requirement on the maximum relative difference in total variance than \eqref{cond_cov_sep_sG} in \cref{thm_cov_mix_sG}.
In the simple example discussed following \cref{thm_cov_mix_sG} where $[\Sigma_1]_{jj} = [\Sigma_2]_{jj} + \delta_n$ for all $j\in[d]$ and some sequence $\delta_n > 0$, condition \eqref{cond_cov_sep} simplifies to 
$
      \delta_n  ~=  ~    \omega(\sqrt{{ n \log(n) / d} }),
$
which becomes less stringent as $d/(n \log n)$ increases. 

\subsection{An Example of Dependent Factors under a Bai-Sarandasa Type Distribution}\label{app_example_Bai_Sarandasa}  
    The second condition on the latent factors $Z$ in \cref{Bai-Sarandasa}  is satisfied when $Z$ consists of coordinates which are either independent or possess some mild form of dependence. An example where the condition is satisfied when the $(Z_\ell)_{\ell \in [m]}$ are not independent is as follows: Suppose $(T_\ell)_{\ell \in [m]} \independent U$, where $(T_\ell)_{\ell \in [m]}$ are independent random variables satisfying the first condition of \cref{Bai-Sarandasa} in lieu of $(Z_\ell)_{\ell \in [m]}$, and $\PP(U = -1) = \PP(U = 1) = 1/2$. Letting $Z_\ell = U T_\ell$ for $\ell \in [m]$, it can be verified that both the marginal- and product-moment conditions of \cref{Bai-Sarandasa} hold and that the $Z_\ell$ are not independent.

\section{Proofs} \label{app_proofs}

\subsection{A Basic Lemma on Effective Ranks of $\Sigma$}

The following lemma establishes relationships among the following effective ranks of $\Sigma$:
\[
    \rho_1(\Sigma) := {\tr(\Sigma) \over \|\Sigma\|_\op},\qquad \rho_2(\Sigma) := {\tr^2(\Sigma) \over \tr(\Sigma^2)},\qquad \rho_3(\Sigma) := {\tr^3(\Sigma^2) \over \tr^2(\Sigma^3)}. 
\]
\begin{lemma}\label{lem_ranks}
    Let $\rho_1(\Sigma), \rho_2(\Sigma), \rho_3(\Sigma)$ be defined as above. Provided that $\|\Sigma\|_\op >0$, one has 
    \begin{enumerate}
        \item[(1)] $1 \leq \sqrt{\rho_3(\Sigma)} \le \rho_2(\Sigma^2) \le \rho_3(\Sigma)\le \rho_2(\Sigma) \leq \rank(\Sigma),$
        \item[(2)] ${\rho_1^2(\Sigma)/ d}\le \rho_3(\Sigma)\le \rho_1^{3/2}(\Sigma),$
        \item[(3)] $\rho_1(\Sigma^2) \le \rho_1(\Sigma) \le \rho_2(\Sigma) \le \rho_1^2(\Sigma),$
        \item[(4)] $\rho_3^{1/4}(\Sigma) \leq \rho_1(\Sigma^2) \leq \rho_3(\Sigma)$.
    \end{enumerate} 
\end{lemma}
\begin{proof}
    Consider the eigenvalues of $\Sigma$ ordered $\lambda_1\ge \cdots \ge \lambda_d\ge 0$. First, note that $\rho_3(\Sigma) \geq 1$ trivially and $\rho_2(\Sigma) \leq \text{rank}(\Sigma)$ via direct application of the Cauchy-Schwarz inequality. Next, the inequality $\rho_3(\Sigma) \le \rho_2(\Sigma)$ can be seen from 
    \begin{align*}
        \sqrt{\rho_3(\Sigma)\over \rho_2(\Sigma)} = {\tr^2(\Sigma^2) \over \tr(\Sigma) \tr(\Sigma^3)} = {(\sum_j\lambda_j^2)^2 \over (\sum_j \lambda_j)(\sum_j\lambda_j^3)}\le 1,
    \end{align*}
    using the Cauchy-Schwarz inequality in the last step. Similarly, we have 
    \begin{align*}
        {\rho_2(\Sigma^2)\over \rho_3(\Sigma)} =  {(\sum_j\lambda_j^3)^2 \over (\sum_j \lambda_j^2)(\sum_j\lambda_j^4)}\le 1.
    \end{align*}
    To complete the proof of the first chain of inequalities, note that  $\tr(\Sigma^4) \le \lambda_1 \tr(\Sigma^3)$, and thus
    \begin{align*}
        \rho_2(\Sigma^2) \ge {\tr^2(\Sigma^2) \over \tr(\Sigma^3) \lambda_1} = {\tr^{3/2}(\Sigma^2) \over \tr(\Sigma^3)} \sqrt{\tr(\Sigma^2) \over \lambda_1^2} \ge \sqrt{\rho_3(\Sigma)}.
    \end{align*}
For the second set of inequalities, first use the fact that $\tr(\Sigma^3) \le \tr(\Sigma^2) \|\Sigma\|_\op$ to obtain
    \[
        \rho_3(\Sigma) \ge {\tr(\Sigma^2) \over \|\Sigma\|_\op^2} = {\sum_j \lambda_j^2 \over \lambda_1^2} \ge {(\sum_j\lambda_j)^2\over d\lambda_1^2} = {\rho_1^2(\Sigma) \over d}.
    \]
    On the other hand, by $(\sum_j\lambda_j^2)^2 \le (\sum_j\lambda_j)(\sum_j\lambda_j^3)$,
    \[
        \rho_3(\Sigma) \le \sqrt{\tr^3(\Sigma) \over \tr(\Sigma^3)} \le \sqrt{\tr^3(\Sigma) \over \|\Sigma\|_\op^3} = \rho_1^{3/2}(\Sigma).
    \]
    The third chain of inequalities follows by noting that $\lambda_1^2 \le \tr(\Sigma^2) \le \lambda_1\tr(\Sigma)$. Finally, the chain of inequalities (4) follows from application of inequalities (1) and (3), thereby completing the proof.
\end{proof}

\subsection{Proof of \cref{thm_range_limit}: Gaussian Approximation for the Range-Type Test Statistic with the Population Dispersion Index Parameter}\label{app_proof_thm_range_limit}

	\begin{proof}
		Recall that $R_{(q)}$ is the $q^\text{th}$ order statistics of $R_i = \|X_i - \oX\|_2$ for $i\in [n]$. 
		Further recall that 
		\begin{equation*}
			a_n = \sqrt{2\log n},\qquad b_n = a_n- {\log\log n + \log(4\pi) \over 2a_n}.
		\end{equation*}
		Our proof consists of the following principal steps: 
		\begin{enumerate}
			\item  First, defining the random vector $Y = (Y_1, \ldots, Y_n)^\T$ via 
			\[
			Y_i := {1\over \sqrt{2 \tr(\Sigma^2)}}\left(  {n\over n-1}R_i^2 - \tr(\Sigma) \right),\qquad \forall\ i\in [n],
			\]
			we establish the limiting distributions of $a_n (Y_{(n)} - b_n)$ and $a_n (Y_{(1)} + b_n)$, and bound
			\[
					\sup_{t\in \RR} \left| \PP\left(
				Y_{(n)} - Y_{(1)} \le t
				\right) - \PP\left(
				S_{(n)} - S_{(1)}\le t
				\right) \right|
			\]
            from above.
			
			\item Secondly, we establish the ratio-consistency of $R_{(q)}$ for $\sqrt{\tr(\Sigma)}$, for each $q \in \{1, n\}$; in particular, 
			\begin{equation}\label{eq_radii_ratio_consistency}
				{R_{(q)} \over \sqrt{\tr(\Sigma)}}  =  1 + \cO_\PP\left(
				{b_n \over  \sqrt{\rho_2(\Sigma)}}
				\right) + \cO\left(1\over n\right).
			\end{equation}
			
			\item Finally, we use the ratio consistency property of \textbf{Step 2} to further bound
			\[
					\sup_{t\in \RR} \left| \PP\left(
				\bar T \le t
				\right) - \PP\left(
			 U_n \le t
				\right) \right|
			\]
			from above, from which we then establish $\bar T\distrto E+E'$.
			
		\end{enumerate}
		
		\paragraph{Proof of Step 1:}
		Recall that the spectral decomposition of $\Sigma$ is $\Sigma = U \Lambda U^{\T}$ with $\Lambda = \diag(\lambda_1, \ldots, \lambda_d)$ and $U \in \bbO^{d}$. 
		Under $\cH_0$, there exist $Z_1,\ldots, Z_n \in \RR^d$ which are \text{i.i.d.} realizations of $\cN_d(0_d, \bI_d)$ such that, by the rotational invariance of standard Gaussian random vectors,
		\begin{align}\label{def_xi}\nonumber
			Y_i  &=  {1\over \sqrt{2 \tr(\Sigma^2)}}\left( {n\over n-1} \| \Lambda^{1/2}(Z_i- \oZ) \|^2 -  \tr(\Sigma) \right)\\\nonumber
			& = \sum_{j=1}^d  {\lambda_j  \over \sqrt{2 \tr(\Sigma^2)}}  \left( {n \over n-1}(Z_{ij} - \oZ_j)^2 - 1\right)\\
			& =: \sum_{j=1}^d \xi_{ij},
		\end{align}
		where $\oZ_j := n^{-1}\sum_{i=1}^n Z_{ij}$. 
		In \cref{lem_xi}, we verify that, for any $i,i'\in [n]$ and $j\in [d]$,
		\begin{equation}\label{eq_moments_xi}
		\EE[\xi_{ij}]= 0, \qquad \Cov(\xi_{ij},\xi_{i'j}) = \frac{ \lambda^2_j}{\tr(\Sigma^2)}  1_{\{i = i'\}}.
	\end{equation}
		Moreover, observe that  $\xi_{ij}$ is independent of $\xi_{ij'}$ for any $i\in [n]$ and any $j\ne j'$.  Since
		\[
			Y_{(n)} = \max_{i \in[n]} {1\over \sqrt{ d}} \sum_{j=1}^d  \xi_{ij}\sqrt{d},
		\]
		we seek to invoke \cref{thm_CCKK} to bound 
		$
			\sup_{t\in \RR} | \PP(
				Y_{(n)} \le t
			) - \PP(
			S_{(n)} \le t
			) |$.
		To do so, we first verify the Conditions E and M in Assumptions \ref{ass_E} \& \ref{ass_M}. Since 
		\begin{equation}\label{dist_Z_Zbar}
			\sqrt{n \over n-1}(Z_{ij}-\bar Z_j) \sim \cN\left(
			0,  1
			\right),
		\end{equation}
		we know $(n/(n-1)) (Z_{ij}-\bar Z_j)^2$
		is sub-exponential, implying that $
			\EE \exp(
			|\xi_{ij}| \sqrt{d} / B_d
			) \le 2
		$
		holds for 
		\begin{equation}\label{def_B_d}
			B_d = C \sqrt{d \lambda_1^2 \over \tr(\Sigma^2)} \overset{\eqref{def_rhos}}{=} C \sqrt{d   \over \rho_1(\Sigma^2)},
		\end{equation}
		where $C>0$ is an absolute constant. Moreover, 
		by \eqref{eq_moments_xi}, we have 
		\[
			{1\over d}\sum_{j=1}^d \EE\left[ 
			 d 	 ~ \xi_{ij}^2 
			\right]  = \sum_{j=1}^d  {\lambda_j^2 \over \tr(\Sigma^2)}= 1 
		\]
		and, by \eqref{def_xi},
		\begin{align*}
			{1\over d}\sum_{j=1}^d \EE\left[ 
			d^2  \xi_{ij}^4
			\right] & ~ \lesssim ~ d \sum_{j=1}^d { \lambda_j^4  \over 4 \tr^2(\Sigma^2)}\EE\left[
			\left({n\over n-1}\right)^4(Z_{ij}-\bar Z_j)^8  + 1		\right]\\
			&~ \lesssim~  {d ~ \tr(\Sigma^4) \over \tr^2(\Sigma^2)} &&\text{by \eqref{dist_Z_Zbar}}\\
			&~ \le ~ B_d^2 {\tr(\Sigma^4) \over \lambda_1^2 \tr(\Sigma^2)} &&\text{by \eqref{def_B_d}}\\
			&~  \le~  B_d^2.
		\end{align*}
		Therefore, invoking \cref{thm_CCKK} with $p = n$, $N = d$,  $X_{ij} = \xi_{ij}\sqrt{d}$, $b_1 \asymp b_2 \asymp 1$, and $B_N = B_d$ as per \eqref{def_B_d} gives
		\begin{equation}\label{clt_Yn}
			\sup_{t\in \RR} \left| \PP\left(
			Y_{(n)} \le t
			\right) - \PP\left(
			S_{(n)} \le t
			\right) \right| ~ \le~  C \left(
				\log^5(nd)  \over \rho_1(\Sigma^2)
			\right)^{1/4}.
		\end{equation}
		Regarding $Y_{(1)}$, since $Y_{(1)} = -  \max_{i \in[n]} (-Y_i)$ and the preceding results apply to  $(-\xi_{ij})$ as well, we also have 
		\begin{equation}\label{clt_Y1}
			\sup_{t\in \RR} \left| \PP\left(
			Y_{(1)} \le t
			\right) - \PP\left(
			S_{(1)} \le t
			\right) \right| ~ \le ~  C \left(
			\log^5(nd)  \over \rho_1(\Sigma^2)
			\right)^{1/4}.
		\end{equation}
		Furthermore,  observe that 
		\[
			Y_{(n)} - Y_{(1)} = \max_{i,j \in [n]} (Y_i - Y_j) = \max_{i\ne j \in [n]} (Y_i - Y_j) =  \max_{i\ne j \in [n]}{1\over \sqrt{d}}\sum_{t=1}^d (\xi_{it} - \xi_{jt}) \sqrt{d}.
		\]
		By repeating the same arguments in the preceding in conjunction with the triangle inequality, one can verify that both Conditions E and M are satisfied by $(\xi_{it} - \xi_{jt}) \sqrt{d}$ for all $i\ne j\in [n]$ and $t\in [d]$, with $b_1 \asymp b_2 \asymp 1$ and $B_d$ as per \eqref{def_B_d}, and that these variates are independent across $t \in [d]$. Invoking  \cref{thm_CCKK} again yields 
		\begin{equation}\label{clt_Range}
					\sup_{t\in \RR} \left| \PP\left(
				Y_{(n)} - Y_{(1)} \le t
				\right) - \PP\left(
				S_{(n)} - S_{(1)}\le t
				\right) \right| ~ \le~  C \left(
				\log^5(n^2 d)  \over \rho_1(\Sigma^2)
				\right)^{1/4}.
		\end{equation}
	

		\paragraph{Proof of Step 2:} We only present the proof for the case of $q=n$ as the same arguments can be used to prove  the $q = 1$ case. Define the event 
		\[
		\cE_{(n)} := \left\{
		|Y_{(n)}| \le   2\sqrt{\log n}
		\right\}.
		\] 
		Note  that \eqref{clt_Yn} and a standard tail-bound for the maximum of centered $n$-dimensional Gaussian random vectors entail that 
		\begin{align}\label{bd_En_comp}\nonumber
		 	\PP\left(\cE_{(n)}^c\right) &=  \PP\left(
		 	|Y_{(n)}|  >  2\sqrt{\log n}
		 	\right)\\ \nonumber
		 	 &\le  \PP\left(
		 	S_{(n)}>  2\sqrt{\log n}
		 	\right) +\PP\left(
		 	S_{(n)} <  -2\sqrt{\log n}
		 	\right) +   C \left(
		 	\log^5(nd)  \over \rho_1(\Sigma^2)
		 	\right)^{1/4} \\ 
		 	& \le  {2\over n}+   C \left(
		 	\log^5(nd)  \over \rho_1(\Sigma^2)
		 	\right)^{1/4}.
	 	\end{align} 
 		Since $b_n \le \sqrt{2\log n}$, when the event $\cE_{(n)}$ holds, we have
		\begin{equation}\label{bd_Yn}
			|Y_{(n)} -b_n| =  \left| {n\over n-1} \frac{R^2_{(n)}}{\sqrt{2 \tr(\Sigma^2)}} - \beta_n\right|  \le 4\sqrt{\log n}, 
		\end{equation}
		where  
		\begin{equation}\label{def_betan}
		\beta_n := \frac{\tr(\Sigma)}{\sqrt{2 \tr(\Sigma^2)}} +  b_n = \frac{\tr(\Sigma)}{\sqrt{2 \tr(\Sigma^2)}}  \left(1 + o(1)\right).
		\end{equation}
		The last step is due to $b_n \le \sqrt{2\log n}$, condition \eqref{cond_Sigma}, \cref{lem_ranks}, and
		\begin{equation}\label{lb_rho_2}
			\rho_2(\Sigma) \ge \rho_2(\Sigma^2) \geq \rho_1(\Sigma^2) = \omega(b_n^2).
		\end{equation}
		We proceed to work under the event $\cE_{(n)}$ since \eqref{bd_En_comp} entails that it holds with probability converging to one as $n \to \i$. 
		A Taylor expansion for the square-root function at  $nR^2_{(n)}/[(n-1)\sqrt{2 \tr(\Sigma^2)}]$ about $\beta_n$ is given by  
		\begin{align*} \nonumber 
			\sqrt{{n\over n-1}} \frac{R_{(n)}}{(2 \tr(\Sigma^2))^{1/4}} &= \sqrt{\beta_n} + \frac{1}{2 \sqrt{\beta_n}}  \left({n\over n-1}\frac{R^2_{(n)}}{\sqrt{2 \tr(\Sigma^2)}} - \beta_n \right)\\
			&\quad - \frac{1}{8} \wt \beta_n^{-3/2} \left({n\over n-1}\frac{R^2_{(n)}}{\sqrt{2 \tr(\Sigma^2)}} - \beta_n\right)^2, 
		\end{align*}
		where, for some $t \in [0,1]$ and by using  \eqref{bd_Yn} and \eqref{def_betan},
		$$
		\wt{\beta}_n = \beta_n + t \left({n\over n-1}\frac{R^2_{(n)}}{\sqrt{2 \tr(\Sigma^2)}} - \beta_n\right)  = \cO(\beta_n).
		$$ 
		By using \eqref{bd_Yn} and  \eqref{lb_rho_2}  again, we further have that  
		\begin{align*}
			\sqrt{{n\over n-1}} \frac{R_{(n)}}{(2 \tr(\Sigma^2))^{1/4}} &=   \sqrt{\beta_n} + \cO\left(
			\sqrt{\log n \over \beta_n}
			\right)  
		\end{align*}
		such that
		\begin{align}\label{bd_ratio_Rn}\nonumber
			\sqrt{n\over n-1}R_{(n)}  &= \sqrt{\tr(\Sigma) - b_n \sqrt{2\tr(\Sigma^2)}} + \cO\left(
			\sqrt{\log n} \sqrt{\tr(\Sigma^2) \over  \tr(\Sigma)}
			\right) \\ 
			&= \sqrt{\tr(\Sigma)} + \cO\left(
			\sqrt{\log n} \sqrt{\tr(\Sigma^2) \over \tr(\Sigma)}
			\right)
		\end{align}
		by Talyor expansion and $b_n \le \sqrt{2\log n}$.
		By similar arguments,  we can show that 
		\begin{equation}\label{bd_ratio_R1}
			\sqrt{n\over n-1}R_{(1)}   = \sqrt{\tr(\Sigma)} + \cO\left(
			\sqrt{\log n}\sqrt{\tr(\Sigma^2) \over \tr(\Sigma)}
			\right)
		\end{equation}
		under the event 
		$
		\cE_{(1)} := \{
		|Y_{(1)}| \le   2\sqrt{\log n}
		\},
		$
		which, by similar arguments to that used in \eqref{bd_En_comp},  satisfies
		\begin{equation}\label{bd_E1_comp}
			\PP(\cE_{(1)}^c) \le {2\over n} +  C\left(\log^5(nd)  \over  \rho_1(\Sigma^2)\right)^{1/4}.
		\end{equation}
        Thus, for every $q \in [n]$,
        \begin{equation}
				{R_{(q)} \over \sqrt{\tr(\Sigma)}}  =  1 + \cO_\PP\left(
				{b_n\over  \sqrt{\rho_2(\Sigma)}}
				\right) + \cO\left(1\over n\right).
			\end{equation}

		\paragraph{Proof of Step 3:}
		We next relate the distribution of $Y_{(n)} - Y_{(1)}$ to that of $\bar T$.  Define 
		\begin{equation}\label{def_zeta}
			\zeta_n  := {n-1 \over n}{2\sqrt{\tr(\Sigma)} \over R_{(n)} + R_{(1)}}.
		\end{equation}   
		Note that \eqref{bd_ratio_Rn} and \eqref{bd_ratio_R1} gives that, under the event $\cE_{(n)} \cap \cE_{(1)}$, 
		\begin{equation}\label{rate_zeta_n}
			\left|{1\over \zeta_n} - 1\right| =  \left|{n \over n-1}{R_{(n)} + R_{(1)} \over 2\sqrt{\tr(\Sigma)}}-1\right| =  \cO\left(
			{b_n \over \sqrt{\rho_2(\Sigma)}} + {1\over n}\right) =: \eta_n.
		\end{equation}
		By  definition,  for any $t_+\ge 0$,
		\begin{align}\label{eq_distr_range_Tn}\nonumber
			\PP\left(
			Y_{(n)} - Y_{(1)} \le t_+\right) & = \PP\left(
			{n \over n-1}  \frac{R_{(n)}^2 - R_{(1)}^2}{\sqrt{2 \tr(\Sigma^2)}}\le t_+\right) \\\nonumber
			&=  \PP\left(
		 2\Delta^{-1/2}\left(R_{(n)}-R_{(1)}\right)  {n \over n-1}  \frac{R_{(n)} + R_{(1)}}{2\sqrt{ \tr(\Sigma)}}\le t_+\right)\\\nonumber
		 &= \PP\left(
		 2a_n \Delta^{-1/2}\left(R_{(n)}-R_{(1)}\right)	 \le   a_n  \zeta_n    t_+ \right)\\
		 &= \PP\left(
		 \bar T	 \le  a_n   \zeta_n  t_+ - 2a_nb_n\right).
		\end{align}
		Recall $U_n$ from \eqref{def_Un}.
		It then follows that, for all $t\in \RR$,
		\begin{align*}
		 	&\PP\left(
			\bar T	 \le  t \right) - \PP\left(
			U_n \le t 
			\right) \\
			&= \PP\left(
			Y_{(n)} - Y_{(1)} \le  {t + 2a_nb_n\over a_n \zeta_n}\right) -  
			\PP\left(
			a_n(S_{(n)} - S_{(1)} - 2b_n) \le  t\right)  &&\text{by \eqref{eq_distr_range_Tn}}\\
            &\le  \PP\left(
			Y_{(n)} - Y_{(1)} \le  {t + 2a_nb_n\over a_n}(1+\eta_n)\right) + \PP\left(
			\cE_{(n)}^c  \cup 
			\cE_{(1)}^c 
			\right) &&\text{by \eqref{rate_zeta_n}}\\
            &\quad -  
			\PP\left(
			S_{(n)} - S_{(1)} \le  {t+2a_nb_n\over a_n}\right)\\
			&\le    \PP\left(
			S_{(n)} - S_{(1)} \le  {t + 2a_nb_n\over a_n}(1+\eta_n)\right) -  
			\PP\left(
			S_{(n)} - S_{(1)} \le  {t+2a_nb_n\over a_n}\right) \\
			& \quad +  C\left(
			\log^5(n^2 d)  \over\rho_1(\Sigma^2)
			\right)^{1/4} + \PP\left(
			\cE_{(n)}^c  \cup 
			\cE_{(1)}^c 
			\right)  &&\text{by \eqref{clt_Range}}.
		\end{align*} 
		Note that
		$
		S_{(n)} - S_{(1)} = \max_{i\ne j} (S_i - S_j)
		$
		and $S_i - S_j \sim \cN(0, 2)$. 
		Invoking  \cref{lem_anti_ratio} with $t_0 = C\sqrt{\log n}$ and $\xi = 1/(1+\eta_n)$ yields 
		\begin{align*}
			&\sup_{t\in \RR} \left| \PP\left(
			S_{(n)} - S_{(1)} \le  {t + 2a_nb_n\over a_n}(1+\eta_n)\right) -  
			\PP\left(
			S_{(n)} - S_{(1)} \le  {t+2a_nb_n\over a_n}\right)\right|\\
			&   \le ~  C\eta_n \log n  + 2\exp\left(
			- {C' \log n}
			\right) .
		\end{align*}
		Together with \eqref{rate_zeta_n}, \eqref{bd_En_comp}, and \eqref{bd_E1_comp}, we hence obtain that for all $t\in \RR$,
		\[
			  \PP\left(
			\bar T	 \le  t \right) - \PP\left(
			U_n \le t 
			\right)   = \cO\left(\left(
			\log^5(n^2 d)  \over \rho_1(\Sigma^2)
			\right)^{1/4} +  
			{\sqrt{\log^3n \over  \rho_2(\Sigma)}} + {\log n\over  n}\right).
		\]
        By symmetric arguments, we also have 
        \begin{align*}
            & \PP\left(
			U_n \le t 
			\right) - \PP\left(
			\bar T	 \le  t \right) \\
            &= \PP\left(
			U_n \le t 
			\right) - 1 + \PP\left(
			Y_{(n)} - Y_{(1)} >  {t + 2a_nb_n\over a_n \zeta_n}\right) \\
            &\le \PP\left(
			U_n \le t 
			\right) -   \PP\left(
			Y_{(n)} - Y_{(1)} \le  {t + 2a_nb_n\over a_n}(1-\eta_n)\right) +  \PP\left(
			\cE_{(n)}^c  \cup 
			\cE_{(1)}^c 
			\right).
        \end{align*}
        Similar arguments with $\xi = 1/(1-\eta_n)$
        yield the same upper bound for $\PP(
			U_n\le t ) - \PP(
			\bar T	 \le  t)$.  
		Using \eqref{cond_Sigma} and \eqref{lb_rho_2} simplifies the expression and completes the proof of \eqref{rate_CLT_range}.

		Finally,  to prove the claim $\bar T \distrto E + E'$,  classical extreme value theory for standard normal random variables (see, for instance, \cite[page 409]{Gumbel} and \cite[page 313]{David}) yield
		\begin{equation}\label{limit_distr_S}
			a_n \begin{pmatrix}
				S_{(n)} - b_n \vspace{1.5mm}\\ S_{(1)} + b_n
			\end{pmatrix} 
			\distrto  \begin{pmatrix}
				E \vspace{1.5mm}\\ -E'
			\end{pmatrix}
		\end{equation}
		where  the random variables $E$ and $E'$ satisfy $E \stackrel{\rm d}{=} E'$, $E \independent E'$, and  
		\begin{equation*}
			\PP\left\{
			E \le x
			\right\} = \exp(-\exp(-x)),\qquad -\i < x < \i.
		\end{equation*}
		Since \eqref{limit_distr_S} ensures that
		\[
				\sup_{t\in \RR} \left|
				\PP\left(
				U_n \le t 
				\right) - \PP\left(
				E+E' \le t 
				\right)
				\right| = o(1),
		\]
		the proof is complete.
	\end{proof}

\subsubsection{A Moment Calculation Lemma Used in the Proof of \cref{thm_range_limit}}

    The following lemma provides the first two moments of the random vectors $\xi_{\cdot j} = (\xi_{1j}, \ldots, \xi_{nj})^\T \in \RR^n$, for $j\in [d]$, as defined in \eqref{def_xi}.
	\begin{lemma}\label{lem_xi}
		For each $j\in [d]$, the random vector $\xi_{\cdot j} = (\xi_{1j}, \ldots, \xi_{nj})^\T \in \RR^n$ defined in \eqref{def_xi} satisfies
		\[
			\EE[\xi_{\cdot j}] = 0_n \ \ \ \ \text{ and } \ \ \ \ \Cov(\xi_{\cdot j}) = \frac{ \lambda^2_j}{\tr(\Sigma^2)} \bI_n.
		\] 
	\end{lemma}
	\begin{proof}
		Let $W_i$,  for $i = 1,\ldots,n$, be i.i.d. samples of $\cN(0,1)$ and write $\oW = n^{-1}\sum_{i = 1}^n W_i$. 
		For any $j\in [d]$, to show $\EE[\xi_{\cdot j}] = 0$, it suffices to prove that for any 
	  $i\in [n]$,
		\begin{align*}
		 \EE [Z_{ij} - \oZ_j]^2  =  \frac{n - 1}{n} .
		\end{align*}
	   This follows from the  fact that 
		\begin{equation}\label{distr_W}
		Z_{ij} - \oZ_j \stackrel{\rm d}{=} W_i - \oW = (1 - \frac{1}{n})W_i - \frac{1}{n}\sum_{k \neq i} W_k \sim \cN\left(0, \frac{n - 1}{n}\right).
		\end{equation}
		Regarding the covariance, pick any $j\in [d]$ and $i,i'\in [n]$. We have		
		\begin{equation*}
			\begin{split}
				  \Cov(\xi_{ij}, \xi_{i'j}) & = \Cov\left( \frac{\lambda_j[(Z_{ij} - \oZ_j)^2 - \frac{n - 1}{n}]}{\frac{n-1}{n} \sqrt{2 \tr(\Sigma^2)}}, \frac{\lambda_j[(Z_{i'j} - \oZ_j)^2 - \frac{n - 1}{n}]}{\frac{n-1}{n} \sqrt{2 \tr(\Sigma^2)}}\right) \\ & = \frac{n^2\lambda_j^2 }{2(n - 1)^2  \tr(\Sigma^2)}    \Cov\left([Z_{ij} - \oZ_j]^2, [Z_{i'j} - \oZ_j]^2\right) \\ & =\frac{n^2\lambda_j^2 }{2(n - 1)^2  \tr(\Sigma^2)}    \Cov\left([W_{i} - \oW]^2, [W_{i'}-\oW]^2\right).
			\end{split}
		\end{equation*}
		Since \eqref{distr_W} implies 
		\begin{align*}
		\Var([W_i - \oW]^2) &= \EE\left[(W_i - \oW)^4\right] - \left(\EE\left[(W_i - \oW)^2\right]\right)^2\\
		& = 3 \left(\frac{n - 1}{n}\right)^2 -\left (\frac{n - 1}{n}\right)^2\\
		& = 2 \left(\frac{n - 1}{n}\right)^2,
	\end{align*}
	we obtain
		\[
			\Cov(\xi_{ij}, \xi_{ij}) = \frac{ \lambda_j^2 }{   \tr(\Sigma^2)},\quad \text{for any }i\in [n], j\in [d].
		\]
		Regarding the off-diagonal terms of $\Cov(\xi_{\cdot j})$, notice that , for any $i\ne i'\in [n]$,  
		\begin{align}\label{eq_cov_offdiag_type_I}\nonumber
			 &\Cov\left([W_{i} - \oW]^2, [W_{i'}-\oW]^2\right)\\ \nonumber
			 &~ =   \Cov(W_1^2 + \oW^2 - 2 W_1 \oW, \ W_2^2 + \oW^2 - 2 W_2 \oW)\\
			 & \stackrel{\text{i.i.d.}}
			= 2\Cov(W_1^2, \oW^2) - 4 \Cov(W_1^2, W_2 \oW) - 4 \Cov(\oW^2, W_2 \oW) + 4 \Cov(W_1 \oW, W_2 \oW).
		\end{align}
		The first term of the preceding display is twice of
		\begin{equation}\label{eq_term_1}
			\begin{split}
				\Cov(W_1^2, \oW^2) & ~=  \Cov\left(W_1^2, \ {1\over n^2} \sum_{k = 1}^n W_k^2   +{1\over n^2}\sum_{k \neq l} W_k W_l \right) \\ & =   {1\over n^2} \Cov(W_1^2, W_1^2) \ +  {1\over n^2} \sum_{k \neq 1}\left[\EE W_1^3 W_k - (\EE W_1^2) (\EE W_1 W_k) \right] \\ & =  \frac{2}{n^2}.
			\end{split}
		\end{equation}
		The second term in \eqref{eq_cov_offdiag_type_I} satisfies
		\begin{equation}\label{eq_term_2}
			\begin{split}
				  \Cov(W_1^2, W_2~ \oW) &~  =   \Cov\left(W_1^2, {1\over n}W_2 \sum_{k = 1}^n W_k\right) \\
				  &=  {1\over n} \Cov(W_1^2, W_1 W_2) \\ & ~=   {1\over n} \left[\EE W_1^3 W_2 - (\EE W_1^2) (\EE W_1 W_2)\right] \\ & =   0.
			\end{split}
		\end{equation}
		Regarding the third term in \eqref{eq_cov_offdiag_type_I}, we find that 
		\begin{equation*}
			\begin{split}
				-4 \Cov(\oW^2, W_1 \oW) & = - {4\over n^3} \Cov\left( \sum_{i=1}^n W_i^2 + \sum_{i \neq k} W_i W_k, \sum_{i=1}^n W_1 W_i\right) \\ & = -{4\over n^3} \left[\sum_{i,k=1}^n \Cov(W_i^2, W_1 W_k) + \sum_{i \neq k} \sum_{j = 1}^n \Cov(W_i W_k, W_1 W_j)\right].
			\end{split}
		\end{equation*}
		Since 
		\begin{align*}
			\sum_{i,k=1}^n \Cov(W_i^2, W_1 W_k)  &= \sum_{i=1}^n \Cov(W_i^2, W_1 W_i) + \sum_{i\ne k}  \Cov(W_i^2, W_1 W_k) = \Cov(W_1^2, W_1^2) = 2,
		\end{align*}
		and
		\begin{align*}
			\sum_{i \neq k} \sum_{j = 1}^n \Cov(W_i W_k, W_1 W_j) &= 
			\sum_{i \neq k} \sum_{j = 1}^n\left[
			\EE(W_1 W_k W_i W_j) - (\EE W_i W_k) (\EE W_1 W_j)
			\right]\\
			&= \sum_{i \neq k} \sum_{j = 1}^n 
			\EE (W_1 W_k W_i W_j)\\
			&= \sum_{k\ne 1}\sum_{j = 1}^n 
			\EE (W_1^2 W_k  W_j)+\sum_{i\ne 1} \sum_{k=1,k\ne i}^n \sum_{j = 1}^n 
			\EE (W_1 W_k W_i W_j)\\
			&= \sum_{k\ne 1}   
			\EE (W_1^2 W_k^2)+\sum_{i\ne 1} 
			\EE (W_1^2 W_i^2)\\
			&= 2(n-1),
		\end{align*} 
		we have 
		\begin{equation}\label{eq_term_3}
			-4 \Cov(\oW^2, W_1 \oW)  = -4 n^{-3}(2 + 2 (n - 1)) = -8 n^{-2}.
		\end{equation}  
		Finally,  the   last term in \eqref{eq_cov_offdiag_type_I} satisfies
		\begin{equation}\label{eq_term_4}
			\begin{split}
				4 \Cov(W_1 \oW, W_2 \oW) & =  -4 n^{-2} \sum_{i, j=1}^n \Cov(W_1 W_i, W_2 W_j) \\ & = -4 n^{-2} \sum_{i, j=1}^n \left[ \EE W_1 W_i W_2 W_j - (\EE W_1 W_i) (\EE W_2 W_j)\right]\\
				&= -4 n^{-2} \left[
				 \EE W_1^2W_2^2 - (\EE W_1)^2 (\EE W_2)^2
				\right]\\
				&= -4n^{-2}.
			\end{split}
		\end{equation}
		 Collecting \eqref{eq_term_1} -- \eqref{eq_term_4} yields
		\begin{equation*}
			\begin{split}
				\Cov\left([W_{i} - \oW]^2, [W_{i'}-\oW]^2\right)&  = \frac{4}{n^2} + 0 - \frac{8}{n^2} + \frac{4}{n^2}  = 0,
			\end{split}
		\end{equation*}
		completing the proof. 
	\end{proof}

\subsubsection{Auxiliary Results on Gaussian Approximation for the Proof of \cref{thm_range_limit}}
    Let $X_1,\ldots, X_N$ be independent random vectors in $\RR^p$. Assume they satisfy the following two conditions.
	
	\begin{ass}[Condition E]\label{ass_E}
		For all $i = 1,\ldots, N$ and $j = 1,\ldots, p$, we have 
		\[
			\EE[\exp(|X_{ij}| / B_N)] \le 2,
		\]
		where $B_N$ is some deterministic sequence that can diverge to infinity.
	\end{ass}

	\begin{ass}[Condition M]\label{ass_M}
		For all $j=1,\ldots,p$, we have 
		\[
			b_1^2 \le {1\over N}\sum_{i=1}^N \EE[X_{ij}^2],\qquad {1\over N}\sum_{i=1}^N \EE[X_{ij}^4] \le B_N^2 b_2^2
		\]
		for some strictly positive constants $b_1 \le b_2$.
	\end{ass}
	Let $a\in \RR^p$ be any deterministic sequence. 
	Further, let $c_{1-\alpha}^G$ be the $(1-\alpha)$th quantile of 
	\[
	\max_{j \in [p]} ~  (S_j + a_j),  
	\]
	where $S = (S_1,\ldots, S_p)^\T$ is a centered Gaussian random vector in $\RR^p$ with covariance matrix 
	\[
		\Cov(S) = {1\over N}\sum_{i=1}^N \Cov(X_i).
	\]
	The following theorem provides a non-asymptotic upper bound on the error of  
	\[
		\left| \PP\left\{
			\max_{j \in [p]} {1\over \sqrt N} \sum_{i=1}^N (X_{ij} + a_j)  > c_{1-\alpha}^G 
		\right\} -  \alpha \right|.
	\] 
	\begin{theorem}[Theorem 2.1 \cite{CCKK}]\label{thm_CCKK}
		Suppose that \cref{ass_E} and \cref{ass_M} are satisfied. Then 
		\[
			\left| \PP\biggl\{
			\max_{j \in [p]} {1\over \sqrt N} \sum_{i=1}^N (X_{ij} + a_j)  > c_{1-\alpha}^G 
			\biggr\} -  \alpha \right|  ~ \le ~  C \left(
			{B_N^2 \log^5(pN) \over N}
			\right)^{1/4}
		\]
		where $C$ is a constant depending only on $b_1$ and $b_2$.
	\end{theorem}

	The following lemma establishes anti-concentration of a centered Gaussian random vector. It is proven in \cite{CCK2017}. For a vector $v\in \RR^p$ and a scalar $r\in \RR$, we write $v + r$ for the vector with its $j$th entry equal to $v_j + r$. 
	
	\begin{lemma}[Gaussian Anti-Concentration Inequality]\label{lem_anti}
		Let $S = (S_1, \ldots, S_p)^\T$ be a centered Gaussian random vector in $\RR^p$ with $p\ge 2$ such that $\EE[S_j^2] \ge b$ for all $j = 1,\ldots, p$ and some constant $b>0$. Then for every $s \in \RR^p$ and $t>0$, 
		\[
			\PP(S \le s + t) - \PP(S \le s) \le C ~ t\sqrt{\log p},
		\]
		where $C$ is a constant depending only on $b$.
	\end{lemma}

    	For the univariate case ($p=1$), we have the following simple result. 
	\begin{lemma}\label{lem_anti_univ}
		Let $S \sim \cN(0, \sigma^2)$ for some $\sigma  > 0$. Then for every $s\in \RR$ and $t>0$, 
		\[
			\PP(S \le s + t) - \PP(S \le s) \le {t\over \sigma \sqrt{2\pi}}.
		\]
	\end{lemma}

	The following lemma establishes a comparison inequality between the maximum of two centered Gaussian random vectors whose respective covariance matrices differ only by a multiplicative constant. It improves upon  \cite[Proposition 2.1]{CCKK}  and \cite[Theorem 2]{CCK2015} for a comparison of this particular type.

    \begin{lemma}\label{lem_anti_ratio}
        Let $S = (S_1,\ldots, S_p)^\T$ be a centered Gaussian random vector in $\RR^p$ with covariance matrix $\Sigma$  such that  $\Sigma_{jj} \ge c$ for  all $j \in [p]$ and some constant $c>0$. Then, for any $t_0>0$, $t\in \RR$, and $\xi > 0$, one has  
        \[
         \left| \PP\left(
        \max_{j \in[p]} S_j \le  {t\over \xi}\right) -  
        \PP\left(
        \max_{j \in[p]} S_j   \le  t\right)\right|  \le C {|1-\xi| \over \xi} t_0 \sqrt{1+\log p} + 2p \exp\left(-{t_0^2 \over C'(\xi \vee 1)^2}\right),
        \]
        where the constants $C,C'$ depend only on $c$.
    \end{lemma}
    \begin{proof}
       To establish the result, it suffices to upper-bound the following: 
        \begin{align*}
            &\rI = \sup_{|t| \le t_0}\left| \PP\left(
            \max_{j \in[p]} S_j \le  {t \over  \xi}\right) -  
            \PP\left(
            \max_{j \in[p]} S_j   \le  t\right)\right|,\\
            &\rII = \sup_{|t| > t_0}\left| \PP\left(
            \max_{j \in[p]} S_j \le  {t \over  \xi}\right) -  
            \PP\left(
            \max_{j \in[p]} S_j   \le  t\right)\right|.
        \end{align*}
        To bound $\rI$, application of the anti-concentration property in \cref{lem_anti} for $p\ge 2$ or \cref{lem_anti_univ} for $p =1$ gives 
        \[
        \rI \le \sup_{|t| \le t_0} C~{t\over \xi}|1-\xi|\sqrt{1 + \log p} \le C' {|1-\xi|\over \xi} t_0\sqrt{1 + \log p}.
        \]
        Regarding $\rII$, note that 
        \begin{align*}
            \rII  &\le  \sup_{t > t_0} \PP\left(
            \max_{j \in[p]} S_j > {t\over \xi \vee 1}
            \right) +  \sup_{t <  -t_0} \PP\left(
            \max_{j \in[p]} S_j \le {t\over \xi \wedge 1}
            \right)\\
            &\le \sup_{t > t_0} \PP\left(
            \max_{j \in[p]} S_j > {t\over \xi \vee 1}
            \right) + \sup_{t \ge   t_0} \PP\left(
            - \max_{j \in[p]} S_j \ge {t\over \xi \wedge 1}
            \right)\\
            &\le 2 \sup_{t > t_0} \PP\left(
            \max_{j \in[p]} S_j > {t\over \xi \vee 1}
            \right)  \\
            &\le 2p \exp\left(-{t_0^2 \over 2(\xi \vee 1)^2}\right).
        \end{align*} 
        Combining the bounds of $\rI$ and $\rII$ completes the proof. 
    \end{proof}

\subsection{Proof of \cref{prop_Delta_Null}: Ratio Consistency of the Dispersion Index Estimator under the Null Hypothesis}\label{app_proof_prop_Delta_Null}

 \begin{proof}
    Recall that 
    $$
     {\Delta \over \wh \Delta}= {\tr(\wh \Sigma) \over \tr(\Sigma)}  {\tr(\Sigma^2)\over \widehat{\tr(\Sigma^2)}}. 
    $$ 	 	
    In bounding the relative error in $\tr(\wh \Sigma) / \tr(\Sigma)$, Chebyshev's inequality in conjunction with the facts (see \cite{HimenoYamada}) 
    $\EE[\tr(\wh \Sigma)] = \tr(\Sigma)$ and 
    $$\Var\left(\tr(\wh \Sigma)\right) = \EE\left(\tr^2(\wh \Sigma)\right) - \left[\EE\left(\tr(\wh \Sigma)\right)\right]^2 = \frac{2}{n - 1} \tr(\Sigma^2),
    $$
    yields that for all $t > 0$,
    \begin{align*}
    \PP\left\{\left| {\tr(\wh \Sigma) \over \tr(\Sigma)} - 1\right| \ge {t\over \sqrt{\rho_2(\Sigma)}}
    \right\}  \le  {2 \over (n-1) ~ t^2}.
    \end{align*}
    To control $\widehat{\tr(\Sigma^2)}/\tr(\Sigma^2)$, we first note that $\EE[\widehat{\tr(\Sigma^2)}] = \tr(\Sigma^2)$ and $$\Var\left(\wh{\tr(\Sigma^2)}\right) = \cO\left(\frac{\tr(\Sigma^4)}{n} + \frac{\tr^2(\Sigma^2)}{n^2}\right)$$ from Proposition A.2 of \cite{Chen2010}.  Chebyshev's inequality  then entails that for all $t > 0$,
    $$ 
    \PP\left\{\left| 	\frac{\wh{\tr(\Sigma^2)}}{\tr(\Sigma^2)} - 1\right| \ge {t\over \sqrt{\rho_2(\Sigma^2)}} + {t\over \sqrt n}
    \right\}  = \cO\left(
    {1\over n  t^2}
    \right).
    $$
    The preceding two upper-tail bounds in conjunction with the fact that $|a^2-1| \ge |a-1|$ for $a\ge 0$ entail that, for all $t\in (0,1)$,   
    \begin{align*}
        \PP\left\{
        \left|
            \sqrt{\Delta \over \wh \Delta} - 1 
        \right| \ge {t\over \sqrt n} + {t\over \sqrt{\rho_2(\Sigma^2)} }+ {t\over \sqrt{\rho_2(\Sigma)}}
        \right\}  = \cO\left(
        {1\over n  t^2}
        \right).
    \end{align*}  
    Finally, using   $\rho_2(\Sigma) \ge \rho_2(\Sigma^2)$  from  \cref{lem_ranks} completes the proof. 
 \end{proof}

\subsection{Proof of \cref{thm_range_stat_limit}}\label{app_proof_thm_range_stat_limit}
	
\begin{proof}
    By definition,  for any $t\in \RR$, 
    \[
        \PP(T \le t) = \PP\left(
        \bar T \le \sqrt{\wh \Delta / \Delta}~\left(t + 2a_nb_n\right) - 2a_nb_n
        \right).
    \]
    For some constant $C>0$, let 
    \[
        \cE_{\Delta} = \left\{
         \Bigl| 1- \sqrt{\Delta / \wh\Delta} \right| \le \epsilon_n
        \Bigr\},\quad \text{with }\quad   \epsilon_n = {C\over \sqrt{\rho_2(\Sigma^2)}} + {C\over \sqrt n}.
    \]
    Invoking \cref{prop_Delta_Null} with $t = \sqrt{n}$ yields $\PP(\cE_\Delta^c) =\cO(1/n)$.
    By repeating the arguments in the proof of \cref{thm_range_limit}, we find that 
    \begin{align*}
        &\PP(T \le t) - \PP(U_n \le t) \\
        &~ \le ~  
        \PP\left(
        \bar T \le {1 \over 1-\epsilon_n}\left(t + 2a_nb_n\right) - 2a_nb_n
        \right) - \PP(U_n \le t)  + \PP(\cE^c_\Delta)\\
        &~ \le ~ C \left(
        \log^5(nd)  \over \rho_1(\Sigma^2)
        \right)^{1/4} + \PP(\cE^c_\Delta)  &&\text{by \cref{thm_range_limit}}  \\
        &\qquad +  
        \PP\left(
        U_n \le  {1 \over 1-\epsilon_n} \left(t + 2a_nb_n\right) - 2a_nb_n
        \right) - \PP(U_n \le t) \\
        &~ \le ~ C \left(
        \log^5(nd)  \over \rho_1(\Sigma^2)
        \right)^{1/4}  + {C'\over n}\\
        &\qquad +  
        \PP\left(
        S_{(n)} - S_{(1)} \le {1 \over 1-\epsilon_n}\left({t\over a_n} + 2b_n\right)  
        \right) - \PP\left(	S_{(n)} - S_{(1)} \le {t\over a_n} + 2b_n\right)  .
    \end{align*}
    \cref{lem_anti_ratio} with $\xi = 1-\epsilon_n$ and $t_0 = C\sqrt{\log n}$ implies that, for all $t\in \RR$,
    \begin{align*}
        & 
        \PP(T \le t) - \PP(U_n \le t) \\ 
        &~~ \le  ~   C \left(
        \log^5(nd)  \over \rho_1(\Sigma^2)
        \right)^{1/4} +C'\left(  {\log n\over n}\right) +  C {\epsilon_n \over 1-\epsilon_n} \log n + 2 \exp\left(-{C'\log n }\right)\\
        &~~ \le ~ C \left(
        \log^5(nd)  \over \rho_1(\Sigma^2)
        \right)^{1/4} +  C {\log n\over \sqrt{\rho_2(\Sigma^2)}} + C'{\log n\over \sqrt{n}}. 
    \end{align*} 
    Since a symmetric argument proves the upper-bound for the reverse direction, using $\rho_2(\Sigma^2) \ge  \rho_1(\Sigma^2)$ from \eqref{cond_rhos} completes the proof. 
\end{proof}

\subsection{Proof of \cref{thm_range_typeI}}\label{app_proof_thm_range_typeI}

    \begin{proof}  
        For arbitrary $\alpha_0\in (0,1)$, let 
        $\wh F^{-1}_{M, n}(\alpha_0)$ be the $\alpha_0$ quantile of $M$ i.i.d. copies of $U_n$, and let $\wh F_{M,n}$ denote its associated empirical cumulative density function. Further, let $F_n^{-1}(\alpha_0)$ be the $\alpha_0$ quantile of the distribution of $U_n$, whose c.d.f. is denoted by $F_n$. To establish the result, we note that it is sufficient to bound $\left| \PP\left( T > \wh F^{-1}_{M, n}(\alpha_0) \right) - (1 - \alpha_0) \right|$ from above, for every $\alpha_0 \in (0,1)$. By the triangle inequality, we have
        \begin{align*}
            \left|  
                \PP\left(
                    T > \wh F^{-1}_{M, n}(\alpha_0) 
                \right) - (1 - \alpha_0)
            \right| 
            &\le ~ \left|  
                \PP\left(
                    T > \wh F^{-1}_{M, n}(\alpha_0) 
                \right) - \PP\left(
                    U_n > \wh F^{-1}_{M, n}(\alpha_0) 
                \right) 
            \right|\\ 
            &\qquad + \left|  
                \PP\left(
                    U_n > \wh F^{-1}_{M, n}(\alpha_0) 
                \right) - (1 - \alpha_0)
            \right|.
        \end{align*}
        The first term can be bounded by invoking \cref{thm_range_stat_limit}, while the second term equals 
        \begin{align*}
             &\left| 1 -  F_n\left(\wh F^{-1}_{M, n}(\alpha_0)\right) - (1 - \alpha_0) \right|\\ 
            &\le  \left|  F_n\left(\wh F^{-1}_{M, n}(\alpha_0)\right) -  \EE_M\left[ \wh F_{M,n}\left(\wh F^{-1}_{M, n}(\alpha_0)\right) \right]\right| + \left|\EE_M\left[ \wh F_{M,n}\left(\wh F^{-1}_{M, n}(\alpha_0)\right) \right] - \alpha_0 \right|\\
            &\le \EE_M\left[  \sup_{t\in \RR} \left| F_n(t) - \wh F_{M,n}(t) \right|\right]   + \left|\EE_M\left[ \wh F_{M,n}\left(\wh F^{-1}_{M, n}(\alpha_0)\right) - \alpha_0 \right]\right|,
        \end{align*}
        where $\EE_M$ denotes the expectation with respect to $M$ \text{i.i.d.} copies of $U_n$. By the Dvoretzky–Kiefer–Wolfowitz inequality, we know that for all $\epsilon\ge 0$,
        \begin{equation*}
            \PP\left\{
            \sup_{t\in \RR}
               \left| F_n(t) -  \wh F_{M,n}(t)\right| 
               > \epsilon
            \right\} \le 2 e^{-2M\epsilon^2}
        \end{equation*}
        which implies  
        \begin{align}\label{bd_CDF_exp}\nonumber
            \EE_M\left[  \sup_{t\in \RR} \left| F_n(t) - \wh F_{M,n}(t) \right|\right] & \le \epsilon + \int_{\epsilon}^{\i} 2e^{-2Mt^2} {\rm d} t\\\nonumber
            &\le \epsilon + {1\over 2M \epsilon }e^{-2M\epsilon^2}\\
            & \le   {2 \over \sqrt M} &&\text{by }\epsilon = 1/\sqrt{M}.
        \end{align}
        On the other hand, we know that (see, for instance, \cite{David}),
        \begin{equation}\label{bd_emp_quantile_1}
            \wh F_{M,n}\left(\wh F^{-1}_{M, n}(\alpha_0)\right) \ge \alpha_0,\qquad \text{almost surely}.
        \end{equation}
        Since $U_n$ has a probability density function, we know that, with probability one, 
        \begin{equation}\label{bd_emp_quantile_2}
            \left| \wh F_{M,n}\left(\wh F^{-1}_{M, n}(\alpha_0)\right) - \alpha_0 \right| \le {1\over M}.
        \end{equation}
        Combining \eqref{bd_CDF_exp},  \eqref{bd_emp_quantile_1}, and \eqref{bd_emp_quantile_2} and invoking \cref{thm_range_stat_limit} completes the proof.        
	   \end{proof}

	\subsection{Proof of \cref{thm_IQR}}\label{app_proof_thm_IQR}
	
	\begin{proof}
		
	The proof largely follows a similar structure to that of \cref{thm_range_limit} and in the sequel we only emphasize the differences.\\
		
	\noindent{\bf Proof of Step 1:} We distinguish between two cases depending on which condition of \eqref{effectiveRank_IQR} is satisfied.
    
    {\bf Case 1:} Suppose that $\rho_1(\Sigma^2) = \omega(n)$. In the proof of {\bf Step 1} towards proving \cref{thm_range_limit}, recall that
		\[ 
		Y_i = \sum_{j=1}^d {\lambda_j\over \sqrt{2 \tr(\Sigma^2)}}\left(  {n\over n-1}(Z_{ij}-\oZ_j)^2 -1 \right) = V_i + Q_i
		\] 
		for each $i\in [n]$, 
		where we write
		\begin{align*}
			 V_i &:= \sum_{j=1}^d {\lambda_j\over \sqrt{2\tr(\Sigma^2)}}\left(
			Z_{ij}^2 - 1
			\right),\\
			Q_i &:= \sum_{j=1}^d {\lambda_j\over \sqrt{2\tr(\Sigma^2)}}\left(
			{n \over n-1} \oZ_j^2 - {2n\over n-1}Z_{ij}\oZ_j + {1\over n-1}Z_{ij}^2
			\right).
		\end{align*} 
	Note that the $V_i$ for $i \in [n]$ are i.i.d. copies of a random variable $V$ satisfying $\EE[V] = 0$ and $\EE[V^2] = 1$. Further, we have that for all $j\in [n]$,
		\[
			\EE\left[\left(
			{\lambda_j\over \sqrt{2\tr(\Sigma^2)}}\left(
			Z_{ij}^2 - 1
			\right)
			\right)^2\right] = {\lambda_j^2 \over \tr(\Sigma^2)}
		\]
		and 
		\[
			\EE \bigg |
			{\lambda_j\over \sqrt{2\tr(\Sigma^2)}}
			\left( Z_{ij}^2 - 1
			\right)
			\bigg |^3 \le  {C\lambda_j^{3} \over (\tr(\Sigma^2))^{3/2}} \le {C\over \sqrt{\rho_1(\Sigma^2)}} {\lambda_j^2 \over \tr(\Sigma^2)}.
		\]
		Thus, by the Berry-Esseen theorem, we have 
		\begin{align*}
			\sup_{t\in \RR} \left|\PP(V \le t) -  \PP(W \le t)\right| = \cO\left(
			1 \over \sqrt{\rho_1(\Sigma^2)}
			\right),
		\end{align*}
		where $W \sim \cN(0,1)$. Moreover, since each $Q_i$ is sub-exponential with sub-exponential constant $c/n$, 
		by taking a union bound over $i\in [n]$, we have 
		\[
		\PP\left(
		\max_{i\in [n]}|Q_i| \ge C\log(n) / n
		\right) \le  n^{-1}.
		\]
		Therefore, conditions (b) and (c) in \cref{thm_order_clt} hold with $V$, $V_i$, and $Q_i$ in lieu of $U$, $U_i$, and $R_i$ respectively, for $i \in [n]$, $\alpha_n = 1/\sqrt{\rho_1(\Sigma^2)}$, $\beta_n = \log(n) / n$, and $\gamma_n = 1/n$. Invoking \cref{thm_order_clt} with $(p_1,p_2) = (1/4,3/4)$, $(r_1, r_2) = (\lfloor n/4 \rfloor, \lfloor 3n/4 \rfloor)$, $F_W = \Phi$, and $f_W = \phi$, as well as using $ \Phi^{-1}(1/4)  = - \Phi^{-1}(3/4)$, we obtain 
		\begin{align*}
			\sqrt{n}\begin{pmatrix} 
				Y_{(\lfloor n/4 \rfloor)}  +  \Phi^{-1}(3/4)  \\
				\\
				Y_{(\lfloor 3n/4 \rfloor)} - \Phi^{-1}(3/4)
			\end{pmatrix} \distrto \cN_2\left(0_2,~ {1\over 16 \phi^2(\Phi^{-1}(3/4))}\begin{pmatrix}
				3 & 1 \\ 1 & 3
			\end{pmatrix}\right),
		\end{align*} 
		so that 
		\begin{equation}\label{distr_limit_Y_IQR}
			\sqrt{n}\left(
			Y_{(\lfloor 3n/4 \rfloor)} - Y_{(\lfloor n/4 \rfloor)} - 2\Phi^{-1}(3/4)
			\right) \distrto 
			\cN(0, \sigma_*^2),
		\end{equation}
		where $\sigma_* = [2 \phi(\Phi^{-1}(3/4))]^{-1}$. This completes the proof of \textbf{Step 1} for \textbf{Case 1}.\\

        {\bf Case 2:} Suppose that $\rho_3(\Sigma) = \omega(n^2 \log^2 n)$.  To establish the analog of {\bf Step 1} in the proof of \cref{thm_range_limit} for \textbf{Case 2}, we invoke the Yurinskii coupling result of \cref{lem_Yurinskii} in \cref{app_proof_tech_Yuri_lemmas} in conjunction with the 1-Lipschitz property of order statistics with respect to the sup-norm as established in \cref{lem_Lipschitz}. This coupling argument yields a Gaussian approximation from which the desired quantile convergence properties follow. 
        
        In particular, let $S \sim \cN_n(0, \bI_n)$, and define $Y_{\pi} := (Y_{(q_1)}, \ Y_{(q_2)})^\T$ and $S_{\pi} := (S_{(q_1)}, S_{(q_2)})^\T$, where $q_1 = \lfloor 3n/4 \rfloor$ and $q_2 = \lfloor n/4 \rfloor$. Recall the decomposition $Y_i = \sum_{j = 1}^d \xi_{ij}$ for each $i \in [n]$, as defined by \eqref{def_xi} in the proof of \cref{thm_range_limit}. To invoke \cref{lem_Yurinskii}, for any $\epsilon > 0$, let $t  = 2 \sqrt{\log n}$ and $\delta  = \epsilon$. In conjunction with \cref{lem_Lipschitz}, \cref{lem_xi}, and the bound established in \cref{lem_beta} for $\beta$ as defined in \cref{lem_Yurinskii}, this yields 
		\begin{align}\label{Yurinskii App2}\nonumber
				\PP (  \| Y_{\pi} - S_{\pi} \|_{\infty} > \epsilon) &  \le \PP (\| Y - S \|_{\infty} > \epsilon) \\ \nonumber
				& \lesssim  \PP \bigl(\| S \|_{\infty} > 2\sqrt{\log n} \bigr) \ + \ { n  \log n  \over  \epsilon^3\sqrt{\rho_3(\Sigma)}} \\ & \lesssim {1\over n} + {n \log n \over \epsilon^3\sqrt{\rho_3(\Sigma)}}.
		\end{align} 
		 Since $n^2\log^2n = o({\rho_3(\Sigma)})$, this entails that 
		\[
			 \| Y_{\pi} - S_{\pi} \|_{\infty} = o_\PP(1).
		\]
		We thus have 
		\begin{align}\label{distr_limit_Y_q1}
			&\sup_{t\in \RR} \left| \PP\left(
			Y_{(q_1)} \le t
			\right) - \PP\left(
			S_{(q_1)} \le t
			\right) \right| ~ \le~  C  \epsilon_n,\\\label{distr_limit_Y_q2}
			&\sup_{t\in \RR} \left| \PP\left(
			Y_{(q_2)} \le t
			\right) - \PP\left(
			S_{(q_2)} \le t
			\right) \right| ~ \le~  C \epsilon_n,\\\label{distr_limit_Y_range}
			&\sup_{t\in \RR} \left| \PP\left(
			Y_{(q_1)} - Y_{(q_2)}\le t
			\right) - \PP\left(
			S_{(q_1)} - S_{(q_2)} \le t
			\right) \right| ~ \le~  C \epsilon_n,
		\end{align}
		where we may take $\epsilon_n = o(1)$. For future reference for the proof of \textbf{Step 3}, we again note that the asymptotic properties of the joint distribution of a finite number of central order statistics for \text{i.i.d.} standard Gaussian samples \cite[Theorem 10.3]{David} imply
		\begin{equation}\label{Joint Limit 2} 
			\sqrt{n} \begin{pmatrix}
				S_{(q_1)} - 	\Phi^{-1}(3/4)  \vspace{1.5mm}\\
				S_{(q_2)} + 	\Phi^{-1}(3/4)
			\end{pmatrix}  \distrto  \cN_{2}(0, \Sigma_{\pi}),  
		\end{equation}
		where we use $\Phi^{-1}(3/4) = -\Phi^{-1}(1/4)$ and 
		\[
		\	\Sigma_{\pi} := \begin{pmatrix}
			\sigma_{11} & \sigma_{12} \\  \sigma_{12} & \sigma_{22} 
		\end{pmatrix} ={1 \over 16\phi^2(\Phi^{-1}(3/4))}\begin{pmatrix}
			3   & 1
			\\ 1 & 3
		\end{pmatrix}.
		\]
        This completes the proof of \textbf{Step 1} for \textbf{Case 2}.\\

	\noindent{\bf Proof of Step 2:} In the proof of  {\bf Step 2}, we largely repeat the arguments analogous to that found in the proof of \cref{thm_range_limit}, except that in this context we take $a_n \equiv  \sqrt{n}$ and $b_n \equiv \Phi^{-1}(3/4)$, and use 
		\begin{align}\label{rho2_IQR}
		    \rho_2(\Sigma)   \ge  \rho_1(\Sigma^2) \vee \rho_3(\Sigma) = \omega(n) = \omega(b_n^2),
		\end{align}
as per \cref{lem_ranks}, and 
		\[
		\max\left\{\left|Y_{(\lfloor n/4 \rfloor)} + b_n\right|,~  \left|Y_{(\lfloor 3n/4 \rfloor)} - b_n\right|\right\} = \cO_\PP(1)
		\]
		 in analogy to the analysis of the event $\cE_{(n)} \cap \cE_{(1)}$ in the proof of \cref{thm_range_limit}. We can thus similarly deduce that
		\[
		\sqrt{n\over n-1}{R_{(q)} \over \sqrt{\tr(\Sigma)}}  =  1 + \cO_\PP\left(
		{1\over  \sqrt{\rho_2(\Sigma)}}  \right),\quad \forall~  q\in \left\{\lfloor n/4 \rfloor, \lfloor 3n/4 \rfloor\right\},
		\]
		so as to obtain that, under the event $\cE_{(n)} \cap \cE_{(1)}$, 
		\begin{equation}\label{bd_zeta_n_IQR}
			 \left|{1\over \zeta_n} -1 \right| =  \left|{n \over n-1}{R_{(\lfloor n/4 \rfloor)} + R_{(\lfloor 3n/4 \rfloor)} \over 2\sqrt{\tr(\Sigma)}} -1 \right| =  \cO\left(
			{1\over  \sqrt{\rho_2(\Sigma)}} +{1\over n}\right) := \eta_n.
		\end{equation}

        \noindent{\bf Proof of Step 3:} For the proof of {\bf Step 3}, define
		\[
		\bar T_* := 2 a_n   \Delta^{-1/2}\left(R_{(\lfloor 3n/4 \rfloor)} - R_{(\lfloor n/4 \rfloor)}\right) - 2 a_n  b_n
		\]
		and 
		let $U \sim \cN(0, \sigma_*^2)$.
		Repeating similar arguments to that of the proof of \cref{thm_range_limit} yields that, for all $t\in \RR$,
		\begin{align*}
			& \PP\left(
			\bar T_* \le t\right) -  
			\PP\left( U \le  t\right)\\ 
			& =    \PP\left(
			Y_{(\lfloor 3n/4 \rfloor)} - Y_{(\lfloor n/4 \rfloor)} \le  {t + 2a_nb_n\over a_n}(1+\eta_n)\right) -   \PP\left( U \le  t\right) +o(1) &&\text{by \eqref{eq_distr_range_Tn}}\\
			&=   \PP\left(
			U \le  t(1+\eta_n) + 2a_nb_n\eta_n\right) -  
			\PP\left( U  \le  t\right)   + o(1),
		\end{align*}  
		where the final equality is due to \eqref{distr_limit_Y_IQR} in \textbf{Case 1} and \eqref{distr_limit_Y_range} with \eqref{Joint Limit 2} in \textbf{Case 2}, under the same case separation considered in the preceding proof of \textbf{Step 1}. Invoking \cref{lem_anti} and \cref{lem_anti_ratio} with $p=1$, $t_0 =  a_n$ and $\xi = 1/(1+\eta_n)$ gives 
		\begin{align*}
			&\sup_{t\in \RR} \left| \PP\left(
			U \le  t(1+\eta_n) + 2a_nb_n\eta_n \right) -  
			\PP\left( U  \le  t\right)\right|   \\
			&\le ~ \sup_{t\in \RR}  \left| \PP\left(
			U \le  t(1+\eta_n) + 2a_nb_n\eta_n \right) -  
			\PP\left( U  \le  t(1+\eta_n)\right)\right|\\
            &\qquad + \sup_{t\in \RR}  \left| \PP\left(
			U \le  t(1+\eta_n)  \right) -  
			\PP\left( U  \le  t\right)\right| \\
			&\le ~   C a_n \eta_n +  2 \exp\left(-{a_n^2 / C'}\right) 
		\end{align*}
		for some constants $C, C'$ depending only on $\Phi^{-1}(3/4)$ and $\sigma_*^2$. In conjunction with \eqref{bd_zeta_n_IQR} and $
		\rho_2(\Sigma) = \omega(n)
		$
		as per \eqref{rho2_IQR}, using symmetric arguments to upper-bound the reverse direction, we obtain
		\[
		\sup_{t\in \RR}\left| \PP\left(
		\bar T_* \le t\right) -  
		\PP\left( U \le  t\right)\right|  = o(1)  + \cO\left(
		\sqrt{n \over \rho_2(\Sigma)} +  {1\over \sqrt n}\right).
		\] 
		Finally, \cref{prop_Delta_Null} yields 
		\[
		a_n b_n \sqrt{\frac{\Delta}{\wh \Delta}} = a_n b_n + o_\PP\left(\frac{a_n b_n}{\sqrt{n}} \right) = a_n b_n + o_\PP\left( 1 \right),
		\]
		again considering $a_n \equiv   \sqrt{n}$ and $b_n \equiv \Phi^{-1}(3/4)$ in this setting. Invoking Slutsky's theorem completes the proof. 
	\end{proof}

    \subsubsection{Generalized Theory for the Limiting Distribution of Central Order Statistics used in the Proof of \cref{thm_IQR}}

		For any fixed percentile $p\in (0,1)$, the following lemma is the key result that proves the limiting distribution of the $r$-th order statistics $Y_{(r)}$ with  any $r /n - p = o(n^{-1/2})$. It generalizes the classical result on empirical quantile statistics in \cite{David} by relaxing the assumption of independence of the samples and allowing the random samples to not be distributed according to a given absolutely continuous distribution but instead be approximated by this distribution in the limit. 
		
		\begin{lemma}\label{lem_order_clt}
			Let $Y_i := U_i + R_i$, for $i\in[n]$, be a sequence of random variables satisfying 
			\begin{enumerate}
				\item[(a)] $U_1,\ldots, U_n$ are i.i.d. copies of some random variable $U$,
				\item[(b)] the random variable $U$ satisfies 
				\begin{equation}\label{Kolm_bd_UW}
					\sup_{t\in \RR} \left| \PP(U \le t) - \PP(W\le t)\right| = \cO(\alpha_n)
				\end{equation}
				for some random variable $W$ that has c.d.f. $F_W$ and quantile function $F^{-1}_W$.
				\item[(c)]  the random variables $R_1,\ldots, R_n$ satisfy
				\begin{equation}\label{eq_tail_Ri}
					\PP\left(\max_{i\in [n]} |R_i| \ge \beta_n\right)  \le \gamma_n.
				\end{equation}
			\end{enumerate}  
            The deterministic sequences $\alpha_n, \beta_n$ and $\gamma_n$ satisfy 
            \[
                    (\alpha_n + \beta_n + \gamma_n)\sqrt{n} = o(1).
            \]
			For any fixed percentile $0<p<1$ with any order $r\in[n]$ satisfying  $r /n - p=  o(n^{-1/2})$, assume that $F_W$ is differentiable at its $p^\text{th}$ quantile, $\xi_p = F_W^{-1}(p)$, with the derivative satisfying $f_W(\xi_p) > 0$. Further assume that the second-order derivative of $F_W$ at $x$, $f'_W(x)$, is bounded for all $ \xi_p - c\le x \le  \xi_p+ c$ with some small constant $c>0$. 
			Then we have 
			\[
			\sqrt{n}\left(
			Y_{(r)} - \xi_p
			\right) = \sqrt{n}~ {p - \wh F_Y(\xi_p) \over f_W(\xi_p)} + o_\PP(1)
			\]
			where  $\wh F_Y$ is the empirical c.d.f. of $Y_1,\ldots, Y_n$. 
		\end{lemma}

        The proof of \cref{lem_order_clt} uses the following lemma,  proved in  \cite{David}, for convergence in probability between two sequences of random variables.
		\begin{lemma}\label{lem_diff_conv_prob}
			Let $V_n$ and $ W_n$ be two sequences of random variables such that 
			\begin{enumerate}
				\item[(a)] $W_n = \cO_\PP(1)$; 
				\item[(b)] For every $y$ and every $\epsilon>0$, 
				\begin{align*}
					&(i) \quad \lim_{n\to \i} \PP\left(
					V_n \le y, W_n \ge y  + \epsilon 
					\right) = 0,\\
					&(ii) \quad \lim_{n\to \i} \PP\left(
					V_n \ge  y+\epsilon, W_n \le \epsilon
					\right) = 0.
				\end{align*}
			\end{enumerate}
			Then,
			\[
			V_n - W_n = o_\PP(1).
			\]
		\end{lemma}

        \medskip 
        
		\begin{proof}[Proof of \cref{lem_order_clt}]
			Define two sequences of random variables
			\begin{align}\label{def_Vn_Wn}
				V_n :=	\sqrt{n}\left(
				Y_{(r)} - \xi_p
				\right),\qquad W_n := \sqrt{n} ~ {p - \wh F_Y(\xi_p) \over f_W(\xi_p)}.
			\end{align}
			We aim to invoke \cref{lem_diff_conv_prob} by verifying the conditions in (a) and (b). 
			
			\paragraph{Verification of (a).}
			To verify condition (a), we first note that 
			\[
			W_n =  \sqrt{n} ~ {F_W(F_W^{-1}(p)) - \wh F_Y(\xi_p) \over f_W(\xi_p)} = \sqrt{n} ~ {F_W(\xi_p) - \wh F_Y(\xi_p) \over f_W(\xi_p)}  
			\] 
			by the fact that $F_W$ is differentiable at $\xi_p$. By adding and subtracting terms,  we have $W = W_{n,1} + W_{n,2} + W_{n,3}$ with 
			\begin{align*}
				&W_{n,1} := \sqrt{n} ~ {\wh F_U(\xi_p) - \wh F_Y(\xi_p) \over f_W(\xi_p)} \\
				&W_{n,2} := \sqrt{n} ~ {F_U(\xi_p) - \wh F_U(\xi_p) \over f_W(\xi_p)}\\
				&W_{n,3}:= \sqrt{n} ~ {F_W(\xi_p) - F_U(\xi_p) \over f_W(\xi_p)}.
			\end{align*}
			Here $F_U$ denotes the c.d.f. of $U$ with  $\wh F_U$ being its empirical counterpart. We proceed to bound the three terms separately. 
			
			For $W_{n,3}$,  the Kolmogorov distance bound in  \eqref{Kolm_bd_UW} of part (b) and $\alpha_n = o(1/\sqrt n)$ gives 
			\begin{equation}\label{bd_Wn3}
				W_{n,3} = \cO\left(
				\alpha_n \sqrt{n} \over f_W(\xi_p)
				\right) = o(1).
			\end{equation}
			Regarding $W_{n,2}$, since for any $y \in \RR$,
			\begin{equation}\label{cdf_clt}
				\sqrt{n} ~ \left(F_U(y) - \wh F_U(y)\right)  \distrto \cN\left(0, F_U(y)(1- F_U(y))\right),
			\end{equation}
			we have $W_{n,2} = \cO_\PP(1)$. Finally, to bound $W_{n,1}$  from above, we note that for any $y\in \RR$,
			\begin{align}\label{diff_emp_F_U}\nonumber
				\wh F_U(y) - \wh F_Y(y)  &= {1\over n}\sum_{i=1}^n \left(
				\b1\{U_i \le y\} - \b1\{U_i  + R_i\le y\} 
				\right)\\
				&=  {1\over n}\sum_{i=1}^n \Bigl(
				\b1\{ y - R_i\le U_i \le y\} - \b1\{ y\le U_i  \le y-R_i\} 
				\Bigr).
			\end{align} 
			By using \eqref{eq_tail_Ri} in part (c), we obtain that 
			\begin{align*}
				\left|\wh F_U(y) - \wh F_Y(y) \right| &\le 
				{1\over n}\sum_{i=1}^n \b1\{ y - \beta_n \le U_i \le y + \beta_n \} +   \gamma_n\\
				&= \wh F_U(y+\beta_n) - \wh F_U(y-\beta_n) + \gamma_n.
			\end{align*} 
			It then follows that $W_{n,1} $ is bounded from above by
			\begin{align*}
				& {\gamma_n \sqrt{n} \over f_W(\xi_p)} + {\sqrt{n}\over f_W(\xi_p)}\left[
				\wh F_U(\xi_p+\beta_n) - \wh F_U(\xi_p-\beta_n)
				\right]\\
				&~ \le  ~ {\gamma_n \sqrt{n} \over f_W(\xi_p)} + {\sqrt{n}\over f_W(\xi_p)}\left|
				\wh F_U(\xi_p+\beta_n) - F_U(\xi_p+\beta_n) - \wh F_U(\xi_p-\beta_n) + F_U(\xi_p-\beta_n)   
				\right| \\
				&\qquad + {\sqrt{n}\over f_W(\xi_p)}\left|
				F_U(\xi_p+\beta_n) - F_W(\xi_p+\beta_n)  
				\right| +  {\sqrt{n}\over f_W(\xi_p)}\left|
				F_U(\xi_p-\beta_n) - F_W(\xi_p-\beta_n)  
				\right|\\
				&\qquad +  {\sqrt{n}\over f_W(\xi_p)}\left|
				F_W(\xi_p+\beta_n) - F_W(\xi_p-\beta_n)  
				\right|. 
			\end{align*} 
			By using \eqref{cdf_clt} and \eqref{Kolm_bd_UW} in part (b),  we conclude that  for some $\bar \xi \in[\xi_p - \beta_n, \xi_p + \beta_n]$,
			\begin{align}\label{bd_Wn1}\nonumber
				W_{n,1} &= \cO_\PP\left( \gamma_n \sqrt{n} +   \alpha_n \sqrt{n}\right)   + {\sqrt{n}\over f_W(\xi_p)} \left[\beta_n f_W(\xi_p) + \cO\left(\beta_n^2 f'_W(\bar \xi)\right)\right] + A(\xi_p, \beta_n)\\
				&=\cO_\PP\left( \gamma_n \sqrt{n} +   \alpha_n \sqrt{n} + \beta_n\sqrt{n}\right) + A(\xi_p, \beta_n)
			\end{align}
			where we write
			\[
			A(\xi_p, \beta_n) := 
			{\sqrt{n}\over f_W(\xi_p)}\left|
			\wh F_U(\xi_p+\beta_n) - F_U(\xi_p+\beta_n) - \wh F_U(\xi_p-\beta_n) + F_U(\xi_p-\beta_n)   
			\right|.
			\]
			By writing 
			\[
			L_i = \b1\left\{
			U_i \le \xi_p + \beta_n
			\right\} - \b1\left\{
			U_i \le \xi_p  -  \beta_n
			\right\},\qquad \text{for each }i\in [n],
			\]
			we know that $\sum_{i=1}^n L_i \sim \text{Binomial}(n, p_n^*)$ with 
			\begin{align}\label{bd_pn_star}\nonumber
				p_n^* &= F_U(\xi_p + \beta_n) - F_U(\xi_p-\beta_n)\\\nonumber
				&= F_U(\xi_p + \beta_n) -F_W(\xi_p+ \beta_n) - F_U(\xi_p-\beta_n) + F_W(\xi_p-\beta_n)\\\nonumber
				&\qquad  + F_W(\xi_p + \beta_n) - F_W(\xi_p-\beta_n)\\\nonumber
				& = \cO(\alpha_n )+ \beta_n \left( f_W(\xi_p) + o(1)\right) &&\text{by \eqref{Kolm_bd_UW} and \eqref{bd_Wn1}}\\
				& = \cO(\alpha_n + \beta_n).
			\end{align}
			It then follows that $\EE[\rII] = 0$ and 
			\[
			\EE\left[[A(\xi_p, \beta_n)]^2\right]= {1 \over nf_W^2(\xi_p)} \EE\left[
			\left( \sum_{i=1}^n 
			L_i - n p_n^*
			\right)^2
			\right] = {p^*_n(1-p^*_n) \over f_W^2(\xi_p)} =  \cO(\alpha_n +\beta_n),
			\]
			so that Chebyshev's inequality yields 
			\begin{equation}\label{bd_A_term}
				A(\xi_p, \beta_n)= \cO_\PP(\sqrt{\alpha_n + \beta_n}).
			\end{equation}
			In view of \eqref{bd_Wn3}, \eqref{cdf_clt}, \eqref{bd_Wn1}, and \eqref{bd_A_term}, we thus have verified condition (a) in \cref{lem_diff_conv_prob}. 
			
			\paragraph{Verification of (b).}
			We verify the part (i) of condition (b) as the same argument can be used to prove part (ii). Fix arbitrary $y\in \RR$ and $\epsilon > 0$. By recalling \eqref{def_Vn_Wn}, we note that 
			\begin{align*}
				V_n \le y & \quad \iff \quad   Y_{(r)} \le \xi_p + y/\sqrt{n} \\
				&\quad \iff \quad  \wh F_Y(\xi_p + y/\sqrt n) \ge r/n\\
				&\quad \iff \quad  Z_n \le y_n
			\end{align*}
			where 
			\begin{align*}
				Z_n &= {\sqrt{n} \over f_W(\xi_p)}\left[
				F_W(\xi_p + y/\sqrt n) - \wh F_Y (\xi_p + y/\sqrt n) 
				\right],\\
				y_n &= {\sqrt{n} \over f_W(\xi_p)}\left[
				F_W(\xi_p + y/\sqrt n) -{r\over n}
				\right].
			\end{align*}
			Further note that 
			\begin{align}\label{diff_y_yn}\nonumber
				y_n - y &= {\sqrt{n} \over f_W(\xi_p)}\left[
				F_W(\xi_p )  + {y\over \sqrt n} f_W(\xi_p) + \cO\left(
				{y^2\over n} f_W'(\bar \xi)
				\right)-{r\over n}
				\right] - y\\\nonumber
				&= {\sqrt{n} \over f_W(\xi_p)}\left[
				p-{r\over n}  + {y\over \sqrt n} \left( f_W(\xi_p) + o(1)
				\right)
				\right]- y\\
				& = o(1)
			\end{align}
			where the last step uses $r/n - p = o(n^{-1/2})$. 
			We find 
			\[
			\PP\left(
			V_n \le y, W_n \ge y  + \epsilon 
			\right)   =  \PP\left(
			Z_n \le y_n, W_n \ge y  + \epsilon 
			\right)   
			\]
			so that part (i) of condition (b) follows 
			\begin{equation}\label{target_ZW_diff}
				Z_n - W_n = {\sqrt{n} \over f_W(\xi_p)}\left[
				F_W(\xi_p + y/\sqrt n) - \wh F_Y (\xi_p + y/\sqrt n) 
				-  F_W(\xi_p) + \wh F_Y(\xi_p) \right]  = o_\PP(1).
			\end{equation} 
			By similar arguments, \eqref{target_ZW_diff} also ensures part (ii). It thus remains to show \eqref{target_ZW_diff}.
			
			Following the preceding arguments for bounding $W_n$,  we need to show 
			\begin{align*}
				\rI &=  {\sqrt{n} \over f_W(\xi_p)}\left[
				\wh F_U(\xi_p + y/\sqrt n) - \wh F_Y (\xi_p + y/\sqrt n) 
				-  \wh F_U(\xi_p) + \wh F_Y(\xi_p) \right] = o_\PP(1)\\
				\rII &=  {\sqrt{n} \over f_W(\xi_p)}\left[
				F_U(\xi_p + y/\sqrt n) - \wh F_U (\xi_p + y/\sqrt n) 
				-  F_U(\xi_p) + \wh F_U(\xi_p) \right]= o_\PP(1)\\
				\rIII &=  {\sqrt{n} \over f_W(\xi_p)}\left[
				F_W(\xi_p + y/\sqrt n) - F_U (\xi_p + y/\sqrt n) 
				-  F_W(\xi_p) + F_U(\xi_p) \right]= o_\PP(1).
			\end{align*}
			For $\rIII$, \eqref{Kolm_bd_UW} in part (b) ensures 
			\[
			\rIII = \cO\left(
			\alpha_n \sqrt n
			\right) = o(1)
			\]
			while for $\rII$, repeating the arguments for bounding $A(\xi_p,\beta_n)$ above gives
			\[
			\rII = \cO_\PP(n^{-1/4}).
			\]
			Finally,  by the decomposition of $\wh F_U(\cdot) - \wh F_Y(\cdot)$ in \eqref{diff_emp_F_U}, the term $	(f_W(\xi_p)~ \rI / \sqrt n)  $ equals 
			\begin{align*}
				& {1\over n}\sum_{i=1}^n \Bigl[
				\b1\{\xi_p + y/\sqrt n - R_i\le U_i \le \xi_p + y/\sqrt n\} - \b1\{ \xi_p + y/\sqrt n\le U_i  \le \xi_p + y/\sqrt n-R_i\} \Bigr]\\
				&\qquad   - 	{1\over n}\sum_{i=1}^n \Bigl[\b1\{ \xi_p  - R_i\le U_i \le \xi_p \}- \b1\{ \xi_p\le U_i  \le \xi_p -R_i\} 
				\Bigr].
			\end{align*}
			By analogous bounding of $W_{n,1}$ and using \eqref{eq_tail_Ri} in part (c), we obtain 
			\begin{align*}
				\rI &\le 	{2\gamma_n \sqrt{n} \over f_W(\xi_p)} + {\sqrt{n}\over f_W(\xi_p)} {1\over n}\sum_{i=1}^n  
				\b1\left\{\xi_p + y/\sqrt n - \beta_n \le U_i \le \xi_p + y/\sqrt n + \beta_n \right\}  \\
				&\qquad + {\sqrt{n}\over f_W(\xi_p)} {1\over n}\sum_{i=1}^n  
				\b1\left\{\xi_p - \beta_n \le U_i \le \xi_p   + \beta_n \right\}\\
				& = {2\gamma_n \sqrt{n} \over f_W(\xi_p)}  + 
				{2\sqrt{n}\over f_W(\xi_p)} \max_{x \in \{\xi_p, ~ \xi_p + y/\sqrt n\}} \left(
				\wh F_U(x+ \beta_n) - \wh F_U(x - \beta_n) 
				\right).
			\end{align*}
			For any $x \in \{\xi_p, ~ \xi_p + y/\sqrt n\}$, since
			\begin{align*}
				\wh F_U(x+\beta_n) - \wh F_U(x -\beta_n)  
				& \le \left| \wh F_U(x+\beta_n)  - F_U(x+\beta_n)- \wh F_U(x -\beta_n)+F_U(x - \beta_n)  \right|\\
				&\quad  +  F_U(x+\beta_n) - F_U(x-\beta_n),
			\end{align*}
			repeating the argument for bounding $p_n^*$ in \eqref{bd_pn_star} yields 
			\[
			\wh F_U(x+\beta_n) - \wh F_U(x -\beta_n)  = \cO_\PP\left(
			\alpha_n + \beta_n f_W(x)
			\right) = \cO_\PP(\alpha_n + \beta_n),
			\] 
			so that  the analogous arguments for bounding $\rII$ gives 
			\begin{align*}
				{\sqrt{n}\over f_W(\xi_p)}\left| \wh F_U(x+\beta_n)  - F_U(x+\beta_n)- \wh F_U(x -\beta_n)+F_U(x - \beta_n)  \right|= \cO_\PP\left(
				\sqrt{\alpha_n + \beta_n}
				\right).
			\end{align*} 
			We thus conclude that 
			\[
			\rI = \cO_\PP\left(
			\gamma_n \sqrt{n} + \sqrt{n}(\alpha_n + \beta_n) + \sqrt{\alpha_n + \beta_n}
			\right) = o_\PP(1).
			\]
			Combining the bounds of $\rI$, $\rII$, and $\rIII$ proves \eqref{target_ZW_diff}, thereby completing the proof.
		\end{proof}

        \medskip

		An immediate corollary of \cref{lem_order_clt} is the following multivariate central limit theorem for a fixed number of  order statistics, which generalizes the classical result, as described in the preceding (see, for instance, Theorem 10.3 of \cite{David}).
		
		\begin{theorem}\label{thm_order_clt}
			Grant conditions (a) -- (c) in \cref{lem_order_clt}.
			For any finite integer $s\ge 1$, let $0<p_1 < \cdots < p_s<1$ be fixed percentiles   with corresponding order $r_i\in [n]$ satisfying $(r_i/n - p_i) = o(n^{-1/2})$ for all $i\in [s]$. Assume $F_W$ is differentiable at $\xi_{p_i} := F_W^{-1}(p_i)$ for all $i\in [s]$ with $0<f_W(\xi_{p_i})<\i$ and its second-order derivative is bounded for all $\xi_{p_i} -c \le x \le \xi_{p_i} + c$ with some small constant $c>0$. Then we have 
			\[
			\sqrt{n}\begin{pmatrix} 
				Y_{(r_1)} - \xi_{p_1} \\
				\vdots \\
				Y_{(r_s)} - \xi_{p_s}
			\end{pmatrix} \distrto \cN_s(0_s, \Sigma)
			\]
			where $Y_{(r_i)}$ is the $r_i$-th order statistic and 
			\[
			\Sigma_{ij} = {p_i(1-p_j)  \over f_W(\xi_{p_i})f_W(\xi_{p_j})},\qquad\text{for all } i\le j.
			\]
		\end{theorem}

\subsubsection{Other Technical Lemmas used in the Proof of \cref{thm_IQR}}\label{app_proof_tech_Yuri_lemmas}

The following theorem is a variant of the Yurinskii Coupling with respect to the sup-norm. It is proven in \cite{Belloni}.

\begin{theorem}[Yurinksii Coupling in Sup-Norm]\label{lem_Yurinskii}
		
  Let $\xi_1, \ldots, \xi_d \in \RR^n$ be independent zero-mean random vectors, and suppose 
		\[ 
		\beta := \sum_{j = 1}^d \EE\|\xi_j \|^2_2 \| \xi_j \|_{\infty} \ + \  \sum_{j = 1}^d \EE\| g_j \|^2_2 \| g_j \|_{\infty} \]
		is finite, where $g_j$ are drawn independently from $\cN_n(0_n, \Cov(\xi_j))$.  Let $V_n = \sum_{j = 1}^d \xi_j$. Then for all $\delta > 0$, there exists a random vector $S_n \sim \cN_n(0_n, \Cov(V_n))$ such that 
		\begin{equation}\label{Yurinskii Sup Bound}
			\begin{split}
				\PP \left(\| V_n - S_n \|_{\infty} > 3 \delta   \right) \leq \min_{t \geq 0} \left\{2 \PP \left(\|  Z  \|_{\infty} > t   \right) + \beta t^2 \delta^{-3} \right\}        
			\end{split}
		\end{equation}
		where $ Z  \sim \cN_n(0_n, \bI_n)$.
	\end{theorem}

\begin{lemma}\label{lem_beta}
		Let $\xi_{\cdot j}\in \RR^n$, for $1\le j\le d$, be defined in \eqref{def_xi}. Let $g_{\cdot j}$, for $1\le j\le d$, be independent realizations from $\cN_n(0, \Cov(\xi_{\cdot j}))$. Then under $\cH_0$ and the conditions of \cref{thm_range_limit},
		\[
		\beta := \sum_{j = 1}^d \EE\left[\|\xi_{\cdot j} \|^2_2 \| \xi_{\cdot j} \|_{\infty}\right] +   \sum_{j = 1}^d \EE\left[\| g_{\cdot j}\|^2_2 \|g_{\cdot j} \|_{\infty}\right] = \cO\left(
		{n\log n\over \sqrt{\rho_3(\Sigma)}}
		\right)
		\]
	\end{lemma}
	\begin{proof}
	We first bound $\EE\| \xi_{\cdot j}\|^2_2 \| \xi_{\cdot j} \|_{\infty} \leq \sqrt{\EE\| \xi_{\cdot j} \|_2^4} \sqrt{\EE\| \xi_{\cdot j} \|^2_{\infty}}$ from above. Note that
	\begin{equation}\label{Xi 4th Moment}
		\begin{split}
			\EE\| \xi_{\cdot j} \|_2^4 &= \frac{\EE\Bigl[\sum_{i = 1}^n \lambda^2_j \bigl( (Z_{ij} - \oZ_j)^2 - \frac{n - 1}{n} \bigr)^2   \Bigr]^2}{4 \tr^2(\Sigma^2) (\frac{n - 1}{n})^4}\\ & \leq \frac{ \lambda^4_j }{4 \tr^2(\Sigma^2) (\frac{n - 1}{n})^4}\left[ \sum_{i = 1}^n  \sqrt{\EE\left[(Z_{ij} - \oZ_j)^2 - \frac{n - 1}{n} \right]^4}   \right]^2 \\ & =  \frac{ \lambda^4_j n^2 }{4 \tr^2(\Sigma^2) (\frac{n - 1}{n})^4} \EE\left[(Z_{11} - \oZ_1)^2 - \frac{n - 1}{n} \right]^4  \\ & = \cO\left( \frac{ \lambda^4_j n^2}{ \tr^2(\Sigma^2)} \right),
		\end{split}
	\end{equation}
	where the second step uses  Minkowski's inequality and the last step uses \eqref{distr_W}. Furthermore, we find that
	\begin{equation}\label{Expected Sup-Norm}
		\begin{split}
			\EE\| \xi_{\cdot j} \|^2_{\infty} & = \frac{\lambda^2_j}{2 \tr(\Sigma^2) (\frac{n - 1}{n})^2} \EE\left[\max_{i \in [n]}  \left| (Z_{ij} - \oZ_j)^2 - \frac{n - 1}{n}\right| \right]^2 \\ & \leq \frac{ \lambda^2_j}{  \tr(\Sigma^2) (\frac{n - 1}{n})^2} \left(\EE\left[\max_{i \in [n]}  (Z_{ij} - \oZ_j)^4\right] + \left(\frac{n - 1}{n}\right)^2 \right)  
			\\ & \lesssim \frac{ \lambda^2_j}{  \tr(\Sigma^2) (\frac{n - 1}{n})^2} \left( \EE\left[\max_{i \in[n]} Z^4_{ij}\right]+ \EE\left[\oZ^4_j\right]\right) +   \frac{ \lambda^2_j}{  \tr(\Sigma^2)  }.
			\\ & = \cO\left(\frac{\lambda^2_j}{ \tr(\Sigma^2)} \log^2 n  \right).
		\end{split}
	\end{equation}
	Here, the last steps uses  
	$$
	\EE\left[\oZ^4_j \right]= \frac{1}{n^4}\EE\left(\sum_{i = 1}^n Z_{ij}\right)^4 = \frac{3}{n^2},$$
	a consequence of the fact that $\sum_{i = 1}^n Z_{ij} \sim \cN(0,n)$, as well as 
	$\EE[\max_{i \in[n]} Z^4_{ij}] = \cO (\log^2 n)$ for any $j\in [d]$
	from \cref{lemma_Z_moment}. Combining \eqref{Xi 4th Moment} and \eqref{Expected Sup-Norm} together with \cref{def_rhos} yields  
	\begin{equation}\label{bd_xi}
		\begin{split}
			\sum_{j = 1}^d \EE\| \xi_{\cdot j} \|^2_2 \| \xi_{\cdot j} \|_{\infty} &  = \cO\left(\frac{\tr(\Sigma^3)}{\sqrt{\tr^3(\Sigma^2)}} n \log n \right)    = \cO\Bigl(\frac{n \log n}{\sqrt{\rho_3(\Sigma)}} \Bigr).
		\end{split}
	\end{equation} 
We proceed to  bound $\EE\| g_{\cdot j} \|^2_2 \| g_{\cdot j} \|_{\infty} \leq \sqrt{\EE\| g_{\cdot j} \|_2^4} \sqrt{\EE\| g_{\cdot j} \|^2_{\infty}}$, where 
	\[
		g_{\cdot j} \sim \cN_n\left(
		0_n, {\lambda_j^2 \over \tr(\Sigma^2)}\bI_n 
		\right),
	\]
	due to \cref{lem_xi}. First, 
	\begin{equation*}
		\begin{split}
			\EE\| g_{\cdot j} \|_2^4 = \Var\left(\| g_{\cdot j} \|^2_2\right) + \Bigl( \EE\| g_{\cdot j} \|^2_2 \Bigr)^2 =\frac{2n \lambda^4_j}{\tr^2(\Sigma^2)} + \frac{n^2 \lambda^4_j}{\tr^2(\Sigma^2)}    = \cO\Bigl(\frac{n^2 \lambda^4_j}{\tr^2(\Sigma^2)} \Bigr).
		\end{split}
	\end{equation*}
	Secondly, for $Z\sim \cN_n(0, \bI_n)$, we have 
	\begin{equation*}
		\begin{split}
			\EE\| g_{\cdot j} \|^2_{\infty} =  {\lambda_j^2 \over 2\tr(\Sigma^2)} \EE\| Z \|^2_{\infty}    =   \frac{\lambda^2_j}{2 \tr(\Sigma^2)} \EE\max_{i \in [n]} Z^2_i  = \cO\Bigl(\frac{\lambda^2_j}{\tr(\Sigma^2)} \log n \Bigr),
		\end{split}
	\end{equation*}
	where the last step uses the classical result on the maximum of $n$ i.i.d.  $\chi^2_1$ random variables \cite{Boucheron}.
	These two facts imply 
	\begin{equation}\label{bd_g}
		\begin{split}
			\sum_{j = 1}^d \EE\| g_{\cdot j}\|^2_2 \| g_{\cdot j} \|_{\infty}  &  = \cO\Bigl(\frac{ \tr(\Sigma^3)}{\sqrt{\tr^3(\Sigma^2)}}n \sqrt{\log n} \Bigr)    =\cO\Bigl(\frac{n \sqrt{\log n}}{\sqrt{\rho_3(\Sigma)}} \Bigr).
		\end{split}
	\end{equation}
	Thus, combining \eqref{bd_xi} and \eqref{bd_g} completes the proof. 
	\end{proof}	
        
	\bigskip

    \begin{lemma}\label{lemma_Z_moment}
		Let $W_1,\ldots, W_n$ be i.i.d. from $\cN(0,1)$. Then 
		\[
			\EE\left[
			\max_{i \in[n]} W_i^4
			\right] = \cO(\log^2 n).
		\]
	\end{lemma}
	\begin{proof}
		Start with 
		\[
			\EE\left[
			\max_{i \in[n]} W_i^4
			\right]   = \EE\left[
			\left( \max_{i \in[n]} W_i^2\right)^2
			\right]  = \Var\left(
			\max_{i \in[n]} W_i^2
			\right) + \left(
			\EE \max_{i \in[n]} W_i^2
			\right)^2.
		\]
		Notice that 
		\[
			\Var\left(
			\max_{i \in[n]} W_i^2
			\right)  \le \Var\left(
			\max_{i \in[n]} (W_i^2 + V_i^2)
			\right), 
		\]
		for some i.i.d. $V_i \sim \cN(0,1)$ with $i\in [n]$, that are also independent of $W_i$. 
		Since $W_i^2 \sim \chi^2_1$ and $W_i^2 + V_i^2 \sim \exp(1/2) \equiv 2 \exp(1)$ for $i \in [n]$, standard results on the maxima of independent samples generated from the unit exponential and $\chi^2_1$ distributions (see, for instance, \cite{Boucheron}) in conjunction with the preceding imply
		\begin{equation*}
			\begin{split}
			\EE\left[
			\max_{i \in[n]} W_i^4
			\right] =  \cO\left(\sum_{k = 1}^n \frac{1}{k^2} \right)   + \cO\left(\log^2 n \right)    = \cO\left(\log^2 n \right),
			\end{split}
		\end{equation*}
		thus yielding the desired result.  
	\end{proof}

    \bigskip 
 
	\begin{lemma}\label{lem_Lipschitz}
		Order statistics are 1-Lipschitz with respect to the sup-norm $\| \cdot \|_{\infty}$. That is, for any $x, y\in \mathbb{R}^n$,  
		\begin{equation}
			\begin{split}
				|x_{(k)} - y_{(k)}| \leq \|x - y \|_{\infty},\quad \text{for each } k = 1,\ldots,n.
			\end{split}
		\end{equation} 
	\end{lemma}
	
	\begin{proof}
		We begin by establishing the 1-Lipschitz property for the minimum and maximum order statistics. In the case of the minimum, without loss of generality consider $x_{(1)} \leq y_{(1)}$. If they occur at the same coordinate in the original $x$ and $y$ vectors, then the 1-Lipschitz property immediately holds. Otherwise, $x_{(1)}$ occurs at the same coordinate as $y_{(l)} \geq y_{(1)} \geq x_{(1)}$ for some $l = 2,\ldots,n$ in the original $x$ and $y$ vectors, implying the Lipschitz property
		\begin{equation}
			\begin{split}
				|x_{(1)} - y_{(1)}| \leq |x_{(1)} - y_{(l)}| \leq \|x - y \|_{\infty}
			\end{split}
		\end{equation}
		The property analogously holds for the maximum, where we consider $x_{(n)} \geq y_{(n)}$ also without loss of generality. Again, when $x_{(n)}$ and $y_{(n)}$ occur at the same coordinate in the original vectors, the property immediately holds. Otherwise, $x_{(n)}$ occurs at the same coordinate as $y_{(m)} \leq y_{(n)} \leq x_{(n)}$ for some $m = 1,\ldots, n - 1$, which entails
		\begin{equation}
			\begin{split}
				|x_{(n)} - y_{(n)}| \leq |x_{(n)} - y_{(m)}| \leq \|x - y \|_{\infty}
			\end{split}
		\end{equation}
	
		Next, consider the non-minimal lower order statistics $x_{(k)}$ for any $k = 2,\ldots,\lfloor n/2 \rfloor$. Further, since we have already established the result for the maximum and minimum, we can consider $n \geq 3$. As per the preceding, consider $x_{(k)} \leq y_{(k)}$ without loss of generality. As before, when $x_{(k)}$ and $y_{(k)}$ occur at the same coordinate in the original vectors, then the 1-Lipschitz property immediately holds. Otherwise, there are two possible cases: 
		\begin{enumerate}
			\item \textbf{Case 1:} $x_{(k)}$ occurs at the same coordinate as $y_{(l)} \geq y_{(k)} \geq x_{(k)}$ for some $l = k + 1, \ldots, n$. As with the case of the minimum order statistic, this immediately implies $|x_{(k)} - y_{(k)}| \leq |x_{(k)} - y_{(l)}| \leq \|x - y \|_{\infty}$.
			\item \textbf{Case 2:} $x_{(k)}$ occurs at the same coordinate as $y_{(m)} \leq y_{(k)}$ for some $m = 1, \ldots, k - 1$. In this case, the pigeonhole principle implies that at least one of the more extreme lower order statistics $x_{(M)} \leq x_{(k)}$, for some $M \in \{1,\ldots,k-1 \}$, must occur at the same coordinate as $y_{(l)} \geq y_{(k)} \geq x_{(k)} \geq x_{(M)}$, for some $l \in \{k,\ldots,n \}$, in the original $x$ and $y$ vectors. Thus, $|x_{(k)} - y_{(k)}| \leq |x_{(M)} - y_{(l)}| \leq \|x - y \|_{\infty}$.
		\end{enumerate}
		
		Finally, we verify that the 1-Lipschitz property holds for the non-maximal upper order statistics $x_{(k)}$, for each $k = \lfloor n/2 \rfloor + 1,\ldots,n - 1$. While this will hold in direct analogy with the preceding proof for the lower order statistics, we will explicitly verify it for the sake of completeness. Without loss of generality, consider $x_{(k)} \geq y_{(k)}$. As before, when $x_{(k)}$ and $y_{(k)}$ occur at the same coordinates in the original $x$ and $y$, the property immediately follows. Otherwise, as per the lower order statistics, there are two possible cases: 
		\begin{enumerate}
			\item \textbf{Case 1:} $x_{(k)}$ occurs at the same coordinate as $y_{(m)} \leq y_{(k)} \leq x_{(k)}$ for some $m = 1, \ldots, k - 1$. This immediately implies that $|x_{(k)} - y_{(k)}| \leq |x_{(k)} - y_{(m)}| \leq \|x - y \|_{\infty}$.
			
			\item \textbf{Case 2:} $x_{(k)}$ occurs at the same coordinate as $y_{(l)} \geq y_{(k)}$ for some $l = k + 1, \ldots, n$. As per the preceding, the pigeonhole principle implies that at least one of the more extreme upper order statistics $x_{(M)} \geq x_{(k)}$, for some $M \in \{k + 1,\ldots,n \}$, must occur at the same coordinate as $y_{(l)} \leq y_{(k)} \leq x_{(k)} \leq x_{(M)}$, for some $l \in \{1,\ldots,k \}$, in the original $x$ and $y$ vectors. Thus, $|x_{(k)} - y_{(k)}| \leq |x_{(M)} - y_{(l)}| \leq \|x - y \|_{\infty}$.
		\end{enumerate}
		In view of all cases above, the proof is complete.	 
	\end{proof}

\subsection{Proof of Consistency Results under \cref{model_subG} and \cref{finiteMixture_stochastic_rep}} 

    Under \cref{model_subG} and \cref{finiteMixture_stochastic_rep}, we use $C_* := (C_1,\ldots,C_n)^{\T}$ to denote the random allocations of the samples to the $K$ mixture components; that is,  $C_i$ for $i \in [n]$ are \text{i.i.d.} with $\PP(C_{i} = k) = \pi_k$ for each $k \in [K]$. Let 
    $$
        n_k := \sum_{i = 1}^n \b1\{C_i = k\},\qquad \text{for each }k\in [K].
    $$
    so that $(n_1,\ldots,n_K)^\T \sim \text{Multinomial}(n; \pi_1,\ldots,\pi_{K})$. The unconditional covariance matrix of $X$ under either \cref{model_subG} or  \cref{finiteMixture_stochastic_rep} satisfies  
    $$
        \Sigma = \sum_{k < m}^K \pi_k \pi_m (\mu_k - \mu_m) (\mu_k - \mu_m)^{\T} + \sum_{k = 1}^K \pi_k \Sigma_k.
    $$  
    For notational convenience, we define 
    \begin{alignat*}{2}
        &\delta := \max_{k, \ell \in [K]} \| \mu_k - \mu_\ell \|^2, \qquad   &&\tr(\oSigma^2) := \max_{k\in [K]}\tr(\Sigma_k^2),\qquad \|\oSigma\|_\op = \max_{k\in[K]}\|\Sigma_k\|_\op,\\
        &\tr(\oSigma) := \max_{k\in [K]}\tr(\Sigma_k),\qquad 
        &&\tr(\uSigma) := \min_{k\in [K]}\tr(\Sigma_k).
    \end{alignat*} 
    Similarly, we also write 
    \[
        \rho_r(\uSigma) = \min_{k \in [K]}\rho_r(\Sigma_k), \qquad \text{for }r = 1,2.
    \]

\subsubsection{Proof of \cref{thm_loc_mix_sG}: Consistency for Location-Type Sub-Gaussian Mixture Alternatives}\label{app_proof_thm_loc_mix_sG}

\begin{proof}
We prove \cref{thm_loc_mix_sG} under the following set of conditions  
\begin{align}\label{cond_D_sG}
    & {\delta \over \tr (\oSigma)} = \omega\left( {1 \over  \sqrt{\rho_2(\uSigma)}}\right),\\\label{cond_Sigma_diff}
    & \tr(\oSigma)  - \tr(\uSigma) = \cO\left(\min\left\{\tr(\uSigma),~  \delta  \right\}\right),\\\label{cond_rho_2_sG}
    & \rho_1(\uSigma^2) \ge \log n.
\end{align}
When $\Sigma_* = \Sigma_1 = \cdots = \Sigma_K$, \cref{lem_ranks} implies that both \eqref{cond_D_sG} and \eqref{cond_Sigma_diff} are satisfied under \eqref{cond_mean_sep_sG}. Meanwhile, \eqref{cond_rho_2_sG} 
reduces  to  $\rho_1(\Sigma_*^2) \ge \log n$.

We prove consistency of the range-based test associated with $T$, as this is sufficient to establish consistency of the combined test. Recall $\Delta$ from \eqref{def_Delta_null} and 
\[
    T =  2  a_n ~ \wh\Delta^{-1/2} \left(R_{(n)} - R_{(1)}\right) -  2a_nb_n
\]
from \eqref{def_T_range}.
Proof of \cref{thm_loc_mix_sG} involves establishing $
    T  \to -\infty,
$
in probability. 
This is accomplished by showing
\begin{equation}\label{Contrast Property}
   \Delta^{-1/2}\left( R_{(n)} - R_{(1)} \right)= o_\PP\left(\sqrt{\log n}\right)
\end{equation}
and invoking the ratio-consistency of $\wh\Delta$ for $\Delta$ as established in \cref{prop_Delta_Alternatives}. 

To prove \eqref{Contrast Property}, we first bound $\Delta$  from below via  
\begin{align}\label{Big Eps}\nonumber
        {\Delta\over 2}  & ~ =~    \frac{\sum_{k < l, m < q}^K \pi_k \pi_l \pi_m \pi_q [(\mu_k - \mu_l)^{\T} (\mu_m - \mu_q)]^2 + \sum_{k,l} \pi_k \pi_l \tr(\Sigma_k \Sigma_l)}{\sum_{k < l}^K \pi_k \pi_l \| \mu_k - \mu_l \|_2^2 + \sum_{k = 1}^K \pi_k \tr(\Sigma_k)} \\\nonumber 
        & \quad ~  +   \frac{\sum_{m = 1}^K \sum_{k \neq l}^K \pi_k \pi_l \pi_m (\mu_k - \mu_l)^{\T} \Sigma_m (\mu_k - \mu_l)}{\sum_{k < l}^K \pi_k \pi_l \| \mu_k - \mu_l \|_2^2 + \sum_{k = 1}^K \pi_k \tr(\Sigma_k)} \\ \nonumber
        & ~ \gtrsim ~ \frac{\max_{k,l} \| \mu_k - \mu_l \|_2^4 + \max_k \tr(\Sigma_k^2)}{\max_{k,l} \| \mu_k - \mu_l \|_2^2 + \max_{k} \tr(\Sigma_k)} \\ 
        & ~ =~  
        {\delta^2 + \tr(\oSigma^2) \over \delta +  \tr(\oSigma)}.
\end{align} 
We next bound $(R_{(n)} -R_{(1)})$ from above. Pick any $k\in [K]$ and $i\in [n]$ with $C_i = k$. 
Invoking \cref{lem_mean_var} yields
\begin{equation}\label{radii_mean_conditional}
    \begin{split}
        M_{ik}^2  &~ := ~ \EE(R_i^2 \mid  C_i = k, C_*)\\
        &~~  = ~  \left\|  \mu_k - \bar \mu \right\|^2 + \frac{n - 2}{n} ~ \tr(\Sigma_k) + \frac{1}{n} \sum_{\ell = 1}^K \frac{n_\ell}{n} ~ \tr(\Sigma_\ell)\\
        & ~ ~ \asymp ~ \left\|  \mu_k - \bar \mu \right\|^2  + \tr(\Sigma_k) + {\tr(\oSigma) \over n},
    \end{split}
\end{equation}
where we write
\begin{equation}\label{def_mu_bar}
    \bar \mu := \sum_{k = 1}^K {n_k \over n}\mu_k.
\end{equation} 
By invoking \cref{lem_radii_concentration} with $\rho_1(\uSigma^2)\ge \log n$ and using a union bound argument, we find that with probability at least $1- 5K/n^2$, the following holds uniformly over $k\in [K]$ and $i \in [n]$ with $C_i = k$:
\begin{align}\label{def_event_sG_init}
    |R_i^2 - M_{ik}^2| &~ \lesssim~  \sqrt{\tr(\Sigma_k^2)\log n}  + \|\mu_k - \bar \mu\|_2\sqrt{\|\Sigma_k\|_\op  \log n}  \\\nonumber
				&\quad ~ +  {1\over \sqrt n}\left(\sqrt{\tr(\oSigma^2) \log n}   + \|\mu_k - \bar \mu\|_2\sqrt{\|\oSigma\|_\op \log n }\right)\\\label{def_event_sG} 
    &~ \le~  \sqrt{\tr(\oSigma^2)\log n} + \|\mu_k - \bar \mu\|_2\sqrt{\|\oSigma\|_\op  \log n}.
\end{align}
In the rest of the proof, we work under the event that \eqref{def_event_sG_init} and \eqref{def_event_sG} hold. Since 
\begin{equation}\label{bd_centers}
    \|\mu_k - \bar\mu\|_2 = \Bigl\|
        \sum_{\ell = 1}^K {n_\ell \over n}(\mu_k - \mu_\ell)
    \Bigr\|_2 \le \sqrt{\delta} ~  \sum_{\ell = 1}^K {n_\ell \over n} = \sqrt{\delta},
\end{equation} 
and 
\eqref{cond_Sigma_diff} implies 
\begin{equation}\label{eq_ratio_tr}
     \tr(\oSigma) \le \tr(\uSigma) +   \tr(\oSigma-\uSigma) \lesssim \tr(\uSigma),
\end{equation} 
we obtain
\begin{align}\label{bd_M_diff} \nonumber
     M_{ik} - M_{j\ell} &= {\left\|  \mu_k - \bar \mu \right\|_2^2 - \left\|  \mu_\ell - \bar \mu \right\|_2^2 + \frac{n - 2}{n} \tr(\Sigma_k-\Sigma_\ell)  \over 
      M_{ik} +  M_{j\ell} }\\\nonumber
      &\lesssim {( \|  \mu_k - \bar \mu  \|_2  -  \|  \mu_\ell - \bar \mu  \|_2)\|\mu_k-\mu_\ell\|_2 +  \tr(\Sigma_k-\Sigma_\ell)  \over 
          \|  \mu_k - \bar \mu  \|_2  + \|  \mu_\ell - \bar \mu  \|_2  +  \sqrt{\tr(\uSigma)}}\\
      &\lesssim 
      {\delta + \tr(\Sigma_k - \Sigma_\ell) \over 
      \sqrt{\tr(\oSigma)}}
\end{align}
and 
\begin{align}\label{bd_var_R_sG}
    |R_i - M_{ik}| &\lesssim
     \sqrt{\tr(\oSigma^2) +\delta  \|\oSigma\|_\op \over \tr(\oSigma)}\sqrt{\log n}  .
\end{align}
We proceed to consider two cases: 

\paragraph{Case 1:} If $\delta \lesssim \tr(\oSigma)$, then  
$
    \Delta \gtrsim {\delta^2  /  \tr(\oSigma)}
$ 
from \eqref{Big Eps}. Since \eqref{bd_M_diff} and \eqref{cond_Sigma_diff} imply 
\[
    M_{ik} - M_{j\ell}  \lesssim 
        {\delta + \tr(\Sigma_k - \Sigma_\ell) \over \sqrt{\tr(\oSigma)}}  \lesssim 
        {\delta  \over \sqrt{\tr(\oSigma)}},
\] 
we find that
\begin{align*} 
    & \Delta^{-1/2}\left(\max_i R_i - \min_i R_i\right)\\ 
     &\lesssim {\sqrt{\tr(\oSigma)}\over \delta}\max_{k,\ell \in [K]}\max_{i,j: C_i = k, C_j = \ell} 
    \left(
            M_{ik} - M_{j\ell} + |R_i - M_{ik}| + |R_j - M_{j\ell}|
    \right) \\ \nonumber
    & \lesssim 
1   + {\tr(\oSigma) \over \delta}\sqrt{\tr(\oSigma^2) \log n\over \tr^2(\oSigma)} +  \sqrt{{\tr(\oSigma) \over \delta}
        {\|\oSigma\|_\op \log n
    \over   \tr(\oSigma)}}.
\end{align*}  
The claim \eqref{Contrast Property} follows from
 \begin{equation}\label{rho_bars}
 {\|\oSigma\|_\op \over \tr(\oSigma)} \le  \sqrt{\tr(\oSigma^2) \over \tr^2(\oSigma)} \le {1\over \sqrt{\rho_2(\uSigma)}} \overset{\eqref{cond_D_sG}}{=}  
    o\left({\delta\over \tr(\oSigma)} \right). 
\end{equation}

\paragraph{Case 2:} If $\tr(\oSigma) = o(\delta)$, then $
    \Delta \gtrsim \delta$ from \eqref{Big Eps}.
By using 
\[
    M_{ik} - M_{j\ell} \lesssim \sqrt{\delta} + {\tr(\oSigma) - \tr(\uSigma) \over \sqrt{\tr(\oSigma)}} 
\]
deduced from the intermediate steps of \eqref{bd_M_diff}, we have 
\begin{align*}
     &\Delta^{-1/2}\left(\max_i R_i - \min_i R_i\right)\\
     & \lesssim 
      1 + {\tr(\oSigma) - \tr(\uSigma) \over \sqrt{\delta \tr(\oSigma)}} +  \sqrt{\tr(\oSigma^2) + \delta \|\oSigma\|_\op \over \delta \tr(\oSigma)}\sqrt{\log n}\\
    &  \lesssim 
      1 + {\tr(\oSigma) - \tr(\uSigma) \over \tr(\oSigma)} +  \sqrt{{\tr(\oSigma^2) \over \tr^2(\oSigma)} + {\|\oSigma\|_\op \over \tr(\oSigma)}}\sqrt{\log n} &&\text{by }\tr(\oSigma) = o(\delta)\\
    & = o \left( \sqrt{\log n} \right),
\end{align*} 
where the last step uses \eqref{rho_bars} and $\tr(\oSigma^2) \le \tr(\oSigma)\|\oSigma\|_\op$ as well as 
\[
     {\tr(\oSigma)\over \|\oSigma\|_\op} \ge {\tr(\Sigma_{k^*})\over \|\Sigma_{k^*}\|_\op}   \ge  {\tr(\Sigma_{k^*}^2)\over \|\Sigma_{k^*}^2\|_\op}  = \rho_1(\uSigma^2)  \ge \log n
\]
where we choose $k^*$ such that $\|\Sigma_{k^*}\|_\op = \|\oSigma\|_\op$.  

Combining the two cases establishes the claim in \eqref{Contrast Property} as $\lim_{n \to \i} \mathbb{P}\left(\cE \right) = 1$, thereby completing the proof. 
\end{proof}

\subsubsection{Proof of  \cref{thm_cov_mix_sG}: Consistency for Covariance-Type Sub-Gaussian Mixture Alternatives}\label{app_proof_thm_cov_mix_sG}

\begin{proof}
 We prove \cref{thm_cov_mix_sG} under \eqref{cond_rho_2_sG} and the following set of conditions:
\begin{align}\label{cond_SigmaRoot_diff_sG}
    &{\sqrt{\tr(\oSigma)}  - \sqrt{\tr(\uSigma)} } = \omega\left( \sqrt{\tr(\oSigma)}  \over  \sqrt{\rho_2(\uSigma)/\log(n)} \right) ,\\\label{cond_separation}
    & \delta = o\left({\tr(\oSigma)- \tr(\uSigma) \over \sqrt{\log(n)}} \right).  
\end{align}
Note that \eqref{cond_SigmaRoot_diff_sG} is equivalent to the condition \eqref{cond_cov_sep_sG} from the theorem statement, and that when $\mu_1=\cdots = \mu_K$, \eqref{cond_separation} is satisfied automatically. 
We prove \cref{thm_cov_mix} by establishing $
    T  \to \infty,
$
in probability,
under the specified asymptotic regime. 
This is accomplished by showing
\begin{equation}\label{Contrast Property Cov Mix}
   \Delta^{-1/2}\left( R_{(n)} - R_{(1)} \right)= \omega_\PP\left(\sqrt{\log n}\right),
\end{equation}
and invoking the ratio-consistency of $\wh\Delta$ for $\Delta$ as established in \cref{prop_Delta_Alternatives}.

To prove \eqref{Contrast Property Cov Mix}, from \eqref{Big Eps} and by using $\tr(\Sigma_k \Sigma_{\ell}) \leq \sqrt{\tr(\Sigma^2_k) \tr(\Sigma^2_{\ell})}$ as well as
\[ 
    (\mu_k-\mu_l)^\T\Sigma_m (\mu_k-\mu_l) \le \delta \|\Sigma_m\|_\op \le {1\over 2}\left[\delta^2 + \tr(\Sigma_m^2)\right]
\]
for any $k,l,m\in [K]$,
we can deduce that 
\begin{align}\label{Delta2_CovMix_Bound} 
    \Delta & \lesssim {\delta^2 + \tr(\oSigma^2) \over   \tr(\oSigma)}.
\end{align}
Next, we bound $(R_{(n)} -R_{(1)})$ from below under the event $\cE$ in \eqref{def_event_sG}.
Note from \eqref{radii_mean_conditional} and \eqref{bd_centers} that 
\begin{equation}\label{lb_M_diff}
    \begin{split}
     &\max_{k,\ell \in [K]}\max_{i,j: C_i = k, C_j = \ell} \left\{M_{ik} - M_{j\ell} \right\} \\ 
      &= 
     \max_{k,\ell \in [K]}\max_{i,j: C_i = k, C_j = \ell} {\left\|  \mu_k - \bar \mu \right\|^2 - \left\|  \mu_\ell - \bar \mu \right\|^2 + \frac{n - 2}{n} \tr(\Sigma_k-\Sigma_\ell)  \over 
      M_{ik} +  M_{j\ell} }\\ 
      &\gtrsim 
      \max_{k,\ell \in [K]} { \left(\sqrt{\tr(\Sigma_k)} + \sqrt{\tr(\Sigma_\ell)}\right) \left(\sqrt{\tr(\Sigma_k)} - \sqrt{\tr(\Sigma_\ell)}\right) - \delta \over 
      \sqrt{\delta} + \sqrt{\tr(\Sigma_k)} + \sqrt{\tr(\Sigma_\ell)} + \sqrt{\tr(\oSigma)/n}}  \\ 
      &\gtrsim  {  \sqrt{\tr(\oSigma)} \left(\sqrt{\tr(\oSigma)} - \sqrt{\tr(\uSigma)}\right) - \delta \over 
      \sqrt{\delta} + \sqrt{\tr(\oSigma)} } \\
      & \gtrsim \sqrt{\tr(\oSigma)} - \sqrt{\tr(\uSigma)}.
    \end{split}
\end{equation}
The last step uses \eqref{cond_separation}. On the other hand, by \eqref{radii_mean_conditional} and \eqref{def_event_sG_init}, we have
\begin{align*}
    |R_i - M_{ik}|   & ~ \lesssim~  \sqrt{\tr(\Sigma_k^2) +\delta  \|\Sigma_k\|_\op \over \tr(\Sigma_k)}\sqrt{\log n}
    +   \sqrt{\tr(\oSigma^2) +\delta  \|\oSigma\|_\op \over \tr(\oSigma)}\sqrt{\log n}.
\end{align*}
Invoking \eqref{cond_SigmaRoot_diff_sG} \& \eqref{cond_separation} gives
\begin{align*}
    &\sqrt{\tr(\Sigma^2_k) \over \tr(\Sigma_k)}\sqrt{\log n} \le  \sqrt{\tr(\oSigma)\log n\over \rho_2(\Sigma_k)} \le \sqrt{\tr(\oSigma)\log n\over \rho_2(\uSigma)}  = o\left(
        \sqrt{\tr(\oSigma)}  - \sqrt{\tr(\uSigma)} 
    \right),\\
    & {\delta \|\Sigma_k\|_\op \over \tr(\Sigma_k)} {\log n} \overset{\eqref{rho_bars}}{\le} {\delta\sqrt{\log n} \over \sqrt{\tr(\oSigma)}} \sqrt{\tr(\oSigma) \log n\over \rho_2(\uSigma)}  = o\left(\left(
        \sqrt{\tr(\oSigma)}  - \sqrt{\tr(\uSigma)} \right)^2
    \right).
\end{align*}
Since the same bounds hold for the terms involving $\oSigma$, by \eqref{Delta2_CovMix_Bound}, we conclude that 
\begin{align*}
    \Delta^{-1/2}(R_{(n)} - R_{(1)}) &\gtrsim {\sqrt{\tr(\oSigma)} \over \delta + \sqrt{\tr(\oSigma^2)}} \left(
        \sqrt{\tr(\oSigma)}  - \sqrt{\tr(\uSigma)} \right).
\end{align*}
Observing that 
\[
    {\sqrt{\tr(\oSigma)} \over \delta} \left(
        \sqrt{\tr(\oSigma)}  - \sqrt{\tr(\uSigma)} \right) \overset{\eqref{cond_separation}}{=} \omega(\sqrt{\log n}),
\]
when $\delta \neq 0$, as well as 
\[
    {\sqrt{\tr(\oSigma)  \over  \tr(\oSigma^2)}} \left(
        \sqrt{\tr(\oSigma)}  - \sqrt{\tr(\uSigma)} \right) \ge  
        {\sqrt{\tr(\oSigma)}  - \sqrt{\tr(\uSigma)}\over \sqrt{\tr(\oSigma)}}\sqrt{\rho_2(\uSigma)} \overset{\eqref{cond_SigmaRoot_diff_sG}}{=} \omega(\sqrt{\log n}),
\]
we have proven \eqref{Contrast Property Cov Mix}, thereby completing the proof.
\end{proof}

\subsubsection{Proof of  \cref{thm_loc_mix}: Consistency for Location-Type Bai-Sarandasa Mixture Alternatives}\label{app_proof_thm_loc_mix}

\begin{proof}
The proof of \cref{thm_loc_mix} largely follows that of \cref{thm_loc_mix_sG}. We only state the main differences below. First,  \eqref{cond_D_sG} and \eqref{cond_rho_2_sG} are replaced by 
\begin{align}\label{cond_D}
    & \delta = \omega\left( {\tr (\uSigma) \over \min\{\rho_1(\uSigma)/n, \sqrt{\rho_2(\uSigma)/n}\}}\right),\\\label{cond_rho_2}
    & \rho_1(\uSigma) = \omega(n).
\end{align} 
Proof of \cref{thm_loc_mix} involves establishing $
    T  \to -\infty,
$
in probability, which is accomplished by proving \eqref{Contrast Property}
and invoking the ratio-consistency of $\wh\Delta$ for $\Delta$ as established in \cref{prop_Delta_Alternatives}. 

Pick any $k\in [K]$ and any $i\in [n]$ with $C_i = k$. In addition to $M_{ik}$ in \eqref{radii_mean_conditional}, by \cref{lem_mean_var}, we also  have  
\begin{equation}\label{radii_var_conditional}
    \begin{split}
    \sigma_{ik}^2   := \Var(R_i^2 \mid  C_i = k, C_*) 
    &\lesssim   \tr(\Sigma^2_k) + \max_{\ell \in [K]} (\mu_{\ell} -\mu_k)^\T\Sigma_k (\mu_{\ell} -\mu_k)\\&\quad  + {\tr(\oSigma^2) \over n} +  \max_{q,r \in [K]}  {(\mu_q - \mu_r)^{\T} \Sigma_q (\mu_q - \mu_r)  \over n}\\
    &\lesssim \tr(\oSigma^2) + \delta  \|\oSigma\|_\op  
    \end{split}
\end{equation}
An application of Chebyshev's inequality yields that, for any $\epsilon>0$,
\[
    \PP \left\{
            \left|R_i^2 - M_{ik}^2\right| \ge \epsilon ~ \sigma_{ik} \mid C_i = k, C_*
    \right\} \le {1 \over \epsilon^2}.
\]
Taking the union bound over $k \in [K]$ and $i \in \{i \in [n]: C_i = k\}$ and choosing $\epsilon = \sqrt{n\log n}$ gives 
$
    \lim_{n\to \infty}\PP(\cE) = 1, 
$
with
\begin{equation}\label{def_event}
    \cE := \bigcap_{k\in[K]}\bigcap_{i:C_i = k} \left\{
        |R_i^2 - M_{ik}^2| \le \sigma_{ik}\sqrt{n\log n}
    \right\}.
\end{equation}
On the event $\cE$, display \eqref{bd_var_R} gets replaced by  
\begin{align}\label{bd_var_R}
    |R_i - M_{ik}| &\le {\sigma_{ik}\sqrt{n\log n} \over M_{ik}} \lesssim
     \sqrt{\tr(\oSigma^2) +\delta  \|\oSigma\|_\op \over \tr(\oSigma)}\sqrt{n\log n}.
\end{align}
Consider the same two cases as in the proof of \cref{thm_loc_mix_sG}: 

\paragraph{Case 1:} If $\delta \lesssim \tr(\oSigma)$, then  
$
    \Delta \gtrsim {\delta^2  /  \tr(\oSigma)}
$ 
from \eqref{Big Eps}. Repeating the same arguments as in the proof of \cref{thm_loc_mix_sG}, we find that
\begin{align*} 
     \Delta^{-1/2}\left(\max_i R_i - \min_i R_i\right) 
    & \lesssim 
1   + {\tr(\oSigma) \over \delta}\sqrt{\tr(\oSigma^2) n\log n\over \tr^2(\oSigma)} +  \sqrt{{\tr(\oSigma) \over \delta}
        {\|\oSigma\|_\op n\log n
    \over   \tr(\oSigma)}}.
\end{align*}  
The claim \eqref{Contrast Property} follows by invoking \eqref{cond_D} in conjunction with 
\begin{equation}\label{rho_bars_BS}
{\tr(\oSigma^2) \over \tr^2(\oSigma)} \le {1\over \rho_2(\uSigma)},\qquad {\|\oSigma\|_\op \over \tr(\oSigma)} \le {1\over \rho_1(\uSigma)}.
\end{equation}

\paragraph{Case 2:} If $\tr(\oSigma) = o(\delta)$, then $
    \Delta \gtrsim \delta$ from \eqref{Big Eps}. We have 
\begin{align*}
     \Delta^{-1/2}\left(\max_i R_i - \min_i R_i\right)    & \lesssim 
      1 + {\tr(\oSigma) - \tr(\uSigma) \over \sqrt{\delta \tr(\oSigma)}} +  \sqrt{\tr(\oSigma^2) + \delta \|\oSigma\|_\op \over \delta \tr(\oSigma)}\sqrt{n\log n}  \\
    &  \lesssim 
      1 + {\tr(\oSigma) - \tr(\uSigma) \over \tr(\oSigma)} +  \sqrt{{\tr(\oSigma^2) \over \tr^2(\oSigma)} + {\|\oSigma\|_\op \over \tr(\oSigma)}}\sqrt{n\log n}  \\
    & = o( \sqrt{\log n}),
\end{align*} 
where the last step uses \eqref{rho_bars_BS}, $\tr(\oSigma^2) \le \tr(\oSigma)\|\oSigma\|_\op$ and 
$\rho_1(\uSigma) = \omega(n)$. 

Combining the two cases establishes the claim in \eqref{Contrast Property} as $\lim_{n \to \i} \mathbb{P}\left(\cE \right) = 1$, thereby completing the proof. 
\end{proof}

\subsubsection{Proof of  \cref{thm_cov_mix}: Consistency for Covariance-Type Bai-Sarandasa Mixture Alternatives}\label{app_proof_thm_cov_mix}

\begin{proof}
The proof of \cref{thm_cov_mix} largely follows that of \cref{thm_cov_mix_sG}. We only state the main differences below.
 We prove \cref{thm_cov_mix} under \eqref{cond_rho_2}, \eqref{cond_separation}, and the following condition:
\begin{align}\label{cond_SigmaRoot_diff}
    &\sqrt{\tr(\oSigma)}  - \sqrt{\tr(\uSigma)} = \omega\left(  \sqrt{\tr(\oSigma)\log(n)} \over \min\{\rho_1(\uSigma) / n, \sqrt{\rho_2(\uSigma)/n}\} \right).  
\end{align}
Note that \eqref{cond_SigmaRoot_diff} is equivalent to the condition \eqref{cond_cov_sep} from the theorem statement, and that when $\mu_1=\cdots = \mu_K$, \eqref{cond_separation} is satisfied automatically. 
We prove \cref{thm_cov_mix} by establishing $
    T  \to \infty,
$
in probability, which is accomplished by showing \eqref{Contrast Property Cov Mix} 
and invoking the ratio-consistency of $\wh\Delta$ for $\Delta$ 
as established in \cref{prop_Delta_Alternatives}.

To prove \eqref{Contrast Property Cov Mix}, recall \eqref{Delta2_CovMix_Bound}. We bound $(R_{(n)} -R_{(1)})$ from below under the event $\cE$ in \eqref{def_event}. Recall the expressions in \eqref{lb_M_diff}. By \eqref{radii_mean_conditional}, \eqref{radii_var_conditional} and \eqref{bd_var_R}, we have
\begin{align*}
    |R_i - M_{ik}| &\le {\sigma_{ik}\sqrt{n\log n} \over M_{ik}} \\ & \lesssim
    \sqrt{\tr(\Sigma^2_k) +  \delta \|\Sigma_k\|_\op   \over \tr(\Sigma_k)} \sqrt{n\log n}   + \sqrt{\tr(\oSigma^2) +  \delta \|\oSigma\|_\op \over \tr(\oSigma)} \sqrt{n\log n}.
\end{align*}
Since invoking \eqref{cond_SigmaRoot_diff} \& \eqref{cond_separation} gives
\begin{align*}
    &\sqrt{\tr(\Sigma^2_k) \over \tr(\Sigma_k)}\sqrt{n\log n} \le \sqrt{\tr(\oSigma)}\sqrt{n\log n\over \rho_2(\Sigma_k)} \le \sqrt{\tr(\oSigma)\log n\over \rho_2(\uSigma)/n}  = o\left(
        \sqrt{\tr(\oSigma)}  - \sqrt{\tr(\uSigma)} 
    \right),\\
    & {\delta \|\Sigma_k\|_\op \over \tr(\Sigma_k)} {n\log n} \le {\delta\sqrt{\log n} \over \sqrt{\tr(\oSigma)}} {\sqrt{\tr(\oSigma) \log n}\over \rho_1(\uSigma)/n}  = o\Bigl(\left(
        \sqrt{\tr(\oSigma)}  - \sqrt{\tr(\uSigma)} \right)^2
    \Bigr)
\end{align*}
and the same bounds hold for $\sqrt{\tr(\oSigma^2)/\tr(\oSigma)}\sqrt{n\log n}$ and $\delta \|\oSigma\|_\op n \log n/ \tr(\oSigma)$, respectively, in conjunction with \eqref{Delta2_CovMix_Bound}, we conclude that 
\begin{align*}
    \Delta^{-1/2}(R_{(n)} - R_{(1)}) &\gtrsim {\sqrt{\tr(\oSigma)} \over \delta + \sqrt{\tr(\oSigma^2)}} \left(
        \sqrt{\tr(\oSigma)}  - \sqrt{\tr(\uSigma)} \right).
\end{align*}
Repeating the same arguments as in the proof of \cref{thm_cov_mix_sG} proves \eqref{Contrast Property Cov Mix}, thereby completing the proof.
\end{proof}

\subsubsection{Technical Lemmas used in the Proofs of \cref{thm_loc_mix_sG,thm_loc_mix,thm_cov_mix,thm_cov_mix_sG}}

The following lemma states bounds for $\EE(R_i^2\mid C_i = k, C_*)$ and $\Var(R_i^2\mid C_i = k, C_*)$ for any $k\in [K]$ and $i\in[n]$. Recall $\bar \mu$ from \eqref{def_mu_bar}.
\begin{lemma}\label{lem_mean_var}
    Under either \cref{model_subG} or \cref{finiteMixture_stochastic_rep}, for any $i\in [n]$ and $k\in[K]$, we have
    \begin{align}\label{Conditional Exp}
         \EE(R_i^2 \mid C_i = k, C_*)   &  = \left\|  \mu_k - \bar \mu \right\|^2 + \frac{n - 2}{n} ~ \tr(\Sigma_k) + \frac{1}{n} \sum_{\ell = 1}^K \frac{n_\ell}{n} ~ \tr(\Sigma_\ell),
    \end{align}
    and, with probability one,
    \begin{equation}\label{Conditional Variance}
    \begin{split}
         \Var(R_i^2 \mid  C_i = k, C_*) & \lesssim  ~ \tr(\Sigma^2_k)  + \max_{\ell \in [K]} (\mu_k -\mu_\ell)^{\T} \Sigma_k (\mu_k - \mu_\ell) \\ & 
         \qquad +  {\tr(\oSigma^2) + \max_{q,r \in [K]} (\mu_q - \mu_r)^{\T} \Sigma_q (\mu_q - \mu_r) \over n}. 
    \end{split}
\end{equation}
\end{lemma}
\begin{proof}
    We only prove for \cref{finiteMixture_stochastic_rep} as the same proof holds for \cref{model_subG} with $\Gamma_k = \Sigma_k^{1/2}$ and $m_k = d$.  
    Notice that for any $k\in [K]$ and for any $i\in [n]$ with $C_i = k$, we have
    \[
        X_i \stackrel{\rm d}{=} \mu_k + \Gamma_k Z_i,
    \]
    where $Z_i$ is an isotropic random vector satisfying \cref{Bai-Sarandasa}. 
    For any $k \in [K]$, we find that 
    \begin{equation*}
        \begin{split}
            & \EE \Bigl(\| X_i - \oX \|^2 \mid  C_i = k, C_* \Bigr) \\ & = \EE\left(\Bigl\| \Gamma_k  Z_i - \frac{1}{n} \sum_{j = 1}^n \Gamma_{C_j}  Z_j +   \mu_k - \bar \mu  \Bigr\|^2  ~ \Big |~  C_i = k, C_*\right) \\  & = \EE\| \Gamma_k  Z_i \|^2 + \|  \mu_k - \bar \mu\|^2  + \frac{1}{n^2} \EE\left(\Bigl\| \sum_{j = 1}^n \Gamma_{C_j}  Z_j \Bigr \|^2 \ \Big | \ C_i = k, C_*\right)  \\ 
            &\qquad  - \frac{2}{n} \EE \left( \Bigl(\Gamma_k  Z_i \Bigr)^{\T} \sum_{j \neq i} \Gamma_{C_j}  Z_j \ \Big | \ C_*  \right)  - \frac{2}{n} \EE\| \Gamma_k  Z_i \|^2   + 2 (\mu_k - \bar \mu)^\T \EE(\Gamma_k  Z_i)   \\ 
            & \qquad - \frac{2}{n} (\mu_k - \bar \mu)^\T \sum_{j = 1}^n \EE\left( \Gamma_{C_j}  Z_j \mid C_i = k, C_* \right).
        \end{split}
    \end{equation*} 
    By using the fact that $(Z_1,\ldots, Z_n)$ are independent and individually satisfy \cref{Bai-Sarandasa} and $\sum_{j=1}^n \b1\{C_j = \ell\} = n_\ell$, the preceding equals 
    \begin{align*}
            &{n-2 \over n} \EE\| \Gamma_k  Z_i \|^2 + \| \mu_k - \bar \mu \|^2 + \frac{1}{n^2} \sum_{j = 1}^n \EE(\|\Gamma_{C_j}  Z_j \|^2 \mid C_i = k, C_*)  \\ & =  {n-2 \over n} \tr(\Gamma_k \Gamma_k^{\T}) + \| \mu_k - \bar \mu \|^2 + \frac{1}{n^2} \sum_{\ell = 1}^K  n_\ell  \tr(\Gamma_\ell \Gamma_\ell^{\T}),
    \end{align*}   
   thus proving the first result.\\

Regarding the  conditional variance, without loss of generality, we evaluate 
$\Var(R_1^2 \mid C_1 = 1, C_*)$. Since $\| X_1 - \oX \|^2$ is invariant to arbitrary location transformation, we center the data by $\mu_1$, and write
\begin{equation}\label{def_Ti}
    T_j := X_j - \mu_1,\qquad \forall j\in [n].
\end{equation}
Beginning with
\begin{align}\label{display_var_radii} \nonumber
        &\Var(\| X_1 - \oX \|^2 \mid   C_1 = 1, C_*)\\\nonumber
        & = \Var(\| T_1 - \overline{T} \|^2 \mid C_1 = 1, C_*) \\ \nonumber
        & = \Var\left(T_1^{\T} T_1 + \frac{1}{n^2} \sum_{i = 1}^n T_i^{\T} T_i + \frac{1}{n^2} \sum_{i \neq j} T_i^{\T} T_j - \frac{2}{n} T_1^{\T} \sum_{i = 1}^n T_i \ \Big | \ C_1 = 1, C_*\right)  \\ \nonumber
        & \lesssim  \Var\left(T_1^{\T} T_1 \mid  C_1 = 1\right) + \frac{1}{n^4} \sum_{k = 1}^K n_k \Var\left(T_i^{\T} T_i  \mid  C_i = k\right) \\ & \quad +  \frac{1}{n^4} \Var\left(\sum_{i \neq j} T_i^{\T} T_j  \mid  C_1 = 1, C_* \right) + \frac{1}{n^2} \Var\left(T_1^{\T} \sum_{i = 1}^n T_i  \mid  C_1 = 1, C_*\right),  
\end{align}
we proceed to bound each term separately. For the first term, we have
\begin{equation}\label{var_one}
    \begin{split}
        \Var \Bigl(T_1^{\T} T_1 \mid  C_1 = 1 \Bigr) &= \Var \Bigl( Z_1^{\T} \Gamma_1^{\T} \Gamma_1  Z_1 \Bigr)\\
        & = 2 \tr\Bigl((\Gamma_1^{\T} \Gamma_1)^2\Bigr) \ + \ (\kappa_1 - 3) \sum_{j = 1}^{m_1} [(\Gamma_1^{\T} \Gamma_1)_{jj}]^2 \\ & \leq 2 \tr(\Sigma^2_1) + (\kappa_1 - 3)_+ \| \Gamma_1^{\T} \Gamma_1 \|^2_{\text{F}} \\ & = 2 \tr(\Sigma^2_1) + (\kappa_1 - 3)_+ \tr(\Sigma^2_1)  \\ & \lesssim \tr(\Sigma^2_1),
    \end{split}
\end{equation}
where $(x)_+ := \max \{x, 0 \}$ and the second equality follows from Lemma 7.1 of \cite{Sri2014} for the variance of quadratic forms under the model defined by \cref{Bai-Sarandasa}, with $\kappa_1$ and $m_1$ corresponding to $\kappa$ and $m$, respectively. Based on \eqref{var_one}, the summands of the second term in \eqref{display_var_radii} can be bounded via
\begin{equation}\label{var_two}
    \begin{split}
        \Var \Bigl(T_i^{\T} T_i \mid  C_i = k \Bigr)  & = \Var \Bigl(\| \mu_k - \mu_1 \|^2 + \| \Gamma_k  Z_i \|^2 + 2 (\mu_k - \mu_1)^{\T} \Gamma_k  Z_i \Bigr) \\ & = \Var \Bigl(\| \Gamma_k  Z_i \|^2 + 2 (\mu_k - \mu_1)^{\T} \Gamma_k  Z_i \Bigr) \\ & \le 2\Var \Bigl(\| \Gamma_k  Z_i \|^2 \Bigr) + 8 \Var\left((\mu_k - \mu_1)^{\T} \Gamma_k  Z_i\right) \\ & \lesssim \tr(\Sigma^2_k) +  (\mu_k - \mu_{1})^\T \Sigma_k (\mu_k - \mu_{1}),
    \end{split}
\end{equation}
for each $k \in [K]$. For the third variance term, we find that 
\begin{align}\label{var_three} \nonumber
        &\frac{1}{n^4} \Var\left(\sum_{i \neq j} T_i^{\T} T_j \mid  C_1 = 1, C_*\right) \\\nonumber
        & = \frac{1}{n^4} \sum_{i \neq j}  \Var\left(T_i^{\T} T_j \mid  C_1 = 1, C_*\right) + \frac{1}{n^4} \sum_{i \neq j \neq k} \Cov\left(T_i^{\T} T_j, T_j^{\T} Y_k \mid  C_1 = 1, C_*\right) \\\nonumber
        & \lesssim \frac{1}{n} \max_{i \neq j} \Var\left(T_i^{\T} T_j \mid  C_1 = 1, C_* \right) \\ \nonumber
        & \lesssim \frac{1}{n} \max_{i \neq j} \{ \Var\left((\Gamma_{C_i} Z_i)^{\T} (\Gamma_{C_j} Z_j) \mid  C_1 = 1, C_* \right) + \Var\left((\Gamma_{C_i} Z_i)^{\T} \gamma_{C_j} \mid  C_1 = 1, C_* \right) \\\nonumber
        & \ \ \ \ \ \  \   + \Var\left((\Gamma_{C_j} Z_j)^{\T} \gamma_{C_i} \mid  C_1 = 1, C_* \right) \} \\\nonumber
        & = \frac{1}{n} \max_{i \neq j} \{ \EE\left([(\Gamma_{C_i} Z_i)^{\T} (\Gamma_{C_j} Z_j)]^2 \mid  C_1 = 1, C_* \right) + \gamma_{C_j}^{\T} \Sigma_{C_i} \gamma_{C_j} + \gamma_{C_i}^{\T} \Sigma_{C_j} \gamma_{C_i} \} \\\nonumber
        & = \frac{1}{n} \max_{i \neq j} \{ \sum_{q, r} \sigma^{(C_i)}_{qr} \sigma^{(C_j)}_{qr} + \gamma_{C_j}^{\T} \Sigma_{C_i} \gamma_{C_j} + \gamma_{C_i}^{\T} \Sigma_{C_j} \gamma_{C_i} \} \\\nonumber
        & = \frac{1}{n} \max_{i \neq j} \{ \tr(\Sigma_{C_i} \Sigma_{C_j}) + \gamma_{C_j}^{\T} \Sigma_{C_i} \gamma_{C_j} + \gamma_{C_i}^{\T} \Sigma_{C_j} \gamma_{C_i} \} \\
        & \leq {1 \over n}\left(\max_{k,l \in [K]} (\mu_k - \mu_l)^{\T} \Sigma_k (\mu_k - \mu_l) + \tr(\oSigma^2)\right), 
\end{align} 
where independence of the samples is invoked to reduce the $\cO(n^4)$ covariance terms to $\cO(n^3)$ non-zero summands. By  similar arguments, we find that the fourth term in \eqref{display_var_radii} is
\begin{align}\label{var_four}\nonumber
        &\frac{1}{n^2} \Var\left(T_1^{\T} \sum_{i = 1}^n T_i \mid  C_1 = 1, C_*\right)\\\nonumber
        & \le \frac{2}{n^2} \Var\left(T_1^{\T} T_1 \mid C_1 = 1\right) + \frac{2}{n^2} \Var\left(\sum_{j \neq 1}^n T_1^{\T} T_j \mid C_1 = 1, C_*\right)  \\ \nonumber
        &\lesssim \max_{k,l \in [K]}  { (\mu_k - \mu_l)^{\T} \Sigma_k (\mu_k - \mu_l) + \tr(\oSigma^2) \over n} + \frac{1}{n^2} \sum_{i \neq j \neq 1}^n \Cov\left(T_1^{\T} T_i, T_1^{\T} T_j \mid C_1 = 1, C_*\right) \\ \nonumber
        & = \max_{k,l \in [K]}  { (\mu_k - \mu_l)^{\T} \Sigma_k (\mu_k - \mu_l) + \tr(\oSigma^2) \over n} + \frac{1}{n^2} \sum_{i \neq j \neq 1}^n (\mu_{C_i} - \mu_1)^{\T} \Sigma_1 (\mu_{C_j} - \mu_1) \\
        & \lesssim \max_{k,l \in [K]}  {(\mu_k - \mu_l)^{\T} \Sigma_k (\mu_k - \mu_l) + \tr(\oSigma^2) \over n} + \max_{k \in [K]} (\mu_k - \mu_1)^{\T} \Sigma_1 (\mu_k - \mu_1),
\end{align}
where the final inequality is due to the Cauchy-Schwartz inequality. Combining \eqref{var_one}, \eqref{var_two}, \eqref{var_three}, and \eqref{var_four} yields 
\begin{align*}
     \Var(R_i^2 \mid  C_i = 1, C_*) &
     ~ \lesssim~   \tr(\Sigma^2_1) + \max_{k \in [K]} (\mu_k - \mu_1)^{\T} \Sigma_1 (\mu_k - \mu_1) \\ & \ \ \ \ \ \ + \max_{k,\ell \in [K]}  {1\over n}\left[(\mu_k - \mu_\ell)^{\T} \Sigma_k (\mu_k - \mu_\ell) + \tr(\oSigma^2) \right],
\end{align*}
thereby completing the proof.
\end{proof}

\medskip

 The following lemma establishes upper bounds of the quadratic forms of $|X^\T X - \EE[X^\T X]|$ and $|X^\T Y|$ where $X = \Sigma_X^{1/2}\wt X$ and $Y = \Sigma_Y^{1/2}\wt Y$ are independent random vectors with $\wt X$ and $\wt Y$ being $\gamma$-sub-Gaussian. It is proved in \citet[Lemma 9]{royer2017adaptive}.
	
\begin{lemma}\label{lem_quad_dev}
    Let $X = \Sigma_X^{1/2}\wt X$ and $Y = \Sigma_Y^{1/2}\wt Y$ be independent random vectors such that $\wt X$ and $\wt Y$ are $\gamma$-sub-Gaussian. 
    There exists some constant $c > 0$ that depends on $\gamma$ only such that for all $t\ge 0$,
    \begin{align*}
        &\PP\left\{
        |X^\T X- \EE[X^\T X]| \ge  \|\Sigma_X\|_F \sqrt{t} +  \|\Sigma_X\|_\op t
        \right\} \le 2e^{-c t};\\
        &\PP\left\{
        2|X^\T Y| \ge   \sqrt{2\tr(\Sigma_X\Sigma_Y) ~ t} +  \|\Sigma_X^{1/2}\Sigma_Y^{1/2}\|_\op ~ t
        \right\} \le 2e^{-c t}.
    \end{align*} 
\end{lemma}

\bigskip 

The following lemma provides concentration inequalities of the squared radii with exponential tails under \cref{model_subG}.

\begin{lemma}\label{lem_radii_concentration}
		Under \cref{model_subG}, for any $i\in [n]$ and $k\in [K]$, by conditioning on $(C_i = k, C_*)$, the following holds  with probability at least $1-5n^{-3}:$
		\begin{align*}
				\left| R_i^2 -  M_{ik}^2\right| 
				&\lesssim   \sqrt{\tr(\Sigma_k^2)\log n} +  \|\Sigma_k\|_\op \log n+ \|\mu_k - \bar \mu\|_2\sqrt{\|\Sigma_k\|_\op  \log n}  \\
				&\quad +  {1\over \sqrt n}\left(\sqrt{\tr(\oSigma^2) \log n}  + {\|\oSigma\|_\op \log n} + \|\mu_k - \bar \mu\|_2\sqrt{\|\oSigma\|_\op \log n }\right).
		\end{align*} 
	Furthermore, if $\rho_1(\uSigma^2) \ge \log n$, the preceding bound simplifies to
    \begin{align*}
				\left| R_i^2 -  M_{ik}^2\right| 
				&\lesssim   \sqrt{\tr(\Sigma_k^2)\log n} +   \|\mu_k - \bar \mu\|_2\sqrt{\|\Sigma_k\|_\op  \log n}  \\
				&\quad +  {1\over \sqrt n}\left(\sqrt{\tr(\oSigma^2) \log n}  + \|\mu_k - \bar \mu\|_2\sqrt{\|\oSigma\|_\op \log n }\right). 
		\end{align*} 
	\end{lemma}
	\begin{proof}
		Fix any $i\in [n]$ and $k\in [K]$. The whole proof conditions on $C_i = k$ and $C_*$. For simplicity, we drop the conditional notation in probabilities and expectations. Recall that
		\[
			\bar \mu =  \sum_{k=1}^K {n_k \over n}\mu_k,\quad \text{with }\quad n_k = \sum_{i=1}^n \b1\{C_i = k\}.
		\]
		By definition, we have 
		\begin{align*}
			R_i^2 = \| X_i - \bar \mu \|_2^2 + \|\bar \mu - \bar X\|_2^2 - 2 (X_i - \bar \mu)^\T (\bar X - \bar \mu).
		\end{align*}
		We proceed to analyze each term on the right hand side (RHS) separately. 
		
		For the first term, recall that conditioning on $C_i = k$, 
		\[
			 X_i - \bar \mu = \mu_k - \bar \mu + \Sigma_k^{1/2}  Z_i
		\]
		so that 
		\[
			\| X_i - \bar \mu \|_2^2 = \|\mu_k - \bar \mu\|_2^2 + \|\Sigma_k^{1/2} Z_i\|_2^2 + 2 (\mu_k - \bar \mu)^\T \Sigma_k^{1/2}  Z_i.
		\]
		Since $Z_i$ is $\gamma$-sub-Gaussian, we know that 
		\begin{equation}\label{bd_mu_Gamma_Z}
			\PP\left\{
				\left| (\mu_k - \bar \mu)^\T \Sigma_k^{1/2}  Z_i \right| \ge  t\sqrt{(\mu_k - \bar \mu)^\T \Sigma_k (\mu_k - \bar \mu)}
			\right\} \le 2e^{-\gamma^2 t^2/2},\quad \forall ~ t\ge 0.
		\end{equation}
		Moreover, invoking \cref{lem_quad_dev} with $X = \Sigma_k^{1/2} Z_i$ and $\|\Sigma_k\|_F^2 = \tr(\Sigma_k^2)$ gives
		\begin{equation}\label{bd_quad_Gamma_Z}
			\PP\left\{
			\left|  \|\Sigma_k^{1/2}  Z_i\|_2^2 -\EE[\|\Sigma_k^{1/2}  Z_i\|_2^2 ]  \right|  \ge  \sqrt{\tr(\Sigma_k^2)}~ t +  \|\Sigma_k\|_\op~  t^2
			\right\} \le 2e^{-c t^2}, \quad \forall ~ t\ge 0.
		\end{equation}
		By choosing $t = C\sqrt{\log n}$ for some large $C\ge 1$ and noting that 
		$
			\EE[\|\Sigma_k^{1/2}  Z_i\|_2^2 = \tr(\Sigma_k),
		$
		we obtain that with probability at least $1-n^{-3}$,
		\begin{align}\label{bd_Xi_centered}\nonumber
			&\left|\| X_i - \bar \mu \|_2^2  - \|\mu_k - \bar \mu\|_2^2 -  \tr(\Sigma_k)
				\right| \\&\quad \lesssim  \sqrt{\tr(\Sigma_k^2)\log n} +  \|\Sigma_k\|_\op \log n+ \sqrt{(\mu_k - \bar \mu)^\T \Sigma_k (\mu_k - \bar \mu) \log n}.
		\end{align}
		
		Regarding the term $\|\bar \mu - \bar X\|_2^2$, we first note that, conditioning on $C_*$,
		\begin{align}\label{eq_decomp_bar_X}
			\bar X = {1\over n}\sum_{i=1}^n X_i  =  {1\over n}\sum_{k=1}^K \sum_{i:C_i = k} \left(
			\mu_k + \Sigma_k^{1/2}  Z_i
			\right)  = \bar \mu +  \sum_{k=1}^K {n_k \over n} \Sigma_k^{1/2}  \bar Z_k,
		\end{align} 
		where we denote $\bar Z_k := n_k^{-1}\sum_{i:C_i = k}Z_i$.  Since $\bar Z_k$ is $(\gamma/\sqrt{n_k})$-sub-Gaussian, we find that for all $v\in \RR^d$,   
		\begin{align*}
			\EE\left[
			\exp\left(
			\sum_{k=1}^K {n_k \over n} v^\T  \Sigma_k^{1/2}  \bar Z_k
			\right)
			\right] &=  \prod_{k=1}^K \EE\left[
			\exp\left(
			{n_k \over n} v^\T  \Sigma_k^{1/2}  \bar Z_k
			\right)
			\right]\\
			&\le \prod_{k=1}^K 
			\exp\left(
			 {n_k^2 \over n^2} v^\T  \Sigma_k v  {\gamma^2\over n_k}
			\right) \\
			&= \exp\left(
			 \gamma^2  v^\T \left({1\over n} \sum_{k=1}^K {n_k \over n}  \Sigma_k \right) v 
			\right).
		\end{align*} 
		By writing 
		$$\Xi :=\sum_{k=1}^K {n_k \over n}  \Sigma_k,
		$$
		we can deduce that $\bar X - \bar \mu \overset{d}{=} \Xi^{1/2} Y/\sqrt{n}$ for some centered, isotropic $\gamma$-sub-Gaussian random vector $Y \in \RR^d$. Since
		\[
			\EE\left[
			\|\bar X - \bar \mu \|_2^2 
			\right] = \sum_{k=1}^K {n_k^2 \over n^2} 	\EE\left[ \|\Sigma_k^{1/2}  \bar Z_k\|_2 \right] = {\tr(\Xi) \over n},
		\]
		 invoking \cref{lem_quad_dev} with $\Sigma_X = \Xi/n$ and $t = C\sqrt{\log n}$ gives 
		\begin{equation}\label{bd_quad_barX}
			\PP\left\{
			\left|  \|\bar X - \bar \mu\|_2^2 -{ \tr(\Xi) \over n}  \right|  \lesssim   { \sqrt{\tr(\Xi^2)\log n} \over n}+  {\|\Xi\|_\op   \log n\over n}
			\right\} \ge 1 - n^{-3}.
		\end{equation}
	
		Finally, conditioning on $(C_i  = k, C_*)$, we analyze the cross-term
		$$
			 (X_i - \bar \mu)^\T (\bar X - \bar \mu) = Z_i^\T \Sigma_k^{1/2} (\bar X -\bar \mu) + (\mu_k - \bar \mu)^\T (\bar X -\bar \mu).
		$$
		By using the sub-Gaussianity of $(\bar X - \bar \mu)$, we have 
		\begin{equation}\label{bd_tail_crossterm_1}
			\PP\left\{
			 \left|(\mu_k - \bar \mu)^\T (\bar X -\bar \mu) \right| \ge t\sqrt{ (\mu_k - \bar \mu)^\T \Xi  (\mu_k - \bar \mu) \over n}
			\right\} \le 2e^{-\gamma^2 t^2/2},\quad \forall t\ge 0.
		\end{equation}
		Moreover, by  
		decomposing  
		\begin{align*}
			Z_i^\T \Sigma_k^{1/2} (\bar X -\bar \mu)   &= {1\over n}\sum_{k' =1}^K \sum_{j:C_j = k' } Z_i^\T \Sigma_k^{1/2}\Sigma_{k'}^{1/2}  Z_j  &&\text{by \eqref{eq_decomp_bar_X}}\\
			&=  {1\over n} Z_i^\T \Sigma_k Z_i +Z_i^\T \Sigma_k^{1/2} \left(\sum_{k' =1}^K {n'_k\over n}\Sigma_{k'}^{1/2} \bar Z_{k'} - {1\over n}\Sigma_k^{1/2}Z_i
			\right),
		\end{align*} 
		the first term one the RHS can be bounded by \eqref{bd_quad_Gamma_Z}. To control the second 
		term, we notice that $\Sigma_k^{1/2}Z_i$ is independent of the term within the parenthesis. Moreover,  it is easy to verify $$\sum_{k' =1}^K {n'_k\over n}\Sigma_{k'}^{1/2} \bar Z_{k'} - {1\over n}\Sigma_k^{1/2}Z_i ~ \overset{d}{=} ~ \Xi_{(-i)}^{1/2} {Y \over \sqrt n}$$ for some centered, isotropic sub-Gaussian random vector $Y\in \RR^p$ with sub-Gaussian constant $\gamma$, and for 
		\[
			\Xi_{(-i)}  =  {1\over n} \left[ 
			(n_k - 1)\Sigma_k + \sum_{\ell \ne k} {n_\ell }\Sigma_{\ell}
			\right].
		\]
		Invoking \cref{lem_quad_dev} with $t =C\sqrt{\log n}$, $\Sigma_X = \Sigma_k$, and $\Sigma_Y = \Xi_{(-i)}/n$ yields that with probability at least $1-n^{-3}$,  
		\begin{align}\label{bd_cross_term_LOO} \nonumber
				2 \left| Z_i^\T \Sigma_k^{1/2} \left(\sum_{k' =1}^K {n'_k\over n}\Sigma_{k'}^{1/2} \bar Z_{k'} - {1\over n}\Sigma_k^{1/2}Z_i
				\right)\right|  &\lesssim  {\sqrt{2\tr(\Sigma_k \Xi_{(-i)}) \log n \over n}} +  {\|\Sigma_k^{1/2}\Xi_{(-i)}^{1/2}\|_\op \log n\over \sqrt n} \\
				&\lesssim {\sqrt{\tr(\oSigma^2) \log n\over n} } + {\|\oSigma\|_\op \log n \over \sqrt n}.
		\end{align}
	We then conclude that with probability at least $1-3n^{-3}$ 
	\begin{align}\label{bd_cross_term}\nonumber
		 & \left| 2(X_i - \bar \mu)^\T (\bar X - \bar \mu) - {2\over n}\tr(\Sigma_k) \right|  \\\nonumber
		 &\quad \lesssim     \sqrt{\tr(\oSigma^2) \log n \over n}  + {\|\oSigma\|_\op \log n \over \sqrt n} + \sqrt{(\mu_k - \bar \mu)^\T \Xi  (\mu_k - \bar \mu) \log n \over n}\\
		 &\quad \lesssim     \sqrt{\tr(\oSigma^2) \log n \over n}  + {\|\oSigma\|_\op \log n \over \sqrt n} + \|\mu_k - \bar \mu\|_2\sqrt{ \|\oSigma\|_\op \log n \over n}
	\end{align}
	holds. The proof is complete in consideration of  \eqref{bd_Xi_centered}, \eqref{bd_quad_barX}, and \eqref{bd_cross_term}, in conjunction  with \cref{lem_mean_var}.
\end{proof}

\subsection{Proof of \cref{thm_ellip}: Consistency for Elliptical Alternatives}\label{app_elliptical_proof}

\begin{proof}
We prove \cref{thm_ellip} by establishing that $T\to \i$ in probability.
This is accomplished by demonstrating
\begin{align}\label{elliptical_contrastProperty}
    \Delta^{-1/2} \left( R_{(n)} - R_{(1)} \right) =  \omega_{\mathbb{P}}\left(\sqrt{\log n} \right),
\end{align}
and invoking the ratio-consistency of $\wh\Delta$ for $\Delta$ 
as established in \cref{prop_Delta_Alternatives}. Under \cref{elliptical_stochastic_rep},  it is easy to verify that
$$
    \Sigma = \Cov(X) = \EE[\eps^2] ~  \Sigma_*,
$$
so  that 
\begin{equation}\label{eq_Delta2_ellip}
    \Delta \equiv \Delta(\Sigma) := {2 \tr(\Sigma^2) \over \tr(\Sigma)} = \frac{2 \EE[\eps^2] ~ \tr(\Sigma^2_*)}{\tr(\Sigma_*)}.
\end{equation}
Note that by the invariance properties of the proposed test statistics relative to location shift and orthogonal transformation as well as the rotational invariance of standard Gaussian random vectors, we can without loss of generality consider
\begin{equation}\label{Elliptical Conditional Representation}
    X_i ~ \overset{\rm d}{=} ~ \eps_i~  \Lambda^{1/2}  Z_i
\end{equation}
where $Z_i$ for $i \in [n]$ are i.i.d. from $\cN_d(0_d, \bI_d)$  and $\Lambda$ is the diagonal matrix of non-increasing eigenvalues  of $\Sigma_*$.
Let $\eps_* := (\eps_1,\ldots,\eps_n)^{\T}$ and denote its order statistics by $\eps_{(n)} \ge \cdots \ge \eps_{(1)}$. We observe that for each $i \in [n]$,
\begin{equation}\label{distr_Ri_ellip}
    X_i - \oX \mid \eps_* \sim  \cN_d\Bigl(0, \ \nu_i~  \Lambda \Bigr)
\end{equation}
with 
\begin{equation}\label{def_nu}
    \nu_i \equiv \nu_i(\eps_*) :=  {(n - 1)^2\over n^2} \eps^2_i + {1\over n^2}\sum_{j \neq i}^n \eps^2_j.
\end{equation}
It then follows from \eqref{distr_Ri_ellip} that for all $i\in [n]$,
\begin{align}\label{moments_ellip}
    \EE(R^2_i \mid \eps_*) = \nu_i ~ \tr(\Sigma_*).
\end{align}
Invoking \cref{lem_quad_dev} with $\Sigma_X = \nu_i \Lambda$ and $\|\Sigma_X\|_F^2 = \nu_i^2\tr(\Sigma_*^2)$ gives that for every $t > 0$, 
\begin{align*}
     \PP \left(\left| R^2_i   - \nu_i \tr(\Sigma_*) \right|   \geq  \nu_i \sqrt{\tr(\Sigma_*^2) ~ t} + \nu_i  \|\Sigma_*\|_\op~  t ~ \mid \eps_*  \right) &  \leq  2e^{-ct}.
\end{align*}
By taking the union bound over $i\in [n]$, choosing $t = C \log n$, and invoking the dominated convergence theorem, we conclude that 
the event 
\begin{align*}
    \cE' = \bigcap_{i=1}^n \left\{
        \left|R_i^2 - \nu_i \tr(\Sigma_*)\right| \right.&\le     \nu_i \sqrt{\tr(\Sigma_*^2)}\sqrt{C\log n} + \nu_i  \|\Sigma_*\|_\op C\log n\\
        &\left.\le C' \nu_i \sqrt{\tr(\Sigma_*^2)}\sqrt{\log n}
    \right\} &&\text{by }\rho_1(\Sigma_*^2)\ge \log n
\end{align*}
holds with probability tending to one, as $n\to \i$. Thus, we work under the event $\cE'$ in the following to bound $(R_{(n)} - R_{(1)})$ from below. We begin by noting that 
\begin{align*}
    &R_{(n)} - R_{(1)}\\ 
    & \ge \max_{i,j\in [n]} \left[
        \left(\sqrt{\nu_i} - \sqrt{\nu_j}\right)\sqrt{\tr(\Sigma_*)} - {|R_i^2 - \nu_i\tr(\Sigma_*)| \over \sqrt{\nu_i \tr(\Sigma_*)}} - {|R_j^2 - \nu_j\tr(\Sigma_*)| \over \sqrt{\nu_j\tr(\Sigma_*)}}
    \right]\\
    &\ge \max_{i,j\in [n]} \left[
        {\nu_i - \nu_j \over \sqrt{\nu_i}+ \sqrt{\nu_j}}\sqrt{\tr(\Sigma_*)} - C'\left(\sqrt{\nu_i}+\sqrt{\nu_j}\right){\sqrt{\tr(\Sigma_*^2) \log n}    \over  \sqrt{\tr(\Sigma_*)}}\right].
\end{align*}
Since, for any $i,j\in[n]$, \eqref{def_nu} entails
\begin{align*}
    \nu_i - \nu_j &= {(n-1)^2\over n^2}(\eps_i^2-\eps_j^2) + {1\over n^2}\left(\eps_j^2 - \eps_i^2\right) = {n-2 \over n}(\eps_i-\eps_j)(\eps_i+\eps_j),
\end{align*}
and
\[
    \sqrt{\nu_i} \le \eps_i + \sqrt{{1\over n^2} \sum_{\ell =1 } \eps_{\ell}^2} \le  \eps_{(n)}\left( 1 + n^{-1/2}\right),
\]
we further conclude that, with probability tending to one,
\begin{align*}
    R_{(n)} - R_{(1)} &\gtrsim \left(\eps_{(n)} - \eps_{(1)}\right) \sqrt{\tr(\Sigma_*)} - \eps_{(n)} \sqrt{\tr(\Sigma_*^2)\log n \over \tr(\Sigma_*)}  \\
    &= \sqrt{\tr(\Sigma_*)}\left(\eps_{(n)} - \eps_{(1)} - \eps_{(n)} \sqrt{\log n \over \rho_2(\Sigma_*)} \right).
\end{align*}
By invoking \eqref{cond_W_gap} and \eqref{eq_Delta2_ellip}, the following holds with probability tending to one:
\begin{align*}
     \Delta^{-1/2}\left(R_{(n)} - R_{(1)}\right) & 
    \gtrsim \sqrt{\rho_2(\Sigma_*)\over \EE [\eps^2]}\left(\eps_{(n)} - \eps_{(1)}\right) \\ & \ge {\sqrt{\rho_2(\Sigma_*)} \over  \eps_{(n)}}\left(\eps_{(n)} - \eps_{(1)}\right)   \\ & = \omega\left(\sqrt{\log n} \right),
\end{align*}
where the second inequality uses the fact that 
\[
        \eps_{(n)}^2 \ge {1\over n}\sum_{i=1}^n \eps_i^2 \to \EE[\eps^2],\quad \text{almost surely, as }n\to \i.
\]
This establishes \eqref{elliptical_contrastProperty}, thereby completing the proof.
\end{proof}

\subsection{Proof of \cref{thm_kurtosis}: Consistency for Leptokurtic Alternatives}\label{app_kurtosis_proof}

	\begin{proof}
        The proof follows the same arguments as that of \cref{thm_range_limit} with modifications due to the excess kurtosis. 
        For future reference, we note that, as established in \cref{lem_xi_kurtosis},
        \begin{align}
             \Var\left((Z_{11} - \bar Z_1)^2\right) & = (2+\delta_n) \left({n - 1\over n} \right)^2 + \cO\left({1 \over n} \right) \\ & =: (\kappa_n - 1) \left({n - 1\over n} \right)^2 + \cO\left({1 \over n} \right) \\ & =:  \left({n - 1\over n} \right)^2 \nu_n,
        \end{align}
        where
        \[
            \kappa_n := 3 + \delta_n,\qquad \nu_n := (\kappa_n - 1) + \cO\left({1 \over n} \right).
        \]
        Further define 
        \[
            \Delta_{2, \delta_n} := \nu_n {\tr(\Sigma^2) \over \tr(\Sigma)}  
        \]
		Our proof consists of the following principal steps: 
		\begin{enumerate}
			\item  Define  and the random vector $Y = (Y_1, \ldots, Y_n)^\T$ via 
            \begin{align*}
                Y_i & := {1\over \sqrt{\Var\left((Z_{11} - \bar Z_1)^2\right) \tr(\Sigma^2)}}\left(R_i^2 - {n - 1 \over n}\tr(\Sigma) \right) \\ & = {1\over \sqrt{\left(\kappa_n - 1 + \cO(n^{-1}) \right) \tr(\Sigma^2)}}\left(  {n\over n-1}R_i^2 - \tr(\Sigma) \right) \\ & =: {1\over \sqrt{\nu_n \tr(\Sigma^2)}}\left(  {n\over n-1}R_i^2 - \tr(\Sigma) \right) \qquad \forall\ i\in [n].
            \end{align*} 
			We first establish the limiting distributions of $a_n (Y_{(n)} - b_n)$ and $a_n (Y_{(1)} + b_n)$, and bound
			\[
					\sup_{t\in \RR} \left| \PP\left(
				Y_{(n)} - Y_{(1)} \le t
				\right) - \PP\left(
				V_{(n)} - V_{(1)}\le t
				\right) \right|
			\]
            from above, where $V = (V_1,\ldots,V_n)^\T \sim \cN_n(0_n, \text{C}_n)$ is an exchangeable  random vector with $(\text{C}_n)_{ii} = 1$, for all $i \in [n]$, and $$
            (\text{C}_n)_{ii'} = {2 (n - 2) (\kappa_n - 3) \over n^3    \Var\left((Z_{11} - \bar Z_1)^2\right)} =  {2 (n - 2) (\kappa_n - 3) \over n^3    (\kappa_n - 1) ({n - 1\over n})^2 + \cO\left({1 \over n} \right)},\qquad \text{for all $i \neq i'$.}
            $$ 
			
			\item Secondly, we establish the ratio-consistency of $R_{(q)}$ for $\sqrt{\tr(\Sigma)}$, for each $q \in \{1, n\}$, as in \eqref{eq_radii_ratio_consistency}.
			
			\item Next, we use this ratio-consistency property to further bound 
			\[
					\sup_{t\in \RR} \left| \PP\left(
				\bar T_{\delta_n} \le t
				\right) - \PP\left(
			 \wt U_n \le t
				\right) \right|
			\]
			from above, where $\wt U_n := a_n (V_{(n)} - V_{(1)}) - 2a_nb_n$ and 
            $$
            \bar T_{\delta_n} := 2 a_n \Delta_{2, \delta_n}^{-1/2} \left(R_{(n)} - R_{(1)}  \right) - 2 a_n b_n.
            $$  
            From this, with $U_n$ as defined in \eqref{def_Un}, we can deduce that
            \[
            \sup_{t\in \RR} \left| \PP\left(
				\bar T_{\delta_n} \le t
				\right) - \PP\left(
			   U_n \le t
				\right) \right| \to 0 \ \ \ \ \text{and} \ \ \ \ \bar T_{\delta_n} \distrto E+E',
			\]
            using properties of the range of exchangeable Gaussian random vectors.
           
            \item Finally, invoking the ratio-consistency property of $\wh \Delta$ for $\Delta$ under \cref{kurtosis_stochastic_rep} as established by \cref{prop_Delta_Alternatives}, we establish $T \to \i$ in probability using \textbf{Step 3}.
			
		\end{enumerate}
		
		\paragraph{Proof of Step 1:} 
        
	Under \cref{kurtosis_stochastic_rep}, there exist $Z_1,\ldots, Z_n \in \RR^d$ which are \text{i.i.d.} realizations of an isotropic random vector $Z \in \RR^d$ with independent sub-Gaussian coordinates such that
		\begin{align}\label{def_xi_kurtosis}\nonumber
			Y_i  &=  {1\over \sqrt{\nu_n \tr(\Sigma^2)}}\left( {n-1 \over n} \| \Lambda^{1/2}(Z_i- \oZ) \|^2 -  \tr(\Sigma) \right)\\\nonumber
			& = \sum_{j=1}^d  {\lambda_j  \over \sqrt{\nu_n \tr(\Sigma^2)}}  \left( {n \over n-1}(Z_{ij} - \oZ_j)^2 - 1\right)\\
			& =: \sum_{j=1}^d \xi_{ij},
		\end{align}
		where $\oZ_j = n^{-1}\sum_{i=1}^n Z_{ij}$. 
		In \cref{lem_xi_kurtosis}, we verify that, for any $i,i'\in [n]$ and $j\in [d]$,
		\begin{equation}\label{eq_moments_xi_kurtosis}
		\EE[\xi_{ij}]= 0, \qquad \Cov(\xi_{ij},\xi_{i'j}) = \frac{ \lambda^2_j}{\tr(\Sigma^2)} \left( 1_{\{i = i'\}} + (\text{C}_n)_{12} 1_{\{i \neq i'\}} \right).
	\end{equation}
		Moreover, observe that  $\xi_{ij}$ is independent of $\xi_{ij'}$ for any $i\in [n]$ and any $j\ne j'$.  Since
		\[
			Y_{(n)} = \max_{i \in[n]} {1\over \sqrt{ d}} \sum_{j=1}^d  \xi_{ij}\sqrt{d},
		\]
		we seek to invoke \cref{thm_CCKK} to bound 
		$
			\sup_{t\in \RR} | \PP(
				Y_{(n)} \le t
			) - \PP(
			V_{(n)} \le t
			) |$.
		Thus, we first verify the Conditions E and M in Assumptions \ref{ass_E} \& \ref{ass_M}. Since $\sqrt{n / (n-1)}(Z_{ij}-\bar Z_j)$ can be expressed as a linear combination of independent sub-Gaussian random variables, we know that $(n/(n-1)) (Z_{ij}-\bar Z_j)^2$ is sub-exponential, which implies $
			\EE \exp(
			|\xi_{ij}| \sqrt{d} / B_d
			) \le 2
		$
		holds for 
		\begin{equation}\label{def_B_d_kurtosis}
			B_d = C \sqrt{d \lambda_1^2 \over \tr(\Sigma^2)} \overset{\eqref{def_rhos}}{=} C \sqrt{d   \over \rho_1(\Sigma^2)},
		\end{equation}
		where $C>0$ is an absolute constant. Moreover, 
		by \eqref{eq_moments_xi_kurtosis}, we have 
		\[
			{1\over d}\sum_{j=1}^d \EE\left[ 
			 d 	 ~ \xi_{ij}^2 
			\right]  = \sum_{j=1}^d  {\lambda_j^2 \over \tr(\Sigma^2)}= 1 
		\]
		and, by \eqref{def_xi_kurtosis} and the fact that $(Z_{ij}-\bar Z_j)$ is sub-Gaussian,
		\begin{align*}
			{1\over d}\sum_{j=1}^d \EE\left[ 
			d^2  \xi_{ij}^4
			\right] & ~ \lesssim ~ d \sum_{j=1}^d { \lambda_j^4  \over \nu^2_n \tr^2(\Sigma^2)}\EE\left[
			\left({n\over n-1}\right)^4(Z_{ij}-\bar Z_j)^8  + 1		\right]\\
			&~ \lesssim~  {d ~ \tr(\Sigma^4) \over \tr^2(\Sigma^2)} \\
			&~ \le ~ B_d^2 {\tr(\Sigma^4) \over \lambda_1^2 \tr(\Sigma^2)} &&\text{by \eqref{def_B_d_kurtosis}}\\
			&~  \le~  B_d^2.
		\end{align*}
		Therefore, invoking \cref{thm_CCKK} with $p = n$, $N = d$,  $X_{ij} = \xi_{ij}\sqrt{d}$, $b_1 \asymp b_2 \asymp 1$, and $B_N = B_d$ as per \eqref{def_B_d_kurtosis} yields 
		\begin{equation}\label{clt_Yn_kurtosis}
			\sup_{t\in \RR} \left| \PP\left(
			Y_{(n)} \le t
			\right) - \PP\left(
			V_{(n)} \le t
			\right) \right| ~ \le~  C \left(
				\log^5(nd)  \over \rho_1(\Sigma^2)
			\right)^{1/4}.
		\end{equation}
		Regarding $Y_{(1)}$, since $Y_{(1)} = -  \max_{i \in[n]} (-Y_i)$ and the above results also apply to  $(-\xi_{ij})$, we also have 
		\begin{equation}\label{clt_Y1_kurtosis}
			\sup_{t\in \RR} \left| \PP\left(
			Y_{(1)} \le t
			\right) - \PP\left(
			V_{(1)} \le t
			\right) \right| ~ \le ~  C \left(
			\log^5(nd)  \over \rho_1(\Sigma^2)
			\right)^{1/4}.
		\end{equation}
		Furthermore,  observe that 
		\[
			Y_{(n)} - Y_{(1)} = \max_{i,j \in [n]} (Y_i - Y_j) = \max_{i\ne j \in [n]} (Y_i - Y_j) =  \max_{i\ne j \in [n]}{1\over \sqrt{d}}\sum_{t=1}^d (\xi_{it} - \xi_{jt}) \sqrt{d}.
		\]
		By repeating the same arguments above in conjunction with use of the triangle inequality, one can verify that both Conditions E and M in Assumptions \ref{ass_E} \& \ref{ass_M} are satisfied by $(\xi_{it} - \xi_{jt}) \sqrt{d}$ for any $i\ne j\in [n]$ and $t\in [d]$ for $b_1 \asymp b_2 \asymp 1$ and $B_d$ as per \eqref{def_B_d_kurtosis}, and that these variates are independent across $t \in [d]$. Thus, invoking  \cref{thm_CCKK} again yields 
		\begin{equation}\label{clt_Range_kurtosis}
					\sup_{t\in \RR} \left| \PP\left(
				Y_{(n)} - Y_{(1)} \le t
				\right) - \PP\left(
				V_{(n)} - V_{(1)}\le t
				\right) \right| ~ \le~  C \left(
				\log^5(n^2 d)  \over \rho_1(\Sigma^2)
				\right)^{1/4}.
		\end{equation}
	

		\paragraph{Proof of Step 2:} Given \eqref{clt_Range_kurtosis}, the ratio consistency in \eqref{eq_radii_ratio_consistency} follows by the arguments as that in the proof of \cref{thm_range_limit}. In particular, displays \eqref{bd_En_comp} -- \eqref{bd_E1_comp} continue to hold. 

		\paragraph{Proof of Step 3:}
		We next relate the distribution of $Y_{(n)} - Y_{(1)}$ to that of $\bar T_{\delta_n}$.  With $\zeta_n$ given by \eqref{def_zeta}, recall from \eqref{rate_zeta_n} that under the event $\cE_{(n)} \cap \cE_{(1)}$, 
		\begin{equation}\label{rate_zeta_n_kurtosis}
			\left|{1\over \zeta_n} -1 \right| =  \left|{n \over n-1}{R_{(n)} + R_{(1)} \over 2\sqrt{\tr(\Sigma)}} -1 \right| =  \cO\left(
			{b_n \over \sqrt{\rho_2(\Sigma)}} + {1\over n}\right) =: \eta_n.
		\end{equation}
		By  definition,  for any $t_+\ge 0$,       \begin{align}\label{eq_distr_range_Tn_kurtosis}\nonumber
			\PP\left(
			Y_{(n)} - Y_{(1)} \le t_+\right) & = \PP\left(
			{n \over n-1}  \frac{R_{(n)}^2 - R_{(1)}^2}{\sqrt{\nu_n \tr(\Sigma^2)}}\le t_+\right) \\\nonumber
			&=  \PP\left(
		 2\Delta^{-1/2}_{2, \delta_n} \left(R_{(n)}-R_{(1)}\right)  {n \over n-1}  \frac{R_{(n)} + R_{(1)}}{2\sqrt{ \tr(\Sigma)}}\le t_+\right)\\\nonumber
		 &= \PP\left(
		 2a_n \Delta^{-1/2}_{2, \delta_n} \left(R_{(n)}-R_{(1)}\right)	 \le   a_n  \zeta_n    t_+ \right)\\
		 &= \PP\left(
		 \bar T_{\delta_n}	 \le  a_n   \zeta_n  t_+ - 2a_nb_n\right).
		\end{align}
		Recalling the definition of $\wt U_n$ from the outline of \textbf{Step 3}, it then follows that, for all $t\in \RR$,
		\begin{align*}
		 	&\PP\left(
			\bar T_{\delta_n}	 \le  t \right) - \PP\left(
			\wt U_n \le t 
			\right) \\
			&\le  \PP\left(
			Y_{(n)} - Y_{(1)} \le  {t + 2a_nb_n\over a_n }(1+\eta_n)\right) -  
			\PP\left(
			V_{(n)} - V_{(1)} \le  {t+2a_nb_n\over a_n}\right) &&\text{by \eqref{eq_distr_range_Tn_kurtosis}}\\
            &\quad  +\PP\left(
			\cE_{(n)}^c  \cup 
			\cE_{(1)}^c 
			\right)\\
			&\le  \PP\left(
			V_{(n)} - V_{(1)} \le  {t + 2a_nb_n\over a_n}(1+\eta_n)\right) -  
			\PP\left(
			V_{(n)} - V_{(1)} \le  {t+2a_nb_n\over a_n}\right)\\
			& \quad +  C\left(
			\log^5(n^2 d)  \over\rho_1(\Sigma^2)
			\right)^{1/4} + \PP\left(
			\cE_{(n)}^c  \cup 
			\cE_{(1)}^c 
			\right) &&\text{by \eqref{clt_Range_kurtosis}}.
		\end{align*} 
		Note that
		$
		V_{(n)} - V_{(1)} = \max_{i\ne j} (V_i - V_j)
		$
		with $V_i - V_j \sim \cN\left(0, 2 + 2(\text{C}_n)_{12} \right)$. 
		Since $2 + 2(\text{C}_n)_{12} \geq 2$ for all $i, j \in [n]$, we invoke  \cref{lem_anti_ratio} with $t_0 = C\sqrt{\log n}$ and $\xi = 1/(1+\eta_n)$ to obtain
		\begin{align*}
			&\sup_{t\in \RR} \left| \PP\left(
			V_{(n)} - V_{(1)} \le  {t + 2a_nb_n\over a_n}(1+\eta_n)\right) -  
			\PP\left(
			V_{(n)} - V_{(1)} \le  {t+2a_nb_n\over a_n}\right)\right|\\
			&   \le ~  C\eta_n \log n  + 2\exp\left(
			- {C' \log n}
			\right).
		\end{align*}
		Together with \eqref{rate_zeta_n_kurtosis}, \eqref{bd_En_comp}, and \eqref{bd_E1_comp}, by using symmetric arguments to bound the other direction, we hence obtain 
        \begin{align} \label{perturbed_stat_convergence_kurtosis}
            \sup_{t\in \RR}\left| \PP\left(
			\bar T_{\delta_n}	 \le  t \right) - \PP\left(
			\wt U_n \le t 
			\right) \right| & = \cO\left(\left(
			\log^5(n^2 d)  \over \rho_1(\Sigma^2)
			\right)^{1/4} +  
			{\sqrt{\log^3n \over  \rho_2(\Sigma)}} + {\log n\over  n}\right),
        \end{align}
        which tends to zero as $n\to \i$ 
under the conditions of \cref{thm_kurtosis}. To relate the asymptotic properties of $\bar T_{\delta_n}$ to that of $U_n$, we use the fact that $V$ is an exchangeable Gaussian random vector entails \cite{Hartley_Range, David}
\begin{align} \label{exhangeable_range}
    V_{(n)} - V_{(1)} = \sqrt{1 - \rho^*_n} \left(S_{(n)} - S_{(1)} \right),
\end{align}
        for some random vector $S \sim \cN_n(0_n, \bI_n)$ and $\rho^*_n := (\text{C}_n)_{12} \asymp n^{-2}$. And,
        \begin{equation*}
            \sqrt{1 - \rho^*_n} ~ U_n = a_n \sqrt{1 - \rho^*_n}\left(S_{(n)} - S_{(1)} - 2b_n \right) \distrto E + E',
        \end{equation*}
        by \cref{thm_range_limit} and Slutsky's theorem. Thus, since $2a_n b_n \sqrt{1 - \rho^*_n} = 2a_n b_n \sqrt{1 - \cO({n^{-2}})} = 2a_n b_n + o(1)$, \eqref{exhangeable_range} implies
        \begin{equation*}
            \wt U_n \distrto E + E',
        \end{equation*}
        and thus, by \eqref{perturbed_stat_convergence_kurtosis},
        \begin{align}\label{kurtosis_range_limit}
            \bar T_{\delta_n} \distrto E + E'.
        \end{align}

\paragraph{Proof of Step 4:}

We verify that $T \to \i$ in probability by first noting that 
    \begin{equation*}
       2 a_n \wh \Delta^{-\frac{1}{2}} \left(R_{(n)} - R_{(1)}  \right) \sqrt{{\Delta \over \Delta_{2, \delta_n}}} - 2 a_n b_n \left(1 + \cO_\PP\left({1 \over \sqrt{n}} \right) \right) = \cO_\PP\left(1 \right),    
    \end{equation*}
    due to \eqref{kurtosis_range_limit}, \cref{prop_Delta_Alternatives} in conjunction with a Taylor expansion of the function $f(x) = 1 / \sqrt{x}$ about 1, and Slutsky's theorem. In conjunction with the bounded fourth moments $\kappa_n := 3 + \delta_n$ of $Z_{ij}$, this yields
    \begin{align*}
        T & :=  2 a_n \wh \Delta^{-\frac{1}{2}} \left(R_{(n)} - R_{(1)} \right) - 2 a_n b_n   = 2 a_n b_n \left(\sqrt{{\kappa_n - 1 + \cO({1 \over n}) \over 2}} - 1  \right) + \cO_\PP\left( 1 \right),
    \end{align*}
    entailing that $T \to \i$ in probability, due to $\delta_n = \omega(1 / \log(n))$. This completes the proof. 
\end{proof}

\subsubsection{Technical Lemmas used in the Proof of \cref{app_kurtosis_proof}}
	
	\begin{lemma}\label{lem_xi_kurtosis}
		Let $\xi_{\cdot j}\in \RR^n$, for $j \in [d]$, as defined in \eqref{def_xi}. Then, for each $j\in [d]$, we have $\EE[\xi_{\cdot j}] = 0_n$ and 
		\[
			\Cov(\xi_{\cdot j}) = \frac{ \lambda^2_j}{\tr(\Sigma^2)} \text{C}_n,
		\]
        where $(\text{C}_n)_{i i} := 1$ and 
        $$
        (\text{C}_n)_{i i'} := {2 (n - 2) (\kappa_n - 3) \over n^3 \Var\left((Z_{11} - \bar Z_1)^2\right)} = {2 (n - 2) (\kappa_n - 3) \over n^3 \left((\kappa_n - 1) \left(\frac{n - 1}{n}\right)^2  + \cO\left({1 \over n} \right) \right)},$$ 
        for $i \neq i' \in [n]$. 
	\end{lemma} 
	\begin{proof}
	For any $j \in [d]$, let $W_i$,  for $i = 1,\ldots,n$, be \text{i.i.d.} samples of $Z_{ij}$ and write $\oW = n^{-1}\sum_{i = 1}^n W_i$. Given any $j\in [d]$, to show $\EE \xi_{\cdot j} = 0_n$, it suffices to prove that for any
	  $i\in [n]$,
		\begin{align*}
		 \EE \left(Z_{ij} - \oZ_j \right)^2  =  \frac{n - 1}{n},
		\end{align*}
	which follows directly from
        \begin{align} \label{centered Z second moment} \nonumber
            \EE \left(Z_{ij} - \oZ_j \right)^2 & = \EE\left(\left(1 - {1 \over n} \right) Z_{ij} - {1 \over n} \sum_{\ell \neq i}^n Z_{\ell j} \right)^2 \\ \nonumber & = \left(1 - {1 \over n} \right)^2 \EE Z^2_{ij} +  \frac{n - 1}{n^2} \\ \nonumber & = {(n - 1)^2 + n - 1 \over n^2} \\ & = {n - 1 \over n}.
        \end{align}
	Regarding the covariance, pick any $j\in [d]$ and $i,i'\in [n]$. We have		
		\begin{equation*}
			\begin{split}
				  \Cov(\xi_{ij}, \xi_{i'j}) & = \Cov\left( \frac{\lambda_j[(Z_{ij} - \oZ_j)^2 - \frac{n - 1}{n}]}{\sqrt{\Var\left((Z_{11} - \bar Z_1)^2\right) \tr(\Sigma^2)}}, \frac{\lambda_j[(Z_{i'j} - \oZ_j)^2 - \frac{n - 1}{n}]}{\frac{n-1}{n} \sqrt{\Var\left((Z_{11} - \bar Z_1)^2\right) \tr(\Sigma^2)}}\right) \\ & = \frac{\lambda_j^2 }{\Var\left((Z_{11} - \bar Z_1)^2\right)  \tr(\Sigma^2)}    \Cov\left([Z_{ij} - \oZ_j]^2, [Z_{i'j} - \oZ_j]^2\right) \\ & = \frac{\lambda_j^2 }{\Var\left((Z_{11} - \bar Z_1)^2\right) \tr(\Sigma^2)}    \Cov\left([W_{i} - \oW]^2, [W_{i'}-\oW]^2\right).
			\end{split}
		\end{equation*}
        Thus, we obtain 
        \[		\Cov(\xi_{ij}, \xi_{ij}) = \frac{ \lambda_j^2 }{\tr(\Sigma^2)},\quad \text{for any }i\in [n], j\in [d].
		\]
        Further, due to \eqref{centered Z second moment} and an analogous direct expansion of $\EE\left(W_i - \oW \right)^4$, we have 
        \begin{align*}
		\Var \left([W_i - \oW]^2 \right) &= \EE\left(W_i - \oW \right)^4 - \left(\EE\left(W_i - \oW \right)^2 \right)^2\\
		& = \kappa_n \left(\frac{n - 1}{n}\right)^2  + \cO\left({1 \over n} \right) -\left (\frac{n - 1}{n}\right)^2\\
		& = (\kappa_n - 1) \left(\frac{n - 1}{n}\right)^2  + \cO\left({1 \over n} \right).
	\end{align*}
  Regarding the off-diagonal terms of $\Cov(\xi_{\cdot j})$, notice that, for any $i\ne i'\in [n]$,  
		\begin{align}\label{eq_cov_offdiag}\nonumber
			 &\Cov\left([W_{i} - \oW]^2, [W_{i'}-\oW]^2\right)\\ \nonumber
			 &~ =   \Cov(W_1^2 + \oW^2 - 2 W_1 \oW, \ W_2^2 + \oW^2 - 2 W_2 \oW)\\
			 & \stackrel{\text{i.i.d.}}
			= 2\Cov(W_1^2, \oW^2) - 4 \Cov(W_1^2, W_2 \oW) - 4 \Cov(\oW^2, W_2 \oW) + 4 \Cov(W_1 \oW, W_2 \oW).
		\end{align}
		The first term of the preceding display is
		\begin{equation}\label{eq_term_1_kurtosis}
			\begin{split}
				2\Cov(W_1^2, \oW^2) & ~= 2 \Cov\left(W_1^2, \ {1\over n^2}\left[ \sum_{k = 1}^n W_k^2 \ + \ \sum_{k \neq l} W_k W_l\right] \right) \\ & \stackrel{\text{indep.}}{=} 2 \left[{1\over n^2} \Cov(W_1^2, W_1^2) \ +  {2\over n^2} \sum_{k \neq 1}\left[\EE W_1^3 W_k - (\EE W_1^2) (\EE W_1 W_k) \right]\right] \\ & \stackrel{\text{indep.}}{=}   \frac{2(\kappa_n - 1)}{n^2}.
			\end{split}
		\end{equation}
		The second term in \eqref{eq_cov_offdiag} satisfies
		\begin{equation}\label{eq_term_2_kurtosis}
			\begin{split}
				  \Cov(W_1^2, W_2 \oW) &~  =   \Cov\left(W_1^2, {1\over n}W_2 \sum_{k = 1}^n W_k\right) \\
				  &\stackrel{\text{indep.}}{=}  {1\over n} \Cov(W_1^2, W_1 W_2) \\ & ~=   {1\over n} \left[\EE W_1^3 W_2 - (\EE W_1^2) (\EE W_1 W_2)\right] \\ & \stackrel{\text{indep.}}{=}   0.
			\end{split}
		\end{equation}
		Regarding the third term in \eqref{eq_cov_offdiag}, we find that 
		\begin{equation*}
			\begin{split}
				-4 \Cov(\oW^2, W_1 \oW) & = - {4\over n^3} \Cov\left( \sum_{i=1}^n W_i^2 + \sum_{i \neq k} W_i W_k, \sum_{i=1}^n W_1 W_i\right) \\ & = -{4\over n^3} \left[\sum_{i,k=1}^n \Cov(W_i^2, W_1 W_k) + \sum_{i \neq k} \sum_{j = 1}^n \Cov(W_i W_k, W_1 W_j)\right].
			\end{split}
		\end{equation*}
		Since 
		\begin{align*}
			\sum_{i,k=1}^n \Cov(W_i^2, W_1 W_k) &  = \sum_{i=1}^n \Cov(W_i^2, W_1 W_i) + \sum_{i\ne k}  \Cov(W_i^2, W_1 W_k) \\ & = \Cov(W_1^2, W_1^2) \\ & = \kappa_n - 1,
		\end{align*}
		and
		\begin{align*}
			\sum_{i \neq k} \sum_{j = 1}^n \Cov(W_i W_k, W_1 W_j) &= 
			\sum_{i \neq k} \sum_{j = 1}^n\left[
			\EE(W_1 W_k W_i W_j) - (\EE W_i W_k) (\EE W_1 W_j)
			\right]\\
			&= \sum_{i \neq k} \sum_{j = 1}^n 
			\EE (W_1 W_k W_i W_j)\\
			&= \sum_{k\ne 1}\sum_{j = 1}^n 
			\EE (W_1^2 W_k  W_j)+\sum_{i\ne 1} \sum_{k=1,k\ne i}^n \sum_{j = 1}^n 
			\EE (W_1 W_k W_i W_j)\\
			&= \sum_{k\ne 1}   
			\EE (W_1^2 W_k^2)+\sum_{i\ne 1} 
			\EE (W_1^2 W_i^2)\\
			&= 2(n-1),
		\end{align*} 
		we have 
		\begin{equation}\label{eq_term_3_kurtosis}
			-4 \Cov(\oW^2, W_1 \oW)  = -4 n^{-3}(\kappa_n - 1 + 2 (n - 1)) = -8 n^{-2} - 4n^{-3}(\kappa_n - 3).
		\end{equation} 
		Finally,  the   last term in \eqref{eq_cov_offdiag} satisfies
		\begin{equation}\label{eq_term_4_kurtosis}
			\begin{split}
				4 \Cov(W_1 \oW, W_2 \oW) & =  -4 n^{-2} \sum_{i, j=1}^n \Cov(W_1 W_i, W_2 W_j) \\ & = -4 n^{-2} \sum_{i, j=1}^n \left[ \EE W_1 W_i W_2 W_j - (\EE W_1 W_i) (\EE W_2 W_j)\right]\\
				&= -4 n^{-2} \left[
				 \EE W_1^2W_2^2 - (\EE W_1)^2 (\EE W_2)^2
				\right]\\
				&= -4n^{-2}.
			\end{split}
		\end{equation}
		 Collecting \eqref{eq_term_1_kurtosis} -- \eqref{eq_term_4_kurtosis} yields
		\begin{equation*}
			\begin{split}
				\Cov\left([W_{i} - \oW]^2, [W_{i'}-\oW]^2\right) &  = \frac{2(\kappa_n - 1)}{n^2} - \frac{8}{n^2} - \frac{4(\kappa_n - 3)}{n^3} + \frac{4}{n^2} \\ & = {2(\kappa_n - 3)(n - 2) \over n^3},
			\end{split}
		\end{equation*}
		which completes the proof. 

\end{proof}



\subsection{Proof of \cref{prop_Delta_Alternatives}} \label{app_Delta_Alternative_proof}

\begin{proof}

We proceed by considering the different specified alternative models as separate cases in \cref{app_sec_ratio_mixture}, \cref{app_sec_ratio_ellip}, and \cref{app_sec_ratio_lep}. Throughout the proof we will use the fact that $\widehat{\tr(\Sigma^2)}$, as defined in \eqref{def_tr_Sigma_hat} based on \cite{HimenoYamada}, can be equivalently expressed as
\begin{equation}\label{decomp_hat_tr_sq}
\begin{split}
    \widehat{\tr(\Sigma^2)} & = \frac{1}{n(n - 1)} \sum_{i \neq j}  [X_i^\T X_j]^2 \ - \ \frac{2}{n(n - 1)(n - 2)} \sum_{i \neq j \neq k}  X_i^\T X_j X_j^\T X_k \\ & \ \ \ \ \ \ \ + \frac{1}{n(n - 1)(n - 2)(n - 3)} \sum_{i \neq j \neq k \neq \ell}  X_i^\T X_j X_k^\T X_{\ell},
\end{split}
\end{equation}
which is the form of the estimator as originally presented in \cite{Chen2010}.

\subsubsection{Proof for \cref{model_subG} and \cref{finiteMixture_stochastic_rep}}\label{app_sec_ratio_mixture}
Recall $\Delta$ from \eqref{def_Delta_null} and that under either \cref{model_subG} or \cref{finiteMixture_stochastic_rep},
\[
    \Sigma  = \sum_{k < m}^K \pi_k \pi_m (\mu_k - \mu_m) (\mu_k - \mu_m)^{\T} + \sum_{k = 1}^K \pi_k \Sigma_k. 
\] 
Thus, 
\begin{align}\label{decomp_tr_Sigma}
    &\tr(\Sigma) = \sum_{k < l} \pi_k \pi_l \| \mu_k - \mu_l \|_2^2 + \sum_{k = 1}^K \pi_k \tr(\Sigma_k),\\\nonumber
    &\tr(\Sigma^2)  = \sum_{k < l, m < q}^K \pi_k \pi_l \pi_m \pi_q [(\mu_k - \mu_l)^{\T} (\mu_m - \mu_q)]^2 + \sum_{k,l=1}^K \pi_k \pi_l \tr(\Sigma_k \Sigma_l) \\\label{decomp_tr_Sigma_sq} & \ \ \ \ \ \ + \sum_{k, l, m=1}^K \pi_k \pi_l \pi_m (\mu_k - \mu_l)^{\T} \Sigma_m (\mu_k - \mu_l).
\end{align} 
We establish the result in two steps by showing
\begin{equation} \label{Epsilon_est_parts}
          \frac{\widehat{\tr(\Sigma^2)}}{\tr(\Sigma^2)}  = 1 + \cO_{\PP}\left({1 \over \sqrt{n}} \right),\qquad \frac{\tr(\wh \Sigma)}{\tr(\Sigma)} = 1 + \cO_{\PP}\left({1 \over \sqrt{n}} \right),
\end{equation}
from which the result follows after taking a Taylor expansion.

\noindent\textbf{Step 1: Ratio-Consistency of $\widehat{\tr(\Sigma^2)}$.} To prove the first result in \eqref{Epsilon_est_parts},  we first note that 
$$\EE\left(\widehat{\tr(\Sigma^2)}\right) = \tr(\Sigma^2).$$ 
See, for instance, \cite{Chen2010,SongChen2014}. By Chebyshev's inequality, it remains to show
\[
     \Var\left(\widehat{\tr(\Sigma^2)}\right)   =  \Var\left(\EE(\widehat{\tr(\Sigma^2)} \mid C_*)\right) + \EE\left(\Var(\widehat{\tr(\Sigma^2)} \mid C_*)\right)  =   \cO\left({\tr^2(\Sigma^2) \over n}\right).
\]
We invoke invariance of $\widehat{\tr(\Sigma^2)}$ under arbitrary location-transformation of the samples \cite{Chen2010}, so as to shift the samples by $-\mu_1$. This corresponds to evaluating $\widehat{\tr(\Sigma^2)}$ using the transformed samples as introduced in \eqref{def_Ti}; that is, 
$$T_i =   X_i - \mu_1 = \underbrace{\mu_{C_i} - \mu_1}_{=:~ \gamma_{C_i}} + \underbrace{\Gamma_{C_i}  Z_i}_{=: ~ Y_i}.$$ 

\noindent\textbf{Step 1a: Bounding $\Var(\EE(\widehat{\tr(\Sigma^2)} \mid C_*))$.}
    The decomposition in \eqref{decomp_hat_tr_sq} gives  
\begin{equation}\label{Conditional Exp - Trace Squared}
    \begin{split}
          \EE\left(\widehat{\tr(\Sigma^2)} \mid C_* \right) & = \frac{1}{n(n - 1)} \sum_{i \neq j}  \EE\left([T_i^\T T_j]^2 \mid C_* \right)\\
          &\quad -  \frac{2}{n(n - 1)(n - 2)} \sum_{i \neq j \neq k}  \EE\left(T_i^\T T_j T_j^\T T_k \mid C_* \right) \\ & \quad  + \frac{1}{n(n - 1)(n - 2)(n - 3)} \sum_{i \neq j \neq k \neq l}  \EE\left(T_i^\T T_j T_k^\T T_l \mid C_*\right),
    \end{split}
\end{equation}
so that $\Var(\EE(\widehat{\tr(\Sigma^2)} \mid C_*))$ can be bounded from above (in order)  by 
\begin{equation}\label{Variance Conditioned on Expectation}
    \begin{split}
    &\frac{1}{n^4} \Var\left(\sum_{i \neq j}  \EE([T_i^\T T_j]^2 \mid C_*) \right) + \frac{1}{n^6} \Var\left(\sum_{i \neq j \neq k}  \EE(T_i^{\T} T_j T_j^{\T} T_k \mid C_*) \right) \\ & \ \ \ \ \ \ \ + \frac{1}{n^8}  \Var\left(\sum_{i \neq j \neq k \neq l}  \EE(T_i^{\T} T_j T_k^{\T} T_l \mid C_*) \right).  
    \end{split}
\end{equation}
For the first sum, observe that
\begin{equation}\label{Var 73}
    \begin{split}
        &\Var\left(\sum_{i \neq j}^n  \EE([T_i^\T T_j]^2 \mid C_*) \right)\\ 
        & = ~ \sum_{i \neq j}^n \Var(\EE([T_i^\T T_j]^2 \mid C_*))   + \sum_{i \neq j \neq k}^n \Cov\Bigl(\EE([T_i^\T T_j]^2 \mid C_*), \ \EE([T_j^\T T_k]^2 \mid C_*) \Bigr) \\ & \leq  ~ \sum_{i \neq j} \EE\Bigl( (\EE([T_i^\T T_j]^2 \mid C_*))^2 \Bigr) \\ & \ \ \ \ \ \ + \sum_{i \neq j \neq k} \sqrt{\EE\Bigl( (\EE([T_i^\T T_j]^2 \mid C_*))^2 \Bigr) \EE\Bigl( (\EE([T_j^\T T_k]^2 \mid C_*))^2 \Bigr)}, 
    \end{split}
\end{equation}
where the summation over the covariance terms in the right-hand side of the first equality is taken over exactly three distinct indices, as opposed to both three and four distinct indices, as
\begin{equation*}
\begin{split}
    &\Cov\Bigl(\EE([T_i^\T T_j]^2 \mid C_*), \ \EE([T_k^\T T_l]^2 \mid C_*) \Bigr)\\ 
    & = \Cov\Bigl(\EE([T_i^\T T_j]^2 \mid C_i, C_j), \ \EE([T_k^\T T_l]^2 \mid C_k, C_l) \Bigr)  = 0,
\end{split}
\end{equation*}
for $i \neq j \neq k \neq l$ by the mutual independence of the $C_1,\ldots,C_n$ as well as the conditional independence $T_i^\T T_j \independent C^{-(j,k)}_* \mid (C_j, C_k)$ with $C^{-(j,k)}_*$ being the $(n - 2)$-dimensional sub-vector of $C_*$ with the $j^{\text{th}}$ and $k^{\text{th}}$ elements removed. Note that
\begin{equation}\label{Trace Squared First Sum}
    \begin{split}
        \EE\left([T_i^\T T_j]^2 \mid C_* \right) & = \EE\left([(\gamma_{C_i} + Y_i)^\T (\gamma_{C_j} + Y_j)]^2 \mid C_* \right) \\ & = (\gamma_{C_i}^\T \gamma_{C_j})^2 + \EE\left((Y_i^\T Y_j)^2 \mid C_* \right) + \EE\left((\gamma_{C_i}^{\T}Y_j)^2 \mid C_* \right) \\ & \ \ \ \  + \EE\left((\gamma_{C_j}^{\T}Y_i)^2 \mid C_* \right) + 2 \EE\left((\gamma_{C_j}^{\T}Y_i) (Y_i^{\T}Y_j) \mid C_* \right) \\ & \ \ \ \  + 2 \EE\left((\gamma_{C_i}^{\T}Y_j)(Y_j^{\T}Y_i) \mid C_* \right),
    \end{split}
\end{equation}
where we have used the fact that samples $i \neq j$ are independent and the fact that  $\EE\left(Y_i \mid C_* \right) = \EE\left(Y_j \mid C_* \right) = 0$ to reduce the final expression of \eqref{Trace Squared First Sum} to the final 6 terms, with the remaining expectations of the expansion immediately seen to be null. In evaluating these expectations, we often suppress conditioning in intermediate steps, but it is to be understood that we are conditiong on the random $C_*$. For $i \neq j \in [n]$, the first expectation of \eqref{Trace Squared First Sum} is
\begin{equation*}
    \begin{split}
        \EE\left((Y_i^\T Y_j)^2 \mid C_* \right) & = 
        \EE\left(
            Z_i^\T \Gamma_{C_i}^\T \Gamma_{C_j} Z_j Z_j^\T \Gamma_{C_j}^T \Gamma_{C_i} Z_i
        \right)\\
        &= \EE\left(
        Z_i^\T\Gamma_{C_i}^\T \Sigma_{C_j} \Gamma_{C_i} Z_i
        \right)
        \\ & = \tr(\Sigma_{C_i} \Sigma_{C_j}),
    \end{split}
\end{equation*}
while the second expectation is 
\begin{equation}
    \begin{split}
        \EE\left((\gamma_{C_i}^{\T}Y_j)^2 \mid C_* \right) = \EE\left(\gamma_{C_i}^{\T} \Gamma_{C_j}  Z_j  Z_j^{\T} \Gamma_{C_j}^{\T} \gamma_{C_i}\right)  = \gamma_{C_i}^{\T} \Sigma_{C_j} \gamma_{C_i}.
    \end{split}
\end{equation} 
Similarly, the third expectation equals $ \gamma_{C_j}^{\T} \Sigma_{C_i} \gamma_{C_j}$. Moreover, it is easy to see that the final two expectations are zero by the independence between the centered vectors $Z_i$ and $Z_j$.   
Thus, we conclude that
\begin{equation*}
    \begin{split}
        \EE\left([T_i^\T T_j]^2 \mid C_* \right) &  = [(\mu_{C_i} - \mu_1)^{\T} (\mu_{C_j} - \mu_1)]^2 + \tr(\Sigma_{C_i} \Sigma_{C_j}) \\ & \quad  + (\mu_{C_i} - \mu_1)^{\T} \Sigma_{C_j} (\mu_{C_i} - \mu_1)  + (\mu_{C_j} - \mu_1)^{\T} \Sigma_{C_i} (\mu_{C_j} - \mu_1)  \\ & \leq
        \max_{k \in [K]} \| \mu_k - \mu_1 \|_2^4 + \max_{k,\ell}\tr(\Sigma_k \Sigma_\ell) + 2\max_{k,\ell}(\mu_k - \mu_1)^{\T} \Sigma_\ell (\mu_k - \mu_1), 
    \end{split}
\end{equation*}
hence 
\begin{align}\label{Different from Cov-Mix 1} \nonumber 
        & \EE\left(\EE([T_i^\T T_j]^2 \mid C_*)^2 \right) \\ & \lesssim \max_{k \in [K]} \| \mu_k - \mu_1 \|_2^8 + \max_{k,\ell}\tr^2(\Sigma_k \Sigma_\ell) + 2\max_{k,\ell}[(\mu_k - \mu_1)^{\T} \Sigma_\ell (\mu_k - \mu_1)]^2. 
\end{align}
In conjunction with the fact that the upper-bound in \eqref{Var 73} can be bounded by $O(n^3)$ such terms as well as \eqref{decomp_tr_Sigma_sq}, we have 
\begin{align}\label{Different from Cov-Mix 2}\nonumber
        &\frac{1}{n^4} \Var\Bigl(\sum_{i \neq j}  \EE([T_i^\T T_j]^2 \mid C_*)
        \Bigr)\\ \nonumber
        & \lesssim ~ {1\over n}\left( \max_{k} \| \mu_k - \mu_1 \|_2^8 + \max_{k,\ell}\tr^2(\Sigma_k \Sigma_\ell)   + \max_{k,l,m} [(\mu_k - \mu_l)^{\T} \Sigma_m (\mu_k - \mu_l)]^2 \right) \\ & = \cO\left(\tr^2(\Sigma^2) \over n\right).
\end{align}
Similarly, up to the $\cO(n^6)$ scaling factor, the second term of \eqref{Variance Conditioned on Expectation} is 
\begin{equation}\label{Var 76}
    \begin{split}
        & \Var\Bigl(\sum_{i \neq j \neq k}  \EE(T_i^{\T} T_j T_j^{\T} T_k \mid C_*) \Bigr) \\ & = \sum_{i \neq j \neq k}   \Var(\EE(T_i^{\T} T_j T_j^{\T} T_k \mid C_*)) \\ & \quad  + \sum_{i \neq j \neq k \ \ast \  l \neq m \neq q} \Cov\Bigl( \EE(T_i^{\T} T_j T_j^{\T} T_k \mid C_*), \ \EE(T_l^{\T} T_m T_m^{\T} T_q \mid C_*) \Bigr) \\ &  = \sum_{i \neq j \neq k}   \Var(\EE(T_i^{\T} T_j T_j^{\T} T_k \mid C_*)) \\ &  \ \ \ \ + \sum_{i \neq j \neq k \ \ast \  i \neq m \neq q} \Cov\Bigl( \EE(T_i^{\T} T_j T_j^{\T} T_k \mid C_*), \ \EE(T_i^{\T} T_m T_m^{\T} T_q \mid C_*) \Bigr) \\ &  \ \ \ \ + \sum_{i \neq j \neq k \ \ast \  l \neq j \neq q} \Cov\Bigl( \EE(T_i^{\T} T_j T_j^{\T} T_k \mid C_*), \ \EE(T_l^{\T} T_j T_j^{\T} T_q \mid C_*) \Bigr)   \\ &  \leq \sum_{i \neq j \neq k} \EE \Bigl([\EE(T_i^{\T} T_j T_j^{\T} T_k \mid C_*)]^2 \Bigr)  \\ &  \ \ \ \ \ +  \sum_{i \neq j \neq k \ \ast \  i \neq m \neq q} \sqrt{ \EE \Bigl([\EE(T_i^{\T} T_j T_j^{\T} T_k \mid C_*)]^2 \Bigr) \  \EE\Bigl([\EE(T_i^{\T} T_m T_m^{\T} T_q \mid C_*)]^2 \Bigr) } \\ &   \ \ \ \ \  +  \sum_{i \neq j \neq k \ \ast \  l \neq j \neq q} \sqrt{ \EE \Bigl([\EE(T_i^{\T} T_j T_j^{\T} T_k \mid C_*)]^2 \Bigr) \  \EE\Bigl([\EE(T_l^{\T} T_j T_j^{\T} T_q \mid C_*)]^2 \Bigr) },
    \end{split}
\end{equation}
where the indexing notation $\{\{i, j, k, m, q \in [n] \mid i \neq j \neq k \ \ast \  i \neq m \neq q \}$ denotes $\{i, j, k, m, q \in [n] \mid i \neq j \neq k, i \neq m \neq q, \{j, k \} \neq \{m, q \} \}$, and analogously for $i \neq j \neq k \ \ast \  l \neq j \neq q$, in the covariance summations. Note that we have used the fact that, analogous to the reduction of covariance terms discussed for \eqref{Var 73}, the $O(n^6)$ covariance terms of \eqref{Var 76} reduces to only $O(n^5)$ non-null covariance summands. And, using the conditional independence of samples with indices $i \neq k \neq j$, 
\begin{equation*}
    \begin{split}
        \EE(T_i^\T T_j T_j^\T T_k \mid C_*) & = \EE\left(
            \gamma_{C_i}^\T \Gamma_{C_j}Z_jZ_j^\T \Gamma_{C_j}^\T \gamma_{C_k} \right)\\
        &  = \gamma_{C_i}^{\T} \Sigma_{C_j} \gamma_{C_k} + (\gamma_{C_j}^{\T} \gamma_{C_i})(\gamma_{C_j}^{\T} \gamma_{C_k})\\
        &  \leq \max_{q,r\in [K]} \gamma_q^{\T} \Sigma_r \gamma_q  + \max_{q\in [K]} \| \gamma_q \|_2^4.   
    \end{split}
\end{equation*}
Thus, by \eqref{decomp_tr_Sigma_sq}, the second term on the right-hand side of \eqref{Variance Conditioned on Expectation} is 
\begin{align}\label{Different from Cov-Mix 4}\nonumber
       \frac{1}{n^6} \Var\Bigl(\sum_{i \neq j \neq k}  \EE(T_i^{\T} T_j T_j^{\T} T_k \mid C_*)\Bigr) & ~ \lesssim~  {1\over n}\left[ \max_{q\in [K]} \| \gamma_q \|_2^8 + \max_{q,r\in [K]} \bigl(\gamma_q^{\T} \Sigma_r \gamma_q \bigr)^2\right] \\ & ~ = \cO\left(\tr^2(\Sigma^2) \over n \right). 
\end{align}
For the final term on the right-hand side of \eqref{Variance Conditioned on Expectation}, we use a bound analogous to that appearing in \eqref{Var 76}, where we again are able to reduce the $\cO(n^8)$ covariance terms to only $\cO(n^7)$ non-null covariances. Up to the $\cO(n^8)$ normalizing factor, this yields 
\begin{align*}
    & \Var\Bigl(\sum_{i \neq j \neq k \neq l}  \EE(T_i^{\T} T_j T_k^{\T} T_l \mid C_*)  \Bigr)   \leq  \sum_{i \neq j \neq k \neq l}  \EE \left(\EE(T_i^{\T} T_j T_k^{\T} T_l \mid C_*)^2 \right) \\ & \qquad\ \ \ +  \sum_{i \neq j \neq k \neq l \ \ast \ i \neq q \neq r \neq s } \sqrt{ \EE \left(\EE(T_i^{\T} T_j T_k^{\T} T_l \mid C_*)^2 \right) ~  \EE\left(\EE(T_i^{\T} T_q T_r^{\T} T_s \mid C_*)^2 \right) }.
\end{align*}
And, due to the independence of samples $i \neq j \neq k \neq l$,
\begin{equation*}
    \begin{split}
         \EE\left(T_i^\T T_j T_k^\T T_l \mid C_* \right) & =  \EE(T_i^\T T_j \mid C_*)   \EE(T_k^\T T_l \mid C_*)    = \gamma_{C_i}^\T \gamma_{C_j}  \gamma_{C_k}^\T \gamma_{C_l}\leq \max_{q\in [K]} \| \gamma_q \|_2^4.
    \end{split}
\end{equation*}
It then follows that the third sum of \eqref{Variance Conditioned on Expectation} is
\begin{equation}\label{Step 1 Term 3}
    \begin{split}
        \frac{1}{n^8} \Var\Bigl(\sum_{i \neq j \neq k \neq l}  \EE(T_i^{\T} T_j T_k^{\T} T_l \mid C_*)  \Bigr)  \lesssim \max_{q\in [K]}{ \| \gamma_q \|_2^8 \over n} = \cO\left(\tr^2(\Sigma^2) \over n\right).
    \end{split}
\end{equation}
In consideration of \eqref{Variance Conditioned on Expectation}, \eqref{Different from Cov-Mix 2}, \eqref{Different from Cov-Mix 4}, and \eqref{Step 1 Term 3}, we thus have
\begin{equation}\label{Variance of Conditional Exp}
    \begin{split}
       \Var\left(\EE(\widehat{\tr(\Sigma^2)} \mid C_*) \right) =  \cO\left(\tr^2(\Sigma^2) \over n\right). 
    \end{split}
\end{equation}

\noindent \textbf{Step 1b: Bounding $\EE(\Var(\widehat{\tr(\Sigma^2)} \mid C_*))$.} We begin by bounding $ \EE[(T_i^{\T} T_j)^4]$ for any $i \neq j \in [n]$, as this will be seen to be sufficient for controlling $\EE(\Var(\widehat{\tr(\Sigma^2)} \mid C_*))$. For any $i \neq j$,
\begin{equation}\label{Exp 84}
    \begin{split}
        &\EE\left([T_i^{\T} T_j]^4 \mid C_* \right)\\ 
        & = \EE\Bigl( \Bigl[(Y_i^{\T}Y_j + \gamma_{C_i}^{\T} Y_j) + (\gamma_{C_j}^{\T} Y_i + \gamma_{C_i}^{\T} \gamma_{C_j}) \Bigr]^4 \mid C_* \Bigr) \\ 
        & \lesssim \EE([Y_i^{\T}Y_j]^4 \mid C_*)  + \EE([\gamma_{C_i}^{\T} Y_j]^4 \mid C_*) + \EE([\gamma_{C_j}^{\T} Y_j]^4 \mid C_*) + \max_{q\in [K]} \| \gamma_q \|_2^8.
    \end{split}
\end{equation}
The first term on the right-hand side of \eqref{Exp 84} can be controlled via
\begin{equation*}
    \begin{split}
        \EE\left([Y_i^{\T}Y_j]^4 \mid C_* \right) & \leq  B_{C_i, C_j} \tr^2(\Sigma_{C_i} \Sigma_{C_j}) + B^*_{C_i, C_j} \tr([\Sigma_{C_i} \Sigma_{C_j}]^2) \\ & \le   B_1 \max_{q,r}  \tr^2(\Sigma_q \Sigma_r) + B_2 \max_{q} \tr^2(\Sigma^2_q)\\
        &\lesssim \max_{q,r}  \tr^2(\Sigma_q \Sigma_r)
    \end{split}
\end{equation*}
for some constants $(B_{q,r})_{q, r \in [K]}$ and $(B^*_{q,r})_{q, r \in [K]}$ as well as $B_1 := \max_{q,r \in [K]} B_{q, r}$ and $B_2 := \max_{q,r \in [K]} B^*_{q, r}$, following page 831 of \cite{ChenQin} in conjunction with Theorem 1 of \cite{Trace_Ineq}. Next, by writing $a_{ij}  := \Gamma_{C_j}^{\T} \gamma_{C_i}$, we have  
\begin{equation*}
    \begin{split}
        \EE\left([\gamma_{C_i}^{\T} Y_j]^4 \mid C_* \right) & =\EE\left([a_{ij}^{\T}  Z_j]^4 \mid C_* \right) \\ & = \EE\left([ Z_j^{\T} a_{ij} a_{ij}^{\T}  Z_j]^2 \mid C_* \right) \\ & =\tr^2(a_{ij} a_{ij}^{\T}) + 2 \tr([a_{ij} a_{ij}^{\T}]^2) + (\kappa_{C_j} - 3) \tr\left((a_{ij} a_{ij}^{\T}) \odot (a_{ij} a_{ij}^{\T})\right) \\ & = 3 \tr^2( \gamma_{C_i}^{\T} \Gamma_{C_j} \Gamma_{C_j}^{\T} \gamma_{C_i} \gamma_{C_i}^{\T} \Gamma_{C_j} \Gamma_{C_j}^{\T} \gamma_{C_i}) + (\kappa_{C_j} - 3) \sum_{q = 1}^{m_{C_j}} (a_{ij})_q^4 \\ & \leq 3(\gamma_{C_i}^{\T} \Sigma_{C_j} \gamma_{C_i})^2 + (\kappa_{C_j} - 3)_+ \Bigl(\sum_{q = 1}^{m_{C_j}} (a_{ij})_q^2 \Bigr)^2 \\ & = 3(\gamma_{C_i}^{\T} \Sigma_{C_j} \gamma_{C_i})^2 + (\kappa_{C_j} - 3)_+ \tr^2(a_{ij} a_{ij}^{\T}) \\ & = 3(\gamma_{C_i}^{\T} \Sigma_{C_j} \gamma_{C_i})^2 + (\kappa_{C_j} - 3)_+ (\gamma_{C_i}^{\T} \Sigma_{C_j} \gamma_{C_i})^2 \\ & \leq (\kappa_{C_j} + 3)(\gamma_{C_i}^{\T} \Sigma_{C_j} \gamma_{C_i})^2 \\ & \leq \max_{q \in [K]} (\kappa_q + 3)  \max_{q, r \in [K]} (\gamma_q^{\T} \Sigma_r \gamma_q)^2,
    \end{split}
\end{equation*}
where $\odot$ denotes the element-wise Hadamard product, $(x)_+ := \max\{0, x\}$, and the third step follows from Proposition A.1 of \cite{Chen2010}. Thus,
\begin{equation*}
    \begin{split}
        \EE\left(\EE([\gamma_{C_i}^{\T} Y_j]^4 \mid C_*) \right) ~ \lesssim~  \max_{q,r,s} \left[(\mu_q - \mu_r)^{\T} \Sigma_s (\mu_q - \mu_r)\right]^2, 
    \end{split}
\end{equation*}
and the same bound holds for $ \EE(\EE([\gamma_{C_j}^{\T} Y_i]^4 \mid C_*))$ analogously. Combining the preceding bounds for the terms of \eqref{Exp 84} yields
\begin{align}\label{Exp 84 Unconditional}\nonumber
        \EE\left([T_i^{\T} T_j]^4  \right)   & \lesssim \max_{k \in [K]} \| \gamma_k \|_2^8 + \max_{k,l \in [K]}  \tr^2(\Sigma_k \Sigma_l)   + \max_{k,l,m} \left[(\mu_k - \mu_l)^{\T} \Sigma_m (\mu_k - \mu_l)\right]^2 \\ & \lesssim \tr^2(\Sigma^2).
\end{align}
Next, by defining 
\begin{align*}
    &S_1:= \frac{1}{n(n - 1)} \sum_{i \neq j}  [T_i^\T T_j]^2,\\
    &S_2 :=   -   \frac{2}{n(n - 1)(n - 2)} \sum_{i \neq j \neq k} T_i^\T T_j T_j^\T T_k,\\ 
    &S_3 := \frac{1}{n(n - 1)(n - 2)(n - 3)} \sum_{i \neq j \neq k \neq l}T_i^\T T_j T_k^\T T_l,
\end{align*}
we have, by Cauchy-Schwartz inequality,
\begin{equation*}
    \begin{split}
        \Var(\widehat{\tr(\Sigma^2)} \mid C_*) & \le \sum_{m = 1}^3 \Var(S_m \mid C_*) \ + \  3 \max_{1\le m \le 3} \Var(S_m \mid C_*) \\ & \le  \sum_{m = 1}^3 \Var(S_m \mid C_*) \ + \  3  \sum_{m = 1}^3 \Var(S_m \mid C_*),  
    \end{split}
\end{equation*}
so that 
\begin{equation*}
        \EE(\Var(\widehat{\tr(\Sigma^2)} \mid C_*)) \leq 4 \sum_{m = 1}^3 \EE(\Var(S_m \mid C_*)).  
\end{equation*}
First, note that
\begin{equation*}
    \begin{split}
        &\Var(S_1 \mid C_*)\\
        & =  \frac{1}{n^2(n-1)^2} \Bigl( \sum_{i \neq j}  \Var([T_i^{\T} T_j]^2 \mid C_*) \ + \ \sum_{i \neq j \neq k} \Cov([T_i^{\T} T_j]^2, [T_i^{\T} T_k]^2  \mid C_*) \Bigr) \\ & \leq \frac{1}{n^2(n-1)^2} \Bigl( \sum_{i \neq j}  \Var([T_i^{\T} T_j]^2 \mid C_*) \ + \ \sum_{i \neq j \neq k} | \Cov([T_i^{\T} T_j]^2, [T_i^{\T} T_k]^2  \mid C_*) | \Bigr) \\ & \leq  {1\over n^2(n - 1)^2} \left[ \sum_{i \neq j}  \EE([T_i^{\T} T_j]^4 \mid C_*) \ + \ \sum_{i \neq j \neq k} \sqrt{\EE([T_i^{\T} T_j]^4 \mid C_*) \EE([T_i^{\T} T_k]^4 \mid C_*)}\right],
    \end{split}
\end{equation*}
where as before, we are able to make the reduction from $O(n^4)$ covariances to $O(n^3)$ non-null covariance terms. Thus, by \eqref{Exp 84 Unconditional} and \eqref{decomp_tr_Sigma_sq}, we have
\begin{equation*}
    \begin{split}
          \EE\left(\Var(S_1 \mid C_*) \right)  & \lesssim {1\over n}\left\{\max_{k \in [K]} \| \gamma_k \|_2^8 + \max_{k,l \in [K]}  \tr^2(\Sigma_k \Sigma_l) + \max_{k,l,m}\left[ (\mu_k - \mu_l)^{\T} \Sigma_m (\mu_k - \mu_l)\right]^2 \right\} \\ & = \cO\left(\tr^2(\Sigma^2) \over n \right),
    \end{split}
\end{equation*}
Similarly, using 
$$
    \Var(T_i^{\T} T_j T_j^{\T} T_k \mid C_*) \leq \EE([T_i^{\T} T_j T_j^{\T} T_k]^2 \mid C_*) \leq \sqrt{\EE([T_i^{\T} T_j]^4 \mid C_*) \EE([T_j^{\T} T_k]^4 \mid C_*)},
$$ and the fact that the $O(n^6)$ covariance terms arising from $\Var(S_2 | C_*)$ reduces to $O(n^5)$, in contrast to the $O(n^6)$ normalizing factor, we analogously have
\begin{equation*} 
    \EE\left(\Var(S_2 \mid C_*) \right)  + \EE\left(\Var(S_3 \mid C_*) \right)  
     = \cO\left(\tr^2(\Sigma^2) \over n \right). 
\end{equation*}
Thus, in conjunction with \eqref{Variance of Conditional Exp}, it follows that
\begin{equation*}
    \begin{split}
         \Var\left(\widehat{\tr(\Sigma^2)} \right) & = \Var\left(\EE(\widehat{\tr(\Sigma^2)} \mid C_*) \right) + \EE\left(\Var(\widehat{\tr(\Sigma^2)} \mid C_*) \right) = \cO\left(\tr^2(\Sigma^2) \over n\right). 
    \end{split}
\end{equation*}

\noindent\textbf{Step 2: Ratio-Consistency of $\tr(\wh \Sigma)$.} We prove the second claim in \eqref{Epsilon_est_parts} by first noting that 
\[
    \EE\left(\tr(\wh \Sigma)\right) = \tr(\Sigma),
\]
based on, for example, \cite{Chen2010}, and by establishing
\begin{equation*}
    \begin{split}
        \Var(\tr(\wh \Sigma)) =\Var\left(\EE(\tr(\wh \Sigma) \mid C_*) \right) + \EE\left(\Var(\tr(\wh \Sigma) \mid C_*) \right)   = \cO\left(\tr^2(\Sigma) \over n \right).
    \end{split}
\end{equation*}
Since $\tr(\wh \Sigma)$ is invariant under arbitrary translation of the samples, we consider the samples $T_1,\ldots, T_n$ as defined in \textbf{Step 1}, and note that 
\begin{equation}\label{decomp_tr_S} 
    \tr(\wh \Sigma) =  \frac{1}{n} \sum_{i = 1}^n T_i^{\T} T_i - \frac{1}{n (n - 1)} \sum_{i \neq j} T_i^{\T} T_j.
\end{equation}

\noindent\textbf{Step 2a: Bounding $\Var(\EE(\tr(\wh \Sigma) \mid C_*))$.} Since
\begin{equation*} 
        \EE(\tr(\wh \Sigma) \mid C_*) = \frac{1}{n} \sum_{i = 1}^n \EE(T_i^{\T} T_i \mid C_*) - \frac{1}{n (n - 1)} \sum_{i \neq j} \EE(T_i^{\T} T_j \mid C_*),
\end{equation*}
we have 
\begin{align}\label{Variance of Conditional Trace Exp}\nonumber 
        & \Var(\EE(\tr(\wh \Sigma) \mid C_*)) \\\nonumber
        & \lesssim \frac{1}{n^2} \left(\sum_{i = 1}^n \Var(\EE(\| T_i \|_2^2 \mid C_*)) +  \sum_{i \neq j} \Cov(\EE(\| T_i \|_2^2 \mid C_*), \EE(\| T_j \|_2^2 \mid C_*)) \right)  \\\nonumber
        & \ \ \ \ \ \ + \frac{1}{n^4}\left( \sum_{i \neq j} \Var(\EE(T_i^{\T} T_j \mid C_*)) +  \sum_{i \neq j \neq k} \Cov(\EE(T_i^{\T} T_j \mid C_*), \EE(T_i^{\T} T_k \mid C_*)) \right) \\\nonumber
        & \lesssim  \frac{1}{n^2} \sum_{i = 1}^n \Var(\EE(\| T_i \|_2^2 \mid C_*)) \\\nonumber & \ \ \ \ + \frac{1}{n^4}\left( \sum_{i \neq j} \Var(\EE(T_i^{\T} T_j \mid C_*)) +  \sum_{i \neq j \neq k} \Cov(\EE(T_i^{\T} T_j \mid C_*), \EE(T_i^{\T} T_k \mid C_*)) \right) \\
        & \lesssim \frac{1}{n^2} \sum_{i = 1}^n \EE\left(\EE(\| T_i \|_2^2 \mid C_*)^2 \right) \\\nonumber
        & \ \ \ \ + \frac{1}{n^4}\left( \sum_{i \neq j} \EE\left(\EE(T_i^{\T} T_j \mid C_*)^2 \right) +  \sum_{i \neq j \neq k} \sqrt{\EE\left(\EE(T_i^{\T} T_j \mid C_*)^2 \right)\EE\left(\EE(T_i^{\T} T_k \mid C_*)^2 \right)} \right),  
\end{align}
where the reduction of the $O(n^4)$ covariance terms $\Cov(\EE(T_i^{\T} T_j \mid C_*), \EE(T_k^{\T} T_l \mid C_*))$ to $O(n^3)$ non-null covariances $\Cov(\EE(T_i^{\T} T_j \mid C_*), \EE(T_i^{\T} T_k \mid C_*))$ follows in the same manner as in the preceding, and 
$$
    \Cov(\EE(\| T_i \|_2^2 \mid C_*), \EE(\| T_j \|_2^2 \mid C_*)) = \Cov(\EE(\| T_i \|_2^2 \mid C_i), \EE(\| T_j \|_2^2 \mid C_j)) = 0
$$ holds due to the conditional independence of $\| T_i \|_2^2$ and $(C_1,\ldots,C_{i - 1}, C_{i + 1}, \ldots, C_n)$ given $C_i$ and the independence of $C_i$ and $C_j$ for $i \neq j$. Using the general expression for the expectation of quadratic forms, we find that the conditional expectations in the summands of the first term of the upper-bound in \eqref{Variance of Conditional Trace Exp} satisfy
\begin{equation*} 
        \EE(\| T_i \|_2^2 \mid C_*)  =  \tr(\Sigma_{C_i}) + \| \gamma_{C_i} \|_2^2 \leq \max_{k \in [K]} \| \gamma_k \|_2^2 + \max_{k \in [K]} \tr(\Sigma_k), 
\end{equation*}
implying
\begin{equation}\label{Iterated Expectation - Squared Norm}
         \EE([\EE(\| T_i \|_2^2 \mid C_*)]^2) \lesssim \max_{k \in [K]} \| \gamma_k \|_2^4 + \max_{k \in [K]} \tr^2(\Sigma_k). 
\end{equation}
For the summands of the second and third terms in the bound of \eqref{Variance of Conditional Trace Exp}, we have  
\begin{equation*} 
         \EE(T_i^{\T} T_j \mid C_*)  =   \gamma_{C_i}^{\T} \gamma_{C_j}  \le  \max_{k \in [K]}\| \gamma_k \|_2^2
\end{equation*}
 for $i \neq j$, which entails
\begin{equation*} 
         \EE([\EE(T_i^{\T} T_j \mid C_*)]^2) \leq \max_{k \in [K]} \| \gamma_k \|_2^4. 
\end{equation*}
Thus, \eqref{Variance of Conditional Trace Exp} together with \eqref{decomp_tr_Sigma} yields
\begin{equation}\label{Trace - Variance Conditional Expectation - Covariance-Mixtures} 
        \Var\left(\EE(\tr(\wh \Sigma) \mid C_*) \right) \lesssim {1\over n}\left(\max_{k \in [K]} \| \gamma_k \|_2^4 + \max_{k \in [K]} \tr^2(\Sigma_k)\right)   = \cO\left(\tr^2(\Sigma)\over n \right). 
\end{equation}

\noindent \textbf{Step 2b: Bounding $\EE(\Var(\tr(\wh \Sigma) \mid C_*))$.}  By  
$$
    \Cov(\| T_i \|_2^2, \| T_j \|_2^2 \mid C_*) = 0, \qquad \Cov(T_i^{\T} T_j, T_k^{\T} T_l \mid C_*) = 0$$ 
for any $i \neq j \neq k \neq l$, we find that
\begin{align}\label{Variance Trace Conditioned} \nonumber
        &\Var(\tr(\wh \Sigma) \mid C_*)\\ \nonumber
        & \le   \frac{3}{n^2} \left(\sum_{i = 1}^n \Var(\| T_i \|_2^2 \mid C_*) +  \sum_{i \neq j} \Cov(\| T_i \|_2^2, \| T_j \|_2^2 \mid C_*) \right) \\\nonumber & \ \ \ \ \ \ + \frac{3}{n^2 (n - 1)^2} \left( \sum_{i \neq j} \Var(T_i^{\T} T_j \mid C_*) +  \sum_{i \neq j \neq k} | \Cov(T_i^{\T} T_j, T_i^{\T} T_k \mid C_*) | \right) \\\nonumber 
        & =  \frac{3}{n^2} \sum_{i = 1}^n \Var(\| T_i \|_2^2 \mid C_*) \\ & \ \ \ \ \ \ + \frac{3}{n^2 (n - 1)^2} \left( \sum_{i \neq j} \Var(T_i^{\T} T_j \mid C_*) +  \sum_{i \neq j \neq k} | \Cov(T_i^{\T} T_j, T_i^{\T} T_k \mid C_*) | \right). 
\end{align}
By \eqref{var_two} with $\gamma_k = \mu_k - \mu_1$, we also have 
\begin{align}\label{Expectation Conditional Variance - Norm - Covariance Mixture} \nonumber
         \EE(\Var(\| T_i \|_2^2 \mid C_*)) & ~ \lesssim ~ \max_{k \in [K]}  \tr(\Sigma^2_k) + \max_{k \in [K]} \gamma_k^\T \Sigma_k \gamma_k \\ \nonumber
         & ~ \le ~ \max_{k \in [K]}  \tr^2(\Sigma_k) + \max_{k\in[K]}\|\gamma_k\|_2^2 \|\Sigma_k\|_\op &&\text{by \cref{lem_ranks}}\\ 
         &~  \lesssim ~ \tr^2(\Sigma) &&\text{by \eqref{decomp_tr_Sigma}}.
\end{align}
Further, by \eqref{var_three},  we have that for any $i\ne j$,
\begin{equation}\label{Expectation Conditional Variance - Inner-Product - Covariance Mixture}
    \begin{split}
        \EE(\Var(T_i^{\T} T_j \mid C_*)) & \leq 4 \max_{k \in [K]} \tr(\Sigma^2_k) + 8 \max_{k,l \in [K]} \gamma_k^\T \Sigma_l \gamma_k 
        \lesssim  \tr^2(\Sigma).
    \end{split}
\end{equation}
Finally, using 
\begin{equation*}
    \begin{split}
          \sum_{i \neq j \neq k} \Cov(T_i^{\T} T_j, T_i^{\T} T_k \mid C_*)  & \le  n (n - 1) (n -  2)\max_{i \neq j \in [n]} \Var(T_i^{\T} T_j \mid C_*) \\ & \le  4 n (n - 1) (n -  2) \left( \max_{k \in [K]} \tr(\Sigma^2_k) + 2 \max_{k,l \in [K]} \gamma_k^\T \Sigma_l \gamma_k \right),
    \end{split}
\end{equation*}
gives
\begin{equation*}
    \begin{split}
        \frac{1}{n^4} \EE\Bigl(\sum_{i \neq j \neq k} \Cov(T_i^{\T} T_j, T_i^{\T} T_k \mid C_*) \Bigr) & \lesssim {1 \over n} \left(\max_{k \in [K]} \tr(\Sigma^2_k) +  \max_{k,l \in [K]} \gamma_k^\T \Sigma_l \gamma_k \right)  \lesssim {\tr^2(\Sigma) \over n}.
    \end{split}
\end{equation*}
Collecting these results, \eqref{Variance Trace Conditioned} yields
\begin{equation*}
    \begin{split}
        \EE\left(\Var(\tr(\wh \Sigma) \mid C_*) \right) & ~ \lesssim ~  \frac{1}{n^2} \sum_{i = 1}^n  \EE(\Var(\| T_i \|_2^2 \mid C_*)) \\ & \quad  + {1\over n^4} \left( \sum_{i \neq j}  \EE(\Var(T_i^{\T} T_j \mid C_*)) +  \sum_{i \neq j \neq k}  \EE(\Cov(T_i^{\T} T_j, T_i^{\T} T_k \mid C_*))\right) \\ &~ = ~\cO\left( {\tr^2(\Sigma) \over n} +  {\tr^2(\Sigma) \over n^2}  +{\tr^2(\Sigma) \over n} \right).
    \end{split}
\end{equation*}
Thus, combining the results of \textbf{Step 2a} and \textbf{Step 2b},
\begin{equation*}
    \begin{split}
        \Var(\tr(\wh \Sigma)) \ = \ \Var\left(\EE(\tr(\wh \Sigma) \mid C_*) \right) + \EE\left(\Var(\tr(\wh \Sigma) \mid C_*) \right) \ = \ \cO\left(\tr^2(\Sigma) \over n \right),
    \end{split}
\end{equation*}
thereby completing the proof for \cref{finiteMixture_stochastic_rep}. 

\subsubsection{Proof for \cref{elliptical_stochastic_rep}}\label{app_sec_ratio_ellip}

\begin{proof}
We establish the result by showing that
\begin{equation} \label{Epsilon_est_parts_ellip}
         \frac{\widehat{\tr(\Sigma^2)}}{\tr(\Sigma^2)}  = 1 + \cO_{\PP}\left({1 \over \sqrt{n}} \right),\qquad \frac{\tr(\wh \Sigma)}{\tr(\Sigma)} = 1 + \cO_{\PP}\left({1 \over \sqrt{n}} \right),
\end{equation}
from which the result follows after taking a Taylor expansion. Recall that under \cref{elliptical_stochastic_rep}, $\Sigma = \EE[\eps^2] \Sigma_*$ whence  
\[
    \Delta  = {2 \tr(\Sigma^2) \over \tr(\Sigma)} = \frac{2 \EE[\eps^2] ~ \tr(\Sigma^2_*)}{\tr(\Sigma_*)}. 
\]
First, note that by the location and unitary invariance properties of the proposed test statistics as discussed in \cref{rem_invariance} and the rotational invariance of standard Gaussian random vectors, we can without loss of generality assume that 
\begin{equation}\label{Elliptical Conditional Representation_Delta_prop}
    X_i \stackrel{\rm d}{=} \eps_i~ \Lambda^{\frac{1}{2}}_*  Z_i,
\end{equation}
where $ Z_i \stackrel{\text{i.i.d.}}{\sim} \cN_d(0_d, \bI_d)$ for $i \in [n]$ and $\Lambda_*$ is the diagonal matrix consisting of the non-increasing eigenvalues $\lambda_1 \geq \ldots \geq \lambda_d$ of $\Sigma_*$. We make use of the fact that the unconditional stochastic representation of \eqref{Elliptical Conditional Representation_Delta_prop} is 
\begin{equation}\label{Elliptical Unconditional Representation} 
        X_i \stackrel{\rm d}{=} \sqrt{\EE[\eps^2]} ~ \Lambda^{\frac{1}{2}}_*  S_i, 
\end{equation}
for $i \in [n]$, where $S_1, \ldots, S_n$ are i.i.d. from $\text{E}_d(0,  (\EE[\eps^2])^{-1} \bI_d)$, an elliptical distribution centered at zero with covariance matrix $\bI_d$ \cite{Brenner}. As a result, each $S_i$ is a rotationally invariant isotropic random vector, with $\EE[S_i] = 0_d$ and $\Cov(S_i) = \bI_d$, for $i \in [n]$.\\

\noindent\textbf{Step 1: Ratio-consistency of $\widehat{\tr(\Sigma^2)}$.}
Since $\Sigma := \Cov(X) = \EE[\epsilon^2] \Sigma_*$, we have $\EE (\widehat{\tr(\Sigma^2)}) = \left( \EE[\epsilon^2] \right)^2 \tr(\Sigma^2_*) =\tr(\Sigma^2)$ \cite{Chen2010,HimenoYamada}. Thus, it remains to prove
\begin{align}\label{ellip_trSq_relVar} 
   \Var\left(\widehat{\tr(\Sigma^2)} \right)  & =  \cO\left( \left(\EE[\epsilon^2] \right)^4 \tr^2(\Sigma^2_*)  \over n \right)   = \cO\left(\tr^2(\Sigma^2)  \over n\right) 
\end{align}
to establish the ratio-consistency property for $\widehat{\tr(\Sigma^2)}$. In consideration of \eqref{Elliptical Unconditional Representation}, Lemma 1 of \cite{HimenoYamada} gives
\begin{align}\label{Variance Trace Squared Elliptical} \nonumber
        &\Var(\widehat{\tr(\Sigma^2)})\\\nonumber 
        & \lesssim \left(\EE[\epsilon^2] \right)^4 \Bigl( \frac{\EE (S_1^{\T} \Lambda_* S_2)^4}{n^2} + \frac{| \EE (S_1^{\T} \Lambda_*^2 S_1)^2 - 2 \tr(\Sigma^4_*) - \tr^2(\Sigma^2_*) |}{n} \Bigr) \\\nonumber 
        & \ \ \ \ \ \ \  +  \left(\EE[\epsilon^2] \right)^4 \Bigl( \frac{|\EE (S_1^{\T} \Lambda_* S_2)^2 (S_1^{\T} \Lambda^2_* S_2)|}{n^3} + \frac{\tr^2(\Sigma^2_*)}{n^2} + \frac{\tr(\Sigma^4_*)}{n} \Bigr) \\ 
        & \lesssim  \left(\EE[\epsilon^2] \right)^4 \Bigl( \frac{\EE (S_1^{\T} \Lambda_* S_2)^4}{n^2} + \frac{\EE (S_1^{\T} \Lambda^2_* S_1)^2}{n} + \frac{\sqrt{\EE (S_1^{\T} \Lambda_* S_2)^4 \EE (S_1^{\T} \Lambda^2_* S_2)^2}}{n^3} + {\tr^2(\Sigma^2_*)\over n} \Bigr). 
\end{align} 
By the independence of $S_1$ and $S_2$, we further have
\begin{align}\label{Elliptical Inner-Prod 1} \nonumber
        \EE (S_1^{\T} \Lambda_* S_2)^4 & = (\EE S^4_{11})^2 \tr(\Sigma^4_*) \ + \ (\EE S^3_{11} S_{12})^2 \sum_{j \neq k} \lambda^3_j \lambda_k \ + \ (\EE S^2_{11} S^2_{12})^2 \sum_{j \neq k} \lambda^2_j \lambda^2_k \\\nonumber 
        & \ \ \ \ + (\EE S^2_{11} S_{12} S_{13})^2 \sum_{j \neq k \neq l} \lambda^2_j \lambda_k \lambda_l \ + \  (\EE S_{11} S_{12} S_{13} S_{14})^2 \sum_{j \neq k \neq l \neq m} \lambda_j \lambda_k \lambda_l \lambda_m \\ \nonumber
        & = (\EE S^4_{11})^2 \tr(\Sigma^4_*) \ + \  (\EE S^2_{11} S^2_{12})^2 \sum_{j \neq k} \lambda^2_j \lambda^2_k \\ 
        & \lesssim \tr^2(\Sigma^2_*), 
\end{align}
where we used the fact that the fourth moments of $S_1$ exist and are uniformly bounded, due to the moment conditions of \cref{elliptical_stochastic_rep}. Note this in turn implies that all the moments appearing in \eqref{Elliptical Inner-Prod 1} exist and that the product moments involving at least one odd power are zero, due to the fact that $S_1$ is rotationally invariant \cite{Brenner}. Similarly,  
\begin{equation}\label{Elliptical Inner-Prod 2} 
        \EE (S_1^{\T} \Lambda^2_* S_1)^2  = (\EE S^4_{11}) \tr(\Sigma^4_*) \ + \ \EE S^2_{11} S^2_{12} \sum_{j \neq k} \lambda^2_j \lambda^2_k = \cO\left( \tr^2(\Sigma^2_*) \right). 
\end{equation}
Finally, we also have
\begin{equation}\label{Elliptical Inner-Prod 3}
        \EE (S_1^{\T} \Lambda^2_* S_2)^2 = (\EE S^2_{11})^2 \tr(\Sigma^4_*) + (\EE S_{11} S_{12})^2 \sum_{j \neq k} \lambda^2_j \lambda^2_k = \tr(\Sigma^4_*) \le   \tr^2(\Sigma^2_*), 
\end{equation}
where we have made use of the fact that $S_1$ is isotropic. Thus, \eqref{Variance Trace Squared Elliptical} in conjunction with \eqref{Elliptical Inner-Prod 1}, \eqref{Elliptical Inner-Prod 2}, and \eqref{Elliptical Inner-Prod 3} yields
\begin{align}
    \Var\left(\widehat{\tr(\Sigma^2)} \right)  & =  \cO\left( \left(\EE[\epsilon^2] \right)^4 \tr^2(\Sigma^2_*)  \over n \right)   = \cO\left(\tr^2(\Sigma^2)  \over n\right) 
\end{align}
thereby completing proof of this step, in light of the remarks pertaining to \eqref{ellip_trSq_relVar}. \\

\noindent\textbf{Step 2: Ratio-consistency of $\tr(\wh \Sigma)$ for $\tr(\Sigma)$.} First, note that for the unconditional covariance matrix $\Sigma$, we have unbiased estimation $\EE \tr(\wh \Sigma) = \EE[\epsilon^2] \tr(\Sigma_*) = \tr(\Sigma)$ \cite{Chen2010}. To establish ratio-consistency, it therefore only remains to show that 
\begin{align}\label{Elliptical Trace}
    \Var(\tr(\wh \Sigma)) & = \cO\left( \left( \EE[\epsilon^2] \right)^2 \tr^2(\Sigma_*) \over n\right)   = \cO\left( \tr^2(\Sigma) \over n\right).
\end{align}
By \eqref{decomp_tr_S}, we obtain
\begin{align*}
     \Var(\tr(\wh \Sigma)) & ~ \lesssim~ \Var\Bigl(\frac{1}{n} \sum_{i = 1}^n X_i^{\T} X_i\Bigr) + \Var\Bigl(\frac{1}{n (n - 1)} \sum_{i \neq j } X_i^{\T} X_j\Bigr)   \\ 
     & ~ \lesssim~  \frac{\left( \EE[\epsilon^2] \right)^2}{n^2} \Var\Bigl(\sum_{i = 1}^n \| \Lambda^{\frac{1}{2}}_* S_i \|_2^2 \Bigr) + \frac{\left( \EE[\epsilon^2] \right)^2}{n^4} \Var\Bigl(\sum_{i \neq j } S_i^{\T} \Lambda_* S_j  \Bigr)  \\ & ~ \lesssim~ { \left( \EE[\epsilon^2] \right)^2 \over n}\Var\Bigl(\| \Lambda^{\frac{1}{2}}_* S_1 \|_2^2 \Bigr)   + {\left( \EE[\epsilon^2] \right)^2 \over n^4}\Var\Bigl(\sum_{i \neq j } S_i^{\T} \Lambda_* S_j \Bigr), 
\end{align*}
since the $S_i$ are \text{i.i.d.}, for $i \in [n]$. And, using the fact that the fourth moments of $S_1$ are uniformly bounded and $S_1$ is isotropic,
\begin{equation*}
    \begin{split}
         \Var(\| \Lambda^{\frac{1}{2}}_* S_1 \|_2^2) & = \EE \| \Lambda^{\frac{1}{2}}_* S_1 \|_2^4 -  \bigl(\EE \| \Lambda^{\frac{1}{2}}_* S_1 \|_2^2 \bigr)^2 \\ & = \EE (S_1^{\T} \Lambda_* S_1)^2 - \tr^2(\Sigma_*) \\ & = \EE S^4_{11} \sum_{j = 1}^d \lambda^2_j \ + \ \EE S^2_{11} S^2_{12} \sum_{j \neq k} \lambda_j \lambda_k \ - \ \tr^2(\Sigma_*) \\ & = (\EE S^4_{11} - 1) \tr(\Sigma^2_*) \\ & \lesssim \tr^2(\Sigma_*).
    \end{split}
\end{equation*}
Moreover, 
\begin{align*}
       \Var\Bigl(\sum_{i \neq j } S_i^{\T} \Lambda_* S_j \Bigr) & = \sum_{i \neq j } \Var(S_i^{\T} \Lambda_* S_j) \ + \ \sum_{i \neq j} \sum_{k \neq l} \Cov(S_i^{\T} \Lambda_* S_j, S_k^{\T} \Lambda_* S_l) \\ & = \sum_{i \neq j } \Var(S_i^{\T} \Lambda_* S_j) \ + \ \sum_{i \neq j \neq k} \Cov(S_i^{\T} \Lambda_* S_j, S_j^{\T} \Lambda_* S_k) \\ & \lesssim n^2 \Var(S_1^{\T} \Lambda_* S_2) + n^3 \Var(S_1^{\T} \Lambda_* S_2) \\ & \lesssim n^3 \EE (S_1^{\T} \Lambda_* S_2)^2 \\ & = n^3  (\EE S^2_{11})^2 \sum_{j = 1}^d \lambda^2_j + n^3  (\EE S_{11} S_{12})^2 \sum_{j \neq k} \lambda_j \lambda_k    \\ & = n^3 \tr(\Sigma^2_*), 
\end{align*}
by Cauchy-Schwartz inequality, the independence of $(S_i, S_j)$ and $(S_k, S_l)$ for $i \neq j \neq k \neq l$, and the isotropy of $S_1$. Combining the preceding results, we have 
\begin{align*} 
         \Var(\tr(\wh \Sigma)) & \lesssim  \frac{\left( \EE[\epsilon^2] \right)^2}{n}\Var(\| \Lambda^{\frac{1}{2}}_* S_1 \|_2^2)  + {\left( \EE[\epsilon^2] \right)^2 \over n^4} \Var\Bigl(\sum_{i \neq j } S_i^{\T} \Lambda_* S_j \Bigr) \\ &  \lesssim   \frac{\left( \EE[\epsilon^2] \right)^2 \tr^2(\Sigma_*) }{n} + \frac{\left( \EE[\epsilon^2] \right)^2 \tr^2(\Sigma_*) }{n} \\ & \lesssim \frac{\tr^2(\Sigma) }{n}, 
\end{align*}
which establishes \eqref{Elliptical Trace} and completes the proof of the proposition for this case.
\end{proof}

\subsubsection{Proof for \cref{kurtosis_stochastic_rep}}\label{app_sec_ratio_lep} 
    Under the \cref{kurtosis_stochastic_rep} alternatives,  
    we adopt the same approach used to prove \cref{prop_Delta_Null}, as found in \cref{app_proof_prop_Delta_Null}. 
    It suffices to show that  
    \begin{align*}
        &\sqrt{\tr(\wh \Sigma) \over \tr(\Sigma)} = 1 +  \cO_\PP\left(\sqrt{ \frac{\tr(\Sigma^2)}{\tr^2(\Sigma)}}\frac{1}{\sqrt n} \right),\\
        &\sqrt{\frac{\tr(\Sigma^2)}{\widehat{\tr(\Sigma^2)}}} = 1 +\cO_\PP\left(
	 {\sqrt{\tr(\Sigma^4)  \over \tr^2(\Sigma^2)}} {1\over \sqrt n} +  {1\over n}
	 \right)
    \end{align*} 
    to establish the desired result
    \begin{align*}
        \sqrt{\Delta \over \wh \Delta} &= 1 + \cO_\PP\left( {1\over \sqrt{n \rho_2(\Sigma^2)}}  +  {1\over n}\right).
    \end{align*}  
    Due to the stochastic representation imposed by \cref{kurtosis_stochastic_rep} and the invariance of both the empirical and population trace quantities under orthogonal transformation of the samples, we can without loss of generality assume that $\Sigma = \Lambda$ in what follows. Moreover, we let $$\kappa_n := 3 + \delta_n$$ denote the marginal kurtosis parameter in the context of \cref{kurtosis_stochastic_rep}.
    
    In bounding the relative error for $\tr(\wh \Sigma) / \tr(\Sigma)$, by using the facts (see \cite{HimenoYamada}) 
	that  $\EE[\tr(\wh \Sigma)] = \tr(\Sigma)$ and
   \begin{align*}
       \Var(\tr(\wh \Sigma)) & = \EE\tr^2(\wh \Sigma) - [\EE\tr(\wh \Sigma)]^2 \\ & = \frac{1}{n} \left(\EE \| \Lambda^{\frac{1}{2}} Z \|_2^4 - 2 \tr(\Sigma^2) - \tr^2(\Sigma) \right) + \frac{2}{n - 1} \tr(\Sigma^2) +  \tr^2(\Sigma) -  \tr^2(\Sigma) \\ & = \frac{1}{n} \left( (\kappa_n - 1) \tr(\Sigma^2) + \tr^2(\Sigma) - 2 \tr(\Sigma^2) - \tr^2(\Sigma) \right) + \frac{2}{n - 1} \tr(\Sigma^2) \\ & =  \frac{(n - 1)\kappa_n - n + 3}{n(n - 1)} \tr(\Sigma^2),
   \end{align*}
    Chebyshev's inequality, the fact that $\kappa_n := 3 + \delta_n$ is uniformly bounded under \cref{kurtosis_stochastic_rep}, and \cref{lem_ranks} entail 
     \begin{align*}
         {\tr(\wh \Sigma) \over \tr(\Sigma)} & = 1 + \cO_\PP\left(\sqrt{\tr(\Sigma^2) \over \tr^2(\Sigma)}{1\over \sqrt n}\right)  = 1 + \cO_\PP\left({1\over \sqrt n}\right). 
     \end{align*}
   Taking the Taylor expansion of $f(x) = \sqrt{x}$ at $\tr(\wh \Sigma)/ \tr(\Sigma)$ about 1 yields the first result. 
 
 To control $\widehat{\tr(\Sigma^2)}/\tr(\Sigma^2)$, we first note that $\EE[\widehat{\tr(\Sigma^2)}] = \tr(\Sigma^2)$ and 
    \begin{align*}
        \Var\left(\wh{\tr(\Sigma^2)}\right) & = \cO\left(\frac{\tr(\Sigma^4)}{n} + \frac{\tr(\Lambda^2 \odot \Lambda^2)}{n} + \frac{\tr^2(\Sigma^2)}{n^2} \right) \\ & = \cO\left(\frac{\tr(\Sigma^4)}{n} + \frac{\tr(\Sigma^4)}{n} + \frac{\tr^2(\Sigma^2)}{n^2} \right),
    \end{align*}
    due to Proposition A.2 of \cite{Chen2010}, where $\odot$ denotes the element-wise Hadamard product. Chebyshev's inequality and \cref{lem_ranks} entail that
    \begin{align*}
      \frac{\wh{\tr(\Sigma^2)}}{\tr(\Sigma^2)} & = 1 +\cO_\PP\left(
	 {\sqrt{\tr(\Sigma^4)} \over \tr(\Sigma^2)} {1\over \sqrt n} + \frac{1}{\sqrt n} \right)   = 1 + \cO_\PP\left({1\over \sqrt n}\right).
    \end{align*}
Taking a Taylor expansion of the function $f(x) =\sqrt{1/x}$ at  $\widehat{\tr(\Sigma^2)}/\tr(\Sigma^2)$ about $1$ yields the desired result. 
\end{proof}

\end{document}